\documentclass[letterpaper,11pt]{article}
\usepackage{fullpage}
\usepackage[ansinew]{inputenc}
\usepackage{verbatim}
\usepackage{graphicx}
\usepackage{tabularx}
\usepackage{array}
\usepackage{delarray}
\usepackage{dcolumn}
\usepackage{float}
\usepackage{amsmath}
\usepackage{psfrag}
\usepackage{booktabs}
\usepackage{multirow,hhline}
\usepackage{amstext}
\usepackage{amssymb}
\usepackage{fancyhdr}
\usepackage[onehalfspacing]{setspace}
\usepackage{amsthm}
\usepackage{makeidx}
\usepackage{url}
\usepackage{adjustbox}
\usepackage{enumitem}
\usepackage{ragged2e}
\usepackage{booktabs}
\usepackage{titling}
\usepackage[english]{babel}
\usepackage{threeparttable}
\usepackage{pdflscape}
\usepackage[round,authoryear]{natbib}
\usepackage{authblk}
\usepackage{hyperref}
\usepackage[svgnames]{xcolor}
\hypersetup{
	colorlinks=true,
	linkcolor=DarkRed,
	citecolor=NavyBlue}

\makeindex
\allowdisplaybreaks
\numberwithin{equation}{section}
\newtheorem{lem}{Lemma}[section]
\newtheorem{thm}{Theorem}[section]

\newtheorem{ass}{Assumption}

\newtheorem{ex}{Example}

\newcommand{\diag}{\text{diag}}
\newcommand{\Supp}{\text{Supp}}
\newcommand{\indep}{\perp\!\!\!\perp}

\newcommand{\convP}{\stackrel{p}{\longrightarrow}}
\newcommand{\convD}{\rightsquigarrow}

\newcommand{\eps}{\varepsilon}

\renewcommand{\epsilon}{\varepsilon}
\DeclareMathOperator*{\argmax}{arg\,max}
\DeclareMathOperator*{\argmin}{arg\,min}
\makeatletter
\newcommand*{\rom}
[1]{\expandafter\@slowromancap\romannumeral #1@}
\makeatother
\usepackage{xr}
\begin{document}
	\setlength{\droptitle}{-1.5cm}
	
	\title{Regression-Adjusted Estimation of Quantile Treatment Effects under Covariate-Adaptive Randomizations}
	
	\author[a]{Liang Jiang} 
	\author[b,c,d,e]{Peter C.~B.~Phillips}
	\author[f]{Yubo Tao\thanks{Corresponding author. Email: yubotao@um.edu.mo.}}
	\author[b]{Yichong Zhang}
	
	\affil[a]{\footnotesize Fanhai International School of Finance, Fudan University. 220 Handan Rd, Shanghai, China 200437.}
	\affil[b]{\footnotesize School of Economics, Singapore Management University. 90 Stamford Rd, Singapore 178903.}
	\affil[c]{\footnotesize Yale University, New Haven, CT 06520-8281, United States of America.}
	\affil[d]{\footnotesize University of Auckland, 12 Grafton Rd, Auckland Central, Auckland, 1010, New Zealand. }
	\affil[e]{\footnotesize University of Southampton, University Rd, Southampton SO17 1BJ, United Kingdom. }
	\affil[f]{\footnotesize Department of Economics, University of Macau. Avenida da Universidade, Taipa, Macao SAR, China.}

	\date{\today}
	
	\maketitle
	
	\def \baselinestretch{1.05}
	\begin{abstract}
		
		Datasets from field experiments with covariate-adaptive randomizations (CARs) usually contain extra covariates in addition to the strata indicators. We propose to incorporate these additional covariates via auxiliary regressions in the estimation and inference of unconditional quantile treatment effects (QTEs) under CARs. We establish the consistency and limit distribution of the regression-adjusted QTE estimator and prove that the use of multiplier bootstrap inference is non-conservative under CARs. The auxiliary regression may be estimated parametrically, nonparametrically, or via regularization when the data are high-dimensional. Even when the auxiliary regression is misspecified, the proposed bootstrap inferential procedure still achieves the nominal rejection probability in the limit under the null. When the auxiliary regression is correctly specified, the regression-adjusted estimator achieves the minimum asymptotic variance. We also discuss forms of adjustments that can improve the efficiency of the QTE estimators. The finite sample performance of the new estimation and inferential methods is studied in simulations, and an empirical application to a well-known dataset concerned with expanding access to basic bank accounts on savings is reported.
		
		% 	Our estimation and inferential methods can be implemented without tuning parameters and they allow for common choices of auxiliary regressions such as linear, probit and logit regressions despite the fact that these regressions may be misspecified.  
		\bigskip
		
		
	\end{abstract}
	\clearpage
	
	\section{Introduction} \label{sec:intro}
	
	Covariate-adaptive randomizations (CARs) have recently seen growing use in a wide variety of randomized experiments in economic research. Examples include \cite{CCFNT16}, \cite{greaney2016}, \cite{jakiela2016}, \cite{burchardi2019}, \cite{anderson2021}, among many others. In CAR modeling, units are first stratified using some baseline covariates, and then, within each stratum, the treatment status is assigned (independent of covariates) to achieve the balance between the numbers of treated and control units. 
	
	In many empirical studies, apart from the average treatment effect (ATE), researchers are often interested in using randomized experiments to estimate quantile treatment effects (QTEs). The QTE has a useful role as a robustness check for the ATE and characterizes any heterogeneity that may be present in the sign and magnitude of the treatment effects according to their position within the distribution of outcomes. See, for example, \cite{B06}, \cite{MS11}, \cite{DGPR13}, \cite{BDGK15}, \cite{CDDP15}, and \cite{CF17}.
	
	Two practical issues arise in estimation and inference concerning QTEs under CARs. First, other covariates in addition to the strata indicators are collected during the experiment. It is possible to incorporate these covariates in the estimation of treatment effects to reduce variance and improve efficiency. In the estimation of ATE, the usual practice is to run a simple  ordinary least squares (OLS) regression of the outcome on treatment status, strata indicators, and additional covariates as in the analysis of covariance (ANCOVA). \cite{F08b, F081} pointed out that such an OLS regression adjustment can degrade the precision of the ATE estimator. \cite{L13} reexamined Freedman's critique and showed that, in order to improve efficiency, the linear regression adjustment should include a full set of interactions between the treatment status and covariates. However, because the quantile function is a nonlinear operator, even when the assignment of treatment status is completely random, a similar linear quantile regression with a full set of interaction terms is unable to provide a consistent estimate of the \textit{unconditional} QTE, not to mention the improvement of estimation efficiency. Second, in order to achieve balance in the respective number of treated and control units within each stratum, treatment statuses under CARs usually exhibit a (negative) cross-sectional dependence. Standard inference procedures that rely on cross-sectional independence are therefore conservative and lack power. These two issues raise questions of how to use the additional covariates to consistently and more efficiently estimate QTE in CAR settings and how to conduct valid statistical procedures that mitigate conservatism in inference.\footnote{For example, \cite{BCS17} and \cite{ZZ20} have shown that the usual two-sample $t$-test for inference concerning ATE and multiplier bootstrap inference concerning QTE are in general conservative under CARs.}
	
	The present paper addresses these issues by proposing a regression-adjusted estimator of the QTE, deriving its limit theory, and establishing the validity of multiplier bootstrap inference under CARs. Even under potential misspecification of the auxiliary regressions, the proposed QTE estimator is shown to maintain its consistency, and the multiplier bootstrap procedure is shown to have an asymptotic size equal to the nominal level under the null. When the auxiliary regression is correctly specified, the QTE estimator achieves minimum asymptotic variance. 
	
	We further investigate efficiency gains that materialize from the regression adjustments in three scenarios: (1) parametric regressions, (2) nonparametric regressions, and (3) regressions with regularization in high-dimensional settings. Specifically, for parametric regressions with a potentially misspecified linear probability model, we propose to compute the optimal linear coefficient by minimizing the variance of the QTE estimator. Such an adjustment is optimal within the class of linear adjustments but does not necessarily achieve the global minimum asymptotic variance. However, as no adjustment is a special case of the linear regression with all the coefficients being zero, our optimal linear adjustment is guaranteed to be weakly more efficient than the QTE estimator with no adjustments, which addresses Freedman's critique. We also consider a potentially misspecified logistic regression with fixed-dimensional regressors and strata- and quantile-specific regression coefficients, which is then estimated by the quasi maximum likelihood estimation (QMLE). Although the QMLE does not necessarily minimize the asymptotic variance of the QTE, such a flexible logistic model can closely approximate the true specification. Therefore, in practice, the corresponding regression-adjusted QTE estimator usually has a smaller variance than that with no adjustments. Last, we propose to treat the logistic QMLE adjustments as new linear regressors and re-construct the corresponding optimal linear adjustments. We then show the QTE estimator with the new adjustments are weakly more efficient than that with both the original logistic QMLE adjustments and no adjustments.  
	
	In nonparametric regressions, we further justify the QMLE by letting the regressors in the logistic regression be a set of sieve basis functions with increasing dimension and show how such a nonparametric regression-adjusted QTE estimator can achieve the global minimum asymptotic variance. For high-dimensional regressions with regularization, we consider logistic regression under $\ell_1$ penalization, an approach that also achieves the global minimum asymptotic variance. All the limit theories hold uniformly over a compact set of quantile indices, implying that our multiplier bootstrap procedure can be used to conduct inference on QTEs involving single, multiple, or a continuum of quantile indices.
	
	These results, including the limit distribution of the regression-adjusted QTE estimator and the validity of the multiplier bootstrap, provide novel contributions to the literature in three respects. First, the data generated under CARs are different from observational data as the observed outcomes and treatment statuses are cross-sectionally dependent due to the randomization schemes. Recently \cite{BCS17} established a rigorous asymptotic framework to study the ATE estimator under CARs and pointed out the conservatism of the two-sample t-test except for some special cases. (See \citet[Remark 4.2]{BCS17} for more detail.) Our analysis follows this new framework, which departs from the literature of causal inference under an i.i.d. treatment structure.  
	
	Second, we contribute to the literature on causal inference under CARs by developing a new methodology that includes additional covariates in the estimation of the \textit{unconditional} QTE and by establishing a general theory for regression adjustments that allow for parametric, nonparametric, and regularized estimations of the auxiliary regressions. As mentioned earlier, unlike ATE estimation, the naive linear quantile regression with additional covariates cannot even produce a consistent estimator of the QTE. Instead, we propose a new way to incorporate additional covariates based on the Neyman orthogonal moment and investigate the asymptotic properties and the efficiency gains of the proposed regression-adjusted estimator under CARs. This new machinery allows us to study the QTE regression, which is nonparametrically specified, with both linear (linear probability model) and nonlinear (logit and probit models) regression adjustments. To clarify this contribution to the literature we note that \cite{HH12,MHZ15,MQLH18,O21,SY13,ZZ20,Y18,YS20} considered inference of various causal parameters under CARs but without taking into account additional covariates. \cite{BCS17}, \cite{BCS18}, and \cite{BG21} considered saturated regressions for ATE and local ATE, which can be viewed as regression-adjustments where strata indicators are interacted with the treatment
	or instrument. \cite{SYZ10} showed that if a test statistic is constructed based on the correctly specified model between outcome and additional covariates and the covariates used for CAR are functions of additional covariates, then the test statistic is valid conditional on additional covariates. \cite{BL16,F18,L13,L16,LD20,LiD20,LTM20,LY20,NW20,YYS20,ZD20} studied various estimation methods based on regression adjustments, but these studies all focused on ATE estimation. Specifically, \cite{LTM20} considered \textit{linear} adjustments for \textit{ATE} under CARs in which the covariates can be high-dimensional and the adjustments can be estimated by Lasso. \cite{AHL18} considered regression adjustment using additional covariates for ATE and Local ATE. We differ from them by considering \textit{QTE} with \textit{nonlinear} adjustments such as logistic Lasso. 
	
	Third, we establish the validity of the multiplier bootstrap inference for the regression-adjusted QTE estimator under CARs. To the best of our knowledge, \cite{SYZ10} and \cite{ZZ20} are the only works in the literature that studied bootstrap inference under CARs. \cite{SYZ10} considered the covariate-adaptive bootstrap for the linear regression model. \cite{ZZ20} proposed to bootstrap inverse propensity score weighted (IPW) QTE estimator with the estimated target fraction of treatment even when the truth is known. They showed that the asymptotic variance of the IPW estimator is the same under various CARs. Thus, even though the bootstrap sample ignores the cross-sectional dependence and behaves as if the randomization scheme is simple, the asymptotic variance of the bootstrap analogue is still the same. We complement this research by studying the validity of multiplier bootstrap inference for our \textit{regression-adjusted QTE} estimator. We establish analytically that the multiplier bootstrap with the estimated fraction of treatment is not conservative in the sense that it can achieve an asymptotic size equal to the nominal level under the null even when the auxiliary regressions are misspecified. 
	
	The present paper also comes under the umbrella of a growing literature that has addressed estimation and inference in randomized experiments. In this connection, we mention the studies of \cite{HHK11, athey2017, abadie2018, T18, BRS19, B19, JL20} among many others. \cite{B19} showed an `optimal' matched-pair design can minimize the mean-squared error of the \textit{difference-in-means} estimator for ATE, conditional on covariates. \cite{T18} designed an  adaptive randomization procedure which can minimize the variance of the \textit{weighted} estimator for ATE. Both works rely on a pilot experiment to design the optimal randomization. In contrast, we take the randomization scheme (i.e., CARs) as given and search for new estimators (other than \textit{difference-in-quantile} and \textit{weighted} estimators) for QTE that have smaller variance. In addition, our approach does not require a pilot experiment. Therefore, our and their methods are applied to different scenarios depending on the definition of `optimality' and the data available, and thus, complement to each other.  
	
	From a practical perspective, our estimation and inferential methods have four advantages. First, they allow for common choices of auxiliary regressions such as linear probability, logit, and probit regressions, even though these regressions may be misspecified. Second, the methods can be implemented without tuning parameters. Third, our (bootstrap) estimator can be directly computed via the subgradient condition, and the auxiliary regressions need not be re-estimated in the bootstrap procedure, both of which save considerable computation time. Last, our estimation and inference methods can be implemented without the knowledge of the exact treatment assignment rule used in the experiment. This advantage is especially useful in subsample analysis, where sub-groups are defined using variables other than those to form the strata and the treatment assignment rule for each sub-group becomes unknown. See, for example, the anemic subsample analysis in \cite{CCFNT16} and \cite{ZZ20}. These last three points carry over from \cite{ZZ20} and are logically independent of the regression adjustments. One of our contributions is to show these results still hold for our regression-adjusted estimator.  
	
	The remainder of the paper is organized as follows. Section \ref{sec:setup} describes the model setup and notation. Section \ref{sec:est} develops the asymptotic properties of our regression-adjusted QTE estimator. Section \ref{sec:boot} studies the validity of the multiplier bootstrap inference. Section \ref{sec:aux} considers parametric, nonparametric, and regularized estimation of the auxiliary regressions. Section \ref{sec:sim} reports simulation results, and an empirical application of our methods to the impact of expanding access to basic bank accounts on savings is provided in Section \ref{sec:app}. Section \ref{sec:concl} concludes. Proofs of all results and some additional simulation results are given in the Online Supplement.
	
	\section{Setup and Notation} \label{sec:setup}
	
	Potential outcomes for treated and control groups are denoted by $Y(1)$ and $Y(0)$, respectively. Treatment status is denoted by $A$, with $A=1$ indicating treated and $A=0$ untreated. The stratum indicator is denoted by $S$, based on which the researcher implements the covariate-adaptive randomization. The support of $S$ is denoted by $\mathcal{S}$, a finite set. After randomization, the researcher can observe the data $\{Y_i,S_i,A_i,X_i\}_{i \in [n]}$ where $[n]=\{1,2,...n\}$, $Y_i = Y_i(1)A_i + Y_i(0)(1-A_i)$ is the observed outcome, and $X_i$ contains extra covariates besides $S_i$ in the dataset. The support of $X$ is denoted $\Supp(X)$. In this paper, we allow $X_i$ and $S_i$ to be dependent. For $i \in [n]$, let $p(s) = \mathbb{P}(S_i = s)$, $n(s) = \sum_{i \in [n]}1\{S_i = s\}$, $n_1(s) = \sum_{i \in [n]}A_i1\{S_i=s\}$, and $n_0(s) = n(s) - n_1(s)$.  We make the following assumptions on the data generating process (DGP) and the treatment assignment rule.
	\begin{ass}
		\begin{enumerate}[label=(\roman*)]
			\item $\{Y_i(1),Y_i(0),S_i,X_i\}_{i \in [n]}$ is i.i.d.
			\item $\{Y_i(1),Y_i(0),X_i\}_{i \in [n]} \indep \{A_i\}_{i \in [n]}|\{S_i\}_{i \in [n]}$.
			\item Suppose $p(s)$ is fixed with respect to (w.r.t.) $n$ and is positive for every $s \in \mathcal{S}$.
			\item  Let $\pi(s)$ denote the target fraction of treatment for stratum $s$. Then, $c<\min_{s \in \mathcal{S}}\pi(s) \leq \max_{s \in \mathcal{S}}\pi(s)<1-c$ for some constant $c \in (0,0.5)$ and $\frac{D_n(s)}{n(s)} = o_p(1)$ for $s \in \mathcal{S}$, where $D_n(s) = \sum_{i \in [n]} (A_i-\pi(s))1\{S_i = s\}$.
		\end{enumerate}
		\label{ass:assignment1}
	\end{ass}
	
	Several remarks are in order. First, Assumption \ref{ass:assignment1}(i) allows for cross-sectional dependence among treatment statuses ($\{A_i\}_{i \in [n]}$), thereby accommodating many covariate-adaptive randomization schemes as discussed below. Second, although treatment statuses are cross-sectionally dependent, they are independent of the potential outcomes and additional covariates conditional on the stratum indicator $S$. Therefore, data are still experimental rather than observational. Third, Assumption \ref{ass:assignment1}(iii) requires the size of each stratum to be proportional to the sample size. Fourth, we can view $\pi(s)$ as the target fraction of treated units in stratum $s$. Similar to \cite{BCS18}, we allow the target fractions to differ across strata. Just as for the overlapping support condition in an observational study, the target fractions are assumed to be bounded away from zero and one. In randomized experiments, this condition usually holds because investigators can determine $\pi(s)$ in the design stage. In fact, in most CARs, $\pi(s)  = 0.5$ for $s \in \mathcal{S}$. Fifth, $D_n(s)$ represents the degree of imbalance between the real and target factions of treated units in the $s$th stratum. \cite{BCS17} show that Assumption \ref{ass:assignment1}(iv) holds under several covariate-adaptive treatment assignment rules such as simple random sampling (SRS), biased-coin design (BCD), adaptive biased-coin design (WEI), and stratified block randomization (SBR). For completeness, we briefly repeat their descriptions below. Note we only require $D_n(s)/n(s) = o_p(1)$, which is weaker than the assumption imposed by \cite{BCS17} but the same as that imposed by \cite{BCS18} and \cite{ZZ20}.
	
	\begin{ex}[SRS]
		\label{ex:srs}
		Let $\{A_i\}_{i \in [n]}$ be drawn independently across $i$ and of $\{S_i\}_{i \in [n]}$ as Bernoulli random variables with success rate $\pi$, i.e., for $k=1,\cdots,n$,
		\begin{align*}
		\mathbb{P}\left(A_k = 1\big|\{S_i\}_{i \in [n]}, \{A_{j}\}_{j \in [k-1]}\right) = \mathbb{P}(A_k = 1) = \pi(S_i).
		\end{align*}
	\end{ex}
	
	\begin{ex}[WEI]
		\label{ex:wei}
		This design was first proposed by \cite{W78}. Let $n_{k-1}(S_k) = \sum_{i \in [k-1]}1\{S_i = S_k\}$, $D_{k-1}(s) = \sum_{i \in [k-1]}\left(A_i - \frac{1}{2} \right) 1\{S_i = s\}$, and
		\begin{align*}
		\mathbb{P}\left(A_k = 1\big| \{S_i\}_{i \in [k]},\{A_i\}_{i \in [k-1]}\right) = \phi\biggl(\frac{2D_{k-1}(S_k)}{n_{k-1}(S_k)}\biggr),
		\end{align*}
		where $\phi(\cdot):[-1,1] \mapsto [0,1]$ is a pre-specified non-increasing function satisfying $\phi(-x) = 1- \phi(x)$ and $\frac{D_0(S_1)}{0}$ is understood to be zero. 
	\end{ex}
	
	\begin{ex}[BCD]
		\label{ex:bcd}
		The treatment status is determined sequentially for $1 \leq k \leq n$ as
		\begin{align*}
		\mathbb{P}\left(A_k = 1| \{S_i\}_{i \in [k]},\{A_i\}_{i \in [k-1]}\right) = \begin{cases}
		\frac{1}{2} & \text{if }D_{k-1}(S_k) = 0 \\
		\lambda & \text{if }D_{k-1}(S_k) < 0 \\
		1-\lambda & \text{if }D_{k-1}(S_k) > 0,
		\end{cases}
		\end{align*}
		where $D_{k-1}(s)$ is defined as above and $\frac{1}{2}< \lambda \leq 1$. 
	\end{ex}
	
	\begin{ex}[SBR]
		\label{ex:sbr}
		For each stratum, $\lfloor \pi(s) n(s) \rfloor$ units are assigned to treatment and the rest are assigned to control. 
	\end{ex}
	
	Denote the $\tau$th quantile of $Y(a)$ by $q_a(\tau)$ for $a=0,1$. We are interested in estimating and inferring the $\tau$th quantile treatment effect defined as $q(\tau) = q_1(\tau) - q_0(\tau)$. The testing problems of interest involve single, multiple, or even a continuum of quantile indices, as in the following null hypotheses
	\begin{align*}
	&\mathcal{H}_0: q(\tau) = \underline{q} \quad \text{versus} \quad q(\tau) \neq \underline{q}, \\
	&\mathcal{H}_0: q(\tau_1) - q(\tau_2) = \underline{q} \quad \text{versus} \quad q(\tau_1) - q(\tau_2) \neq \underline{q}, \;\textrm{and} \\
	&\mathcal{H}_0: q(\tau) = \underline{q}(\tau)~\forall \tau \in \Upsilon \quad \text{versus} \quad q(\tau) \neq \underline{q}(\tau)~\text{for some}~\tau \in \Upsilon,
	\end{align*}
	for some pre-specified value $\underline{q}$ or function $\underline{q}(\tau)$, where $\Upsilon$ is some compact subset of $(0,1)$. We can also test constant QTE by letting $\underline{q}(\tau)$ in the last hypothesis be a constant $\underline{q}$. 
	
	\section{Estimation} 
	\label{sec:est}
	Define $m_a(\tau,s,x) = \tau - \mathbb{P}(Y_i(a)\leq q_a(\tau)|S_i=s,X_i=x)$ for $a=0,1$ which are the true specifications but unknown to researchers. Instead, researchers specify working models $\{\overline{m}_a(\tau,s,x)\}_{a=0,1}$\footnote{We view $\overline{m}_a(\cdot)$ as some function with inputs $\tau,s,x$. For example, researchers can specify a linear probability model with $\overline{m}_a(\tau,s,x) = \tau - x^\top \beta_{a,s}$, where $\beta_{a,s}$ is the linear coefficient that varies across treatment status $a$ and stratum $s$.} for the true specification, which can be misspecified. Last, the researchers estimate the (potentially misspecified) working models via some forms of regression, and the estimators are denoted as $\{\widehat{m}_a(\tau,s,x)\}_{a=0,1}$. We also refer to $\overline{m}_a(\cdot)$ as the auxiliary regression. 
	
	Our regression-adjusted estimator of $q_1(\tau)$, denoted as $\hat{q}_1^{adj}(\tau)$, can be defined as
	\begin{align} \label{eq:q1adj}
	\hat{q}_1^{adj}(\tau) = & \argmin_q \sum_{i \in [n]}\left[\frac{A_i}{\hat{\pi}(S_i)}\rho_\tau(Y_i - q) + \frac{(A_i - \hat{\pi}(S_i))}{\hat{\pi}(S_i)}\widehat{m}_1(\tau,S_i,X_i)q\right], 
	\end{align}
	where $\rho_\tau(u) = u(\tau-1\{u \leq 0\})$ is the usual check function and $\hat{\pi}(s) = n_1(s)/n(s)$. We emphasize that $\widehat{m}_1(\cdot)$ may not consistently estimate the true specification $m_1(\cdot)$. Similarly, we can define
	\begin{align} \label{eq:q0adj}
	\hat{q}_0^{adj}(\tau) = & \argmin_q \sum_{i \in [n]}\left[\frac{1-A_i}{1-\hat{\pi}(S_i)}\rho_\tau(Y_i - q) - \frac{(A_i - \hat{\pi}(S_i))}{1-\hat{\pi}(S_i)}\widehat{m}_0(\tau,S_i,X_i)q\right].
	\end{align}
	Then, our regression adjusted QTE estimator is
	\begin{align} \label{eq:qadj}
	\hat{q}^{adj}(\tau) = \hat{q}_1^{adj}(\tau)-\hat{q}_0^{adj}(\tau).
	\end{align}
	
	Several remarks are in order. First, in observational studies with i.i.d. data and $A_i \indep X_i|S_i$, \cite{F07}, \cite{BCFH13}, and \cite{KMU19} showed that the doubly robust moment for $q_1(\tau)$ is
	\begin{align}
	\mathbb{E}\left[\frac{A_i(\tau-1\{Y_i(1)\leq q\})}{\overline{\pi}(S_i)} - \frac{A_i - \overline{\pi}(S_i)}{\overline{\pi}(S_i)}\overline{m}_1(\tau,S_i,X_i)\right] = 0,
	\label{eq:db1}
	\end{align}
	where $\overline{\pi}(s)$ and $\overline{m}_1(\tau,s,x)$ are the working models for the target fraction ($\pi(s)$) and conditional probability ($m_1(\tau,s,x)$), respectively. Our estimator is motivated by this doubly robust moment, but our analysis differs from that for the observational data as CARs introduces cross-sectional dependence among observations. Second, as our target fraction estimator $\hat{ \pi}(s) = n_1(s)/n(s)$ is consistent, it means $\overline{\pi}(s)$ is correctly specified as $\pi(s)$. Then, due to the double robustness, our regression adjusted estimator is consistent even when $\overline{m}_a(\cdot)$ is misspecified and $\widehat{m}_a(\cdot)$ is an inconsistent estimator of $m_a(\cdot)$. Third, we use the estimated target fraction $\hat{ \pi}(s)$ even when $\pi(s)$ is known because this guarantees that the bootstrap inference is not conservative. Further discussion is provided after Theorem \ref{thm:boot}. 
	
	\begin{ass}
		For $a=0,1$, denote $f_a(\cdot)$, $f_a(\cdot|s)$, and $f_a(\cdot|x,s)$ as the PDFs of $Y_i(a)$, $Y_i(a)|S_i=s$, and $Y_i(a)|S_i=s,X_i=x$, respectively.
		\begin{enumerate}[label=(\roman*)]
			\item $f_a(q_a(\tau))$ and $f_a(q_a(\tau)|s)$ are bounded and bounded away from zero uniformly over $\tau \in \Upsilon$ and $s \in \mathcal{S}$, where $\Upsilon$ is a compact subset of $(0,1)$.
			\item $f_a(\cdot)$ and $f_a(\cdot|s)$ are Lipschitz over $\{q_j(\tau):\tau \in \Upsilon\}.$
			\item $\sup_{y \in \Re, x \in \Supp(X), s \in \mathcal{S}}f_a(y|x,s)<\infty$.
		\end{enumerate}
		\label{ass:tau}
	\end{ass}
	
	\begin{ass} 
		\begin{enumerate}[label=(\roman*)]
			\item For $a = 0,1$, there exists a function $\overline{m}_a(\tau,s,x)$ such that for  $\overline{\Delta}_a(\tau,s,X_i) = \widehat{m}_a(\tau,s,X_i) - \overline{m}_a(\tau,s,X_i)$, we have
			\begin{align*}
			\sup_{\tau \in \Upsilon,s\in \mathcal{S}}\biggl|\frac{\sum_{i\in I_1(s)}\overline{\Delta}_a(\tau,s,X_i)}{n_1(s)} - \frac{\sum_{i \in I_0(s)}\overline{\Delta}_a(\tau,s,X_i)}{n_0(s)}\biggr| = o_p(n^{-1/2}),
			\end{align*}
			where $I_a(s) = \{i\in [n]: A_i = a, S_i=s\}$.
			\item For $a = 0,1$, let $\mathcal{F}_a = \{\overline{m}_a(\tau,s,x): \tau \in \Upsilon \}$ with an envelope $F_{a}(s,x)$. Then, $\max_{s \in \mathcal{S}}\mathbb{E}(|F_{a}(S_i,X_i)|^q|S_i=s)<\infty$ for $q > 2$ and there exist fixed constants $(\alpha,v)>0$ such that
			\begin{align*}
			\sup_Q N(\mathcal{F}_a,e_Q,\eps||F_a||_{Q,2}) \leq \left(\frac{\alpha}{\eps}\right)^v, \quad \forall \eps \in (0,1],
			\end{align*}
			where $N(\cdot)$ denotes the covering number, $e_Q(f,g) = ||f-g||_{Q,2}$, and the supremum is taken over all finitely discrete probability measures $Q$.
			\item For $a=0,1$ and any $\tau_1,\tau_2 \in \Upsilon$, there exists a constant $C>0$ such that $$\mathbb{E}((\overline{m}_a(\tau_2,S_i,X_i) - \overline{m}_a(\tau_1,S_i,X_i))^2|S_i=s) \leq C|\tau_2-\tau_1|.$$
		\end{enumerate}
		\label{ass:mhat}
	\end{ass}
	
	Several remarks are in order. First, Assumption \ref{ass:tau} is standard in the quantile regression literature. We do not need $f_a(y|x,s)$ to be bounded away from zero because we are interested in the unconditional quantile $q_a(\tau)$, which is uniquely defined as long as the unconditional density $f_a(q_a(\tau))$ is positive. Second, Assumption \ref{ass:mhat}(i) is high-level. If we consider a linear probability model such that $\overline{m}_a(\tau,s,X_i) = \tau-X_i^\top \theta_{a,s}(\tau)$ and $\widehat{m}_{a}(\tau,s,X_i)= \tau-X_i^\top \hat{ \theta}_{a,s}(\tau)$, then Assumption \ref{ass:mhat}(i) is equivalent to
	\begin{align*}
	\sup_{\tau \in \Upsilon, a =0,1, s\in \mathcal{S}}\left| \left(\frac{\sum_{i\in I_1(s)}X_i}{n_1(s)}  - \frac{\sum_{i\in I_0(s)}X_i}{n_0(s)}\right)^\top \left( \hat{ \theta}_{a,s}(\tau) - \theta_{a,s}(\tau)\right)   \right| = o_p(n^{-1/2}),
	\end{align*}
	which is similar to \citet[Assumption 3]{LTM20} and holds intuitively if $\hat{ \theta}_{a,s}(\tau)$ is a consistent estimator of the pseudo true value $ \theta_{a,s}(\tau)$. Third, Assumptions \ref{ass:mhat}(ii) and \ref{ass:mhat}(iii) impose mild regularity conditions on $\overline{m}_a(\cdot)$. Assumption \ref{ass:mhat}(ii) holds automatically if $\Upsilon$ is a finite set. In general, both Assumption \ref{ass:mhat}(ii) and \ref{ass:mhat}(iii) hold if
	\begin{align*}
	\sup_{a = 0,1, s \in \mathcal{S}, x \in \Supp(X)}|\overline{m}_a(\tau_2,s,x) - \overline{m}_a(\tau_1,s,x)| \leq L|\tau_2 -\tau_1|
	\end{align*}
	for some constant $L>0$. Such Lipschitz continuity holds for the true specification ($\overline{m}_a(\cdot) = m_a(\cdot)$) under Assumption \ref{ass:tau}. Fourth, we provide primitive sufficient conditions for Assumption \ref{ass:mhat} in Section \ref{sec:aux}.
	
	\begin{thm} 
		Suppose Assumptions \ref{ass:assignment1}--\ref{ass:mhat} hold. Then, uniformly over $\tau \in \Upsilon$,
		\begin{align*}
		\sqrt{n}(\hat{q}^{adj}(\tau) - q(\tau)) \convD \mathcal{B}(\tau),
		\end{align*}
		where $\mathcal{B}(\tau)$ is a tight Gaussian process with covariance kernel $\Sigma(\tau,\tau')$ defined in Section \ref{sec:thm_est_pf} of the Online Supplement. In addition, for any finite set of quantile indices $(\tau_1,\cdots,\tau_K)$, the asymptotic covariance matrix of $(\hat{q}^{adj}(\tau_1),\cdots,\hat{q}^{adj}(\tau_K))$ is denoted as $[\Sigma(\tau_k,\tau_l)]_{k,l \in [K]}$, where we use $[U_{kl}]_{k,l\in [K]}$ to denote a $K \times K$ matrix whose $(k,l)$th entry is $U_{kl}$. Then,  $[\Sigma(\tau_k,\tau_l)]_{k,l \in [K]}$ is minimized  in the matrix sense\footnote{For two symmetric matrices $A$ and $B$, we say $A$ is greater than or equal to $B$ if $A-B$ is positive semidefinite.} when the auxiliary regressions are correctly specified at $(\tau_1,\cdots,\tau_K)$, i.e., $\overline{m}_a(\tau_k,s,x) = m_a(\tau_k,s,x)$ for $a = 0,1$, $k \in [K]$, and all $(s,x)$ in the joint support of $(S_i,X_i)$. \label{thm:est}
	\end{thm}
	
	Three remarks are in order. First, the expression for the asymptotic variance of $\hat{q}^{adj}(\tau)$ can be found in the proof of Theorem \ref{thm:est}. It is the same whether the randomization scheme achieves strong balance\footnote{We refer readers to \cite{BCS17} for the definition of strong balance.} or not. This robustness is due to the use of the estimated target fraction ($\hat{\pi}(s)$). The same phenomenon was discovered in the simplified setting by \cite{ZZ20}. Second, although our estimator is still consistent and asymptotically normal when the auxiliary regression is misspecified, it is meaningful to pursue the correct specification as it achieves the minimum variance. As the estimator with no adjustments can be viewed as a special case of our estimator with $\overline{m}_a(\cdot) = 0$, Theorem \ref{thm:est} implies that the adjusted estimator with the correctly specified auxiliary regression is more efficient than that with no adjustments. If the auxiliary regression is misspecified, the adjusted estimator can sometimes be less efficient than the unadjusted one, which is known as the Freedman's critique. In Section \ref{sec:aux}, we discuss how to make adjustments that do not harm the precision of the QTE estimator. Third, the asymptotic variance of $\hat{q}^{adj}(\tau)$ depends on $(f_a(q_a(\tau)),m_a(\tau,s,x))_{a = 0,1}$, which are infinite-dimensional nuisance parameters. To conduct analytic inference, it is necessary to nonparametrically estimate these nuisance parameters, which requires tuning parameters. Nonparametric estimation can be sensitive to the choice of tuning parameters and rule-of-thumb tuning parameter selection may not be appropriate for every DGP or every quantile. The use of cross-validation in selecting the tuning parameters is possible in principle but, in practice, time-consuming. These practical difficulties of analytic methods of inference provide strong motivation to investigate bootstrap inference procedures that are much less reliant on tuning parameters.
	
	\section{Multiplier Bootstrap Inference}
	\label{sec:boot}
	We approximate the  asymptotic distributions of $\hat{q}^{adj}(\tau)$ via the multiplier bootstrap. Let $\{\xi_i\}_{i \in [n]}$ be a sequence of bootstrap weights which will be specified later. Define
	$n_1^w(s) = \sum_{i \in [n]}\xi_iA_i1\{S_i=s\}$, $n_0^w(s) = \sum_{i \in [n]}\xi_i(1-A_i)1\{S_i=s\}$, $n^w(s) = \sum_{i \in [n]} \xi_i 1\{S_i=s\} = n_1^w(s) + n_0^w(s)$, and $\hat{\pi}^w(s) = n_1^w(s)/n^w(s)$. The multiplier bootstrap counterpart of $\hat{q}^{adj}(\tau)$ is denoted by $\hat{q}^{w}(\tau)$ and defined as
	\begin{align*}
	\hat{q}^{w}(\tau) = \hat{q}_1^{w}(\tau) - \hat{q}_0^{w}(\tau),
	\end{align*}
	where
	\begin{align}
	\hat{q}_1^{w}(\tau) = & \argmin_q \sum_{i \in [n]}\xi_i\left[\frac{A_i}{\hat{\pi}^w(S_i)}\rho_\tau(Y_i - q) + \frac{(A_i - \hat{\pi}^w(S_i))}{\hat{\pi}^w(S_i)}\widehat{m}_1(\tau,S_i,X_i)q\right],
	\label{eq:q1w}
	\end{align}
	and
	\begin{align}
	\hat{q}_0^{w}(\tau) = & \argmin_q \sum_{i \in [n]}\xi_i\left[\frac{1-A_i}{1-\hat{\pi}^w(S_i)}\rho_\tau(Y_i - q) - \frac{(A_i - \hat{\pi}^w(S_i))}{1-\hat{\pi}^w(S_i)}\widehat{m}_0(\tau,S_i,X_i)q\right].
	\label{eq:q0w}
	\end{align}
	Two comments on implementation are noted here: (i) we do not re-estimate $\widehat{m}_a(\cdot)$ in the bootstrap sample, which is similar to the multiplier bootstrap procedure proposed by \cite{BCFH13}; and (ii) in Section \ref{sec:compute} of the Online Supplement we propose a way to directly compute $(\hat{q}_a^{w}(\tau))_{a=0,1}$ from the subgradient conditions of \eqref{eq:q1w} and \eqref{eq:q0w}, thereby avoiding the optimization. Both features considerably reduce computation time of our bootstrap procedure.
	
	Next, we specify the bootstrap weights.
	\begin{ass}
		Suppose $\{\xi_i\}_{i \in [n]}$ is a sequence of nonnegative i.i.d. random variables with unit expectation and variance and a sub-exponential upper tail.
		\label{ass:weight}
	\end{ass}
	
	\begin{ass}
		Recall $\overline{\Delta}_a(\tau,s,x)$ defined in Assumption \ref{ass:mhat}. We have, for $a=0,1$,
		\begin{align*}
		\sup_{\tau \in \Upsilon,s\in \mathcal{S}}\biggl|\frac{\sum_{i\in I_1(s)}\xi_i\overline{\Delta}_a(\tau,s,X_i)}{n_1^w(s)} - \frac{\sum_{i \in I_0(s)}\xi_i\overline{\Delta}_a(\tau,s,X_i)}{n_0^w(s)}\biggr| = o_p(n^{-1/2}).
		\end{align*}
		\label{ass:mhatw}
	\end{ass}
	
	We require the bootstrap weights to be nonnegative so that the objective functions in \eqref{eq:q1w} and \eqref{eq:q0w} are convex. In practice, we generate $\xi_i$ independently from the standard exponential distribution. Assumption \ref{ass:mhatw} is the bootstrap counterpart of Assumption \ref{ass:mhat}. Continuing with the linear model example considered after Assumption \ref{ass:mhat}, Assumption \ref{ass:mhatw} requires
	\begin{align*}
	\sup_{\tau \in \Upsilon, a =0,1, s\in \mathcal{S}}\left| \left(\frac{\sum_{i\in I_1(s)}\xi_iX_i}{n_1^w(s)}  - \frac{\sum_{i\in I_0(s)}\xi_iX_i}{n_0^w(s)}\right)^\top \left( \hat{ \theta}_{a,s}(\tau) - \theta_{a,s}(\tau)\right)   \right| = o_p(n^{-1/2}),
	\end{align*}
	which holds if $\hat{ \theta}_{a,s}(\tau)$ is a uniformly consistent estimator of $\theta_{a,s}(\tau)$.
	
	\begin{thm}
		Suppose Assumptions \ref{ass:assignment1}--\ref{ass:mhatw} hold. Then, uniformly over $\tau \in \Upsilon$,
		\begin{align*}
		\sqrt{n}(\hat{q}^{w}(\tau) - \hat{q}^{adj}(\tau)) \underset{\xi}{\overset{\mathbb{P}}{\convD}} \mathcal{B}(\tau),
		\end{align*}
		where $\mathcal{B}(\tau)$ is the same Gaussian process defined in Theorem \ref{thm:est}.\footnote{We view $\sqrt{n}(\hat{q}^w(\tau) - \hat{q}^{adj}(\tau))$ and $\mathcal{B}(\tau)$ as two processes indexed by $\tau \in \Upsilon$ and denote them as $G_n$ and $G$, respectively. Then, following \citet[Chapter 2.9]{VW96}, we say $G_n \underset{\xi}{\overset{\mathbb{P}}{\convD}} G$ uniformly over $\tau \in \Upsilon$ if
			\begin{align*}
			\sup_{h \in \text{BL}_1}|\mathbb{E}_\xi h(G_n) - \mathbb{E}h(G)| \convP 0,
			\end{align*}
			where $\text{BL}_1$ is the set of all functions $h:\ell^\infty(\Upsilon) \mapsto [0,1]$ such that $|h(z_1)-h(z_2)| \leq |z_1-z_2|$ for every $z_1,z_2 \in \ell^\infty(\Upsilon)$, and $\mathbb{E}_\xi$ denotes expectation with respect to the bootstrap weights $\{\xi\}_{i \in [n]}$.}
		\label{thm:boot}
	\end{thm}
	
	Two remarks are in order. First, Theorem \ref{thm:boot} shows the limit distribution of the bootstrap estimator conditional on data can approximate that of the original estimator uniformly over $\tau \in \Upsilon$. This is the theoretical foundation for the bootstrap confidence intervals and bands described in Section \ref{sec:compute} in the Online Supplement. Specifically, denote $\{\hat{q}^{w,b}(\tau)\}_{b \in [B]}$ as the bootstrap estimates where $B$ is the number of bootstrap replications. Let $\widehat{\mathcal{C}}(\nu)$ and $\mathcal{C}(\nu)$ be the $\nu$th empirical quantile of the sequence  $\{\hat{q}^{w,b}(\tau)\}_{b \in [B]}$ and the $\nu$th standard normal critical value, respectively. Then, we suggest using the bootstrap estimator to construct the standard error of $\hat{q}^{adj}(\tau)$ as $\hat{\sigma} = \frac{\widehat{\mathcal{C}}(0.975)- \widehat{\mathcal{C}}(0.025)}{\mathcal{C}(0.975) - \mathcal{C}(0.025)}$. Note that, unlike \cite{HL21}, our bootstrap standard error is not conservative.  In our context, the bootstrap estimator of $\sigma^2$ considered by \cite{HL21} is $\mathbb{E}^*(\sqrt{n}(\hat{q}^w(\tau) - \hat{q}^{adj}(\tau))^2)$, where $\mathbb{E}^*$ is the conditional expectation given data. It is well-known that weak convergence does not imply convergence in $L_2$-norm, which explains why they can show their estimator is in general conservative. Instead, we use a different estimator of the standard error and can show it is consistent given weak convergence. Second, such a bootstrap approximation is consistent under CAR. \cite{ZZ20} showed that for the QTE estimation without regression adjustment, bootstrapping the IPW QTE estimator with the estimated target fraction results in non-conservative inference, while bootstrapping the IPW estimator with the true fraction is conservative under CARs. As the estimator considered by \cite{ZZ20} is a special case of our regression-adjusted estimator with $\widehat{m}_a(\cdot) =0$, we conjecture that the same conclusion holds. A proof of conservative bootstrap inference with the true target fraction is not included in the paper due to the space limit.\footnote{Full statements and proofs are lengthy because we need to derive the limit distributions of not only the bootstrap but also the original estimator with the true target fraction. Although the negative result is theoretically interesting, we are not aware of any empirical papers using the true target fraction while making regression adjustments. Moreover, our method is shown to have better performance than the one with the true target fraction in simulations. So the practical value of proving the negative result is limited.} Our simulations confirm both the correct size coverage of our inference method using the bootstrap with the estimated target fraction and the conservatism of the bootstrap with the true target fraction. The standard error of the QTE estimator is found to be 34.9\% larger on average by using the true rather than the estimated target fraction in the simulations (see Tables \ref{tab:Sim_Point1} below and \ref{tab:Sim_Point_0.5'} in the Online Supplement).

	\section{Auxiliary Regressions} \label{sec:aux}
	
	In this section, we consider two approaches to estimation for the auxiliary regressions: (1) a parametric method and (2) a nonparametric method. In Section \ref{sec:reg} of the Online Supplement, we further consider a regularization method. For the parametric method, we do not require the model to be correctly specified. We propose ways to estimate the pseudo true value of the auxiliary regression. For the other two methods, we (nonparametrically) estimate the true model so that the asymptotic variance of $\hat{q}^{adj}(\tau)$ achieves its minimum based on Theorem \ref{thm:est}. For all three methods, we verify Assumptions \ref{ass:mhat} and \ref{ass:mhatw}.
	
	\subsection{Parametric method}
	\label{sec:aux_par}
	In this section, we consider the case where $X_i$ is finite-dimensional. Recall $m_a(\tau,s,x) \equiv \tau - \mathbb{P}\left(Y_i(a) \leq q_a(\tau)|X_i=x,S_i=s\right)$ for $a=0,1$. We propose to model $\mathbb{P}\left(Y_i(a) \leq q_a(\tau)|X_i,S_i=s\right)$ as $\Lambda_{\tau,s}(X_i,\theta_{a,s}(\tau))$, where $\theta_{a,s}(\tau)$ is a finite dimensional parameter that depends on $(a,s,\tau)$ so that our model for $m_a(\tau,s,X_i)$ is
	\begin{align}
	\overline{m}_a(\tau,s,X_i) = \tau - \Lambda_{\tau,s}(X_i,\theta_{a,s}(\tau)).
	\label{eq:moverline_par}
	\end{align}
	We note that, as we allow for misspecification, the researchers have the freedom to choose any functional forms for $\Lambda_{\tau,s}(\cdot)$ and any pseudo true values for $\theta_{a,s}(\tau)$, both of which can vary with respect to $(\tau,s)$. For example, if we assume a logistic regression with $\Lambda_{\tau,s}(X_i,\theta_{a,s}(\tau)) = \lambda(X_i^\top\theta_{a,s}(\tau))$, where $\lambda(\cdot)$ is the logistic CDF, then there are various choices of $\theta_{a,s}(\tau)$ such as the maximizer of the population pseudo likelihood, the maximizer of the population version of the least squares objective function, or the minimizer of the asymptotic variance of the adjusted QTE estimator. As the logistic model is potentially misspecified, these three pseudo true values are not necessarily the same and can lead to different adjustments, and thus, different asymptotic variances of the corresponding adjusted QTE estimators. 
	
	Next, we first state a general result for generic choices of $\Lambda_{\tau,s}(\cdot)$ and $\theta_{a,s}(\tau)$. Suppose we estimate $\theta_{a,s}(\tau)$ by $\hat{\theta}_{a,s}(\tau)$. Then, the corresponding $\widehat{m}_a(\tau,s,X_i)$ can be written as
	\begin{align}
	\widehat{m}_a(\tau,s,X_i) = \tau - \Lambda_{\tau,s}(X_i,\hat{\theta}_{a,s}(\tau)).
	\label{eq:mhat_par}
	\end{align}
	\begin{ass}
		\begin{enumerate}[label=(\roman*)]
			\item Suppose there exist a positive random variable $L_i$ and a positive constant $C>0$ such that 
			\begin{align*}
			&	\sup_{\tau_1,\tau_2 \in \Upsilon,s\in \mathcal{S},||\theta||\leq C}|| \Lambda_{\tau_1,s}(X_i,\theta)- \Lambda_{\tau_2,s}(X_i,\theta)||_2 \leq L_i|\tau_1-\tau_2|, \quad \sup_{\tau \in \Upsilon, s\in \mathcal{S},||\theta||\leq C}|\Lambda_{\tau,s}(X_i,\theta)| \leq  L_i, \\
			& \sup_{\tau \in \Upsilon, s\in \mathcal{S},||\theta||\leq C}|\partial_\theta \Lambda_{\tau,s}(X_i,\theta)| \leq  L_i, \quad \text{and} \quad \mathbb{E}(L_i^d|S_i=s) \leq C <\infty \quad \text{for some $d > 2$.}
			\end{align*}
			\item $\sup_{\tau_1,\tau_2 \in \Upsilon, a=0,1,s\in \mathcal{S}}|\theta_{a,s}(\tau_1) - \theta_{a,s}(\tau_2)| \leq C|\tau_1-\tau_2|$.
			\item $\sup_{\tau \in \Upsilon, a=0,1,s\in \mathcal{S}}||\hat{\theta}_{a,s}(\tau)-\theta_{a,s}(\tau)||_2 \convP 0$.
		\end{enumerate}
		\label{ass:par}
	\end{ass}
	
	Three remarks are in order. First, common choices for auxiliary regressions are linear probability, logistic, and probit regressions, corresponding to $\Lambda_{\tau,s}(X_i,\theta_{a,s}(\tau)) = X_i^\top\theta_{a,s}(\tau)$,  $\lambda(\vec{X}_i^\top\theta_{a,s}(\tau))$, and $\Phi(\vec{X}_i^\top\theta_{a,s}(\tau))$, respectively, where $\Phi(\cdot)$ is the standard normal CDF and $\vec{X}_i = (1,X_i^\top)^\top$. For these models, the functional form $\Lambda_{\tau,s}(\cdot)$ does not depend on $(\tau,s)$, and Assumption \ref{ass:par}(i) holds automatically. For the linear regression case, we do not include the intercept because our regression adjusted estimators (\eqref{eq:q1adj} and \eqref{eq:q0adj}) and their bootstrap counterparts (\eqref{eq:q1w} and \eqref{eq:q0w}) are numerically invariant to location shift of the auxiliary regressions. Second, it is also important to allow the functional form $\Lambda_{\tau,s}(\cdot)$ to vary across $\tau$ to incorporate the case in which the regressor $X_i$ in the linear, logistic, and probit regressions is replaced by $W_{i,s}(\tau)$, a function of $X_i$ that depends on $(\tau,s)$. We give a concrete example for this situation in Section \ref{sec:lpml}. Third, Assumption \ref{ass:par}(ii) also holds automatically if $\Upsilon$ is finite. When $\Upsilon$ is infinite, this condition is still mild.
	
	\begin{thm}
		Denote $\hat{q}^{par}(\tau)$ and $\hat{q}^{par,w}(\tau)$ as the $\tau$th QTE estimator and its multiplier bootstrap counterpart defined in Sections \ref{sec:est} and \ref{sec:boot}, respectively, with $\overline{m}_a(\tau,S_i,X_i)$ and $\widehat{m}_a(\tau,S_i,X_i)$ defined in \eqref{eq:moverline_par} and \eqref{eq:mhat_par}, respectively. Suppose Assumptions \ref{ass:assignment1}, \ref{ass:tau}, \ref{ass:weight}, and \ref{ass:par} hold. Then, Assumptions \ref{ass:mhat} and \ref{ass:mhatw} hold, which further implies Theorems \ref{thm:est} and \ref{thm:boot} hold for $\hat{q}^{par}(\tau)$ and $\hat{q}^{par,w}(\tau)$, respectively.
		\label{thm:par}
	\end{thm}
	Theorem \ref{thm:par} shows that, as long as the estimator of the pseudo true value ($\hat{\theta}_{a,s}(\tau)$) is uniformly consistent, under mild regularity conditions, all the general estimation and bootstrap inference results established in Sections \ref{sec:est} and \ref{sec:boot} hold. 
	
	\subsubsection{Linear probability model}
	\label{sec:lp}
	In this section, we consider linear adjustment with parameter $t_{a,s}(\tau)$ such that 
	\begin{align}
	\Lambda_{\tau,s}(X_i,t_{a,s}(\tau)) = W_{i,s}^\top(\tau) t_{a,s}(\tau) \quad \text{and} \quad \overline{m}_a(\tau,s,X_i) = \tau - W_{i,s}^\top(\tau) t_{a,s}(\tau),
	\label{eq:adj_linear}
	\end{align}
	where the regressor $W_{i,s}(\tau)$ is a function of $X_i$ but the functional form may vary across $s,\tau$. For example, we can consider $W_{i,s}(\tau) = X_i$, the transformations of $X_i$ such as quadratic and interaction terms, and some prediction of $(1\{Y_i(1) \leq q_1(\tau)\},1\{Y_i(0) \leq q_0(\tau)\})$ given $X_i$ and $S_i = s$. The last example is further explained in Section \ref{sec:lpml}.   
	% As we allow for the adjustment to be misspecified, researchers have the freedom to determine the pseudo true value  $\theta_{a,s}(\tau)$. 
	
	We note that the asymptotic variance (denoted as $\sigma^2$) of the $\hat{q}^{adj}(\tau)$ is a function of the working model ($\overline{m}_{a}(\tau,s,\cdot)$), which is further indexed by its parameters (denoted as $\{t_{a,s}(\tau)\}_{a = 0,1, s \in \mathcal{S}}$), i.e., $\sigma^2 = \sigma^2(\{\overline{m}_{a}(\tau,s,\cdot;t_{a,s})\}_{a = 0 ,1, s \in \mathcal{S}})$. Our optimal linear adjustment corresponds to parameter value $\theta_{a,s}(\tau)$ such that it minimizes $\sigma^2(\{\overline{m}_{a}(\tau,s,\cdot;t_{a,s})\}_{a = 0 ,1, s \in \mathcal{S}})$, i.e.,
	\begin{align*}
	\{\theta_{a,s}(\tau)\}_{a =0,1, s \in \mathcal{S}} \in \argmin_{t_{a,s}: a =0,1, s \in \mathcal{S}} \sigma^2(\{\overline{m}_{a}(\tau,s,\cdot;t_{a,s})\}_{a = 0 ,1, s \in \mathcal{S}}).
	\end{align*}
	
	\begin{ass}
		Define $\tilde{W}_{i,s}(\tau) = W_{i,s}(\tau) - \mathbb{E}(W_{i,s}(\tau)|S_i=s)$. 	There exist constants $0<c< C<\infty$ such that
		\begin{align*}
		c< \inf_{a = 0,1, s \in \mathcal{S},\tau \in \Upsilon}\lambda_{\min}(\mathbb{E}\tilde{W}_{i,s}(\tau)\tilde{W}_{i,s}(\tau)^\top|S_i=s) \leq \sup_{a =0,1, s\in \mathcal{S},\tau \in \Upsilon}\lambda_{\max}(\mathbb{E}\tilde{W}_{i,s}(\tau)\tilde{W}_{i,s}(\tau)^\top|S_i=s) \leq C
		\end{align*}
		and $\mathbb{E}(||\tilde{W}_{i,s}||_2^d|S_i=s) \leq C$ for some $d>2$, where for a generic symmetric matrix $U$, $\lambda_{\min}(U)$ and $\lambda_{\max}(U)$ denote the minimal and maximal eigenvalues of $U$, respectively.
		\label{ass:regX}
	\end{ass}

	The next theorem derives the closed-form expression for the optimal linear coefficient. 
	\begin{thm}
		Suppose Assumptions \ref{ass:assignment1}, \ref{ass:tau}, \ref{ass:weight}, \ref{ass:par}, \ref{ass:regX} hold, and $\Lambda_{\tau,s}(\cdot)$ is defined in \eqref{eq:adj_linear}. Further denote the asymptotic covariance matrix of $(\hat{q}^{par}(\tau_1),\cdots,\hat{q}^{par}(\tau_K))$ for any finite set of quantile indices $(\tau_1,\cdots,\tau_K)$ as $[\Sigma^\textit{LP}(\tau_k,\tau_l)]_{k,l \in [K]}$. Then, $[\Sigma^\textit{LP}(\tau_k,\tau_l)]_{k,l \in [K]}$ is minimized in the matrix sense at $\left(\theta_{1,s}(\tau_k),\theta_{0,s}(\tau_k)\right)_{k \in [K]}$ such that 
		\begin{align*}
		\frac{\theta_{1,s}(\tau_k)}{f_1(q_1(\tau_k))} + \frac{\pi(s) \theta_{0,s}(\tau_k)}{(1-\pi(s))f_0(q_0(\tau_k))} = \frac{\theta_{1,s}^\textit{LP}(\tau_k)}{f_1(q_1(\tau_k))} + \frac{\pi(s)\theta_{0,s}^\textit{LP}(\tau_k)}{(1-\pi(s))f_0(q_0(\tau_k))},~k \in [K],
		\end{align*}
		where for $\tau = \tau_1,\cdots,\tau_K$ and $a=0,1$,
		\begin{align*}
		& \theta_{a,s}^\textit{LP}(\tau) = \left[\mathbb{E}(\tilde{W}_{i,s}(\tau) \tilde{W}_{i,s}(\tau)^\top|S_i=s)\right]^{-1}\mathbb{E}\left[\tilde{W}_{i,s}(\tau) 1\{Y_i(a) \leq q_a(\tau)\}|S_i=s\right].
		\end{align*}
		\label{thm:theta_as^star}
	\end{thm}
	
	Four remarks are in order. First, the optimal linear coefficients $\{\theta_{a,s}(\tau)\}_{a = 0 ,1, s \in \mathcal{S}}$ are not uniquely defined. In order to achieve the minimal variance, we only need to consistently estimate one of the minimizers. We choose
	\begin{align*}
	(\theta_{1,s}(\tau),\theta_{0,s}(\tau)) = (\theta_{1,s}^\textit{LP}(\tau),\theta_{0,s}^\textit{LP}(\tau)),~s \in \mathcal{S},
	\end{align*}
	as this choice avoids estimation of the densities $f_1(q_1(\tau))$ and $f_0(q_0(\tau))$. In Theorem \ref{prop:lp} below, we propose estimators of $\theta_{1,s}^\textit{LP}(\tau)$ and $\theta_{0,s}^\textit{LP}(\tau)$ and show they are consistent uniformly over $s$ and $\tau$. Second, note that no adjustment is nested by our linear adjustment with zero coefficients. Due to the optimality result established in Theorem \ref{thm:theta_as^star}, our regression-adjusted QTE estimator with (consistent estimators of) $\{\theta_{a,s}^\textit{LP}(\tau)\}_{a = 0 ,1, s \in \mathcal{S}}$ is more efficient than that with no adjustments. Third, we also need to clarify that the optimality of $\{\theta_{a,s}^\textit{LP}(\tau)\}_{a = 0 ,1, s \in \mathcal{S}}$ is only within the class of linear regressions. It is possible that the QTE estimator with some nonlinear adjustments are more efficient than that with the optimal linear adjustments, especially because the linear probability model is likely misspecified. Fourth, the optimal linear coefficients $\{\theta_{a,s}^\textit{LP}(\tau)\}_{a = 0 ,1, s \in \mathcal{S},\tau \in \Upsilon}$ minimize (over the class of linear models) not only the asymptotic variance of $\hat{q}^{par}(\tau)$ but also the covariance matrix of $(\hat{q}^{par}(\tau_1),\cdots,\hat{q}^{par}(\tau_K))$ for any finite-dimension quantile indices $(\tau_1,\cdots,\tau_K)$. This implies we can use the same (estimators of) optimal linear coefficients for hypothesis testing involving single, multiple, or even a continuum of quantile indices.
	
	In the rest of this subsection, we focus on the estimation of $\{\theta_{a,s}^\textit{LP}(\tau)\}_{a = 0,1,s\in \mathcal{S}}$. Note that $\theta_{a,s}^\textit{LP}(\tau)$ is the projection coefficient of $1\{Y_i \leq q_a(\tau)\}$ on $\tilde{W}_{i,s}(\tau)$ for the sub-population with $S_i=s$ and $A_i = a$. We estimate them by sample analog. Specifically, the parameter $q_a(\tau)$ is unknown and is replaced by some $\sqrt{n}$-consistent estimator denoted by $\hat{q}_a(\tau)$. 
	
	\begin{ass}
		Assume that $\sup_{\tau \in \Upsilon, a = 0,1}|\hat{q}_a(\tau) - q_a(\tau)| = O_p(n^{-1/2})$.
		\label{ass:qhat}
	\end{ass}
	
	In practice, we compute $\{\hat{q}_a(\tau)\}_{a=0,1}$ based on \eqref{eq:q1adj} and \eqref{eq:q0adj} by setting $\widehat{m}_a(\tau,S_i,X_i) \equiv 0$. Then, Assumption \ref{ass:qhat} holds automatically by Theorem \ref{thm:est} with $\widehat{m}_a(\tau,S_i,X_i) = \overline{m}_a(\tau,S_i,X_i)=0$. Analysis throughout this section takes into account that the estimator $\hat{q}_a(\tau)$ is used in place of $q_a(\tau)$.  
	
	Next, we define the estimator of $\theta_{a,s}^\textit{LP}(\tau)$. Recall $I_a(s)$ is defined in Assumption \ref{ass:mhat}. Let
	\begin{align}
	& \overline{m}_a(\tau,s,X_i) = \tau - W_{i,s}^\top(\tau)\theta_{a,s}^\textit{LP}(\tau), \label{eq:ma_lp}\\
	& \widehat{m}_a(\tau,s,X_i) = \tau - W_{i,s}^\top(\tau)\hat{\theta}_{a,s}^\textit{LP}(\tau), \label{eq:mahat_lp}\\
	& \dot{W}_{i,a,s}(\tau)= W_{i,s}(\tau)- \frac{1}{n_a(s)}\sum_{i \in I_a(s)}W_{i,s}(\tau), \label{eq:xahat} 
	\end{align}
	and
	\begin{align}
	\hat{\theta}_{a,s}^\textit{LP}(\tau) = \left[\frac{1}{n_a(s)}\sum_{i \in I_a(s)} \dot{W}_{i,a,s}(\tau) \dot{W}_{i,a,s}^\top(\tau)\right]^{-1} \left[\frac{1}{n_a(s)}\sum_{i \in I_a(s)} \dot{W}_{i,a,s}(\tau)1\{Y_i \leq \hat{q}_a(\tau)\} \right].
	\label{eq:thetahat_lp_est}
	\end{align}
	
	\begin{ass}
		Suppose there exists a positive random variable $L_i$ and a positive constant $C>0$ such that for $a=0,1$, 
		\begin{align*}
		\sup_{\tau_1,\tau_2 \in \Upsilon, a=0,1,s\in \mathcal{S}}||W_{i,s}(\tau_1) - W_{i,s}(\tau_2)||_2 \leq & L_i|\tau_1-\tau_2|, \quad \sup_{\tau \in \Upsilon, a=0,1,s\in \mathcal{S}}||W_{i,s}(\tau)||_2 \leq  L_i,
		\end{align*}
		and $\mathbb{E}(L_i^d|S_i=s) \leq C <\infty$ for some $d > 2$. \label{ass:regX2} 
	\end{ass}
	We note that Assumption \ref{ass:regX2} holds automatically if the regressor $W_{i,s}(\tau)$ does not depend on $\tau$. 
	\begin{thm}
		Suppose Assumptions \ref{ass:assignment1}, \ref{ass:tau}, \ref{ass:regX}--\ref{ass:regX2} hold. Then Assumption \ref{ass:par} holds for $(\theta_{a,s}(\tau), \hat{\theta}_{a,s}(\tau)) = (\theta_{a,s}^\textit{LP}(\tau), \hat{\theta}_{a,s}^\textit{LP}(\tau)),~a=0,1,s \in \mathcal{S},\tau \in \Upsilon$.
		\label{prop:lp}
	\end{thm}
	
	We refer to the QTE estimator adjusted by this linear probability model with optimal linear coefficients $\theta_{a,s}^\textit{LP}(\tau)$ and estimators $\hat{\theta}_{a,s}^\textit{LP}(\tau)$ as the LP estimator and denote it and its bootstrap counterpart as $\hat{q}^\textit{LP}(\tau)$ and $\hat{q}^\textit{LP,w}(\tau)$, respectively. Theorem \ref{prop:lp} verifies Assumption \ref{ass:par} for the proposed estimator of the optimal linear coefficient. Then, by Theorem \ref{thm:par}, Theorems \ref{thm:est} and \ref{thm:boot} hold for $\hat{q}^\textit{LP}(\tau)$ and $\hat{q}^\textit{LP,w}(\tau)$, which implies all the estimation and inference methods established in the paper are valid for the LP estimator. Theorem \ref{thm:theta_as^star} further shows $\hat{q}^\textit{LP}(\tau)$ is the estimator with the optimal linear adjustment and weakly more efficient than the QTE estimator with no adjustments.

	\subsubsection{Logistic probability model}
	\label{sec:logistic}
	It is also common to consider the logistic regression as the adjustment and estimate the model by maximum likelihood (ML). The main goal of the working model is to approximate the true model as closely as possible. It is, therefore, useful to include additional technical regressors such as interactions in the logistic regression. The set of regressors used is defined as $H_i = H(X_i)$, which is allowed to contain the intercept. Let $\hat{\theta}_{a,s}^\textit{ML}(\tau)$ and $\theta_{a,s}^\textit{ML}(\tau)$ be the quasi-ML estimator and its corresponding pseudo true value, respectively, i.e.,
	\begin{align}
	\hat{\theta}_{a,s}^\textit{ML}(\tau) = \argmax_{\theta_a}\frac{1}{n_a(s)}\sum_{i \in I_a(s)}\left[1\{Y_i \leq \hat{q}_a(\tau)\}\log(\lambda(H_i^\top\theta_a)) + 1\{Y_i > \hat{q}_a(\tau)\}\log(1-\lambda(H_i^\top\theta_a))\right],
	\label{eq:mlhat}
	\end{align}
	and
	\begin{align}
	\theta_{a,s}^\textit{ML}(\tau) = \argmax_{\theta_a}\mathbb{E}\left[1\{Y_i(a) \leq q_a(\tau)\}\log(\lambda(H_i^\top\theta_a)) + 1\{Y_i(a) > q_a(\tau)\}\log(1-\lambda(H_i^\top\theta_a))|S_i=s\right].
	\label{eq:ml}
	\end{align}
	We then define 
	\begin{align}
	\overline{m}_a(\tau,s,X_i) = \tau -\lambda(H_i^\top\theta_{a,s}^\textit{ML}(\tau)) \quad \text{and} \quad 
	\widehat{m}_a(\tau,s,X_i) = \tau -\lambda(H_i^\top\hat{\theta}_{a,s}^\textit{ML}(\tau)). \label{eq:mahat_ml}
	\end{align}
	In addition to the inclusion of technical regressors, we allow the pseudo true value ($\theta_{a,s}^\textit{ML}(\tau)$) to vary across quantiles $\tau$, giving another layer of flexibility to the model. Such a model is called the distribution regression and was first proposed by \cite{CFM13}. We emphasize here that, although we aim to make the regression model as flexible as possible, our theory and results do not require the model to be correctly specified.
	
	\begin{ass}
		Suppose $\theta_{a,s}^\textit{ML}(\tau)$ is the unique minimizer defined in \eqref{eq:ml} for $a=0,1$.
		\label{ass:logit2}
	\end{ass}
	
	\begin{thm}
		Suppose Assumptions \ref{ass:assignment1}, \ref{ass:tau}, \ref{ass:qhat}, \ref{ass:logit2} hold and there exist constants $c,C$ such that 
		\begin{align*}
		0<c \leq \lambda_{\min}(\mathbb{E}H_iH_i^\top) \leq  \lambda_{\max}(\mathbb{E}H_iH_i^\top) \leq C<\infty,
		\end{align*}
		then Assumption \ref{ass:par}(iii) holds for $(\theta_{a,s}(\tau), \hat{\theta}_{a,s}(\tau)) = (\theta_{a,s}^\textit{ML}(\tau), \hat{\theta}_{a,s}^\textit{ML}(\tau)),~a=0,1,s \in \mathcal{S},\tau \in \Upsilon$.
		\label{prop:lg2}
	\end{thm}
	
	Four remarks are in order. First, we refer to the QTE estimator adjusted by the logistic model with QMLE as the ML estimator and denote it and its bootstrap counterpart as $\hat{q}^\textit{ML}(\tau)$ and $\hat{q}^\textit{ML,w}(\tau)$, respectively. Assumptions \ref{ass:par}(i) holds automatically for the logistic regression. If we further impose Assumption \ref{ass:par}(ii), then Theorem \ref{prop:lg2} implies that all the estimation and bootstrap inference methods established in the paper are valid for the ML estimator. Second, we take into account that $\hat{\theta}_{a,s}^\textit{ML}(\tau)$ is computed when the true $q_a(\tau)$ is replaced by its estimator $\hat{q}_a(\tau)$ and derive the results in Theorem \ref{prop:lg2} under Assumption \ref{ass:qhat}. Third, the ML estimator is not guaranteed to be optimal or be more efficient than QTE estimator with no adjustments. On the other hand, as we can include additional technical terms in the regression and allow the regression coefficients to vary across $\tau$, the logistic model can be close to the true model $m_a(\tau,s,X_i)$, which achieves the global minimum asymptotic variance based on Theorem \ref{thm:est}. Fourth, in Section \ref{sec:np}, we further justify the use of the ML estimator with a flexible logistic model by letting the number of technical terms (or equivalently, the dimension of $H_i$) diverge to infinity, showing by this means that the ML estimator can indeed consistently estimate the true model and thereby achieve the global minimum covariance matrix of the adjusted QTE estimator. 
	
	\subsubsection{Further improved logistic model}
	\label{sec:lpml}
	Although in simulations, we cannot find a DGP in which the QTE estimator with logistic adjustment is less efficient than that with no adjustments, theoretically such a scenario still exists. In this section, we follow the idea of \cite{CF21} and construct an estimator which is weakly more efficient than both the ML estimator and the estimator with no adjustments. We denote $W_{i,s}(\tau) = (\lambda(H_i^\top \theta_{1,s}^\textit{ML}(\tau)),\lambda(H_i^\top \theta_{0,s}^\textit{ML}(\tau)))^\top$ and treat it as the regressor in a linear adjustment, i.e., define 
	$\overline{m}_{a}(\tau,s,X_i) = \tau - W_{i,s}^\top(\tau) t_{a,s}(\tau)$. Then, the logistic adjustment in Section \ref{sec:logistic} and no adjustments correspond to $t_{a,s}(\tau) = a(1,0)^\top + (1-a)(0,1)^\top$ and $t_{a,s}(\tau) = (0,0)^\top$ for $a=0,1$, respectively. However, following Theorem \ref{thm:theta_as^star}, the optimal linear coefficient with regressor $W_{i,s}(\tau)$ is 
	\begin{align}
	\label{eq:theta^lpml}
	& \theta_{a,s}^\textit{LPML}(\tau) = \left[\mathbb{E}(\tilde{W}_{i,s}(\tau)\tilde{W}_{i,s}^\top(\tau)|S_i=s)\right]^{-1}\mathbb{E}\left[\tilde{W}_{i,s}(\tau) 1\{Y_i(a) \leq q_a(\tau)\}|S_i=s\right],
	\end{align}
	where $\tilde{W}_{i,s}(\tau) = W_{i,s}(\tau) - \mathbb{E}(W_{i,s}(\tau)|S_i=s)$. Using the adjustment term 
	$\overline{m}_{a}(\tau,s,X_i) = \tau - W_{i,s}^\top(\tau) t_{a,s}(\tau)$ with $t_{a,s}(\tau) =  \theta_{a,s}^\textit{LPML}(\tau)$ is asymptotically weakly more efficient than any other choices of $t_{a,s}(\tau)$. In practice, we do not observe $W_{i,s}(\tau)$, but can replace it by its feasible version $\hat{W}_{i,s}(\tau) = (\lambda(H_i^\top \hat{\theta}_{1,s}^\textit{ML}(\tau)),\lambda(H_i^\top \hat{\theta}_{0,s}^\textit{ML}(\tau)))^\top$. We then define 
	\begin{align}
	& \overline{m}_a(\tau,s,X_i) = \tau - W_{i,s}^\top(\tau)\theta_{a,s}^\textit{LPML}(\tau), \label{eq:ma_lp_lg}\\
	& \widehat{m}_a(\tau,s,X_i) = \tau - \hat{W}_{i,s}^\top(\tau)\hat{\theta}_{a,s}^\textit{LPML}(\tau), \label{eq:mahat_lp_lg}\\
	& \breve{W}_{i,a,s}(\tau) = \hat{W}_{i,s}(\tau) - \frac{1}{n_a(s)}\sum_{i \in I_a(s)}\hat{W}_{i,s}(\tau), \label{eq:xahat_lp_lg}
	\end{align}
	and
	\begin{align}
	\label{eq:thetahat^lpml}
	\hat{\theta}_{a,s}^\textit{LPML}(\tau) = \left[\frac{1}{n_a(s)}\sum_{i \in I_a(s)}\breve{W}_{i,a,s}(\tau)\breve{W}_{i,a,s}^\top(\tau)\right]^{-1} \left[\frac{1}{n_a(s)}\sum_{i \in I_a(s)}\breve{W}_{i,a,s}(\tau)1\{Y_i \leq \hat{q}_a(\tau)\} \right].
	\end{align}
	
	\begin{ass}
		\begin{enumerate}[label=(\roman*)]
			\item There exist constants $c,C$ such that 
			\begin{align*}
			0<c< & \inf_{a = 0,1, s \in \mathcal{S},\tau \in \Upsilon} \lambda_{\min}(\mathbb{E}\tilde{W}_{i,s}(\tau)\tilde{W}_{i,s}^\top(\tau)|S_i=s) \\
			\leq & \sup_{a = 0,1, s \in 
				\mathcal{S},\tau \in \Upsilon}\lambda_{\max}(\mathbb{E}\tilde{W}_{i,s}(\tau)\tilde{W}_{i,s}^\top(\tau)|S_i=s) \leq C<\infty.
			\end{align*}
			\item Suppose 
			\begin{align*}
			& \sup_{\tau_1,\tau_2 \in \Upsilon, a=0,1,s\in \mathcal{S}}||\theta_{a,s}^\textit{ML}(\tau_1) - \theta_{a,s}^\textit{ML}(\tau_2)||_2 \leq C|\tau_1-\tau_2| \\
			& \sup_{\tau_1,\tau_2 \in \Upsilon, a=0,1,s\in \mathcal{S}}||\theta_{a,s}^\textit{LPML}(\tau_1) - \theta_{a,s}^\textit{LPML}(\tau_2)||_2 \leq C|\tau_1-\tau_2|.
			\end{align*}
		\end{enumerate}
		\label{ass:lpml}
	\end{ass}
	\begin{thm}
		Denote $\hat{q}^\textit{LPML}(\tau)$ and $\hat{q}^\textit{LPML,w}(\tau)$ as the $\tau$th QTE estimator and its multiplier bootstrap counterpart defined in Sections \ref{sec:est} and \ref{sec:boot}, respectively, with $ \overline{m}_a(\tau,s,X_i) $ and $\widehat{m}_a(\tau,s,X_i)$ defined in \eqref{eq:ma_lp_lg} and \eqref{eq:mahat_lp_lg}, respectively. Suppose Assumptions \ref{ass:assignment1}, \ref{ass:tau}, \ref{ass:qhat}, \ref{ass:logit2}, and \ref{ass:lpml} hold, and there exist constants $c,C$ such that 
		\begin{align*}
		0<c \leq \lambda_{\min}(\mathbb{E}H_iH_i^\top) \leq  \lambda_{\max}(\mathbb{E}H_iH_i^\top) \leq C<\infty.
		\end{align*}
		Then, Assumptions \ref{ass:mhat} and \ref{ass:mhatw} hold, which further implies Theorems \ref{thm:est} and \ref{thm:boot} hold for $\hat{q}^\textit{LPML}(\tau)$ and $\hat{q}^\textit{LPML,w}(\tau)$, respectively. Further denote the asymptotic covariance matrices of $(\hat{q}^\textit{J}(\tau_1),\cdots,\hat{q}^\textit{J}(\tau_K))$ for any finite set of quantile indices $(\tau_1,\cdots,\tau_K)$ as $[\Sigma^\textit{J}(\tau_k,\tau_l)]_{k,l \in [K]}$ for $J \in \{ \text{LPML,ML,NA}\}$, where $\hat{q}^\textit{NA}(\tau)$ is the $\tau$th QTE estimator without adjustments. Then we have
		\begin{align*}
		[\Sigma^\textit{LPML}(\tau_k,\tau_l)]_{k,l \in [K]} \leq [\Sigma^\textit{ML}(\tau_k,\tau_l)]_{k,l \in [K]} \quad \text{and} \quad [\Sigma^\textit{LPML}(\tau_k,\tau_l)]_{k,l \in [K]} \leq [\Sigma^\textit{NA}(\tau_k,\tau_l)]_{k,l \in [K]}   
		\end{align*}  in the matrix sense. 
		\label{prop:lpml}
	\end{thm}
	
	In practice when $n$ is small $\breve{W}_{i,a,s}(\tau)$ may be nearly multicollinear within some stratum, which can lead to size distortion in inference concerning QTE. We therefore suggest first normalizing $\breve{W}_{i,a,s}(\tau)$ by its standard deviation (denoting the normalized $\breve{W}_{i,a,s}(\tau)$ as $\ddot{W}_{i,a,s}(\tau)$) and then running a ridge regression 
	\begin{align*}
	\tilde{\theta}_{a,s}^\textit{LPML}(\tau) = \left[\frac{1}{n_a(s)}\sum_{i \in I_a(s)}\ddot{W}_{i,a,s}(\tau)\ddot{W}_{i,a,s}^\top(\tau)+\delta_n I_2\right]^{-1} \left[\frac{1}{n_a(s)}\sum_{i \in I_a(s)}\ddot{W}_{i,a,s}(\tau)1\{Y_i \leq \hat{q}_a(\tau)\} \right],
	\end{align*}
	where $I_2$ is the two-dimensional identity matrix and $\delta_n = 1/n$. Then, the final regression adjustment is 
	\begin{align*}
	\widehat{m}_a(\tau,s,X_i) = \tau - \ddot{W}_{i,a,s}^\top (\tau)\tilde{\theta}_{a,s}^\textit{LPML}(\tau).
	\end{align*}
	Given Assumption \ref{ass:lpml}, such a ridge penalty is asymptotically negligible and all the results in Theorem \ref{prop:lpml} still hold.\footnote{In unreported simulations, we find that when $n\geq 800$, the ridge regularization is unnecessary and the original adjustment (i.e., \eqref{eq:mahat_lp_lg}) has no size distortion, implying that near-multicollinearity is indeed just a finite-sample issue.} 
	
	% We also conjecture that even when Assumption \ref{ass:lpml} does not hold, $\theta_{a,s}^\textit{LPML}(\tau) = \mathbb{E}(\tilde{W}_{i,s}(\tau)\tilde{W}_{i,s}^\top(\tau)|S_i=s)^+\mathbb{E}(\tilde{W}_{i,s}(\tau)1\{Y_i(a) \leq q_a(\tau)\}|S_i=s)$ still achieves the minimum variance of $\hat{q}^{adj}(\tau)$ among all linear adjustments with regressor $\tilde{W}_{i,s}(\tau)$, where $U^+$ is the pseudo inverse of $U$, and the ridge-regularized estimator is consistent for such a $\theta_{a,s}^\textit{LPML}(\tau)$. A formal analysis is left for future research.

	\subsection{Nonparametric method}
	\label{sec:np}
	This section considers nonparametric estimation of $m_a(\tau,s,X_i)$ when the dimension of $X_i$ is fixed as $d_x$. For ease of notation, we assume all coordinates of $X_i$ are continuously distributed. If in an application some elements of $X$ are discrete, the dimension $d_x$ is interpreted as the dimension of the continuous covariates. All results in this section can then be extended in a conceptually straightforward manner by using the continuous covariates only within samples that are homogeneous in discrete covariates.
	
	As  $m_a(\tau,s,X_i)$ is nonparametrically estimated, we have $ \overline{m}_a(\tau,s,X_i) = m_a(\tau,s,X_i)=\tau- \mathbb{P}(Y_i(a)\leq q_a(\tau)|S_i=s,X_i)$. We estimate  $\mathbb{P}(Y_i(a)\leq q_a(\tau)|S_i=s,X_i)$ by the sieve method of fitting a logistic model, as studied by \cite{HIR03}. Specifically, recall $\lambda(\cdot)$ is the logistic CDF and denote  the number of sieve bases by $h_n$, which depends on the sample size $n$ and can grow to infinity as $n \rightarrow \infty$. Let $H_{h_n}(x) = (b_{1n}(x),\cdots,b_{h_nn}(x))^\top$  where $\{b_{hn}(x)\}_{h \in [h_n]}$ is an $h_n$ dimensional basis of a linear sieve space. More details on the sieve space are given in Section \ref{sec:compute} of the Online Supplement. Denote
	\begin{align}
	& \widehat{m}_a(\tau,s,X_i) = \tau- \lambda(H_{h_n}^\top(X_i)\hat{\theta}_{a,s}^\textit{NP}(\tau)) \quad \text{and } \label{eq:mhat_np1}\\
	& \hat{\theta}_{a,s}^\textit{NP}(\tau) = \argmax_{\theta_a} \frac{1}{n_{a}(s)}\sum_{i \in I_{a}(s)} \biggl[1\{Y_i \leq \hat{q}_a(\tau)\}\log(\lambda(H_{h_n}^\top(X_i)\theta_a)) \notag \\
	& \qquad \qquad \qquad \qquad \qquad \qquad + 1\{Y_i > \hat{q}_a(\tau)\}\log(1-\lambda(H_{h_n}^\top(X_i)\theta_a))\biggr]. \label{eq:mhat_np2}
	\end{align}
	
	We refer to the QTE estimator with the nonparametric adjustment as the NP estimator. Note that we use the estimator $\hat{q}_a(\tau)$ of $q_a(\tau)$ in \eqref{eq:mhat_np2}, where $\hat{q}_a(\tau)$ satisfies Assumption \ref{ass:qhat}. All the  analysis in this section takes account of the fact that $\hat{q}_a(\tau)$ instead of $q_a(\tau)$ is used. 
	
	\begin{ass}
		\begin{enumerate}[label=(\roman*)]
			\item There exist constants $0<\kappa_1<\kappa_2 < \infty$ such that with probability approaching one, $$\kappa_1\leq \lambda_{\min}\left( \frac{1}{n_{a}(s)}\sum_{i \in I_{a}(s)}H_{h_n}(X_i)H_{h_n}^\top(X_i)\right)\leq \lambda_{\max}\left( \frac{1}{n_{a}(s)}\sum_{i \in I_{a}(s)}H_{h_n}(X_i)H_{h_n}^\top(X_i)\right) \leq \kappa_2,$$
			and
			$$\kappa_1\leq \lambda_{\min}\left( \mathbb{E}(H_{h_n}(X_i)H_{h_n}^\top(X_i)|S_i=s)\right)\leq \lambda_{\max}\left( \mathbb{E}(H_{h_n}(X_i)H_{h_n}^\top(X_i)|S_i=s)\right)\leq \kappa_2.$$
			\item For $a=0,1$, there exists an $h_n \times 1$ vector $\theta_{a,s}^\textit{NP}(\tau)$ such that for $R_{a}(\tau,s,x) = \mathbb{P}(Y_i(a)\leq q_a(\tau)|S_i=s,X_i=x) - \lambda(H_{h_n}^\top(x) \theta_{a,s}^\textit{NP}(\tau))$, we have $\sup_{a = 0,1, s \in \mathcal{S},\tau \in \Upsilon, x \in \Supp(X)}|R_a(\tau,s,x)| = o(1)$,
			\begin{align*}
			\sup_{a=0,1,\tau \in \Upsilon,s\in \mathcal{S}}\frac{1}{n_a(s)}\sum_{i \in I_{a}(s)}R_{a}^2(\tau,s,X_i) = O_p\left(\frac{h_n \log(n)}{n}\right),
			\end{align*}
			and
			\begin{align*}
			\sup_{a=0,1,\tau \in \Upsilon,s\in \mathcal{S}} \mathbb{E}(R_a^2(\tau,s,X_i)|S_i=s) = O\left(\frac{h_n \log(n)}{n}\right).
			\end{align*}
			\item For $a=0,1$, there exists a constant $c \in (0,0.5)$ such that
			\begin{align*}
			c\leq & \inf_{a = 0,1, s \in \mathcal{S},\tau \in \Upsilon, x \in \Supp(X)}\mathbb{P}(Y_i(a)\leq q_a(\tau)|S_i=s,X_i=x) \\
			\leq & \sup_{a = 0,1, s \in \mathcal{S},\tau \in \Upsilon, x \in \Supp(X)}\mathbb{P}(Y_i(a)\leq q_a(\tau)|S_i=s,X_i=x)  \leq 1-c.
			\end{align*}
			\item Suppose $\mathbb{E}(H^2_{h_n,h}(X_i)|S_i=s) \leq C<\infty$ for some constant $C>0$, $\sup_{x \in \Supp(X)}||H_{h_n}(x)||_2 \leq \zeta(h_n)$, $\zeta^2(h_n) h_n \log(n) = o(n)$, and $h_n^2 \log^2(n) = o(n)$, where $H_{h_n,h}(X_i)$ denotes the $h$th coordinate of $H_{h_n}(X_i)$.
		\end{enumerate}
		\label{ass:np}
	\end{ass}
	
	Four remarks are in order. First, Assumption \ref{ass:np}(i) is standard in the sieve literature. Second, Assumption \ref{ass:np}(ii) means the approximation error of the sieve logistic model vanishes asymptotically, which holds given sufficient smoothness of $\mathbb{P}(Y_i(a)\leq q_a(\tau)|S_i=s,X_i=x)$ in $x$. Third, Assumption \ref{ass:np}(iii) usually holds when $\Supp(X)$ is compact. This condition is also assumed by \cite{HIR03}. Fourth, the quantity $\zeta(h_n)$ in Assumption \ref{ass:np}(iv) depends on the choice of basis functions. For example, $\zeta(h_n) = O(h_n^{1/2})$ for splines and $\zeta(h_n) = O(h_n)$ for power series. Taking splines as an example, Assumption \ref{ass:np}(iv) requires $h_n = o(n^{1/2})$.  
	
	\begin{thm}
		Denote $\hat{q}^\textit{NP}(\tau)$ and $\hat{q}^\textit{NP,w}(\tau)$ as the $\tau$th QTE estimator and its multiplier bootstrap counterpart defined in Sections \ref{sec:est} and \ref{sec:boot}, respectively, with $\overline{m}_a(\tau,S_i,X_i) = m_a(\tau,S_i,X_i)$ and $\widehat{m}_a(\tau,S_i,X_i)$ defined in \eqref{eq:mhat_np1}. Further suppose Assumptions \ref{ass:assignment1}, \ref{ass:tau}, \ref{ass:weight}, \ref{ass:qhat}, and \ref{ass:np} hold. Then, Assumptions \ref{ass:mhat} and \ref{ass:mhatw} hold, which further implies that Theorems \ref{thm:est} and \ref{thm:boot} hold for $\hat{q}^\textit{NP}(\tau)$ and $\hat{q}^\textit{NP,w}(\tau)$, respectively.  In addition, for any finite-dimensional quantile indices $(\tau_1,\cdots,\tau_K)$, the covariance matrix of $(\hat{q}^\textit{NP}(\tau_1),\cdots,\hat{q}^\textit{NP}(\tau_K))$ achieves the minimum (in the matrix sense) as characterized in Theorem \ref{thm:est}.
		\label{thm:np}
	\end{thm}
	Three remarks are in order. First, as the nonparametric regression consistently estimates the true specifications $\{m_a(\cdot)\}_{a=0,1}$, the QTE estimator adjusted by the nonparametric regression achieves the global minimum asymptotic variance, and thus is weakly more efficient than QTE estimation with linear and logistic adjustments studied in the previous section. Second, the practical implementation of NP and ML methods are the same, given that they share the same set of covariates (basis functions). Therefore, even if we include a small number of basis functions so that $h_n$ is better treated as fixed, the proposed estimation and inference methods for the regression-adjusted QTE estimator are still valid, although they may not be optimal. Third, in Section \ref{sec:reg} in the Online Supplement, we consider computing $\widehat{m}_a(\tau,s,x)$ via an $\ell_1$ penalized logistic regression when the dimension of the regressors can be comparable or even higher than the sample size. We then provide primitive conditions under which we verify Assumptions \ref{ass:mhat} and \ref{ass:mhatw}. 
	
	\section{Simulations} \label{sec:sim}
	\subsection{Data generating processes}
	Two DGPs are used to assess the finite sample performance of the estimation and inference methods introduced in the paper. We consider the outcome equation
	\begin{align}
	Y_{i} = \alpha(X_{i}) + \gamma Z_i + \mu(X_{i}) A_{i} + \eta_{i},
	\end{align}
	where $\gamma = 4$ for all cases while $\alpha(X_{i})$, $\mu(X_{i})$, and $\eta_{i}$ are separately specified as follows. 
	\begin{enumerate}[label=(\roman*)]
		\item Let $Z$ be standardized Beta$(2,2)$ distributed, $S_i = \sum_{j = 1}^4 1\{Z_i \leq g_j\}$, and $(g_1, \cdots, g_4) = (-0.25\sqrt{20}, 0, 0.25\sqrt{20}, 0.5\sqrt{20})$. $X_i$ contains two covariates $(X_{1i}, X_{2i})^\top$, where $X_{1i}$ follows a uniform distribution on $[-2,2]$, $X_{2i}$ follows a standard normal distribution, and $X_{1i}$ and $X_{2i}$ are independent. Further define $\alpha(X_{i}) = 1 + X_{2i}$, $\mu(X_i) = 1 + X_i^\top \beta$, $\beta = (3,3)^\top$, and $\eta_i = (0.25 + X_{1i}^2) A_i\varepsilon_{1i} + (1-A_i)\varepsilon_{2i}$, where $(\varepsilon_{1i}, \varepsilon_{2i})$ are jointly standard normal.
		
		\item Let $Z$ be uniformly distributed on $[-2,2]$, $S_i = \sum_{j = 1}^4 1\{Z_i \leq g_j\}$, and $(g_1, \cdots, g_4) = (-1, 0, 1, 2)$. Let $X_i=(X_{1i}, X_{2i})^\top$ be the same as defined in DGP (i). Further define $\alpha(X_{i}) = 1 + X_{1i} + X_{2i}$, $\mu(X_i) = 1 + X_{1i} + X_{2i} + \frac{1}{4} (X_i^\top\beta)^2$ with $\beta = (2,2)^\top$, and $\eta_i = 2(1 + Z_{i}^2) A_i\varepsilon_{1i} + (1 + Z_{i}^2)(1-A_i)\varepsilon_{2i}$, where $(\varepsilon_{1i}, \varepsilon_{2i})$ are mutually independently $T(5)/\sqrt{5}$ distributed.
	\end{enumerate}
	
	For each DGP, we consider the following four randomization schemes as in \cite{ZZ20} with $\pi(s) = 0.5$ for $s \in \mathcal{S}$:
	\begin{enumerate}[label=(\roman*)]
		\item SRS: Treatment assignment is generated as in Example \ref{ex:srs}.
		\item WEI: Treatment assignment is generated as in Example \ref{ex:wei} with $\phi(x) = (1-x)/2$.
		\item BCD: Treatment assignment is generated as in Example \ref{ex:bcd} with $\lambda = 0.75$.
		\item SBR: Treatment assignment is generated as in Example \ref{ex:sbr}.
	\end{enumerate}
	We assess the empirical size and power of the tests for $n = 200$ and $n = 400$. We compute the true QTEs and their differences by simulations with 10,000 sample size and 1,000 replications. To compute power, we perturb the true values by $\Delta=1.5$. We examine three null hypotheses: 
	\begin{enumerate}
		\item[(i)] Pointwise test 
		$$H_0: q(\tau) = \text{truth} \quad \text{v.s.} \quad H_1: q(\tau) = \text{truth} + \Delta, \quad \tau = 0.25,0.5,0.75;$$
		\item[(ii)] Test for the difference
		$$H_0: q(0.75) - q(0.25) = \text{truth} \quad \text{v.s.} \quad H_1: q(0.75) - q(0.25) = \text{truth} + \Delta;$$ 
		\item[(iii)] Uniform test 
		$$H_0: q(\tau) = \text{truth}(\tau) \quad v.s. \quad H_1: q(\tau) = \text{truth}(\tau) + \Delta, \quad \tau \in [0.25,0.75].$$ 
	\end{enumerate}
	For the pointwise test, we report the results for the median ($\tau = 0.5$) in the main text and give the cases $\tau = 0.25$ and $\tau = 0.75$ in the Online Supplement. 
	
	\subsection{Estimation methods}
	We consider the following estimation methods of the auxiliary regression. 
	\begin{enumerate}[label=(\roman*)]
		\item NA: the estimator with no adjustments, i.e., setting $\widehat{m}_{a}(\cdot) = \overline{m}_a(\cdot) = 0$.  
		\item LP: the linear probability model with regressors $X_i$ and the pseudo true value estimated by $\hat{\theta}_{a,s}^\textit{LP}(\tau)$ defined in \eqref{eq:thetahat_lp_est}. 
		\item ML: the logistic model with regressor $H_i = (1,X_{1i}, X_{2i})^\top$ and the pseudo true value estimated by $\hat{\theta}_{a,s}^\textit{ML}(\tau)$ defined in \eqref{eq:mlhat}. 
		\item LPML: the logistic model with regressor $H_i = (1,X_{1i}, X_{2i})^\top$ and the pseudo true value estimated by $\hat{\theta}_{a,s}^\textit{LPML}(\tau)$ defined in \eqref{eq:thetahat^lpml}.
		\item MLX: the logistic model with regressor $H_i = (1,X_{1i}, X_{2i}, X_{1i} X_{2i})^\top$ and the pseudo true value estimated by $\hat{\theta}_{a,s}^\textit{ML}(\tau)$ defined in \eqref{eq:mlhat}. 
		\item LPMLX: the logistic model with regressor $H_i = (1,X_{1i}, X_{2i}, X_{1i} X_{2i})^\top$ and the pseudo true value estimated by $\hat{\theta}_{a,s}^\textit{LPML}(\tau)$ defined in \eqref{eq:thetahat^lpml}. 
		\item NP: the logistic model with regressor $H_{h_n}(X_i) = (1,X_{1i}, X_{2i}, X_{1i} X_{2i},X_{1i} 1\{X_{1i}>t_1\}X_{2i}1\{X_{2i}>t_2\} )^\top$ where $t_1$ and $t_2$ are the sample medians of $\{X_{1i}\}_{i \in [n]}$ and $\{X_{2i}\}_{i \in [n]}$, respectively. The pseudo true value is estimated by $\hat{\theta}_{a,s}^\textit{NP}(\tau)$ defined in \eqref{eq:mhat_np2}. 
	\end{enumerate}
	
	\subsection{Simulation results}
	Table \ref{tab:Sim_Point1} presents the empirical size and power for the pointwise test with $\tau = 0.5$ under DGPs (i) and (ii). We make six observations. First, none of the auxiliary regressions is correctly specified, but test sizes are all close to the nominal level 5\%, confirming that estimation and inference are robust to misspecification. Second, the inclusion of auxiliary regressions improves the efficiency of the QTE estimator as the powers for method ``NA" are the lowest among all the methods for both DGPs and all randomization schemes. This finding is consistent with theory because methods ``LP", ``LPML", ``LPMLX", ``NP" are guaranteed to be weakly more efficient than ``NA". Third, the powers of methods ``LPML" and ``LPMLX" are higher than those of methods ``ML" and ``MLX", respectively. This is consistent with our theory that methods ``LPML" and ``LPMLX" further improve ``ML" and ``MLX", respectively. In addition, methods ``MLX" and ``LPMLX" fit a flexible distribution regression that can approximate the true DGP well. Therefore, the powers of ``MLX" and ``LPMLX" are respectively much larger than those of ``ML" and ``LPML". For the same reason we observe that the power of ``LPMLX" is close to ``NP".\footnote{The results in Section \ref{sec:add_sim} of the Online Supplement show that ``LPMLX" has much smaller bias than ``NP" and its variance is similar to ``NP", which make ``LPMLX" preferable in practice.} Fourth, the powers of method ``NP" are the best because it estimates the true specification and achieves the minimum asymptotic variance of $\hat{q}^{adj}(\tau)$ as shown in Theorem \ref{thm:par}. Fifth, when the sample size is 200, the method ``NP" slightly over-rejects but size becomes closer to nominal when the sample size increases to 400. Sixth, the improvement of power of ``LPMLX" estimator upon ``NA" (i.e., with no adjustments) is due to about 12-15\% reduction of the standard error of the QTE estimator on average.\footnote{The bias and standard errors are reported in the Section \ref{sec:add_sim} in the Online Supplement.} 
	
	\begin{table}[!htp]
		\footnotesize
		\centering
		\caption{Pointwise Test ($\tau = 0.5$)}
		\label{tab:Sim_Point1}%
		\smallskip
		\setlength{\tabcolsep}{2pt}
		\renewcommand*{\arraystretch}{1.2}
		\begin{tabularx}{\linewidth}{@{\extracolsep{\fill}}lcccccccccccccccc}
			\toprule
			& \multicolumn{8}{c}{Size} & \multicolumn{8}{c}{Power} \\ \cmidrule{2-9}\cmidrule{10-17}
			& \multicolumn{4}{c}{$N = 200$} & \multicolumn{4}{c}{$N = 400$} & \multicolumn{4}{c}{$N = 200$} & \multicolumn{4}{c}{$N = 400$} \\ \cmidrule{2-5}\cmidrule{6-9}\cmidrule{10-13}\cmidrule{14-17}
			Methods & SRS   & WEI   & BCD   & SBR   & SRS   & WEI   & BCD   & SBR   & SRS   & WEI   & BCD   & SBR   & SRS   & WEI   & BCD   & SBR \\
			\midrule
			\multicolumn{17}{c}{Panel A: DGP (i)}  \\
			NA    & 0.055 & 0.054 & 0.050 & 0.054 & 0.051 & 0.054 & 0.051 & 0.051 & 0.404 & 0.406 & 0.403 & 0.406 & 0.665 & 0.676 & 0.681 & 0.681 \\
			LP    & 0.052 & 0.050 & 0.049 & 0.052 & 0.048 & 0.053 & 0.051 & 0.052 & 0.491 & 0.497 & 0.502 & 0.492 & 0.779 & 0.788 & 0.790 & 0.791 \\
			ML    & 0.053 & 0.050 & 0.049 & 0.055 & 0.051 & 0.050 & 0.052 & 0.052 & 0.472 & 0.478 & 0.483 & 0.473 & 0.759 & 0.768 & 0.775 & 0.773 \\
			LPML  & 0.054 & 0.052 & 0.052 & 0.057 & 0.052 & 0.054 & 0.051 & 0.053 & 0.506 & 0.509 & 0.523 & 0.513 & 0.802 & 0.812 & 0.814 & 0.809 \\
			MLX   & 0.056 & 0.059 & 0.055 & 0.057 & 0.055 & 0.054 & 0.055 & 0.058 & 0.475 & 0.479 & 0.486 & 0.482 & 0.752 & 0.759 & 0.760 & 0.760 \\
			LPMLX & 0.060 & 0.058 & 0.059 & 0.058 & 0.054 & 0.055 & 0.054 & 0.054 & 0.506 & 0.513 & 0.521 & 0.512 & 0.802 & 0.810 & 0.813 & 0.811 \\
			NP    & 0.063 & 0.059 & 0.062 & 0.064 & 0.055 & 0.054 & 0.054 & 0.056 & 0.523 & 0.523 & 0.531 & 0.526 & 0.804 & 0.811 & 0.814 & 0.809 \\
			\midrule
			\multicolumn{17}{c}{Panel B: DGP (ii)}  \\
			NA    & 0.046 & 0.051 & 0.045 & 0.047 & 0.047 & 0.045 & 0.048 & 0.047 & 0.479 & 0.489 & 0.500 & 0.490 & 0.773 & 0.775 & 0.774 & 0.782 \\
			LP    & 0.049 & 0.051 & 0.050 & 0.050 & 0.045 & 0.048 & 0.050 & 0.045 & 0.572 & 0.581 & 0.589 & 0.579 & 0.851 & 0.856 & 0.857 & 0.854 \\
			ML    & 0.051 & 0.058 & 0.050 & 0.054 & 0.049 & 0.046 & 0.050 & 0.048 & 0.524 & 0.534 & 0.541 & 0.539 & 0.812 & 0.810 & 0.807 & 0.807 \\
			LPML  & 0.051 & 0.058 & 0.054 & 0.053 & 0.050 & 0.049 & 0.053 & 0.047 & 0.574 & 0.581 & 0.588 & 0.580 & 0.862 & 0.863 & 0.863 & 0.863 \\
			MLX   & 0.058 & 0.059 & 0.056 & 0.059 & 0.051 & 0.049 & 0.051 & 0.050 & 0.566 & 0.574 & 0.583 & 0.573 & 0.826 & 0.824 & 0.827 & 0.827 \\
			LPMLX & 0.057 & 0.062 & 0.057 & 0.060 & 0.052 & 0.050 & 0.053 & 0.052 & 0.615 & 0.620 & 0.630 & 0.627 & 0.878 & 0.878 & 0.880 & 0.879 \\
			NP    & 0.063 & 0.066 & 0.062 & 0.062 & 0.056 & 0.055 & 0.056 & 0.051 & 0.622 & 0.625 & 0.632 & 0.628 & 0.883 & 0.880 & 0.882 & 0.879 \\
			\bottomrule
		\end{tabularx}%
	\end{table}%
	
	Tables \ref{tab:Sim_Diff1} and \ref{tab:Sim_Unif1} present sizes and powers of inference on $q(0.75) - q(0.25)$ and on $q(\tau)$ uniformly over $\tau \in [0.25,0.75]$, respectively, for DGPs (i) and (ii) and four randomization schemes. All the observations made above apply to these results. The improvement in power of the ``LPMLX" estimator upon ``NA" (i.e., with no adjustments) is due to a 9\% reduction of the standard error of the difference of the QTE estimators on average. In Section \ref{sec:add_sim} in the Online Supplement, we provide additional simulation results such as the empirical sizes and powers for the pointwise test with $\tau = 0.25$ and $0.75$, the bootstrap inference with the true target fraction, and the adjusted QTE estimator when the DGP contains high-dimensional covariates and the adjustments are computed via logistic Lasso. We also report the biases and standard errors of the adjusted QTE estimators. 
	
	\begin{table}[htp!]
		\footnotesize
		\centering
		\caption{Test for Differences ($\tau_1 = 0.25$, $\tau_2 = 0.75$)}
		\label{tab:Sim_Diff1}%
		\smallskip
		
		\setlength{\tabcolsep}{2pt}
		\renewcommand*{\arraystretch}{1.2}
		\begin{tabularx}{\linewidth}{@{\extracolsep{\fill}}lcccccccccccccccc}
			\toprule
			& \multicolumn{8}{c}{Size} & \multicolumn{8}{c}{Power} \\ \cmidrule{2-9}\cmidrule{10-17}
			& \multicolumn{4}{c}{$N = 200$} & \multicolumn{4}{c}{$N = 400$} & \multicolumn{4}{c}{$N = 200$} & \multicolumn{4}{c}{$N = 400$} \\ \cmidrule{2-5}\cmidrule{6-9}\cmidrule{10-13}\cmidrule{14-17}
			Methods & SRS   & WEI   & BCD   & SBR   & SRS   & WEI   & BCD   & SBR   & SRS   & WEI   & BCD   & SBR   & SRS   & WEI   & BCD   & SBR \\
			\midrule
			\multicolumn{17}{c}{Panel A: DGP (i)}  \\
			NA    & 0.043 & 0.045 & 0.040 & 0.041 & 0.044 & 0.043 & 0.041 & 0.043 & 0.214 & 0.216 & 0.209 & 0.203 & 0.387 & 0.389 & 0.383 & 0.365 \\
			LP    & 0.045 & 0.048 & 0.043 & 0.045 & 0.045 & 0.047 & 0.043 & 0.045 & 0.246 & 0.242 & 0.234 & 0.248 & 0.424 & 0.422 & 0.422 & 0.421 \\
			ML    & 0.045 & 0.045 & 0.043 & 0.042 & 0.046 & 0.047 & 0.040 & 0.048 & 0.234 & 0.233 & 0.231 & 0.239 & 0.415 & 0.422 & 0.417 & 0.426 \\
			LPML  & 0.044 & 0.049 & 0.045 & 0.045 & 0.049 & 0.049 & 0.044 & 0.047 & 0.250 & 0.250 & 0.248 & 0.259 & 0.451 & 0.453 & 0.450 & 0.459 \\
			MLX   & 0.046 & 0.052 & 0.046 & 0.047 & 0.047 & 0.047 & 0.044 & 0.049 & 0.232 & 0.234 & 0.229 & 0.241 & 0.415 & 0.415 & 0.404 & 0.416 \\
			LPMLX & 0.049 & 0.055 & 0.047 & 0.047 & 0.049 & 0.050 & 0.047 & 0.047 & 0.247 & 0.249 & 0.249 & 0.258 & 0.445 & 0.453 & 0.445 & 0.453 \\
			NP    & 0.050 & 0.054 & 0.050 & 0.051 & 0.052 & 0.052 & 0.047 & 0.048 & 0.246 & 0.248 & 0.245 & 0.257 & 0.444 & 0.444 & 0.442 & 0.450 \\
			\midrule
			\multicolumn{17}{c}{Panel B: DGP (ii)}  \\
			NA    & 0.039 & 0.044 & 0.040 & 0.038 & 0.044 & 0.041 & 0.039 & 0.047 & 0.211 & 0.225 & 0.217 & 0.194 & 0.399 & 0.396 & 0.392 & 0.383 \\
			LP    & 0.043 & 0.048 & 0.045 & 0.040 & 0.045 & 0.044 & 0.042 & 0.047 & 0.244 & 0.255 & 0.251 & 0.245 & 0.447 & 0.440 & 0.441 & 0.455 \\
			ML    & 0.049 & 0.046 & 0.046 & 0.043 & 0.044 & 0.045 & 0.042 & 0.048 & 0.217 & 0.228 & 0.213 & 0.212 & 0.379 & 0.386 & 0.386 & 0.396 \\
			LPML  & 0.047 & 0.051 & 0.048 & 0.043 & 0.047 & 0.045 & 0.047 & 0.048 & 0.253 & 0.258 & 0.253 & 0.252 & 0.456 & 0.451 & 0.454 & 0.468 \\
			MLX   & 0.047 & 0.051 & 0.047 & 0.047 & 0.046 & 0.046 & 0.045 & 0.049 & 0.226 & 0.240 & 0.228 & 0.223 & 0.394 & 0.392 & 0.391 & 0.399 \\
			LPMLX & 0.053 & 0.056 & 0.051 & 0.048 & 0.051 & 0.049 & 0.045 & 0.050 & 0.261 & 0.272 & 0.265 & 0.263 & 0.467 & 0.460 & 0.460 & 0.477 \\
			NP    & 0.056 & 0.058 & 0.053 & 0.052 & 0.051 & 0.052 & 0.045 & 0.050 & 0.266 & 0.275 & 0.266 & 0.270 & 0.469 & 0.459 & 0.461 & 0.479 \\
			\bottomrule
		\end{tabularx}%
	\end{table}%
	
	\begin{table}[htp!]
		\footnotesize
		\centering
		\caption{Uniform Test ($\tau \in [0.25, 0.75]$)}
		\label{tab:Sim_Unif1}%
		\smallskip
		\setlength{\tabcolsep}{2pt}
		\renewcommand*{\arraystretch}{1.2}
		\begin{tabularx}{\linewidth}{@{\extracolsep{\fill}}lcccccccccccccccc}
			\toprule
			& \multicolumn{8}{c}{Size} & \multicolumn{8}{c}{Power} \\ \cmidrule{2-9}\cmidrule{10-17}
			& \multicolumn{4}{c}{$N = 200$} & \multicolumn{4}{c}{$N = 400$} & \multicolumn{4}{c}{$N = 200$} & \multicolumn{4}{c}{$N = 400$} \\ \cmidrule{2-5}\cmidrule{6-9}\cmidrule{10-13}\cmidrule{14-17}
			Methods & SRS   & WEI   & BCD   & SBR   & SRS   & WEI   & BCD   & SBR   & SRS   & WEI   & BCD   & SBR   & SRS   & WEI   & BCD   & SBR \\
			\midrule
			\multicolumn{17}{c}{Panel A: DGP (i)} \\ 
			NA    & 0.048 & 0.044 & 0.044 & 0.045 & 0.047 & 0.049 & 0.045 & 0.048 & 0.450 & 0.451 & 0.455 & 0.454 & 0.765 & 0.770 & 0.769 & 0.770 \\
			LP    & 0.045 & 0.044 & 0.043 & 0.045 & 0.047 & 0.051 & 0.047 & 0.046 & 0.589 & 0.588 & 0.589 & 0.581 & 0.902 & 0.901 & 0.904 & 0.900 \\
			ML    & 0.047 & 0.044 & 0.043 & 0.045 & 0.044 & 0.051 & 0.045 & 0.047 & 0.570 & 0.577 & 0.582 & 0.568 & 0.887 & 0.889 & 0.893 & 0.890 \\
			LPML  & 0.046 & 0.046 & 0.045 & 0.047 & 0.046 & 0.050 & 0.046 & 0.051 & 0.603 & 0.605 & 0.616 & 0.607 & 0.916 & 0.917 & 0.915 & 0.915 \\
			MLX   & 0.052 & 0.049 & 0.048 & 0.048 & 0.046 & 0.053 & 0.050 & 0.050 & 0.582 & 0.582 & 0.595 & 0.576 & 0.889 & 0.893 & 0.891 & 0.889 \\
			LPMLX & 0.053 & 0.047 & 0.049 & 0.052 & 0.047 & 0.053 & 0.050 & 0.050 & 0.612 & 0.614 & 0.619 & 0.610 & 0.915 & 0.919 & 0.919 & 0.913 \\
			NP    & 0.056 & 0.055 & 0.054 & 0.055 & 0.050 & 0.057 & 0.052 & 0.054 & 0.633 & 0.627 & 0.633 & 0.629 & 0.916 & 0.919 & 0.918 & 0.915 \\
			\midrule
			\multicolumn{17}{c}{Panel B: DGP (ii)} \\ 
			NA    & 0.038 & 0.039 & 0.039 & 0.038 & 0.045 & 0.039 & 0.040 & 0.045 & 0.572 & 0.571 & 0.579 & 0.574 & 0.878 & 0.882 & 0.879 & 0.879 \\
			LP    & 0.041 & 0.044 & 0.045 & 0.041 & 0.044 & 0.043 & 0.039 & 0.042 & 0.704 & 0.708 & 0.710 & 0.700 & 0.953 & 0.955 & 0.956 & 0.955 \\
			ML    & 0.044 & 0.043 & 0.048 & 0.041 & 0.047 & 0.045 & 0.043 & 0.044 & 0.661 & 0.660 & 0.664 & 0.655 & 0.931 & 0.931 & 0.933 & 0.935 \\
			LPML  & 0.047 & 0.046 & 0.048 & 0.044 & 0.047 & 0.046 & 0.041 & 0.046 & 0.723 & 0.714 & 0.720 & 0.714 & 0.964 & 0.963 & 0.965 & 0.964 \\
			MLX   & 0.052 & 0.050 & 0.052 & 0.049 & 0.048 & 0.046 & 0.045 & 0.045 & 0.703 & 0.710 & 0.708 & 0.704 & 0.946 & 0.949 & 0.946 & 0.951 \\
			LPMLX & 0.056 & 0.054 & 0.054 & 0.051 & 0.052 & 0.048 & 0.046 & 0.048 & 0.761 & 0.761 & 0.766 & 0.754 & 0.972 & 0.972 & 0.972 & 0.974 \\
			NP    & 0.060 & 0.060 & 0.062 & 0.058 & 0.055 & 0.052 & 0.047 & 0.051 & 0.770 & 0.771 & 0.773 & 0.765 & 0.973 & 0.974 & 0.972 & 0.974 \\
			\bottomrule
		\end{tabularx}%
	\end{table}%
	
	\subsection{Practical recommendations}
	When $X$ is finite-dimensional, we suggest using the LPMLX adjustment in which the logistic model includes interaction terms and the regression coefficients are allowed to depend on $(\tau,a,s)$. When $X$ is high-dimensional, we suggest using the logistic Lasso to estimate the regression adjustment.\footnote{The relevant theory and simulation results on high-dimensional covariates are provided in Section \ref{sec:reg} of the Online Supplement.} 
	
	\section{Empirical Application}
	\label{sec:app}
	
	Undersaving has been found to have important individual and social welfare consequences \citep{karlan2014}. Does expanding access to bank accounts for the poor lead to an overall increase in savings? To answer the question, \citet{dupas2018} conducted a covariate-adaptive randomized experiment in Uganda, Malawi, and Chile to study the impact of a bank account subsidy on savings. In their paper, the authors examined the ATEs as well as the QTEs of the subsidy. This section reports an application of our methods to the same dataset to examine the QTEs of the subsidy on household total savings in Uganda.  
	
	The sample consists of 2160 households in Uganda.\footnote{We filter out observations with missing values. Our final sample contains 1952 households.} Within each of 41 strata by gender, occupation, and bank branch, 50 percent of households in the sample were randomly assigned to receive the bank account subsidy and the rest of the sample were in the control group. This is a block stratified randomization design with 41 strata, which satisfies Assumption \ref{ass:assignment1} in Section \ref{sec:setup}. The target fraction of the treated units is 1/2. It is trivial to see that statements (i), (ii), and (iii) in Assumption \ref{ass:assignment1} are satisfied. Because $\max_{s\in\mathcal{S}}|\frac{D_{n}(s)}{n(s)}|\approx 0.056$, it is reasonable to claim that Assumption \ref{ass:assignment1}(iv) is also satisfied in our analysis. 
	
	After the randomization and the intervention, the authors conducted 3 rounds of follow-up surveys in Uganda (See \cite{dupas2018} for a detailed description). In this section, we focus on the first round follow up survey to examine the impact of the bank account subsidy on total savings. 
	
	Tables \ref{tab:emp_qte1x} and \ref{tab:emp_qte4x} present the QTE estimates and their standard errors (in parentheses) estimated by different methods at quantile indices 0.25, 0.5, and 0.75. The description of these estimators is similar to that in Section \ref{sec:sim}.\footnote{Specifically, we have:
		\begin{enumerate}[label=(\roman*)]
			\item NA: the estimator with no adjustments. 
			\item LP: the linear probability model. When there is only one auxiliary regressor, $H_{i}= (1,X_{1i})^\top$, and when there are four auxiliary regressors, $H_{i}= (1,X_{1i},X_{2i},X_{3i},X_{4i})^\top$, where $X_{1i},X_{2i},X_{3i},X_{4i},$ represent four covariates used in the regression adjustment.
			\item ML: the logistic probability model with regressor $H_{i}$, where $H_{i}$ is the same as that in the LP model. 
			\item LPML: the further improved logistic probability model with regressor $H_{i}$, where $H_{i}$ is the same as that in the LP model. 
			\item MLX: the logistic probability model with interaction terms. MLX is only be applied to the case when there are four auxiliary regressors with $H_{i}= (1,X_{1i},X_{2i},X_{3i},X_{4i}, X_{1i}X_{2i},X_{2i}X_{3i})^\top$.
			\item LPMLX: the further improved logistic probability model with interaction terms. LPMLX is only be applied to the case when there are four auxiliary regressors with the same $H_{i}$ as that used in the MLX model.
			\item NP: the nonparametric logistic probability model with regressor $H_{h_n}$. NP is only applied to the case when there are four auxiliary regressors with $H_{h_n} =(1,X_{1i},X_{2i},X_{3i},X_{4i},X_{1i}X_{2i},X_{2i}X_{3i},X_{1i}1\{X_{1i}>t_1\},X_{2i}1\{X_{2i}>t_2\},X_{1i} 1\{X_{1i}>t_1\}X_{2i}1\{X_{2i}>t_2\})^\top$ where $t_1$ and $t_2$ are the sample medians of $\{X_{1i}\}_{i \in [n]}$ and $\{X_{2i}\}_{i \in [n]}$, respectively.
			\item Lasso: the logistic probability model with regressor $H_{p_n}$  and post-Lasso coefficient estimator $\hat{\theta}_{a,s}^{post}(\tau)$. Lasso is only applied to the case when there are four auxiliary regressors with $H_{p_n}(X_i) =(1,X_{1i},X_{2i},X_{3i},X_{4i},X_{1i}^2,X_{2i}^2,X_{3i}^2,X_{1i}X_{2i},X_{2i}X_{3i},X_{1i}1\{X_{1i}>t_1\},X_{2i}1\{X_{2i}>t_2\},X_{1i} 1\{X_{1i}>t_1\}X_{2i}1\{X_{2i}>t_2\})^\top$. The post-Lasso estimator $\hat{\theta}_{a.s}^{post}(\tau)$ is defined in \eqref{eq:post}. The choice of tuning parameter and the estimation procedure are detailed in Section \ref{sec:aux_imp}. 
	\end{enumerate}} In the analysis, we focus on two sets of additional baseline variables: baseline value of total savings only (one auxiliary regressor) and baseline value of total savings, household size, age, and married female dummy (four auxiliary regressors). The first set of regressors follows \cite{dupas2018}. The second one is used to illustrate all the methods discussed in the paper. Tables \ref{tab:emp_qte1x} and \ref{tab:emp_qte4x} report the results with one and four auxiliary regressors, respectively.
	
	\newcolumntype{L}{>{\raggedright\arraybackslash}X}
	\newcolumntype{C}{>{\centering\arraybackslash}X}
	
	\begin{table}[t]
		\centering
		\caption{QTEs on Total Savings (one auxiliary regressor) }
		\vspace{1ex}
		\begin{tabularx}{1\textwidth}{LCCCC}
			\toprule
			& NA & LP & ML & LPML  \\ 
			\midrule
			25\% & 1.105 & 1.105 & 1.105 & 1.105\\
			& (0.564) & (0.564) & (0.470) & (0.470)\\
			50\% & 3.682 & 3.682 & 3.682 & 3.682\\
			& (1.010) & (1.080) & (1.146) & (1.033)\\
			75\% & 7.363 & 9.204 & 9.204 & 9.204\\
			& (3.757) & (4.227) & (3.616) & (3.757)\\
			\bottomrule
		\end{tabularx} \\ 
		
		\vspace{-1ex}
		\justify
		{\footnotesize Notes: The table presents the QTE estimates of the effect of the bank account subsidy on household total savings at quantiles 25\%, 50\%, and 75\% when only one auxiliary regressor (baseline value of total savings) is used in the regression adjustment models. Standard errors are in parentheses. }
		\label{tab:emp_qte1x}
	\end{table}
	
	\begin{table}[t]
		\centering
		\caption{QTEs on Total Savings (four auxiliary regressors) }
		\vspace{1ex}
		\begin{tabularx}{1\textwidth}{LCCCCCCCC}
			\toprule
			& NA & LP & ML & LPML & MLX & LPMLX & NP & Lasso  \\ 
			\midrule
			25\% & 1.105 & 1.473 & 1.105 & 1.105 & 1.105 & 1.105 & 1.105 & 1.105\\
			& (0.564) & (0.564) & (0.564) & (0.564) & (0.357) & (0.319) & (0.188) & (0.564)\\
			50\% & 3.682 & 3.682 & 3.682 & 3.682 & 3.682 & 3.682 & 3.682 & 3.682\\
			& (1.010) & (1.033) & (0.939) & (0.939) & (0.958) & (1.033) & (0.939) & (0.939)\\
			75\% & 7.363 & 8.100 & 7.363 & 7.363 & 7.363 & 7.363 & 7.363 & 7.363\\
			& (3.757) & (3.757) & (3.757) & (3.569) & (3.757) & (3.663) & (3.663) & (3.757)\\
			\bottomrule
		\end{tabularx} \\ 
		
		\vspace{-1ex}
		\justify
		{\footnotesize Notes: The table shows QTE estimates of the effect of the bank account subsidy on household total savings at quantiles 25\%, 50\%, and 75\% when four auxiliary regressors (baseline value of total savings, household size, age, and married female dummy) are used in the regression adjustment models. Standard errors are in parentheses. }
		\label{tab:emp_qte4x}
	\end{table}
	
	\newcolumntype{B}{>{\hsize=1.44\hsize \raggedright\arraybackslash}X}
	\newcolumntype{S}{>{\hsize=.89\hsize \centering\arraybackslash}X}
	
	\begin{table}[t]
		\centering
		\caption{Test for the Difference between Two QTEs on Total Savings}
		\vspace{1ex}
		\begin{tabularx}{1\textwidth}{BSSSSSSSS}
			\toprule
			& NA & LP & ML & LPML & MLX & LPMLX & NP & Lasso  \\ 
			\midrule
			$50\%-25\%$ & 2.577 & 2.209 & 2.577 & 2.577 & 2.577 & 2.577 & 2.577 & 2.577\\
			& (0.939) & (1.104) & (0.939) & (0.939) & (0.958) & (1.033) & (0.845) & (0.911)\\
			$75\%-50\%$ & 3.682 & 4.418 & 3.682 & 3.682 & 3.682 & 3.682 & 3.682 & 3.682\\
			& (3.757) & (3.663) & (3.663) & (3.287) & (3.475) & (3.287) & (3.663) & (3.757)\\
			$75\%-25\%$ & 6.259 & 6.627 & 6.259 & 6.259 & 6.259 & 6.259 & 6.259 & 6.259\\
			& (3.851) & (3.757) & (3.757) & (3.695) & (3.588) & (3.569) & (3.287) & (3.832)\\
			\bottomrule
		\end{tabularx} \\ 
		
		\vspace{-1ex}
		\justify
		{\footnotesize Notes: The table presents tests for the difference between two QTE estimates of the effect of the bank account subsidy on household total savings when there are four auxiliary regressors: baseline value of total savings, household size, age, and married female dummy. Standard errors are in parentheses.}
		\label{tab:emp_qte4x_dif}
	\end{table}
	
	\begin{figure}[t]
		\includegraphics[width=\textwidth]{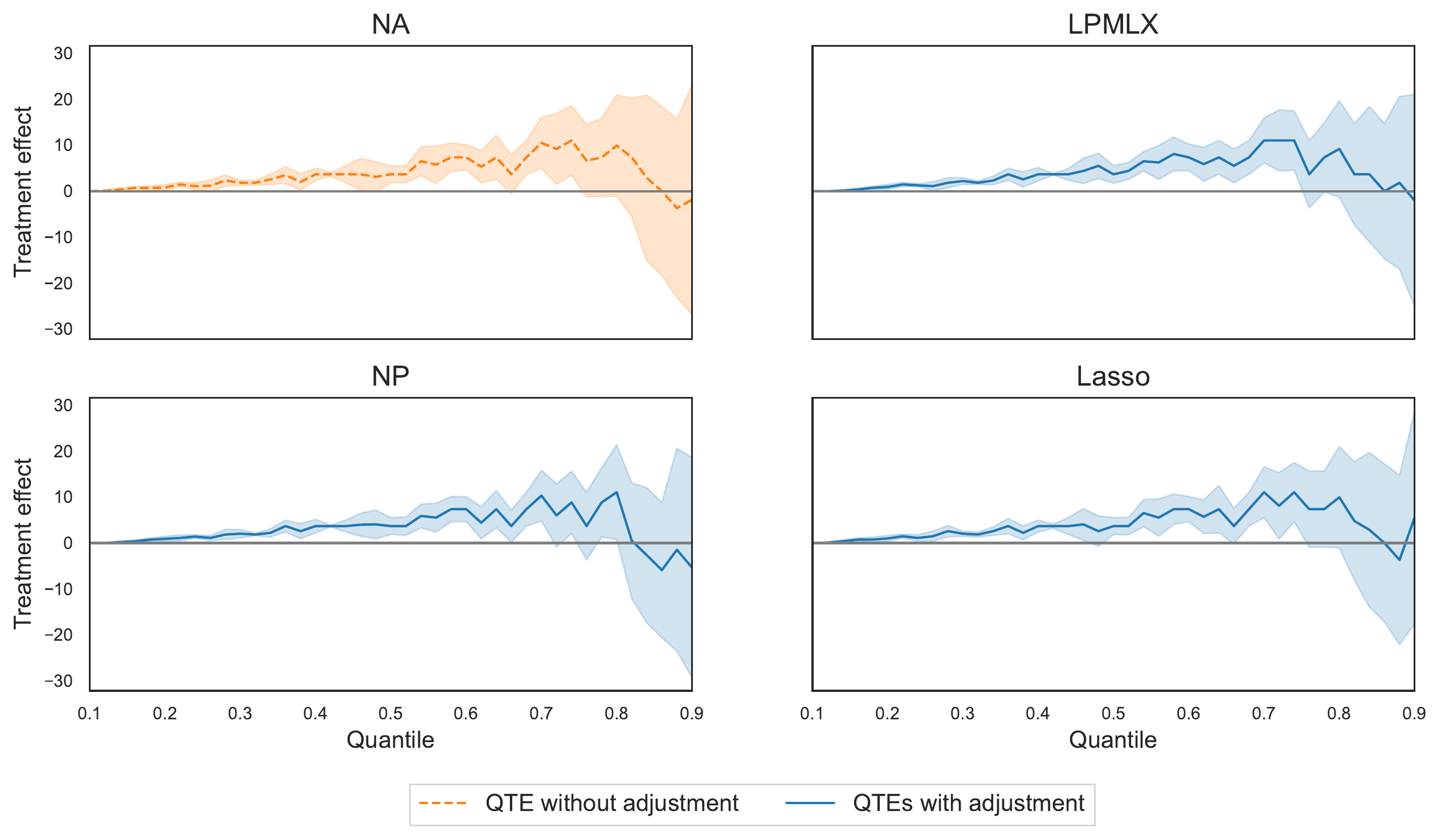}
		\centering
		\caption{Quantile Treatment Effects on the Distribution of Total Savings }
		\label{fig:band}
		
		\vspace{-1ex}
		\justify
		\footnotesize Notes: The graphs in each panel of the figure plot the QTE estimates of the effect of the bank account subsidy on the distribution of household total savings when there are four auxiliary regressors: baseline value of total savings, household size, age, and married female dummy. The shadowed areas display 95\% confidence regions.
	\end{figure}
	
	The results in Tables \ref{tab:emp_qte1x}-\ref{tab:emp_qte4x} prompt two observations. First, consistent with the theoretical and simulation results, the standard errors for the regression-adjusted QTEs are mostly lower than those for the QTE estimate without adjustment. This observation holds for most of the specification and estimation methods of the auxiliary regressions.\footnote{The efficiency gain from the ``NP" adjustment is not the only reason for its small standard error at the 25\% QTE. Another reason is that the treated outcomes around this percentile themselves do not have much variation.} For example, in Table \ref{tab:emp_qte1x}, the standard errors for the ``LPML'' QTE estimates are 16.7\% less than those for the QTE estimate without adjustment at the 25th percentile, respectively. For another example in Table \ref{tab:emp_qte4x}, at the 25th percentile, the standard error for the ``LPMLX'' QTE estimates is 43.4\% less than that for the QTE estimate without adjustment. At the median, the standard error for the ``LPML'' QTE estimates is 7\% less than that for the QTE estimate without adjustment.
	
	Second, there is substantial heterogeneity in the impact of the subsidy on total savings. In particular, we observe larger effects as the quantile indexes increase, which is consistent with the findings in \cite{dupas2018}. For example, Table \ref{tab:emp_qte4x} shows that, although the treatment effects are all positive and significantly different from zero at quantiles 25\%, 50\%, and 75\%, the magnitude of the effects increases by over 200\% from the 25th percentile to the median and by around 100\% from the median to the 75th percentile. 
	
	The second observation suggests that the heterogeneous effects of the subsidy on savings are sizable economically. To evaluate whether these effects are statistically significant, we report statistical tests for the heterogeneity of the QTEs in Table \ref{tab:emp_qte4x_dif}. Specifically, we test the null hypotheses: $H_0: q(0.5) - q(0.25)=0$, $H_0: q(0.75) - q(0.5)=0$, and $H_0: q(0.75) - q(0.25)=0$. Table \ref{tab:emp_qte4x_dif} shows that only the difference between the 50th and 25th QTEs is statistically significant at the 5\% significance level.
	
	How does the variation in the impact of the subsidy appear across the distribution of total savings? The QTEs on the distribution of savings are plotted in Figure \ref{fig:band}, where the shaded areas represent the 95\% confidence region. The figure shows that the QTEs seem insignificantly different from zero below about the 20th percentile. At higher levels to near the 80th percentile, the treatment group savings exceed the control group savings at an accelerated rate, yielding increasingly significantly positive QTEs. Beyond the 80th percentile, the QTEs again become insignificantly different from zero. These findings point to notable distributional heterogeneity in the impact of the subsidy on savings. 
	
	\section{Conclusion}
	\label{sec:concl}
	This paper proposes the use of auxiliary regression to incorporate additional covariates into estimation and inference relating to unconditional QTEs under CARs. The auxiliary regression model may be estimated parametrically, nonparametrically, or via regularization if there are high-dimensional covariates. Both estimation and bootstrap inference methods are robust to potential misspecification of the auxiliary model and do not suffer from the conservatism due to the CAR. It is shown that efficiency can be improved when including extra covariates. When the auxiliary regression is correctly specified, the regression-adjusted estimator further achieves the minimum asymptotic variance. In both the simulations and the empirical application, the proposed regression-adjusted QTE estimator performs well. These results and the robustness of the methods to auxiliary model misspecification reflect the aphorism widespread in scientific modeling that all models may be wrong, but some are useful.\footnote{The aphorism ``all models are wrong, but some are useful" is often attributed to the statistician George \cite{Box1976}. But the notion has many antecedents, including a particularly apposite remark made in 1947 by John \citet{Neumann2019} in an essay on the empirical origins of mathematical ideas to the effect that ``truth ... is much too complicated to allow anything but approximations''.} 
	
	\section*{Acknowledgements}
	
	We thank the Managing Editor, Elie Tamer, the Associate Editor and three anonymous referees for many useful comments that helped to improve this paper. We are also grateful to Michael Qingliang Fan and seminar participants from the 2022 Econometric Society Australasian Meeting, the 2021 Nanyang Econometrics Workshop, University of California, Irvine, and University of Sydney for their comments. 
	
	\paragraph{Funding:} Yichong Zhang acknowledges financial support from the Singapore Ministry of Education under Tier 2 grant No.~MOE2018-T2-2-169, the NSFC under the grant No.~72133002, and a Lee Kong Chian fellowship. Peter C.~B.~Phillips acknowledges support from NSF Grant No.~SES 18-50860, a Kelly Fellowship at the University of Auckland, and a Lee Kong Chian Fellowship. Yubo Tao acknowledges the financial support from the Start-up Research Grant of University of Macau (SRG2022-00016-FSS). Liang Jiang acknowledges support from MOE (Ministry of Education in China) Project of Humanities and Social Sciences (Project No.18YJC790063). 

\newpage
\appendix

\section{Regularization Method for Regression Adjustments}
\label{sec:reg}
This section considers estimation of $m_a(\tau,s,X)$ in a high-dimensional environment. Let $H_{p_n}(X_i)$ be the regressors with dimension $p_n$, which may exceed the sample size. When the number of raw controls is comparable to or exceeds the sample size, we can just let $H_{p_n}(X_i) = X_i$. On the other hand, $H_{p_n}(X_i)$ may be composed of a large dictionary of sieve bases derived from a fixed dimensional vector $X_i$ through suitable transformations such as powers and interactions. Thus, high dimensionality in $H_{p_n}(X_i)$ can arise from the desire to flexibly approximate nuisance functions.  In our approach we follow \cite{BCFH13} and implement a logistic regression with $\ell_1$-penalization. In their notation we view $m_a(\tau,s,x)$ as a function of $q_a(\tau)$, i.e., $m_a(\tau,s,x) = \tau - \mathcal{M}_a(q_a(\tau),s,x)$, where $\mathcal{M}_a(q,s,x) = \mathbb{P}(Y_i(a) \leq q |S_i=s,X_i=x)$. We estimate $\mathcal{M}_a( q_a(\tau),s,x)$ as $\lambda(H_{p_n}(X_i)^\top \hat{\theta}_{a,s}^\textit{HD}(\hat{q}_a(\tau)))$, where $\hat{q}_a(\tau)$ is defined in Assumption \ref{ass:qhat},
\begin{align}
\hat{\theta}_{a,s}^\textit{HD}(q) = & \argmin_{\theta_a} \frac{-1}{n_{a}(s)}\sum_{i \in I_{a}(s)} \biggl[1\{Y_i \leq q\}\log(\lambda(H_{p_n}(X_i)^\top \theta_a)) \notag \\
& \qquad \qquad \qquad \qquad \qquad \qquad + 1\{Y_i > q\}\log(1-\lambda(H_{p_n}(X_i)^\top \theta_a))\biggr] + \frac{\varrho_{n,a}(s)}{n_a(s)}||\hat{ \Omega} \theta_a||_1,
\label{eq:thetahat^hd}
\end{align}
$\varrho_{n,a}(s)$ is a tuning parameter, and $\hat{ \Omega} =  \diag(\hat{ \omega}_1,\cdots,\hat{ \omega}_{p_n})$ is a diagonal matrix of data-dependent penalty loadings. We specify $\varrho_{n,a}(s)$ and $\hat{ \Omega}$ in Section \ref{sec:compute}. Post-Lasso estimation is also considered. Let $\hat{\mathbb{S}}_{a,s}(q)$ be the support of $\hat{\theta}_{a,s}^\textit{HD}(q) = \{h \in [p_n]: \hat{\theta}_{a,s,h}^\textit{HD}(q) \neq 0\}$, where $\hat{\theta}_{a,s,h}^\textit{HD}(q)$ is the $h$th coordinate of $\hat{\theta}_{a,s}^\textit{HD}(q)$. We can complement $\hat{\mathbb{S}}_{a,s}(q)$ with additional variables in $\hat{\mathbb{S}}^+_{a,s}(q)$ that researchers want to control for and define the enlarged set of variables as  $\tilde{\mathbb{S}}_{a,s}(q) = \hat{\mathbb{S}}_{a,s}(q) \cup \hat{\mathbb{S}}^+_{a,s}(q)$. We compute the post-Lasso estimator $\hat{\theta}_{a,s}^{post}(q)$ as
\begin{align}
\hat{\theta}_{a,s}^{post}(q) = & \argmin_{\theta_a \in \tilde{\mathbb{S}}_{a,s}(q)} \frac{-1}{n_{a}(s)}\sum_{i \in I_{a}(s)} \biggl[1\{Y_i \leq q\}\log(\lambda(H_{p_n}(X_i)^\top \theta_a)) \notag \\
& \qquad \qquad \qquad \qquad \qquad \qquad + 1\{Y_i > q\}\log(1-\lambda(H_{p_n}(X_i)^\top \theta_a))\biggr].
\label{eq:post}
\end{align}
Finally, we compute the auxiliary model as 
\begin{align}
\widehat{m}_a(\tau,s,X_i) = \lambda(H_{p_n}^\top(X_i)\hat{\theta}_{a,s}^\textit{HD}(\hat{q}_a(\tau))) \quad \text{or} \quad 
\widehat{m}_a(\tau,s,X_i) = \lambda(H_{p_n}^\top(X_i)\hat{\theta}_{a,s}^{post}(\hat{q}_a(\tau))). 
\label{eq:hd}
\end{align}
We refer to the QTE estimator with the regularized adjustment as the HD estimator. Note that we use the estimator $\hat{q}_a(\tau)$ of $q_a(\tau)$ in \eqref{eq:hd}, where $\hat{q}_a(\tau)$ satisfies Assumption \ref{ass:qhat}. All the analysis in this section takes account of the fact that $\hat{q}_a(\tau)$ instead of $q_a(\tau)$ is used. 
\begin{ass}
	\begin{enumerate}[label=(\roman*)]
		\item Let $\mathcal{Q}^\eps_a = \{q: \inf_{\tau \in \Upsilon}|q - q_a(\tau)|\leq \eps\}$ for $a=0,1$. Suppose $\mathbb{P}(Y_i(a) \leq q|S_i=s,X_i) = \lambda(H_{p_n}(X_i)^\top \theta^\textit{HD}_{a,s}(q)) + r_a(q,s,X_i)$ such that $\sup_{a=0,1,q \in \mathcal{Q}^\eps_a,s\in \mathcal{S}}||\theta^\textit{HD}_{a,s}(q)||_0 \leq h_n$.
		\item Suppose $\sup_{i\in [n]}||H_{p_n}(X_i)||_\infty \leq \zeta_n$ and $\sup_{h \in [p_n]}\mathbb{E}(|H_{p_n,h}^d(X_i)||S_i=s) < \infty$ for $d>2$.
		\item Suppose
		$$\sup_{a = 0,1, q \in\mathcal{Q}^\eps_a, s \in \mathcal{S}}\frac{1}{n_a(s)}\sum_{i \in I_a(s)}r^2_a(q,s,X_i) = O_p(h_n \log(p_n)/n),$$
		$$\sup_{a = 0,1, q \in\mathcal{Q}^\eps_a, s \in \mathcal{S}}\mathbb{E}(r^2_a(q,s,X_i)|S_i=s) = O(h_n \log(p_n)/n),$$
		and
		$$\sup_{a = 0,1, q \in\mathcal{Q}^\eps_a, s \in \mathcal{S},x \in \mathcal{X}}|r_a(q,s,X)| = O(\sqrt{\xi_n^2 h_n^2 \log(p_n)/n}).$$
		\item $\frac{\log(p_n)\xi_n^2 h_n^2}{n} \rightarrow 0$, $\frac{\log^2(p_n)\log^2(n) h_n^2}{n} \rightarrow 0$, $\sup_{a=0,1,q \in \mathcal{Q}_a^\eps, s\in \mathcal{S}} |\hat{\mathbb{S}}^+_{a,s}(q)| = O_p(h_n)$, where $|\hat{\mathbb{S}}^+_{a,s}(q)|$ denotes the number of elements in $\hat{\mathbb{S}}^+_{a,s}(q)$.
		\item There exists a constant $c \in (0,0.5)$ such that
		\begin{align*}
		c\leq & \inf_{a = 0,1, s \in \mathcal{S},\tau \in \Upsilon, x \in \Supp(X)}\mathbb{P}(Y_i(a)\leq q_a(\tau)|S_i=s,X_i=x) \\
		\leq & \sup_{a = 0,1, s \in \mathcal{S},\tau \in \Upsilon, x \in \Supp(X)}\mathbb{P}(Y_i(a)\leq q_a(\tau)|S_i=s,X_i=x)  \leq 1-c.
		\end{align*}
		\item Let $\ell_n$ be a sequence that diverges to infinity. Then, there exist two constants $\kappa_1$ and $\kappa_2$ such that with probability approaching one,
		\begin{align*}
		0< \kappa_1 \leq & \inf_{a = 0,1, s\in \mathcal{S}, ||v||_0 \leq h_n \ell_n} \frac{v^T \left(\frac{1}{n_a(s)}\sum_{i \in I_a(s)}H_{p_n}(X_i)H_{p_n}(X_i)^\top  \right)v }{||v||_2^2} \\
		\leq & \sup_{a = 0,1, s\in \mathcal{S}, ||v||_0 \leq h_n \ell_n} \frac{v^T \left(\frac{1}{n_a(s)}\sum_{i \in I_a(s)}H_{p_n}(X_i)H_{p_n}(X_i)^\top  \right)v }{||v||_2^2} \leq \kappa_2 < \infty,
		\end{align*}
		and
		\begin{align*}
		0< \kappa_1 \leq & \inf_{a = 0,1, s \in \mathcal{S}, ||v||_0 \leq h_n \ell_n} \frac{v^T \mathbb{E}(H_{p_n}(X_i)H_{p_n}(X_i)^\top |S_i=s)v }{||v||_2^2} \\
		\leq & \sup_{a = 0,1, s\in \mathcal{S}, ||v||_0 \leq h_n \ell_n} \frac{v^T \mathbb{E}(H_{p_n}(X_i)H_{p_n}(X_i)^\top |S_i=s)v }{||v||_2^2} \leq \kappa_2 < \infty,
		\end{align*}
		where $||v||_0$ denotes the number of nonzero components in $v$.
		\item For $a=0,1$, let $\varrho_{n,a}(s) = c \sqrt{n_a(s)}\Phi^{-1}\left( 1-\frac{0.1}{4\log(n_a(s))p_n}\right)$ where $\Phi(\cdot)$ is the standard normal CDF and $c>0$ is a constant.
	\end{enumerate}
	\label{ass:hd}
\end{ass}
Assumption \ref{ass:hd} is standard in the literature and we refer interested readers to \cite{BCFH13} for more discussion. Assumption \ref{ass:hd}(i) implies the logistic model is approximately correctly specified. As the approximation is assumed to be sparse, the condition is not innocuous in the high-dimensional setting. As our method is valid even when the auxiliary model is misspecified, we conjecture that Assumption \ref{ass:hd}(i) can be relaxed, which links to the recent literature on the study of regularized estimation in the high-dimensional setting under misspecification: see, for example, \cite{BWZ19} and \cite{T20} and the references therein. An interesting topic for future work is to study misspecification-robust high-dimensional estimators of the conditional probability model and their use to adjust the QTE estimator under CAR based on \eqref{eq:q1adj} and \eqref{eq:q0adj}. The following theorem shows that all the estimation and inference results in Theorems \ref{thm:est} and \ref{thm:boot} hold for the HD estimator.

\begin{thm}
	Denote $\hat{q}^\textit{HD}(\tau)$ and $\hat{q}^\textit{HD,w}(\tau)$ as the $\tau$th QTE estimator and its multiplier bootstrap counterpart defined in Sections \ref{sec:est} and \ref{sec:boot}, respectively, with $\overline{m}_a(\tau,S_i,X_i) = m_a(\tau,S_i,X_i)$ and $\widehat{m}_a(\tau,S_i,X_i)$ defined in \eqref{eq:hd}. Further suppose Assumptions \ref{ass:assignment1}, \ref{ass:tau}, \ref{ass:weight}, \ref{ass:qhat}, and \ref{ass:hd} hold. Then, Assumptions \ref{ass:mhat} and \ref{ass:mhatw} hold, which further imply Theorems \ref{thm:est} and \ref{thm:boot} hold for $\hat{q}^\textit{HD}(\tau)$ and $\hat{q}^\textit{HD,w}(\tau)$, respectively. In addition, for any finite-dimensional quantile indices $(\tau_1,\cdots,\tau_K)$, the covariance matrix of $(\hat{q}^\textit{HD}(\tau_1),\cdots,\hat{q}^\textit{HD}(\tau_K))$ achieves the minimum (in the matrix sense) as characterized in Theorem \ref{thm:est}. 
	\label{thm:hd}
\end{thm}

\section{Practical Guidance and Computation} 
\label{sec:compute}

\subsection{Procedures for estimation and bootstrap inference}
\label{sec:procedure}
We can compute $(\hat{q}_1^{adj}(\tau),\hat{q}_0^{adj}(\tau))$ by solving the subgradient conditions of \eqref{eq:q1adj} and \eqref{eq:q0adj}, respectively. Specifically, we have $(\hat{q}_1^{adj}(\tau),\hat{q}_0^{adj}(\tau)) = (Y_{i_1},Y_{i_0})$ such that $A_{i_1} = 1$, $A_{i_0} = 0$,
\begin{align}
& \tau\left(\sum_{i =1}^n \frac{A_i}{\hat{\pi}(S_i)}\right) - \sum_{i =1}^n \left( \frac{(A_i - \hat{\pi}(S_i))}{\hat{\pi}(S_i)}\widehat{m}_1(\tau,S_i,X_i)\right) \notag \\
& \geq \sum_{i =1}^n \frac{A_i}{\hat{\pi}(S_i)}1\{ Y_i < Y_{i_1}\} \notag \\
& \geq  \tau  \left(\sum_{i =1}^n \frac{A_i}{\hat{\pi}(S_i)}\right) - \frac{1}{\hat{\pi}(S_{i_1})} - \sum_{i =1}^n \left( \frac{(A_i - \hat{\pi}(S_i))}{\hat{\pi}(S_i)}\widehat{m}_1(\tau,S_i,X_i)\right),
\label{eq:q1adj_gradient}
\end{align}
and
\begin{align}
& \tau\left(\sum_{i =1}^n \frac{1-A_i}{1-\hat{\pi}(S_i)}\right) + \sum_{i =1}^n \left( \frac{(A_i - \hat{\pi}(S_i))}{1-\hat{\pi}(S_i)}\widehat{m}_0(\tau,S_i,X_i)\right) \notag \\
& \geq \sum_{i =1}^n \frac{1-A_i}{1-\hat{\pi}(S_i)}1\{ Y_i < Y_{i_0}\} \notag \\
& \geq  \tau  \left(\sum_{i =1}^n \frac{1-A_i}{1-\hat{\pi}(S_i)}\right) - \frac{1}{1-\hat{\pi}(S_{i_0})} + \sum_{i =1}^n \left( \frac{(A_i - \hat{\pi}(S_i))}{1-\hat{\pi}(S_i)}\widehat{m}_0(\tau,S_i,X_i)\right).
\label{eq:q0adj_gradient}
\end{align}
We note $(i_1,i_0)$ are uniquely defined as long as all the inequalities in \eqref{eq:q1adj_gradient} and \eqref{eq:q0adj_gradient} are strict, which is usually the case. If we have
\begin{align*}
\tau\left(\sum_{i =1}^n \frac{A_i}{\hat{\pi}(S_i)}\right) - \sum_{i =1}^n \left( \frac{(A_i - \hat{\pi}(S_i))}{\hat{\pi}(S_i)}\widehat{m}_1(\tau,S_i,X_i)\right) = \sum_{i =1}^n \frac{A_i}{\hat{\pi}(S_i)}1\{ Y_i \leq Y_{i_1}\},
\end{align*}
then both $i_1$ and $i_1'$ satisfy \eqref{eq:q1adj_gradient}, where $i_1'$ is the index such that $A_{i_1'} = 1$ and $Y_{i_1'}$ is the smallest observation in the treatment group that is larger than $Y_{i_1}$. In this case, we let $\hat{q}_1^{adj}(\tau) = Y_{i_1}$.\footnote{In this case, any value that belongs to $[Y_{i_1},Y_{i_1'}]$ can be viewed as a solution. Because $|Y_{i_1}-Y_{i_1'}| = O_p(1/n)$, all choices are asymptotically equivalent.} Similarly, by solving the subgradient conditions of \eqref{eq:q1w_gradient} and \eqref{eq:q0w_gradient}, we have  $(\hat{q}_1^{w}(\tau),\hat{q}_0^{w}(\tau)) = (Y_{i_1^w},Y_{i_0^w})$ such that $A_{i_1^w} = 1$, $A_{i_0^w} = 0$,
\begin{align}
& \tau\left(\sum_{i =1}^n \frac{\xi_i A_i}{\hat{\pi}^w(S_i)}\right) - \sum_{i =1}^n \left( \frac{\xi_i (A_i - \hat{\pi}^w(S_i))}{\hat{\pi}^w(S_i)}\widehat{m}_1(\tau,S_i,X_i)\right) \notag \\
& \geq  \sum_{i =1}^n \frac{\xi_i A_i}{\hat{\pi}^w(S_i)}1\{ Y_i < Y_{i_1^w}\} \notag \\
& \geq  \tau  \left(\sum_{i =1}^n \frac{\xi_i A_i}{\hat{\pi}^w(S_i)}\right) - \frac{\xi_{i_1^w}}{\hat{\pi}^w(S_{i_1^w})} - \sum_{i =1}^n \left( \frac{\xi_i (A_i - \hat{\pi}^w(S_i))}{\hat{\pi}^w(S_i)}\widehat{m}_1(\tau,S_i,X_i)\right)
\label{eq:q1w_gradient}
\end{align}
and
\begin{align}
& \tau\left(\sum_{i =1}^n \frac{\xi_i (1-A_i)}{1-\hat{\pi}^w(S_i)}\right) + \sum_{i =1}^n \left( \frac{\xi_i (A_i - \hat{\pi}(S_i))}{1-\hat{\pi}^w(S_i)}\widehat{m}_0(\tau,S_i,X_i)\right) \notag \\
& \geq \sum_{i =1}^n \frac{\xi_i (1-A_i)}{1-\hat{\pi}^w(S_i)}1\{ Y_i < Y_{i_0^w}\} \notag \\
& \geq \tau  \left(\sum_{i =1}^n \frac{\xi_i (1-A_i)}{1-\hat{\pi}^w(S_i)}\right) - \frac{\xi_{i_0^w}}{1-\hat{\pi}^w(S_{i_0})} + \sum_{i =1}^n \left( \frac{\xi_i (A_i - \hat{\pi}^w(S_i))}{1-\hat{\pi}^w(S_i)}\widehat{m}_0(\tau,S_i,X_i)\right).
\label{eq:q0w_gradient}
\end{align}
The inequalities in \eqref{eq:q1w_gradient} and \eqref{eq:q0w_gradient} are strict with probability one if $\xi_i$ is continuously distributed. In this case, $(\hat{q}_1^{w}(\tau),\hat{q}_0^{w}(\tau))$ are uniquely defined with probability one.

We summarize the steps in the bootstrap procedure as follows.
\begin{enumerate}
	\item  Let $\mathcal{G}$ be a set of quantile indices. For $\tau \in \mathcal{G}$, compute $\hat{q}_1(\tau)$ and $\hat{q}_0(\tau)$ following \eqref{eq:q1adj_gradient} and \eqref{eq:q0adj_gradient} with $\widehat{m}_1(\tau,S_i,X_i)$ and $\widehat{m}_0(\tau,S_i,X_i)$ replaced by zero.
	\item Compute $\widehat{m}_a(\tau,S_i,X_i)$ for $a=0,1$ and $\tau \in \mathcal{G}$ using  $\hat{q}_1(\tau)$ and $\hat{q}_0(\tau)$.
	\item Compute the original estimator $\hat{q}^{adj}(\tau) = \hat{q}^{adj}_1(\tau) - \hat{q}^{adj}_0(\tau)$, following \eqref{eq:q1adj_gradient} and \eqref{eq:q0adj_gradient} for $\tau \in \mathcal{G}$.
	\item Let $B$ be the number of bootstrap replications.  For $b \in [B]$, generate $\{\xi_i\}_{i \in [n]}$. Compute $\hat{q}^{w,b}(\tau) = \hat{q}_1^{w,b}(\tau) - \hat{q}_0^{w,b}(\tau)$ for $\tau \in \mathcal{G}$ following \eqref{eq:q1w_gradient} and \eqref{eq:q0w_gradient}. Obtain $\{\hat{q}^{w,b}(\tau)\}_{\tau \in \mathcal{G}}$.
	\item Repeat the above step for $b \in [B]$ and obtain $B$ bootstrap estimates of the QTE, denoted as
	$\{\hat{q}^{w,b}(\tau)\}_{b \in [B],\tau \in \mathcal{G}}$.
\end{enumerate}

\subsection{Bootstrap confidence intervals} \label{subsec:boostrap+confid_int}

Given the bootstrap estimates, we next discuss how to conduct bootstrap inference for the null hypotheses with single, multiple, and a continuum of quantile indices.

\vspace{1mm}
\textbf{Case (1).} We test the single null hypothesis that $\mathcal{H}_0: q(\tau) = \underline{q}$ vs. $q(\tau) \neq \underline{q}$. Set $\mathcal{G} = \{\tau\}$ in the procedures described above and let $\widehat{\mathcal{C}}(\nu)$ and $\mathcal{C}(\nu)$ be the $\nu$th empirical quantile of the sequence  $\{\hat{q}^{w,b}(\tau)\}_{b \in [B]}$ and the $\nu$th standard normal critical value, respectively. Let $\alpha \in (0,1)$ be the significance level. We suggest using the bootstrap estimator to construct the standard error of $\hat{q}^{adj}(\tau)$ as $\hat{\sigma} = \frac{\widehat{\mathcal{C}}(0.975)- \widehat{\mathcal{C}}(0.025)}{\mathcal{C}(0.975) - \mathcal{C}(0.025)}$.  Then the valid confidence interval and Wald test using this standard error are
\begin{align*}
CI(\alpha) = (\hat{q}^{adj}(\tau)+\mathcal{C}(\alpha/2)\hat{\sigma}, \hat{q}^{adj}(\tau)+\mathcal{C}(1-\alpha/2)\hat{\sigma}),
\end{align*}
and $1\big\lbrace\big|\frac{\hat{q}^{adj}(\tau) - \underline{q}}{\hat{\sigma}}\big| \geq \mathcal{C}(1-\alpha/2)\big\rbrace$, respectively.\footnote{It is asymptotically valid to use standard and percentile bootstrap confidence intervals. But in simulations we found that the confidence interval proposed in the paper has better finite sample performance in terms of coverage rates under the null.}

%Further denote the standard and percentile bootstrap confidence intervals as $CI_2$ and $CI_3$, respectively, where
%\begin{align*}
%CI_2(\alpha) = (2\hat{q}(\tau) - \mathcal{Q}(1-\alpha/2), 2\hat{q}(\tau) - \mathcal{Q}(\alpha/2))
%\end{align*}
%and
%\begin{align*}
%CI_3(\alpha) = (\mathcal{Q}(\alpha/2), \mathcal{Q}(1-\alpha/2)).
%\end{align*}
%Theoretically, $CI_1$, $CI_2$, and $CI_3$ are all valid. When $\alpha = 0.05$, $CI_1$, $CI_2$, and $CI_3$ are centered at $\hat{q}(\tau)$, $2\hat{q}(\tau) - \frac1{2}\{\mathcal{Q}(0.975)+ \mathcal{Q}(0.025)\}$, and $\frac1{2}\{\mathcal{Q}(0.975)+ \mathcal{Q}(0.025)\}$, respectively, but share the same length $\mathcal{Q}(0.975)- \mathcal{Q}(0.025)$. In (unreported) simulations, we found that in small samples, $CI_1$ usually has the best size control while $CI_2$ over-rejects and $CI_3$ under-rejects.

\vspace{2mm}
\textbf{Case (2).} We test the null hypothesis that $\mathcal{H}_0: q(\tau_1) - q(\tau_2) = \underline{q}$ vs. $q(\tau_1) - q(\tau_2) \neq \underline{q}$. In this case, we have $\mathcal{G} = \{\tau_1,\tau_2\}$ in the procedure described in Section \ref{sec:procedure}. Further, let $\widehat{\mathcal{C}}(\nu)$ be the $\nu$th empirical quantile of the sequence  $\{\hat{q}^{w,b}(\tau_1) - \hat{q}^{w,b}(\tau_2)\}_{b \in [B]}$, and let $\alpha \in (0,1)$ be the significance level. We suggest using the bootstrap standard error to construct the valid confidence interval and Wald test as
\begin{align*}
CI(\alpha) = (\hat{q}^{adj}(\tau_1)-\hat{q}^{adj}(\tau_2)+\mathcal{C}(\alpha/2)\hat{\sigma}, \hat{q}^{adj}(\tau_1)-\hat{q}^{adj}(\tau_2)+\mathcal{C}(1-\alpha/2)\hat{\sigma}),
\end{align*}
and $1\big\{\big|\frac{\hat{q}^{adj}(\tau_1)-\hat{q}^{adj}(\tau_2) - \underline{q}}{\hat{\sigma}}\big| \geq \mathcal{C}(1-\alpha/2)\big\}$, respectively, where $\hat{\sigma} = \frac{\widehat{\mathcal{C}}(0.975)- \widehat{\mathcal{C}}(0.025)}{\mathcal{C}(0.975) - \mathcal{C}(0.025)}$.

% In the following, $CI_1(\alpha)$, $CI_2(\alpha)$, $CI_3(\alpha)$ denote the standard bootstrap, percentile bootstrap, and bootstrap standard error confidence intervals, respectively.
%\begin{align*}
%CI_3(\alpha) = (2(\hat{q}(\tau_1)-\hat{q}(\tau_2))-C_{1-\alpha/2}\hat{\sigma}, 2(\hat{q}(\tau_1)-\hat{q}(\tau_2))+C_{\alpha/2}\hat{\sigma}),
%\end{align*}

%
%\begin{align*}
%CI_1(\alpha) = (2(\hat{q}(\tau_1)-\hat{q}(\tau_2)) - \mathcal{Q}(1-\alpha/2), 2(\hat{q}(\tau_1)-\hat{q}(\tau_2)) - \mathcal{Q}(\alpha/2)),
%\end{align*}
%\begin{align*}
%CI_2(\alpha) = (\mathcal{Q}(\alpha/2), \mathcal{Q}(1-\alpha/2)),
%\end{align*}
%and

\vspace{2mm}
\textbf{Case (3).} We test the null hypothesis that
$$\mathcal{H}_0: q(\tau) = \underline{q}(\tau)~\forall \tau \in \Upsilon~\text{vs.}~q(\tau) \neq \underline{q}(\tau)~ \exists \tau \in \Upsilon.$$ In theory, we should let $\mathcal{G} = \Upsilon$. In practice, we let $\mathcal{G} = \{\tau_1,\cdots,\tau_G\}$ be a fine grid of $\Upsilon$ where $G$ should be as large as computationally possible. Further, let $\widehat{\mathcal{C}}_\tau(\nu)$ denote the $\nu$th empirical quantile of the sequence  $\{\hat{q}^{w,b}(\tau)\}_{b \in [B]}$ for $\tau \in \mathcal{G}$. Compute the standard error of $\hat{q}^{adj}(\tau)$ as
\begin{align*}
\hat{\sigma}_\tau = \frac{\widehat{\mathcal{C}}_\tau(0.975)- \widehat{\mathcal{C}}_\tau(0.025)}{\mathcal{C}(0.975) - \mathcal{C}(0.025)}.
\end{align*}
The uniform confidence band with an $\alpha$ significance level is constructed as
\begin{align*}
CB(\alpha) = \{ \hat{q}^{adj}(\tau)-\widetilde{\mathcal{C}}_\alpha  \hat{\sigma}_\tau,  \hat{q}^{adj}(\tau)+\widetilde{\mathcal{C}}_\alpha  \hat{\sigma}_\tau: \tau \in \mathcal{G} \},
\end{align*}
where the critical value $\widetilde{\mathcal{C}}_\alpha$ is computed as
\begin{align*}
\widetilde{\mathcal{C}}_\alpha = \inf\left\{z:\frac{1}{B}\sum_{b=1}^B 1\left\{ \sup_{\tau \in \mathcal{G}}\left| \frac{\hat{q}^{w,b}(\tau) - \tilde{q}(\tau)}{\hat{\sigma}_\tau}\right| \leq z \right\}\geq 1-\alpha \right\},
\end{align*}
and $\tilde{q}(\tau)$ is first-order equivalent to $\hat{q}^{adj}(\tau)$ in the sense that $\sup_{\tau \in \Upsilon}|\tilde{q}(\tau) - \hat{q}^{adj}(\tau)|=  o_p(1/\sqrt{n})$. We suggest choosing $\tilde{q}(\tau) = \widehat{\mathcal{C}}_\tau(0.5)$ over other choices such as $\tilde{q}(\tau) = \hat{q}^{adj}(\tau)$ due to its better finite sample performance. We reject $\mathcal{H}_0$ at significance level $\alpha$ if $\underline{q}(\cdot) \notin CB(\alpha).$

\subsection{Computation of Auxiliary Regressions}
\label{sec:aux_imp}

\paragraph{Parametric regressions.}
For the linear probability model, we compute the LP estimator via \eqref{eq:thetahat_lp_est}. For the logistic model, we consider the ML and the LPML estimators. First, we compute %the LG estimator  $\hat{\theta}_{a,s}^{LG}(\tau)$ as in \eqref{eq:thetahat_lg} in the Online Supplement, which minimizes the sample variance of our adjusted QTE estimator $\hat{q}^{adj}(\tau)$.\footnote{When the null hypothesis in Case (2) is tested, the optimal pseudo true value is defined to minimize the asymptotic variance of $\hat{q}^{par}(\tau_1) - \hat{q}^{par}(\tau_2)$, which is not necessarily equal to the pseudo true values that minimize the asymptotic variances of $\hat{q}^{par}(\tau_1)$ and $\hat{q}^{par}(\tau_2)$, i.e., $(\theta_{a,s}^{LG}(\tau_1),\theta_{a,s}^{LG}(\tau_2))$. We propose a  way to compute the new pseudo true value in Section \ref{sec:pseudo_diff} of the Online Supplement.} Second, we propose to compute 
the ML estimator $ \hat{\theta}_{a,s}^\textit{ML}(\tau)$ as in \eqref{eq:ml}, which is the quasi maximum likelihood estimator of a flexible distribution regression. Second, we propose to compute the logistic function values with $\hat{\theta}_{a,s}^\textit{ML}(\tau)$ and treat them as regressors in a linear adjustment to further improve the ML estimate.

\paragraph{Sieve logistic regressions.}
We provide more detail on the sieve basis. Recall $H_{h_n}(x) \equiv (b_{1n}(x), \cdots, b_{h_nn}(x))^\top $, where $\{b_{hn }(\cdot)\}_{h \in [h_n]}$ are $h_n$ basis functions of a linear sieve space, denoted as $\mathbb{B}$.  Given that all $d_x$ elements of $X$ are continuously distributed, the sieve space $\mathbb{B}$ can be constructed as follows.
\begin{enumerate}
	\item For each element $X^{(l)}$ of $X$, $l=1,\cdots,d_x$, let $\mathcal{B}_l$ be the univariate sieve space of dimension $J_n$. One example of $\mathcal{B}_l$ is the linear span of the $J_n$ dimensional polynomials given by
	$$\mathbb{B}_l = \biggl\{\sum_{k=0}^{J_n}\alpha_k x^k, x \in \Supp(X^{(l)}), \alpha_k \in \Re \biggr\};$$
	Another is the linear span of $r$-order splines with $J_n$ nodes given by
	$$\mathbb{B}_l = \biggl\{\sum_{k=0}^{r-1}\alpha_k x^k + \sum_{j=1}^{J_n}b_j[\max(x-t_j,0)]^{r-1}, x \in \Supp(X^{(l)}), \alpha_k, b_j \in \Re \biggr\},$$
	where the grid $-\infty=t_0 \leq t_1 \leq \cdots \leq t_{J_n} \leq t_{J_n+1} = \infty$ partitions $\Supp(X^{(l)})$ into $J_n+1$ subsets $I_j = [t_j,t_{j+1}) \cap \Supp(X^{(l)})$, $j=1,\cdots,J_n-1$, $I_{0} = (t_{0},t_{1}) \cap \Supp(X^{(l)})$, and $I_{J_{n}} = (t_{J_n},t_{J_n+1}) \cap \Supp(X^{(l)})$.
	\item Let $\mathbb{B}$ be the tensor product of $\{\mathcal{B}_l\}_{l=1}^{d_x}$, which is defined as a linear space spanned by the functions $\prod_{l=1}^{d_x} g_l$, where $g_l \in \mathcal{B}_l$. The dimension of $\mathbb{B}$ is then $K \equiv d_x J_n$.
\end{enumerate}

We refer interested readers to \cite{HIR03} and \cite{c07} for more details about the implementation of sieve estimation. Given the sieve basis, we can compute the $\widehat{m}_a(\tau,s,X_i)$ following \eqref{eq:mhat_np1} and \eqref{eq:mhat_np2}.

\paragraph{Logistic regressions with an $\ell_1$ penalization.}
We follow the estimation procedure and the choice of tuning parameter proposed by \cite{BCFH13}. We provide details below for completeness. Recall $\varrho_{n,a}(s) = c \sqrt{n_a(s)}\Phi^{-1}( 1-1/(p_n\log(n_a(s))))$. We set $c=1.1$ following \cite{BCFH13}.  We then implement the following algorithm to estimate $\hat{\theta}_{a,s}^\textit{HD}(\hat{q}_a(\tau))$ for $\tau \in \Upsilon$:
\begin{enumerate}[label=(\roman*)]
	\item Let $\hat{\sigma}_h^{(0)} = \frac{1}{n_a(s)}\sum_{i \in I_{a}(s)}(1\{Y_i \leq \hat{q}_a(\tau)\} - \bar{Y}_{a,s}(\tau))^2H_{p_n,h}^2$ for $h \in [p_n]$, where $\bar{Y}_{a,s}(\tau) = \frac{1}{n_a(s)}\sum_{i \in I_{a}(s)}1\{Y_i \leq \hat{q}_a(\tau)\}$. Estimate
	\begin{align*}
	\hat{\theta}_{a,s}^\textit{HD,0}(\hat{q}_a(\tau)) = & \argmin_{\theta_a} \frac{-1}{n_{a}(s)}\sum_{i \in I_{a}(s)} \biggl[1\{Y_i \leq \hat{q}_a(\tau)\}\log(\lambda(H_{p_n}(X_i)^\top \theta_a)) \notag \\
	& \qquad \qquad \qquad + 1\{Y_i > \hat{q}_a(\tau)\}\log(1-\lambda(H_{p_n}(X_i)^\top \theta_a))\biggr] + \frac{\varrho_{n,a}(s)}{n_a(s)} \sum_{h \in [p_n]} \hat{\sigma}_h^{(0)}|\theta_{a,h}|.
	\end{align*}
	\item For $k = 1,\cdots,K$, obtain $\hat{\sigma}_h^{(k)} = \sqrt{\frac{1}{n}\sum_{i =1}^n (H_{p_n,h}\hat{\eps}_i^{(k)})^2}$, where $\hat{\eps}_i^{(k)} = 1\{Y_i \leq \hat{q}_a(\tau)\} - \lambda(H_{p_n}^T \hat{\theta}_{a,s}^\textit{HD,k-1}(\hat{q}_a(\tau)))$. Estimate
	\begin{align*}
	\hat{\theta}_{a,s}^\textit{HD,k}(\hat{q}_a(\tau)) = & \argmin_{\theta_a} \frac{-1}{n_{a}(s)}\sum_{i \in I_{a}(s)} \biggl[1\{Y_i \leq \hat{q}_a(\tau)\}\log(\lambda(H_{p_n}(X_i)^\top \theta_a)) \notag \\
	& \qquad \qquad \qquad + 1\{Y_i > \hat{q}_a(\tau)\}\log(1-\lambda(H_{p_n}(X_i)^\top \theta_a))\biggr] + \frac{\varrho_{n,a}(s)}{n_a(s)} \sum_{h \in [p_n]} \hat{\sigma}_h^{(k)}|\theta_{a,h}|.
	\end{align*}
	\item Let $\hat{\theta}_{a,s}^\textit{HD}(\hat{q}_a(\tau)) = \hat{\theta}_{a,s}^\textit{HD,K}(\hat{q}_a(\tau))$.
	\item Repeat the above procedure for $\tau \in \mathcal{G}$.
\end{enumerate}

%\item Given the choice of the penalty loading, choose $\lambda = 1.1\Phi^{-1}(1-\gamma/p)n^{-1/2}$ where $\gamma = o(1)$ and $\Phi(\cdot)$ is the standard normal CDF. One can set $\gamma= C/(n\log(n))$ for some constant $C>0$.

\section{Additional Simulation Results} \label{sec:add_sim}

\subsection{Pointwise tests}
Additional simulation results are provided for pointwise tests at 25\% and 75\% quantiles. The results are summarized in Tables \ref{tab:Sim_Point_0.25} and \ref{tab:Sim_Point_0.75}. The simulation settings are the same as the pointwise test simulations in Section \ref{sec:sim} of the main paper. 

\begin{table}[htp!]
	\footnotesize
	\centering
	\caption{Pointwise Test ($\tau = 0.25$)}
	\label{tab:Sim_Point_0.25}%
	\smallskip
	
	\setlength{\tabcolsep}{2pt}
	\renewcommand*{\arraystretch}{1.2}
	\begin{tabularx}{\linewidth}{@{\extracolsep{\fill}}lcccccccccccccccc}
		\hline
		& \multicolumn{8}{c}{Size} & \multicolumn{8}{c}{Power} \\ \cmidrule{2-9}\cmidrule{10-17}
		& \multicolumn{4}{c}{$N = 200$} & \multicolumn{4}{c}{$N = 400$} & \multicolumn{4}{c}{$N = 200$} & \multicolumn{4}{c}{$N = 400$} \\ \cmidrule{2-5}\cmidrule{6-9}\cmidrule{10-13}\cmidrule{14-17}
		Methods & SRS   & WEI   & BCD   & SBR   & SRS   & WEI   & BCD   & SBR   & SRS   & WEI   & BCD   & SBR   & SRS   & WEI   & BCD   & SBR \\
		\hline
		\multicolumn{17}{c}{Panel A: DGP (i)}  \\
		NA    & 0.053 & 0.057 & 0.052 & 0.052 & 0.056 & 0.054 & 0.053 & 0.052 & 0.334 & 0.341 & 0.346 & 0.344 & 0.579 & 0.591 & 0.595 & 0.607 \\
		LP    & 0.054 & 0.056 & 0.052 & 0.054 & 0.053 & 0.057 & 0.053 & 0.052 & 0.391 & 0.406 & 0.405 & 0.394 & 0.683 & 0.696 & 0.693 & 0.694 \\
		ML    & 0.057 & 0.055 & 0.055 & 0.054 & 0.053 & 0.055 & 0.049 & 0.051 & 0.380 & 0.387 & 0.389 & 0.378 & 0.673 & 0.678 & 0.683 & 0.674 \\
		LPML  & 0.056 & 0.056 & 0.057 & 0.057 & 0.052 & 0.056 & 0.055 & 0.056 & 0.410 & 0.418 & 0.417 & 0.409 & 0.714 & 0.722 & 0.717 & 0.715 \\
		MLX   & 0.058 & 0.060 & 0.057 & 0.057 & 0.051 & 0.055 & 0.053 & 0.057 & 0.387 & 0.394 & 0.398 & 0.388 & 0.668 & 0.674 & 0.677 & 0.677 \\
		LPMLX & 0.060 & 0.059 & 0.060 & 0.060 & 0.052 & 0.060 & 0.056 & 0.057 & 0.414 & 0.423 & 0.431 & 0.415 & 0.718 & 0.722 & 0.725 & 0.718 \\
		NP    & 0.061 & 0.065 & 0.063 & 0.063 & 0.056 & 0.060 & 0.058 & 0.057 & 0.432 & 0.444 & 0.442 & 0.427 & 0.724 & 0.728 & 0.730 & 0.724 \\
		\hline
		\multicolumn{17}{c}{Panel B: DGP (ii)}  \\
		NA    & 0.044 & 0.047 & 0.048 & 0.046 & 0.054 & 0.049 & 0.044 & 0.046 & 0.457 & 0.457 & 0.466 & 0.475 & 0.741 & 0.751 & 0.752 & 0.760 \\
		LP    & 0.050 & 0.050 & 0.050 & 0.049 & 0.054 & 0.052 & 0.045 & 0.044 & 0.541 & 0.542 & 0.545 & 0.538 & 0.824 & 0.830 & 0.831 & 0.825 \\
		ML    & 0.049 & 0.049 & 0.052 & 0.054 & 0.050 & 0.051 & 0.049 & 0.045 & 0.478 & 0.477 & 0.480 & 0.474 & 0.757 & 0.761 & 0.760 & 0.761 \\
		LPML  & 0.053 & 0.052 & 0.056 & 0.054 & 0.053 & 0.050 & 0.048 & 0.044 & 0.542 & 0.540 & 0.544 & 0.536 & 0.832 & 0.837 & 0.840 & 0.833 \\
		MLX   & 0.055 & 0.057 & 0.057 & 0.058 & 0.054 & 0.053 & 0.048 & 0.045 & 0.507 & 0.504 & 0.505 & 0.500 & 0.765 & 0.771 & 0.771 & 0.770 \\
		LPMLX & 0.055 & 0.061 & 0.061 & 0.057 & 0.057 & 0.054 & 0.050 & 0.046 & 0.572 & 0.567 & 0.572 & 0.563 & 0.848 & 0.850 & 0.852 & 0.845 \\
		NP    & 0.063 & 0.065 & 0.063 & 0.061 & 0.058 & 0.056 & 0.052 & 0.051 & 0.575 & 0.571 & 0.576 & 0.572 & 0.847 & 0.852 & 0.854 & 0.847 \\
		\hline
	\end{tabularx}%
\end{table}%

\begin{table}[htp!]
	\footnotesize
	\centering
	\caption{Pointwise Test ($\tau = 0.75$)}
	\label{tab:Sim_Point_0.75}%
	\smallskip
	
	\setlength{\tabcolsep}{2pt}
	\renewcommand*{\arraystretch}{1.2}
	\begin{tabularx}{\linewidth}{@{\extracolsep{\fill}}lcccccccccccccccc}
		\hline
		& \multicolumn{8}{c}{Size} & \multicolumn{8}{c}{Power} \\ \cmidrule{2-9}\cmidrule{10-17}
		& \multicolumn{4}{c}{$N = 200$} & \multicolumn{4}{c}{$N = 400$} & \multicolumn{4}{c}{$N = 200$} & \multicolumn{4}{c}{$N = 400$} \\ \cmidrule{2-5}\cmidrule{6-9}\cmidrule{10-13}\cmidrule{14-17}
		Methods & SRS   & WEI   & BCD   & SBR   & SRS   & WEI   & BCD   & SBR   & SRS   & WEI   & BCD   & SBR   & SRS   & WEI   & BCD   & SBR \\
		\hline
		\multicolumn{17}{c}{Panel A: DGP (i)}  \\
		NA    & 0.054 & 0.055 & 0.056 & 0.054 & 0.057 & 0.052 & 0.051 & 0.051 & 0.348 & 0.342 & 0.352 & 0.338 & 0.583 & 0.594 & 0.601 & 0.576 \\
		LP    & 0.055 & 0.052 & 0.056 & 0.052 & 0.050 & 0.053 & 0.053 & 0.049 & 0.428 & 0.418 & 0.424 & 0.432 & 0.686 & 0.698 & 0.698 & 0.697 \\
		ML    & 0.051 & 0.053 & 0.054 & 0.051 & 0.050 & 0.051 & 0.054 & 0.054 & 0.402 & 0.395 & 0.403 & 0.402 & 0.658 & 0.667 & 0.666 & 0.667 \\
		LPML  & 0.056 & 0.056 & 0.058 & 0.056 & 0.050 & 0.053 & 0.056 & 0.051 & 0.428 & 0.425 & 0.437 & 0.435 & 0.706 & 0.711 & 0.711 & 0.709 \\
		MLX   & 0.059 & 0.056 & 0.057 & 0.057 & 0.054 & 0.054 & 0.053 & 0.055 & 0.404 & 0.399 & 0.404 & 0.406 & 0.654 & 0.663 & 0.662 & 0.657 \\
		LPMLX & 0.058 & 0.058 & 0.061 & 0.057 & 0.052 & 0.056 & 0.056 & 0.053 & 0.425 & 0.421 & 0.432 & 0.432 & 0.702 & 0.711 & 0.710 & 0.710 \\
		NP    & 0.061 & 0.061 & 0.065 & 0.062 & 0.056 & 0.057 & 0.058 & 0.057 & 0.441 & 0.435 & 0.441 & 0.447 & 0.706 & 0.710 & 0.711 & 0.711 \\
		\hline
		\multicolumn{17}{c}{Panel B: DGP (ii)}  \\
		NA    & 0.052 & 0.054 & 0.053 & 0.050 & 0.047 & 0.047 & 0.050 & 0.051 & 0.325 & 0.331 & 0.325 & 0.311 & 0.557 & 0.550 & 0.546 & 0.538 \\
		LP    & 0.055 & 0.055 & 0.062 & 0.053 & 0.051 & 0.053 & 0.053 & 0.054 & 0.389 & 0.402 & 0.396 & 0.394 & 0.626 & 0.626 & 0.634 & 0.634 \\
		ML    & 0.056 & 0.055 & 0.057 & 0.055 & 0.049 & 0.050 & 0.052 & 0.054 & 0.350 & 0.357 & 0.346 & 0.348 & 0.563 & 0.575 & 0.574 & 0.569 \\
		LPML  & 0.054 & 0.057 & 0.060 & 0.055 & 0.053 & 0.051 & 0.055 & 0.053 & 0.390 & 0.400 & 0.394 & 0.388 & 0.635 & 0.640 & 0.639 & 0.644 \\
		MLX   & 0.058 & 0.057 & 0.059 & 0.059 & 0.053 & 0.053 & 0.052 & 0.055 & 0.371 & 0.387 & 0.377 & 0.373 & 0.590 & 0.595 & 0.595 & 0.597 \\
		LPMLX & 0.062 & 0.060 & 0.065 & 0.058 & 0.055 & 0.055 & 0.057 & 0.056 & 0.416 & 0.425 & 0.420 & 0.421 & 0.658 & 0.663 & 0.663 & 0.668 \\
		NP    & 0.068 & 0.068 & 0.067 & 0.063 & 0.056 & 0.054 & 0.058 & 0.059 & 0.429 & 0.436 & 0.427 & 0.426 & 0.663 & 0.670 & 0.670 & 0.672 \\
		\hline
	\end{tabularx}%
\end{table}%

\subsection{Estimation biases and standard errors}
\label{sec:bias}
In this section we report the biases and standard errors of our regression-adjusted estimators under three test settings. Specifically, the biases and standard errors for pointwise tests are summarized in Tables \ref{tab:Bias0.25}-\ref{tab:Bias0.75}. Table \ref{tab:BiasDiff} reports the biases and standard errors for estimating the difference of QTEs. Table \ref{tab:BiasUniform} provides the average estimation bias and standard errors over the interval $\tau \in [0.25, 0.75]$.

\begin{table}[htp!]
	\footnotesize
	\centering
	\caption{Estimation Bias and Standard Errors for Pointwise Tests ($\tau = 0.25$)}
	\label{tab:Bias0.25}%
	\smallskip
	
	\setlength{\tabcolsep}{1pt}
	\renewcommand*{\arraystretch}{1.2}
	\begin{tabularx}{\linewidth}{@{\extracolsep{\fill}}lcccccccccccccccc}
		\hline
		& \multicolumn{8}{c}{Bias} & \multicolumn{8}{c}{Standard Error} \\ \cmidrule{2-9}\cmidrule{10-17}
		& \multicolumn{4}{c}{$N = 200$} & \multicolumn{4}{c}{$N = 400$} & \multicolumn{4}{c}{$N = 200$} & \multicolumn{4}{c}{$N = 400$} \\ \cmidrule{2-5}\cmidrule{6-9}\cmidrule{10-13}\cmidrule{14-17}
		Methods & SRS   & WEI   & BCD   & SBR   & SRS   & WEI   & BCD   & SBR   & SRS   & WEI   & BCD   & SBR   & SRS   & WEI   & BCD   & SBR \\
		\hline
		\multicolumn{17}{c}{Panel A: DGP (i)}  \\
		NA    & 0.010 & 0.004 & -0.010 & -0.023 & 0.011 & -0.001 & -0.001 & -0.026 & 0.984 & 0.975 & 0.972 & 0.975 & 0.688 & 0.685 & 0.685 & 0.686 \\
		LP    & 0.024 & 0.013 & 0.005 & 0.036 & 0.016 & 0.006 & 0.004 & 0.009 & 0.882 & 0.874 & 0.872 & 0.872 & 0.610 & 0.607 & 0.608 & 0.607 \\
		ML    & 0.024 & 0.022 & 0.014 & 0.037 & 0.010 & 0.004 & -0.002 & 0.012 & 0.902 & 0.894 & 0.892 & 0.891 & 0.620 & 0.617 & 0.619 & 0.618 \\
		LPML  & 0.011 & 0.010 & 0.002 & 0.028 & 0.005 & 0.000 & -0.009 & 0.006 & 0.867 & 0.863 & 0.857 & 0.860 & 0.596 & 0.592 & 0.595 & 0.592 \\
		MLX   & -0.001 & -0.006 & -0.015 & 0.019 & 0.003 & -0.002 & -0.016 & 0.000 & 0.904 & 0.896 & 0.893 & 0.894 & 0.626 & 0.624 & 0.626 & 0.624 \\
		LPMLX & 0.005 & -0.002 & -0.012 & 0.015 & 0.000 & -0.007 & -0.014 & 0.000 & 0.867 & 0.861 & 0.857 & 0.858 & 0.594 & 0.591 & 0.593 & 0.592 \\
		NP    & -0.037 & -0.045 & -0.053 & -0.021 & -0.014 & -0.018 & -0.027 & -0.013 & 0.869 & 0.862 & 0.858 & 0.859 & 0.592 & 0.590 & 0.592 & 0.591 \\
		\hline
		\multicolumn{17}{c}{Panel B: DGP (ii)}  \\
		NA    & -0.004 & 0.003 & -0.007 & -0.043 & -0.001 & -0.008 & -0.001 & -0.019 & 0.824 & 0.820 & 0.816 & 0.819 & 0.574 & 0.572 & 0.571 & 0.571 \\
		LP    & -0.038 & -0.036 & -0.040 & -0.034 & -0.022 & -0.021 & -0.015 & -0.008 & 0.754 & 0.751 & 0.747 & 0.752 & 0.523 & 0.522 & 0.521 & 0.521 \\
		ML    & -0.032 & -0.027 & -0.036 & -0.033 & -0.023 & -0.020 & -0.013 & -0.011 & 0.816 & 0.813 & 0.808 & 0.814 & 0.573 & 0.571 & 0.570 & 0.570 \\
		LPML  & -0.013 & -0.007 & -0.013 & -0.003 & -0.008 & -0.009 & -0.002 & 0.004 & 0.741 & 0.739 & 0.734 & 0.740 & 0.513 & 0.511 & 0.511 & 0.511 \\
		MLX   & -0.063 & -0.056 & -0.060 & -0.065 & -0.033 & -0.038 & -0.032 & -0.035 & 0.803 & 0.800 & 0.796 & 0.802 & 0.571 & 0.568 & 0.569 & 0.568 \\
		LPMLX & -0.061 & -0.054 & -0.057 & -0.050 & -0.032 & -0.035 & -0.028 & -0.022 & 0.735 & 0.733 & 0.729 & 0.734 & 0.510 & 0.509 & 0.508 & 0.508 \\
		NP    & -0.068 & -0.066 & -0.072 & -0.062 & -0.039 & -0.041 & -0.033 & -0.026 & 0.734 & 0.732 & 0.729 & 0.733 & 0.510 & 0.509 & 0.508 & 0.508 \\
		\hline
	\end{tabularx}%
\end{table}%

\begin{table}[htp!]
	\footnotesize
	\centering
	\caption{Estimation Bias and Standard Errors for Pointwise Tests ($\tau = 0.5$)}
	\label{tab:Bias0.5}%
	\smallskip
	
	\setlength{\tabcolsep}{1pt}
	\renewcommand*{\arraystretch}{1.2}
	\begin{tabularx}{\linewidth}{@{\extracolsep{\fill}}lcccccccccccccccc}
		\hline
		& \multicolumn{8}{c}{Bias} & \multicolumn{8}{c}{Standard Error} \\ \cmidrule{2-9}\cmidrule{10-17}
		& \multicolumn{4}{c}{$N = 200$} & \multicolumn{4}{c}{$N = 400$} & \multicolumn{4}{c}{$N = 200$} & \multicolumn{4}{c}{$N = 400$} \\ \cmidrule{2-5}\cmidrule{6-9}\cmidrule{10-13}\cmidrule{14-17}
		Methods & SRS   & WEI   & BCD   & SBR   & SRS   & WEI   & BCD   & SBR   & SRS   & WEI   & BCD   & SBR   & SRS   & WEI   & BCD   & SBR \\
		\hline
		\multicolumn{17}{c}{Panel A: DGP (i)}  \\
		NA    & -0.001 & 0.020 & 0.016 & 0.043 & -0.003 & -0.006 & -0.005 & 0.019 & 0.979 & 0.976 & 0.976 & 0.973 & 0.688 & 0.685 & 0.685 & 0.687 \\
		LP    & -0.021 & -0.005 & -0.006 & -0.017 & -0.009 & -0.008 & -0.006 & -0.010 & 0.875 & 0.873 & 0.871 & 0.871 & 0.610 & 0.606 & 0.607 & 0.609 \\
		ML    & 0.006 & 0.023 & 0.013 & 0.014 & 0.006 & 0.004 & 0.005 & 0.004 & 0.897 & 0.894 & 0.893 & 0.893 & 0.627 & 0.623 & 0.623 & 0.625 \\
		LPML  & -0.004 & 0.013 & 0.001 & -0.002 & 0.001 & -0.001 & 0.003 & -0.004 & 0.860 & 0.858 & 0.856 & 0.856 & 0.597 & 0.592 & 0.592 & 0.595 \\
		MLX   & -0.004 & 0.016 & 0.006 & 0.003 & 0.003 & -0.001 & 0.006 & 0.000 & 0.898 & 0.894 & 0.894 & 0.894 & 0.631 & 0.627 & 0.628 & 0.630 \\
		LPMLX & 0.008 & 0.025 & 0.011 & 0.005 & 0.010 & 0.006 & 0.008 & 0.003 & 0.860 & 0.858 & 0.855 & 0.855 & 0.594 & 0.592 & 0.591 & 0.593 \\
		NP    & -0.021 & -0.006 & -0.014 & -0.028 & 0.003 & 0.002 & 0.004 & -0.003 & 0.859 & 0.858 & 0.855 & 0.855 & 0.593 & 0.589 & 0.589 & 0.592 \\
		\hline
		\multicolumn{17}{c}{Panel B: DGP (ii)}  \\
		NA    & 0.032 & 0.017 & 0.032 & 0.091 & 0.003 & 0.013 & 0.019 & 0.039 & 1.036 & 1.026 & 1.023 & 1.025 & 0.723 & 0.720 & 0.718 & 0.718 \\
		LP    & -0.012 & -0.034 & -0.020 & -0.015 & -0.012 & -0.003 & 0.001 & -0.008 & 0.944 & 0.937 & 0.932 & 0.936 & 0.660 & 0.658 & 0.656 & 0.655 \\
		ML    & 0.025 & 0.010 & 0.026 & 0.029 & 0.010 & 0.009 & 0.015 & 0.010 & 1.006 & 1.000 & 0.997 & 0.998 & 0.709 & 0.705 & 0.702 & 0.702 \\
		LPML  & 0.020 & 0.013 & 0.029 & 0.031 & 0.016 & 0.019 & 0.028 & 0.013 & 0.930 & 0.920 & 0.919 & 0.924 & 0.644 & 0.640 & 0.638 & 0.637 \\
		MLX   & -0.008 & -0.038 & -0.021 & -0.001 & -0.015 & -0.012 & -0.008 & -0.008 & 0.981 & 0.972 & 0.969 & 0.970 & 0.692 & 0.690 & 0.688 & 0.688 \\
		LPMLX & -0.021 & -0.041 & -0.025 & -0.024 & -0.017 & -0.013 & -0.006 & -0.017 & 0.916 & 0.912 & 0.906 & 0.908 & 0.636 & 0.635 & 0.633 & 0.631 \\
		NP    & -0.043 & -0.059 & -0.042 & -0.041 & -0.026 & -0.019 & -0.012 & -0.023 & 0.910 & 0.903 & 0.901 & 0.903 & 0.634 & 0.632 & 0.631 & 0.630 \\
		\hline
	\end{tabularx}%
\end{table}%

\begin{table}[htp!]
	\footnotesize
	\centering
	\caption{Estimation Bias and Standard Errors for Pointwise Tests ($\tau = 0.75$)}
	\label{tab:Bias0.75}%
	\smallskip
	
	\setlength{\tabcolsep}{1pt}
	\renewcommand*{\arraystretch}{1.2}
	\begin{tabularx}{\linewidth}{@{\extracolsep{\fill}}lcccccccccccccccc}
		\hline
		& \multicolumn{8}{c}{Bias} & \multicolumn{8}{c}{Standard Error} \\ \cmidrule{2-9}\cmidrule{10-17}
		& \multicolumn{4}{c}{$N = 200$} & \multicolumn{4}{c}{$N = 400$} & \multicolumn{4}{c}{$N = 200$} & \multicolumn{4}{c}{$N = 400$} \\ \cmidrule{2-5}\cmidrule{6-9}\cmidrule{10-13}\cmidrule{14-17}
		Methods & SRS   & WEI   & BCD   & SBR   & SRS   & WEI   & BCD   & SBR   & SRS   & WEI   & BCD   & SBR   & SRS   & WEI   & BCD   & SBR \\
		\hline
		\multicolumn{17}{c}{Panel A: DGP (i)}  \\
		NA    & 0.011 & 0.010 & 0.001 & 0.003 & 0.004 & 0.004 & -0.004 & -0.008 & 0.885 & 0.882 & 0.883 & 0.883 & 0.622 & 0.619 & 0.618 & 0.619 \\
		LP    & 0.002 & 0.002 & -0.006 & 0.011 & 0.005 & 0.003 & 0.000 & -0.008 & 0.779 & 0.776 & 0.775 & 0.776 & 0.545 & 0.542 & 0.541 & 0.542 \\
		ML    & 0.012 & 0.012 & 0.000 & 0.018 & 0.009 & 0.007 & 0.002 & -0.006 & 0.795 & 0.791 & 0.790 & 0.793 & 0.557 & 0.554 & 0.553 & 0.553 \\
		LPML  & 0.004 & 0.006 & -0.007 & 0.009 & 0.004 & 0.005 & 0.002 & -0.004 & 0.759 & 0.756 & 0.755 & 0.757 & 0.528 & 0.525 & 0.524 & 0.525 \\
		MLX   & 0.000 & 0.002 & -0.006 & 0.002 & 0.006 & 0.001 & 0.003 & -0.004 & 0.796 & 0.791 & 0.790 & 0.793 & 0.561 & 0.559 & 0.558 & 0.558 \\
		LPMLX & 0.004 & 0.007 & -0.007 & 0.007 & 0.006 & 0.006 & 0.001 & -0.004 & 0.758 & 0.755 & 0.753 & 0.756 & 0.527 & 0.524 & 0.523 & 0.524 \\
		NP    & -0.023 & -0.017 & -0.028 & -0.015 & 0.004 & 0.004 & 0.000 & -0.005 & 0.758 & 0.755 & 0.753 & 0.755 & 0.526 & 0.524 & 0.523 & 0.524 \\
		\hline
		\multicolumn{17}{c}{Panel B: DGP (ii)}  \\
		NA    & 0.024 & 0.019 & 0.007 & 0.025 & 0.005 & 0.006 & 0.014 & 0.004 & 0.783 & 0.777 & 0.777 & 0.777 & 0.547 & 0.545 & 0.545 & 0.546 \\
		LP    & -0.009 & -0.012 & -0.029 & -0.011 & -0.013 & -0.008 & -0.004 & -0.010 & 0.711 & 0.707 & 0.707 & 0.707 & 0.495 & 0.494 & 0.493 & 0.494 \\
		ML    & 0.009 & 0.009 & -0.008 & 0.005 & -0.008 & 0.001 & 0.003 & 0.001 & 0.744 & 0.74  & 0.739 & 0.74  & 0.523 & 0.522 & 0.522 & 0.522 \\
		LPML  & 0.026 & 0.025 & 0.013 & 0.021 & 0.011 & 0.018 & 0.020 & 0.017 & 0.697 & 0.696 & 0.696 & 0.693 & 0.481 & 0.48  & 0.479 & 0.482 \\
		MLX   & -0.037 & -0.039 & -0.049 & -0.037 & -0.028 & -0.019 & -0.018 & -0.018 & 0.727 & 0.724 & 0.723 & 0.723 & 0.517 & 0.516 & 0.515 & 0.516 \\
		LPMLX & -0.045 & -0.043 & -0.058 & -0.046 & -0.027 & -0.019 & -0.015 & -0.019 & 0.686 & 0.685 & 0.684 & 0.684 & 0.479 & 0.476 & 0.476 & 0.479 \\
		NP    & -0.056 & -0.061 & -0.074 & -0.061 & -0.036 & -0.027 & -0.023 & -0.027 & 0.685 & 0.681 & 0.682 & 0.681 & 0.474 & 0.473 & 0.473 & 0.473 \\
		\hline
	\end{tabularx}%
\end{table}%

\begin{table}[htp!]
	\footnotesize
	\centering
	\caption{Estimation Bias and Standard Errors for Tests of Differences ($\tau_1 = 0.25$, $\tau_2 = 0.75$)}
	\label{tab:BiasDiff}%
	\smallskip
	
	\setlength{\tabcolsep}{1pt}
	\renewcommand*{\arraystretch}{1.2}
	\begin{tabularx}{\linewidth}{@{\extracolsep{\fill}}lcccccccccccccccc}
		\hline
		& \multicolumn{8}{c}{Bias} & \multicolumn{8}{c}{Standard Error} \\ \cmidrule{2-9}\cmidrule{10-17}
		& \multicolumn{4}{c}{$N = 200$} & \multicolumn{4}{c}{$N = 400$} & \multicolumn{4}{c}{$N = 200$} & \multicolumn{4}{c}{$N = 400$} \\ \cmidrule{2-5}\cmidrule{6-9}\cmidrule{10-13}\cmidrule{14-17}
		Methods & SRS   & WEI   & BCD   & SBR   & SRS   & WEI   & BCD   & SBR   & SRS   & WEI   & BCD   & SBR   & SRS   & WEI   & BCD   & SBR \\
		\hline
		\multicolumn{17}{c}{Panel A: DGP (i)}  \\
		NA    & -0.011 & 0.015 & 0.029 & 0.067 & -0.013 & -0.007 & 0.000 & 0.045 & 1.310 & 1.301 & 1.299 & 1.299 & 0.911 & 0.905 & 0.906 & 0.909 \\
		LP    & -0.044 & -0.020 & -0.007 & -0.052 & -0.025 & -0.015 & -0.007 & -0.019 & 1.255 & 1.248 & 1.246 & 1.245 & 0.869 & 0.864 & 0.865 & 0.867 \\
		ML    & -0.029 & -0.012 & -0.005 & -0.034 & -0.009 & -0.006 & 0.001 & -0.015 & 1.253 & 1.243 & 1.241 & 1.241 & 0.865 & 0.860 & 0.862 & 0.862 \\
		LPML  & -0.028 & -0.008 & -0.010 & -0.043 & -0.010 & -0.007 & 0.009 & -0.015 & 1.201 & 1.194 & 1.190 & 1.191 & 0.823 & 0.819 & 0.820 & 0.822 \\
		MLX   & -0.012 & 0.000 & 0.012 & -0.016 & 0.003 & 0.000 & 0.024 & -0.001 & 1.250 & 1.244 & 1.239 & 1.241 & 0.870 & 0.865 & 0.867 & 0.868 \\
		LPMLX & -0.019 & 0.006 & 0.006 & -0.029 & 0.007 & 0.011 & 0.025 & 0.002 & 1.197 & 1.190 & 1.187 & 1.187 & 0.822 & 0.817 & 0.818 & 0.819 \\
		NP    & -0.010 & 0.014 & 0.018 & -0.026 & 0.009 & 0.011 & 0.028 & 0.002 & 1.198 & 1.192 & 1.188 & 1.189 & 0.819 & 0.814 & 0.816 & 0.818 \\
		\hline
		\multicolumn{17}{c}{Panel B: DGP (ii)}  \\
		NA    & 0.038 & 0.016 & 0.040 & 0.135 & 0.007 & 0.022 & 0.021 & 0.060 & 1.280 & 1.270 & 1.264 & 1.268 & 0.886 & 0.882 & 0.881 & 0.880 \\
		LP    & 0.029 & 0.002 & 0.021 & 0.020 & 0.012 & 0.018 & 0.017 & 0.000 & 1.201 & 1.192 & 1.186 & 1.191 & 0.831 & 0.829 & 0.827 & 0.826 \\
		ML    & 0.053 & 0.040 & 0.048 & 0.061 & 0.034 & 0.030 & 0.029 & 0.026 & 1.286 & 1.279 & 1.270 & 1.276 & 0.898 & 0.894 & 0.892 & 0.893 \\
		LPML  & 0.030 & 0.014 & 0.038 & 0.027 & 0.024 & 0.027 & 0.031 & 0.010 & 1.180 & 1.169 & 1.164 & 1.172 & 0.811 & 0.808 & 0.806 & 0.806 \\
		MLX   & 0.054 & 0.032 & 0.043 & 0.054 & 0.022 & 0.024 & 0.012 & 0.016 & 1.258 & 1.253 & 1.247 & 1.252 & 0.889 & 0.884 & 0.884 & 0.882 \\
		LPMLX & 0.028 & 0.003 & 0.023 & 0.013 & 0.013 & 0.021 & 0.021 & 0.004 & 1.165 & 1.159 & 1.152 & 1.157 & 0.804 & 0.803 & 0.801 & 0.799 \\
		NP    & 0.019 & 0.001 & 0.023 & 0.010 & 0.014 & 0.021 & 0.021 & 0.003 & 1.160 & 1.153 & 1.149 & 1.152 & 0.804 & 0.801 & 0.799 & 0.799 \\
		\hline
	\end{tabularx}%
\end{table}%

\begin{table}[htp!]
	\footnotesize
	\centering
	\caption{Average Estimation Bias and Standard Errors for Uniform Tests ($\tau \in [0.25, 0.75]$)}
	\label{tab:BiasUniform}%
	\smallskip
	
	\setlength{\tabcolsep}{1pt}
	\renewcommand*{\arraystretch}{1.2}
	\begin{tabularx}{\linewidth}{@{\extracolsep{\fill}}lcccccccccccccccc}
		\hline
		& \multicolumn{8}{c}{Bias} & \multicolumn{8}{c}{Standard Error} \\ \cmidrule{2-9}\cmidrule{10-17}
		& \multicolumn{4}{c}{$N = 200$} & \multicolumn{4}{c}{$N = 400$} & \multicolumn{4}{c}{$N = 200$} & \multicolumn{4}{c}{$N = 400$} \\ \cmidrule{2-5}\cmidrule{6-9}\cmidrule{10-13}\cmidrule{14-17}
		Methods & SRS   & WEI   & BCD   & SBR   & SRS   & WEI   & BCD   & SBR   & SRS   & WEI   & BCD   & SBR   & SRS   & WEI   & BCD   & SBR \\
		\hline
		\multicolumn{17}{c}{Panel A: DGP (i)}  \\
		NA    & 0.003 & 0.012 & 0.000 & 0.010 & 0.002 & 0.000 & -0.004 & -0.004 & 0.931 & 0.927 & 0.925 & 0.925 & 0.653 & 0.651 & 0.651 & 0.652 \\
		LP    & 0.001 & 0.004 & -0.003 & 0.008 & 0.002 & 0.001 & -0.001 & -0.002 & 0.829 & 0.825 & 0.824 & 0.824 & 0.576 & 0.575 & 0.575 & 0.575 \\
		ML    & 0.012 & 0.016 & 0.005 & 0.019 & 0.005 & 0.006 & 0.001 & 0.003 & 0.851 & 0.847 & 0.845 & 0.846 & 0.593 & 0.591 & 0.591 & 0.591 \\
		LPML  & 0.002 & 0.008 & -0.003 & 0.008 & 0.002 & 0.001 & -0.002 & -0.002 & 0.814 & 0.811 & 0.808 & 0.810 & 0.562 & 0.560 & 0.560 & 0.560 \\
		MLX   & -0.002 & 0.006 & -0.008 & 0.009 & 0.000 & 0.001 & -0.002 & -0.003 & 0.850 & 0.847 & 0.845 & 0.846 & 0.596 & 0.594 & 0.594 & 0.594 \\
		LPMLX & 0.003 & 0.010 & -0.004 & 0.007 & 0.003 & 0.001 & -0.001 & -0.002 & 0.813 & 0.810 & 0.808 & 0.808 & 0.561 & 0.559 & 0.559 & 0.560 \\
		NP    & -0.024 & -0.017 & -0.029 & -0.019 & -0.002 & -0.003 & -0.005 & -0.006 & 0.813 & 0.810 & 0.808 & 0.808 & 0.560 & 0.558 & 0.558 & 0.558 \\
		\hline
		\multicolumn{17}{c}{Panel B: DGP (ii)}  \\
		NA    & 0.012 & 0.013 & 0.010 & 0.024 & 0.002 & 0.003 & 0.010 & 0.005 & 0.850 & 0.843 & 0.841 & 0.843 & 0.593 & 0.591 & 0.590 & 0.590 \\
		LP    & -0.024 & -0.022 & -0.027 & -0.017 & -0.015 & -0.011 & -0.007 & -0.010 & 0.772 & 0.767 & 0.765 & 0.766 & 0.537 & 0.535 & 0.534 & 0.534 \\
		ML    & 0.001 & 0.000 & 0.000 & 0.009 & -0.003 & -0.001 & 0.005 & 0.001 & 0.820 & 0.816 & 0.814 & 0.816 & 0.579 & 0.577 & 0.575 & 0.575 \\
		LPML  & 0.015 & 0.016 & 0.015 & 0.023 & 0.010 & 0.014 & 0.018 & 0.014 & 0.761 & 0.757 & 0.755 & 0.757 & 0.525 & 0.523 & 0.522 & 0.522 \\
		MLX   & -0.035 & -0.038 & -0.038 & -0.030 & -0.022 & -0.024 & -0.019 & -0.020 & 0.801 & 0.797 & 0.795 & 0.796 & 0.569 & 0.567 & 0.566 & 0.566 \\
		LPMLX & -0.039 & -0.040 & -0.039 & -0.034 & -0.021 & -0.020 & -0.015 & -0.018 & 0.752 & 0.748 & 0.745 & 0.746 & 0.521 & 0.519 & 0.518 & 0.519 \\
		NP    & -0.055 & -0.059 & -0.057 & -0.052 & -0.033 & -0.030 & -0.024 & -0.028 & 0.746 & 0.742 & 0.741 & 0.742 & 0.516 & 0.515 & 0.514 & 0.514 \\
		\hline
	\end{tabularx}%
\end{table}%

\subsection{Na{\"i}ve Bootstrap Inference} 
\label{sec:add_sim1}
In this section we report the size and power of our regression-adjusted estimator for the median QTE when we replace $\hat{\pi}(s)$ by the true propensity score $1/2$. We then consider the multiplier bootstrap as defined in the main text but with  $\hat{\pi}(s)$ replaced by the true propensity score $1/2$. We call this the naive bootstrap inference because the simulation results below show that it is conservative. Specifically, we report addition simulation results for the pointwise tests with $\tau = 0.25$, $0.5$ and $0.75$ (Tables \ref{tab:Sim_Point_0.25'}-\ref{tab:Sim_Point_0.75'}), tests for differences (Table \ref{tab:Sim_Diff1'}), and uniform tests (Table \ref{tab:Sim_Unif1'}).

Comparing the results with the ones in Section \ref{sec:sim}, we see that using the true, instead of the estimated, propensity score, the multiplier bootstrap inference becomes conservative for randomization schemes ``WEI", ``BCD", and ``SBR". Specifically, the sizes are much smaller than the nominal rate ($5\%$). At the same time, the powers are smaller than their counterparts in Section \ref{sec:sim}. The improvement in powers of the ``LPMLX" estimator with the estimated propensity score over the ``LPMLX" estimator with the true propensity score is due to the 31--38\% reduction in the standard errors. This outcome is consistent with the findings in \cite{BCS17} and \cite{ZZ20} that the naive inference methods under CARs are conservative. 

\begin{table}[htp!]
	\footnotesize
	\centering
	\caption{Pointwise Tests with Na{\"i}ve Bootstrap Inference ($\tau = 0.25$, $\hat{\pi}(s) = 0.5$)}
	\label{tab:Sim_Point_0.25'}%
	\smallskip
	
	\setlength{\tabcolsep}{2pt}
	\renewcommand*{\arraystretch}{1.2}
	\begin{tabularx}{\linewidth}{@{\extracolsep{\fill}}lcccccccccccccccc}
		\hline
		& \multicolumn{8}{c}{Size} & \multicolumn{8}{c}{Power} \\ \cmidrule{2-9}\cmidrule{10-17}
		& \multicolumn{4}{c}{$N = 200$} & \multicolumn{4}{c}{$N = 400$} & \multicolumn{4}{c}{$N = 200$} & \multicolumn{4}{c}{$N = 400$} \\ \cmidrule{2-5}\cmidrule{6-9}\cmidrule{10-13}\cmidrule{14-17}
		Methods & SRS   & WEI   & BCD   & SBR   & SRS   & WEI   & BCD   & SBR   & SRS   & WEI   & BCD   & SBR   & SRS   & WEI   & BCD   & SBR \\
		\hline
		\multicolumn{17}{c}{Panel A: DGP (i)}  \\
		NA    & 0.049 & 0.032 & 0.022 & 0.023 & 0.056 & 0.032 & 0.023 & 0.023 & 0.255 & 0.234 & 0.224 & 0.220 & 0.454 & 0.454 & 0.456 & 0.465 \\
		LP    & 0.048 & 0.018 & 0.008 & 0.006 & 0.053 & 0.019 & 0.006 & 0.007 & 0.222 & 0.173 & 0.149 & 0.113 & 0.399 & 0.391 & 0.372 & 0.331 \\
		ML    & 0.048 & 0.034 & 0.032 & 0.031 & 0.048 & 0.044 & 0.037 & 0.037 & 0.330 & 0.292 & 0.279 & 0.253 & 0.622 & 0.614 & 0.602 & 0.592 \\
		LPML  & 0.051 & 0.034 & 0.031 & 0.029 & 0.050 & 0.041 & 0.038 & 0.038 & 0.346 & 0.311 & 0.296 & 0.270 & 0.643 & 0.630 & 0.622 & 0.606 \\
		MLX   & 0.049 & 0.039 & 0.033 & 0.034 & 0.053 & 0.046 & 0.041 & 0.042 & 0.334 & 0.300 & 0.294 & 0.265 & 0.621 & 0.617 & 0.608 & 0.597 \\
		LPMLX & 0.052 & 0.040 & 0.033 & 0.032 & 0.049 & 0.046 & 0.042 & 0.041 & 0.353 & 0.326 & 0.308 & 0.280 & 0.656 & 0.647 & 0.640 & 0.629 \\
		NP    & 0.054 & 0.045 & 0.037 & 0.035 & 0.053 & 0.049 & 0.045 & 0.044 & 0.370 & 0.347 & 0.327 & 0.302 & 0.679 & 0.672 & 0.663 & 0.650 \\
		\hline
		\multicolumn{17}{c}{Panel B: DGP (ii)}  \\
		NA    & 0.050 & 0.019 & 0.008 & 0.009 & 0.053 & 0.019 & 0.007 & 0.008 & 0.276 & 0.247 & 0.225 & 0.237 & 0.498 & 0.498 & 0.501 & 0.520 \\
		LP    & 0.052 & 0.013 & 0.002 & 0.002 & 0.053 & 0.011 & 0.001 & 0.002 & 0.238 & 0.178 & 0.148 & 0.109 & 0.432 & 0.413 & 0.389 & 0.363 \\
		ML    & 0.046 & 0.041 & 0.041 & 0.038 & 0.045 & 0.039 & 0.037 & 0.033 & 0.443 & 0.443 & 0.446 & 0.435 & 0.720 & 0.726 & 0.726 & 0.724 \\
		LPML  & 0.047 & 0.039 & 0.037 & 0.033 & 0.047 & 0.038 & 0.031 & 0.030 & 0.444 & 0.431 & 0.440 & 0.435 & 0.720 & 0.735 & 0.736 & 0.741 \\
		MLX   & 0.052 & 0.046 & 0.044 & 0.045 & 0.048 & 0.044 & 0.041 & 0.038 & 0.469 & 0.470 & 0.472 & 0.464 & 0.732 & 0.744 & 0.743 & 0.748 \\
		LPMLX & 0.053 & 0.049 & 0.052 & 0.047 & 0.053 & 0.047 & 0.042 & 0.041 & 0.531 & 0.527 & 0.531 & 0.527 & 0.818 & 0.829 & 0.834 & 0.830 \\
		NP    & 0.058 & 0.056 & 0.056 & 0.057 & 0.052 & 0.053 & 0.050 & 0.049 & 0.548 & 0.553 & 0.556 & 0.553 & 0.841 & 0.848 & 0.850 & 0.842 \\
		\hline
	\end{tabularx}%
\end{table}%

\begin{table}[htp!]
	\footnotesize
	\centering
	\caption{Pointwise Tests with Na{\"i}ve Estimator ($\tau = 0.5$, $\hat{\pi}(s) = 0.5$)}
	\label{tab:Sim_Point_0.5'}%
	\smallskip
	
	\setlength{\tabcolsep}{2pt}
	\renewcommand*{\arraystretch}{1.2}
	\begin{tabularx}{\linewidth}{@{\extracolsep{\fill}}lcccccccccccccccc}
		\hline
		& \multicolumn{8}{c}{Size} & \multicolumn{8}{c}{Power} \\ \cmidrule{2-9}\cmidrule{10-17}
		& \multicolumn{4}{c}{$N = 200$} & \multicolumn{4}{c}{$N = 400$} & \multicolumn{4}{c}{$N = 200$} & \multicolumn{4}{c}{$N = 400$} \\ \cmidrule{2-5}\cmidrule{6-9}\cmidrule{10-13}\cmidrule{14-17}
		Methods & SRS   & WEI   & BCD   & SBR   & SRS   & WEI   & BCD   & SBR   & SRS   & WEI   & BCD   & SBR   & SRS   & WEI   & BCD   & SBR \\
		\hline
		\multicolumn{17}{c}{Panel A: DGP (i)}  \\
		NA    & 0.056 & 0.026 & 0.015 & 0.019 & 0.053 & 0.028 & 0.014 & 0.015 & 0.291 & 0.255 & 0.242 & 0.236 & 0.492 & 0.501 & 0.503 & 0.504 \\
		LP    & 0.049 & 0.008 & 0.001 & 0.001 & 0.052 & 0.008 & 0.001 & 0.000 & 0.181 & 0.103 & 0.053 & 0.033 & 0.314 & 0.254 & 0.191 & 0.158 \\
		ML    & 0.050 & 0.018 & 0.008 & 0.007 & 0.051 & 0.018 & 0.006 & 0.007 & 0.273 & 0.226 & 0.198 & 0.166 & 0.479 & 0.479 & 0.479 & 0.458 \\
		LPML  & 0.049 & 0.017 & 0.006 & 0.006 & 0.048 & 0.017 & 0.006 & 0.005 & 0.276 & 0.217 & 0.198 & 0.155 & 0.495 & 0.492 & 0.491 & 0.476 \\
		MLX   & 0.050 & 0.018 & 0.008 & 0.008 & 0.052 & 0.021 & 0.008 & 0.007 & 0.270 & 0.229 & 0.211 & 0.169 & 0.482 & 0.472 & 0.473 & 0.455 \\
		LPMLX & 0.053 & 0.017 & 0.007 & 0.007 & 0.050 & 0.018 & 0.006 & 0.006 & 0.284 & 0.232 & 0.209 & 0.169 & 0.500 & 0.498 & 0.498 & 0.479 \\
		NP    & 0.055 & 0.020 & 0.008 & 0.008 & 0.052 & 0.021 & 0.006 & 0.007 & 0.291 & 0.243 & 0.227 & 0.184 & 0.507 & 0.507 & 0.507 & 0.490 \\
		\hline
		\multicolumn{17}{c}{Panel B: DGP (ii)}  \\
		NA    & 0.051 & 0.017 & 0.005 & 0.005 & 0.048 & 0.016 & 0.006 & 0.006 & 0.284 & 0.244 & 0.223 & 0.211 & 0.499 & 0.491 & 0.486 & 0.494 \\
		LP    & 0.047 & 0.006 & 0.000 & 0.000 & 0.048 & 0.004 & 0.000 & 0.000 & 0.171 & 0.087 & 0.041 & 0.019 & 0.315 & 0.229 & 0.150 & 0.116 \\
		ML    & 0.050 & 0.021 & 0.013 & 0.011 & 0.047 & 0.026 & 0.019 & 0.018 & 0.315 & 0.254 & 0.216 & 0.187 & 0.588 & 0.552 & 0.528 & 0.509 \\
		LPML  & 0.049 & 0.016 & 0.009 & 0.006 & 0.045 & 0.017 & 0.010 & 0.009 & 0.296 & 0.229 & 0.188 & 0.166 & 0.526 & 0.484 & 0.462 & 0.440 \\
		MLX   & 0.056 & 0.021 & 0.014 & 0.013 & 0.052 & 0.027 & 0.019 & 0.021 & 0.336 & 0.275 & 0.243 & 0.213 & 0.606 & 0.564 & 0.554 & 0.535 \\
		LPMLX & 0.055 & 0.020 & 0.013 & 0.011 & 0.054 & 0.021 & 0.017 & 0.016 & 0.341 & 0.280 & 0.240 & 0.210 & 0.598 & 0.565 & 0.550 & 0.525 \\
		NP    & 0.059 & 0.023 & 0.015 & 0.014 & 0.056 & 0.028 & 0.024 & 0.023 & 0.368 & 0.303 & 0.269 & 0.237 & 0.651 & 0.614 & 0.602 & 0.578 \\
		\hline
	\end{tabularx}%
\end{table}%

\begin{table}[htp!]
	\footnotesize
	\centering
	\caption{Pointwise Tests with Na{\"i}ve Estimator ($\tau = 0.75$, $\hat{\pi}(s) = 0.5$)}
	\label{tab:Sim_Point_0.75'}%
	\smallskip
	
	\setlength{\tabcolsep}{2pt}
	\renewcommand*{\arraystretch}{1.2}
	\begin{tabularx}{\linewidth}{@{\extracolsep{\fill}}lcccccccccccccccc}
		\hline
		& \multicolumn{8}{c}{Size} & \multicolumn{8}{c}{Power} \\ \cmidrule{2-9}\cmidrule{10-17}
		& \multicolumn{4}{c}{$N = 200$} & \multicolumn{4}{c}{$N = 400$} & \multicolumn{4}{c}{$N = 200$} & \multicolumn{4}{c}{$N = 400$} \\ \cmidrule{2-5}\cmidrule{6-9}\cmidrule{10-13}\cmidrule{14-17}
		Methods & SRS   & WEI   & BCD   & SBR   & SRS   & WEI   & BCD   & SBR   & SRS   & WEI   & BCD   & SBR   & SRS   & WEI   & BCD   & SBR \\
		\hline
		\multicolumn{17}{c}{Panel A: DGP (i)}  \\
		NA    & 0.054 & 0.033 & 0.024 & 0.023 & 0.054 & 0.030 & 0.022 & 0.019 & 0.265 & 0.241 & 0.232 & 0.218 & 0.458 & 0.462 & 0.462 & 0.441 \\
		LP    & 0.015 & 0.002 & 0.000 & 0.000 & 0.029 & 0.002 & 0.000 & 0.000 & 0.089 & 0.022 & 0.002 & 0.001 & 0.162 & 0.073 & 0.012 & 0.006 \\
		ML    & 0.028 & 0.003 & 0.001 & 0.000 & 0.044 & 0.005 & 0.001 & 0.000 & 0.127 & 0.050 & 0.015 & 0.008 & 0.228 & 0.137 & 0.067 & 0.053 \\
		LPML  & 0.028 & 0.003 & 0.000 & 0.000 & 0.045 & 0.005 & 0.000 & 0.000 & 0.126 & 0.046 & 0.013 & 0.008 & 0.232 & 0.138 & 0.063 & 0.045 \\
		MLX   & 0.028 & 0.003 & 0.001 & 0.000 & 0.045 & 0.005 & 0.000 & 0.000 & 0.127 & 0.050 & 0.016 & 0.008 & 0.228 & 0.141 & 0.066 & 0.053 \\
		LPMLX & 0.028 & 0.003 & 0.000 & 0.000 & 0.045 & 0.005 & 0.000 & 0.000 & 0.127 & 0.049 & 0.014 & 0.009 & 0.232 & 0.140 & 0.064 & 0.047 \\
		NP    & 0.030 & 0.003 & 0.000 & 0.000 & 0.044 & 0.006 & 0.000 & 0.001 & 0.134 & 0.052 & 0.017 & 0.009 & 0.234 & 0.145 & 0.069 & 0.053 \\
		\hline
		\multicolumn{17}{c}{Panel B: DGP (ii)}  \\
		NA    & 0.051 & 0.027 & 0.019 & 0.016 & 0.050 & 0.028 & 0.016 & 0.019 & 0.239 & 0.210 & 0.192 & 0.169 & 0.409 & 0.392 & 0.384 & 0.371 \\
		LP    & 0.009 & 0.001 & 0.000 & 0.000 & 0.024 & 0.002 & 0.000 & 0.000 & 0.067 & 0.018 & 0.002 & 0.000 & 0.144 & 0.056 & 0.007 & 0.002 \\
		ML    & 0.021 & 0.003 & 0.000 & 0.000 & 0.040 & 0.004 & 0.001 & 0.000 & 0.103 & 0.040 & 0.011 & 0.005 & 0.191 & 0.102 & 0.041 & 0.031 \\
		LPML  & 0.027 & 0.003 & 0.000 & 0.000 & 0.041 & 0.004 & 0.000 & 0.000 & 0.101 & 0.034 & 0.008 & 0.004 & 0.185 & 0.094 & 0.030 & 0.026 \\
		MLX   & 0.023 & 0.003 & 0.000 & 0.000 & 0.040 & 0.005 & 0.000 & 0.000 & 0.107 & 0.044 & 0.011 & 0.005 & 0.195 & 0.107 & 0.044 & 0.035 \\
		LPMLX & 0.024 & 0.003 & 0.000 & 0.000 & 0.041 & 0.004 & 0.000 & 0.000 & 0.110 & 0.043 & 0.011 & 0.004 & 0.199 & 0.109 & 0.042 & 0.032 \\
		NP    & 0.027 & 0.002 & 0.000 & 0.000 & 0.040 & 0.003 & 0.000 & 0.000 & 0.113 & 0.045 & 0.013 & 0.005 & 0.205 & 0.113 & 0.044 & 0.035 \\
		\hline
	\end{tabularx}%
\end{table}%

\begin{table}[htp!]
	\footnotesize
	\centering
	\caption{Tests for Differences with Na{\"i}ve Estimator ($\tau_1 = 0.25$, $\tau_2 = 0.75$, $\hat{\pi}(s) = 0.5$)}
	\label{tab:Sim_Diff1'}%
	\smallskip
	
	\setlength{\tabcolsep}{2pt}
	\renewcommand*{\arraystretch}{1.2}
	\begin{tabularx}{\linewidth}{@{\extracolsep{\fill}}lcccccccccccccccc}
		\hline
		& \multicolumn{8}{c}{Size} & \multicolumn{8}{c}{Power} \\ \cmidrule{2-9}\cmidrule{10-17}
		& \multicolumn{4}{c}{$N = 200$} & \multicolumn{4}{c}{$N = 400$} & \multicolumn{4}{c}{$N = 200$} & \multicolumn{4}{c}{$N = 400$} \\ \cmidrule{2-5}\cmidrule{6-9}\cmidrule{10-13}\cmidrule{14-17}
		Methods & SRS   & WEI   & BCD   & SBR   & SRS   & WEI   & BCD   & SBR   & SRS   & WEI   & BCD   & SBR   & SRS   & WEI   & BCD   & SBR \\
		\hline
		\multicolumn{17}{c}{Panel A: DGP (i)}  \\
		NA    & 0.035 & 0.033 & 0.028 & 0.029 & 0.043 & 0.035 & 0.030 & 0.033 & 0.187 & 0.184 & 0.170 & 0.159 & 0.345 & 0.343 & 0.335 & 0.311 \\
		LP    & 0.010 & 0.006 & 0.001 & 0.001 & 0.029 & 0.008 & 0.003 & 0.003 & 0.081 & 0.059 & 0.034 & 0.026 & 0.203 & 0.149 & 0.110 & 0.097 \\
		ML    & 0.027 & 0.008 & 0.003 & 0.002 & 0.042 & 0.009 & 0.002 & 0.002 & 0.105 & 0.060 & 0.030 & 0.025 & 0.190 & 0.127 & 0.080 & 0.072 \\
		LPML  & 0.027 & 0.009 & 0.002 & 0.002 & 0.043 & 0.009 & 0.002 & 0.003 & 0.105 & 0.055 & 0.025 & 0.023 & 0.195 & 0.128 & 0.079 & 0.068 \\
		MLX   & 0.027 & 0.008 & 0.004 & 0.003 & 0.042 & 0.010 & 0.002 & 0.002 & 0.101 & 0.059 & 0.029 & 0.026 & 0.188 & 0.125 & 0.079 & 0.072 \\
		LPMLX & 0.028 & 0.009 & 0.003 & 0.003 & 0.042 & 0.009 & 0.002 & 0.003 & 0.108 & 0.058 & 0.028 & 0.025 & 0.197 & 0.128 & 0.079 & 0.070 \\
		NP    & 0.028 & 0.009 & 0.002 & 0.003 & 0.044 & 0.010 & 0.002 & 0.003 & 0.110 & 0.057 & 0.027 & 0.025 & 0.198 & 0.130 & 0.078 & 0.070 \\
		\hline
		\multicolumn{17}{c}{Panel B: DGP (ii)}  \\
		NA    & 0.037 & 0.026 & 0.020 & 0.021 & 0.040 & 0.028 & 0.023 & 0.026 & 0.167 & 0.165 & 0.152 & 0.122 & 0.330 & 0.318 & 0.306 & 0.294 \\
		LP    & 0.005 & 0.003 & 0.001 & 0.001 & 0.024 & 0.006 & 0.001 & 0.000 & 0.053 & 0.038 & 0.018 & 0.009 & 0.174 & 0.106 & 0.062 & 0.050 \\
		ML    & 0.022 & 0.004 & 0.002 & 0.001 & 0.035 & 0.005 & 0.002 & 0.001 & 0.081 & 0.035 & 0.014 & 0.006 & 0.162 & 0.086 & 0.045 & 0.033 \\
		LPML  & 0.026 & 0.005 & 0.001 & 0.000 & 0.039 & 0.004 & 0.001 & 0.001 & 0.081 & 0.032 & 0.009 & 0.005 & 0.157 & 0.070 & 0.025 & 0.021 \\
		MLX   & 0.023 & 0.005 & 0.001 & 0.001 & 0.034 & 0.007 & 0.001 & 0.001 & 0.082 & 0.038 & 0.013 & 0.008 & 0.165 & 0.092 & 0.043 & 0.038 \\
		LPMLX & 0.023 & 0.005 & 0.001 & 0.000 & 0.037 & 0.005 & 0.001 & 0.001 & 0.084 & 0.037 & 0.012 & 0.007 & 0.170 & 0.089 & 0.037 & 0.031 \\
		NP    & 0.024 & 0.005 & 0.001 & 0.000 & 0.037 & 0.005 & 0.001 & 0.001 & 0.091 & 0.038 & 0.013 & 0.007 & 0.175 & 0.093 & 0.042 & 0.036 \\
		\hline
	\end{tabularx}%
\end{table}%

\begin{table}[htp!]
	\footnotesize
	\centering
	\caption{Uniform Tests with Na{\"i}ve Estimator ($\tau  \in [0.25,0.75]$, $\hat{\pi}(s) = 0.5$)}
	\label{tab:Sim_Unif1'}%
	\smallskip
	
	\setlength{\tabcolsep}{2pt}
	\renewcommand*{\arraystretch}{1.2}
	\begin{tabularx}{\linewidth}{@{\extracolsep{\fill}}lcccccccccccccccc}
		\hline
		& \multicolumn{8}{c}{Size} & \multicolumn{8}{c}{Power} \\ \cmidrule{2-9}\cmidrule{10-17}
		& \multicolumn{4}{c}{$N = 200$} & \multicolumn{4}{c}{$N = 400$} & \multicolumn{4}{c}{$N = 200$} & \multicolumn{4}{c}{$N = 400$} \\ \cmidrule{2-5}\cmidrule{6-9}\cmidrule{10-13}\cmidrule{14-17}
		Methods & SRS   & WEI   & BCD   & SBR   & SRS   & WEI   & BCD   & SBR   & SRS   & WEI   & BCD   & SBR   & SRS   & WEI   & BCD   & SBR \\
		\hline
		\multicolumn{17}{c}{Panel A: DGP (i)}  \\
		NA    & 0.044 & 0.020 & 0.012 & 0.011 & 0.048 & 0.025 & 0.014 & 0.011 & 0.187 & 0.184 & 0.170 & 0.159 & 0.566 & 0.569 & 0.567 & 0.562 \\
		LP    & 0.027 & 0.004 & 0.001 & 0.001 & 0.040 & 0.005 & 0.000 & 0.001 & 0.081 & 0.059 & 0.034 & 0.026 & 0.347 & 0.295 & 0.234 & 0.195 \\
		ML    & 0.031 & 0.011 & 0.007 & 0.006 & 0.043 & 0.016 & 0.009 & 0.010 & 0.105 & 0.060 & 0.030 & 0.025 & 0.630 & 0.588 & 0.562 & 0.541 \\
		LPML  & 0.034 & 0.010 & 0.007 & 0.007 & 0.043 & 0.016 & 0.008 & 0.009 & 0.105 & 0.055 & 0.025 & 0.023 & 0.636 & 0.608 & 0.583 & 0.562 \\
		MLX   & 0.032 & 0.012 & 0.008 & 0.007 & 0.046 & 0.018 & 0.009 & 0.011 & 0.101 & 0.059 & 0.029 & 0.026 & 0.629 & 0.591 & 0.560 & 0.548 \\
		LPMLX & 0.033 & 0.013 & 0.010 & 0.008 & 0.045 & 0.019 & 0.011 & 0.011 & 0.108 & 0.058 & 0.028 & 0.025 & 0.653 & 0.623 & 0.602 & 0.582 \\
		NP    & 0.034 & 0.013 & 0.008 & 0.009 & 0.047 & 0.021 & 0.012 & 0.012 & 0.110 & 0.057 & 0.027 & 0.025 & 0.673 & 0.646 & 0.624 & 0.604 \\
		\hline
		\multicolumn{17}{c}{Panel B: DGP (ii)}  \\
		NA    & 0.043 & 0.013 & 0.005 & 0.003 & 0.050 & 0.011 & 0.003 & 0.004 & 0.308 & 0.251 & 0.213 & 0.208 & 0.562 & 0.570 & 0.566 & 0.572 \\
		LP    & 0.029 & 0.002 & 0.000 & 0.000 & 0.035 & 0.002 & 0.000 & 0.000 & 0.171 & 0.082 & 0.041 & 0.025 & 0.358 & 0.282 & 0.207 & 0.177 \\
		ML    & 0.035 & 0.012 & 0.015 & 0.011 & 0.039 & 0.019 & 0.016 & 0.015 & 0.454 & 0.401 & 0.387 & 0.368 & 0.826 & 0.800 & 0.789 & 0.787 \\
		LPML  & 0.030 & 0.011 & 0.010 & 0.010 & 0.036 & 0.016 & 0.009 & 0.010 & 0.444 & 0.388 & 0.374 & 0.355 & 0.804 & 0.771 & 0.762 & 0.755 \\
		MLX   & 0.037 & 0.017 & 0.019 & 0.014 & 0.040 & 0.022 & 0.017 & 0.019 & 0.500 & 0.450 & 0.439 & 0.418 & 0.853 & 0.830 & 0.826 & 0.819 \\
		LPMLX & 0.038 & 0.018 & 0.019 & 0.016 & 0.040 & 0.023 & 0.018 & 0.019 & 0.534 & 0.492 & 0.482 & 0.462 & 0.889 & 0.870 & 0.864 & 0.860 \\
		NP    & 0.041 & 0.023 & 0.025 & 0.023 & 0.045 & 0.028 & 0.023 & 0.024 & 0.573 & 0.539 & 0.533 & 0.513 & 0.919 & 0.906 & 0.903 & 0.899 \\
		\hline
	\end{tabularx}%
\end{table}%

\subsection{High-dimensional covariates}
\label{sec:add_sim2}
To assess the finite sample performance of the estimation and inference methods introduced in Section \ref{sec:reg}, we consider the outcomes equation
\begin{align}
Y_{i} = \alpha(X_{i}) + \gamma Z_i + \mu(X_{i}) A_{i} + \eta_{i},
\end{align}
where $\gamma = 4$ for all cases while $\alpha(X_{i})$, $\mu(X_{i})$, and $\eta_{i}$ are separately specified as follows. 

Let $Z$ follow the standardized Beta(2, 2) distribution, $S_i = \sum_{j = 1}^4 1\{Z_i \leq g_j\}$, and $(g_1, \cdots, g_4) = (-0.5\sqrt{5}, 0, 0.5\sqrt{5}, \sqrt{5})$. Further suppose that $X_i$ contains twenty covariates $(X_{1i}, \cdots, X_{20,i})^\top$, where $X = \Phi(W)$ with $W \sim N(0_{20 \times 1}, \Omega)$ and the variance matrix $\Omega$ is the Toeplitz matrix
\begin{align*}
\Omega = \begin{pmatrix}
1 & 0.5 & 0.5^2 &\cdots & 0.5^{19} \\
0.5 & 1 & 0.5 & \cdots & 0.5^{18} \\
0.5^2 & 0.5 & 1 & \cdots & 0.5^{17} \\
\vdots & \vdots & \vdots & \ddots & \vdots \\
0.5^{19} & 0.5^{18} & 0.5^{17} & \cdots & 1
\end{pmatrix}
\end{align*}
Further define $\alpha(X_{i}) = 1$, $\mu(X_i) = 1 + \sum_{k = 1}^{20} X_{ki} \beta_k$ with $\beta_k = 4/k^2$, and  $\eta_i = 2 A_i\varepsilon_{1i} + (1-A_i)\varepsilon_{2i}$, where $(\varepsilon_{1i}, \varepsilon_{2i})$ are jointly standard normal.

We consider the post-Lasso estimator $\hat{\theta}_{a,s}^{post}(\hat{q}_a(\tau))$ as defined in \eqref{eq:post} with $H_{p_n}(X_i) = (1,X_i^\top)^\top$ and $\hat{\mathbb{S}}^+_{a,s}(q) = \{2\}$. The choice of tuning parameter and the estimation procedure are detailed in Section \ref{sec:aux_imp}. We assess the empirical size and power of the tests for $n = 200$ and $n = 400$. All simulations are replicated 10,000 times, with the bootstrap sample size being 1,000. We compute the true QTEs or QTE differences by simulations with 10,000 sample size and 1,000 replications. To compute power, we perturb the true values by 1.5. 

In Table \ref{tab:Sim_Lasso}, we report the empirical size and power of all three testing scenarios in the high-dimensional setting. In particular, we compare the methods ``NA" with our post-Lasso estimator and the oracle estimator. Evidently, all sizes for both methods approach the nominal level as the sample size increases. The post-Lasso method dominates ``NA" in all tests with superior power performance. The improvement in power of the ``Post-Lasso" estimator upon ``NA" (i.e., with no adjustments) is due to a 2.5\% reduction of the standard error of the difference of the QTE estimators on average as shown in Table \ref{tab:LassoBias}. This result is consistent with the theory given in Theorem \ref{thm:hd}. The powers of the ``Post-Lasso" and ``Oracle" estimators are similarly, which also confirms that the ``Post-Lasso" estimator achieves the minimum asymptotic variance. 

\begin{table}[!htp]
	\footnotesize
	\centering
	\caption{Empirical Size and Power for High-dimensional Covariates}
	\label{tab:Sim_Lasso}%
	\smallskip
	
	\setlength{\tabcolsep}{2pt}
	\renewcommand*{\arraystretch}{1.1}
	\begin{tabularx}{\linewidth}{@{\extracolsep{\fill}}lcccccccccccccccc}
		\hline
		& \multicolumn{8}{c}{Size} & \multicolumn{8}{c}{Power} \\ \cmidrule{2-9}\cmidrule{10-17}
		& \multicolumn{4}{c}{$N = 200$} & \multicolumn{4}{c}{$N = 400$} & \multicolumn{4}{c}{$N = 200$} & \multicolumn{4}{c}{$N = 400$} \\\cmidrule{2-5}\cmidrule{6-9}\cmidrule{10-13}\cmidrule{14-17}
		Cases & SRS   & WEI   & BCD   & SBR   & SRS   & WEI   & BCD   & SBR   & SRS   & WEI   & BCD   & SBR   & SRS   & WEI   & BCD   & SBR \\
		\hline
		\multicolumn{17}{c}{Panel A: NA}  \\
		$\tau = 0.25$ & 0.049 & 0.046 & 0.046 & 0.050 & 0.050 & 0.049 & 0.046 & 0.045 & 0.649 & 0.646 & 0.645 & 0.660 & 0.915 & 0.911 & 0.915 & 0.917 \\
		$\tau = 0.50$ & 0.047 & 0.044 & 0.045 & 0.043 & 0.042 & 0.044 & 0.044 & 0.046 & 0.732 & 0.732 & 0.726 & 0.736 & 0.955 & 0.960 & 0.960 & 0.960 \\
		$\tau = 0.75$ & 0.050 & 0.045 & 0.046 & 0.047 & 0.044 & 0.046 & 0.051 & 0.047 & 0.620 & 0.635 & 0.638 & 0.627 & 0.895 & 0.904 & 0.903 & 0.898 \\
		Diff  & 0.038 & 0.041 & 0.038 & 0.039 & 0.041 & 0.042 & 0.042 & 0.037 & 0.365 & 0.373 & 0.369 & 0.351 & 0.643 & 0.644 & 0.644 & 0.628 \\
		Uniform & 0.035 & 0.034 & 0.036 & 0.036 & 0.040 & 0.038 & 0.040 & 0.040 & 0.852 & 0.860 & 0.857 & 0.865 & 0.994 & 0.996 & 0.994 & 0.994 \\
		\multicolumn{17}{c}{Panel B: Post-Lasso}  \\
		$\tau = 0.25$ & 0.060 & 0.055 & 0.058 & 0.058 & 0.054 & 0.054 & 0.048 & 0.052 & 0.661 & 0.655 & 0.659 & 0.656 & 0.923 & 0.916 & 0.924 & 0.918 \\
		$\tau = 0.50$ & 0.056 & 0.054 & 0.056 & 0.052 & 0.048 & 0.050 & 0.051 & 0.049 & 0.739 & 0.744 & 0.728 & 0.741 & 0.960 & 0.964 & 0.963 & 0.963 \\
		$\tau = 0.75$ & 0.059 & 0.055 & 0.055 & 0.056 & 0.052 & 0.050 & 0.056 & 0.055 & 0.627 & 0.644 & 0.648 & 0.643 & 0.902 & 0.911 & 0.907 & 0.904 \\
		Diff  & 0.052 & 0.050 & 0.048 & 0.050 & 0.048 & 0.046 & 0.046 & 0.043 & 0.377 & 0.380 & 0.381 & 0.373 & 0.657 & 0.665 & 0.655 & 0.659 \\
		Uniform & 0.051 & 0.051 & 0.048 & 0.046 & 0.049 & 0.046 & 0.048 & 0.049 & 0.872 & 0.881 & 0.877 & 0.882 & 0.996 & 0.998 & 0.996 & 0.996 \\
		\multicolumn{17}{c}{Panel C: Oracle}  \\
		$\tau = 0.25$ & 0.048 & 0.043 & 0.045 & 0.048 & 0.048 & 0.047 & 0.040 & 0.046 & 0.668 & 0.663 & 0.660 & 0.660 & 0.925 & 0.921 & 0.929 & 0.922 \\
		$\tau = 0.50$ & 0.041 & 0.044 & 0.044 & 0.040 & 0.042 & 0.044 & 0.043 & 0.044 & 0.745 & 0.746 & 0.739 & 0.749 & 0.962 & 0.967 & 0.967 & 0.967 \\
		$\tau = 0.75$ & 0.049 & 0.044 & 0.046 & 0.045 & 0.046 & 0.046 & 0.050 & 0.047 & 0.640 & 0.652 & 0.648 & 0.649 & 0.907 & 0.916 & 0.914 & 0.911 \\
		Diff  & 0.052 & 0.050 & 0.048 & 0.050 & 0.041 & 0.041 & 0.042 & 0.038 & 0.387 & 0.390 & 0.392 & 0.385 & 0.661 & 0.663 & 0.656 & 0.655 \\
		Uniform & 0.041 & 0.044 & 0.044 & 0.040 & 0.040 & 0.038 & 0.040 & 0.041 & 0.873 & 0.881 & 0.883 & 0.883 & 0.997 & 0.998 & 0.997 & 0.997 \\
		\hline
	\end{tabularx}%
\end{table}%

\begin{table}[!htp]
	\footnotesize
	\centering
	\caption{Estimation Bias and Standard Errors for High-dimensional Covariates}
	\label{tab:LassoBias}%
	\smallskip
	
	\setlength{\tabcolsep}{1pt}
	\renewcommand*{\arraystretch}{1.1}
	\begin{tabularx}{\linewidth}{@{\extracolsep{\fill}}lcccccccccccccccc}
		\hline
		& \multicolumn{8}{c}{Bias} & \multicolumn{8}{c}{Standard Error} \\ \cmidrule{2-9}\cmidrule{10-17}
		& \multicolumn{4}{c}{$N = 200$} & \multicolumn{4}{c}{$N = 400$} & \multicolumn{4}{c}{$N = 200$} & \multicolumn{4}{c}{$N = 400$} \\\cmidrule{2-5}\cmidrule{6-9}\cmidrule{10-13}\cmidrule{14-17}
		Cases & SRS   & WEI   & BCD   & SBR   & SRS   & WEI   & BCD   & SBR   & SRS   & WEI   & BCD   & SBR   & SRS   & WEI   & BCD   & SBR \\
		\hline
		\multicolumn{17}{c}{Panel A: NA}  \\
		$\tau = 0.25$ & -0.009 & -0.006 & -0.001 & -0.033 & -0.005 & 0.000 & -0.005 & -0.018 & 0.652 & 0.649 & 0.647 & 0.651 & 0.456 & 0.455 & 0.454 & 0.456 \\
		$\tau = 0.50$ & 0.004 & 0.005 & 0.018 & -0.004 & 0.004 & -0.002 & -0.002 & -0.001 & 0.588 & 0.584 & 0.583 & 0.583 & 0.408 & 0.407 & 0.407 & 0.407 \\
		$\tau = 0.75$ & 0.022 & 0.006 & 0.016 & 0.028 & 0.017 & 0.001 & 0.007 & 0.016 & 0.652 & 0.651 & 0.650 & 0.648 & 0.457 & 0.456 & 0.454 & 0.456 \\
		Diff  & 0.011 & 0.022 & 0.012 & 0.067 & 0.007 & 0.008 & 0.013 & 0.032 & 0.922 & 0.917 & 0.916 & 0.917 & 0.644 & 0.642 & 0.641 & 0.643 \\
		Uniform & 0.008 & 0.000 & -0.004 & -0.002 & 0.000 & -0.001 & 0.000 & 0.008 & 0.624 & 0.621 & 0.620 & 0.621 & 0.436 & 0.435 & 0.435 & 0.435 \\
		\multicolumn{17}{c}{Panel B: Post-Lasso}  \\
		$\tau = 0.25$ & 0.004 & 0.005 & 0.005 & -0.004 & -0.002 & 0.004 & -0.003 & -0.002 & 0.639 & 0.636 & 0.633 & 0.637 & 0.446 & 0.445 & 0.445 & 0.446 \\
		$\tau = 0.50$ & 0.014 & 0.010 & 0.029 & 0.011 & 0.007 & 0.001 & 0.003 & 0.005 & 0.576 & 0.573 & 0.572 & 0.571 & 0.400 & 0.398 & 0.398 & 0.399 \\
		$\tau = 0.75$ & 0.039 & 0.021 & 0.026 & 0.026 & 0.023 & 0.008 & 0.014 & 0.015 & 0.639 & 0.638 & 0.636 & 0.635 & 0.447 & 0.446 & 0.445 & 0.446 \\
		Diff  & 0.013 & 0.025 & 0.008 & 0.039 & 0.010 & 0.007 & 0.018 & 0.012 & 0.907 & 0.901 & 0.900 & 0.902 & 0.632 & 0.630 & 0.630 & 0.631 \\
		Uniform & 0.020 & 0.011 & 0.006 & 0.010 & 0.004 & 0.005 & 0.006 & 0.014 & 0.611 & 0.609 & 0.608 & 0.608 & 0.427 & 0.426 & 0.426 & 0.426 \\
		\multicolumn{17}{c}{Panel C: Oracle}  \\
		$\tau = 0.25$ & -0.010 & -0.003 & -0.001 & -0.003 & -0.006 & -0.001 & -0.007 & -0.004 & 0.633 & 0.631 & 0.629 & 0.633 & 0.445 & 0.443 & 0.443 & 0.444 \\
		$\tau = 0.50$ & 0.001 & 0.003 & 0.022 & 0.006 & 0.002 & -0.003 & -0.002 & 0.004 & 0.572 & 0.569 & 0.568 & 0.568 & 0.397 & 0.396 & 0.396 & 0.396 \\
		$\tau = 0.75$ & 0.022 & 0.008 & 0.020 & 0.022 & 0.018 & 0.001 & 0.009 & 0.013 & 0.637 & 0.635 & 0.633 & 0.632 & 0.446 & 0.445 & 0.443 & 0.445 \\
		Diff  & 0.017 & 0.026 & 0.013 & 0.040 & 0.008 & 0.009 & 0.016 & 0.014 & 0.900 & 0.894 & 0.894 & 0.895 & 0.629 & 0.628 & 0.627 & 0.628 \\
		Uniform & 0.007 & 0.002 & -0.002 & 0.004 & -0.001 & 0.000 & 0.000 & 0.011 & 0.607 & 0.604 & 0.603 & 0.604 & 0.425 & 0.424 & 0.424 & 0.424 \\
		\hline
	\end{tabularx}%
\end{table}%

\section{Additional Notation}
\label{sec:notation}
Throughout the supplement the collection $(\xi_i^s,X_i^s,Y_i^s(1),Y_i^s(0))_{i \in [n]}$ denotes an i.i.d. sequence with marginal distribution equal to the conditional distribution of $(\xi_i,X_i,Y_i(1),Y_i(0))$ given $S_i=s$. In addition, $\{(\xi_i^s,X_i^s,Y_i^s(1),Y_i^s(0))_{i \in [n]}\}_{s \in \mathcal{S}}$ are independent across $s$ and with $\{A_i,S_i\}_{i \in [n]}$. We further denote $\mathcal{F}$ as a generic class of functions which differs in different contexts. The envelope of $\mathcal{F}$ is denoted as $F$. We say $\mathcal{F}$ is of VC-type with coefficients $(\alpha_n,v_n)$ if 
\begin{align*}
\sup_Q N\left(\mathcal{F},e_Q,\eps||F||_{Q,2}\right) \leq \left(\frac{\alpha_n}{\eps}\right)^{v_n}, \quad \forall \eps \in (0,1],
\end{align*}
where $N(\cdot)$ denote the  covering number, $e_Q(f,g) = ||f-g||_{Q,2}$, and the supremum is taken over all finitely discrete probability measures.

\section{Proof of Theorem \ref{thm:est}}
\label{sec:thm_est_pf}
\begin{table}[H]
	%\label{Table:constants}
	\begin{adjustbox}{max width=\textwidth}
		\begin{tabular}{c|l}
			\hline
			Name & Description \\ \hline
			$n_a(s)$ & For $a \in \{0,1\}, s\in \mathcal{S}$, $n_a(s)$ is the number of individuals with $A_i = a$ in stratum $s \in \mathcal{S}$ \\ 
			$n(s)$ & For $s\in \mathcal{S}$, $n(s)$ is the number of individuals in stratum $s \in \mathcal{S}$ \\
			$\hat{\pi}(s)$ & For $s \in \mathcal{S}$, $\hat{\pi}(s) = n_1(s)/n(s)$ \\
			$\hat{q}_a^{adj}(\tau)$ & For $a \in \{0,1\}$ and $\tau \in \Upsilon$, $\hat{q}_a^{adj}(\tau)$ is the regression-adjusted estimator of $q_a(\tau)$ with a generic regression adjustment \\
			$m_a(\tau,s,x)$ & For $a \in \{0,1\}$,  $s \in \mathcal{S}$, $\tau \in \Upsilon$, and $x \in \Supp(X)$, $m_a(\tau,s,x) = \tau - \mathbb{P}(Y_i(a)\leq q_a(\tau)|S_i=s,X_i=x)$ is the true specification \\
			$\overline{m}_a(\tau,s,x)$ & For $a \in \{0,1\}$,  $s \in \mathcal{S}$, $\tau \in \Upsilon$, and $x \in \Supp(X)$, $\overline{m}_a(\tau,s,x)$ is the model for $m_a(\tau,s,x)$ specified by researchers \\
			$m_a(\tau,s)$ & For $a \in \{0,1\}$,  $s \in \mathcal{S}$, and $\tau \in \Upsilon$, $m_a(\tau,s)=\mathbb{E}(m_a(\tau,S_i,X_i)|S_i=s)$ \\
			$\overline{m}_a(\tau,s)$ & For $a \in \{0,1\}$,  $s \in \mathcal{S}$, and $\tau \in \Upsilon$,
			$\overline{m}_a(\tau,s)=\mathbb{E}(\overline{m}_a(\tau,S_i,X_i)|S_i=s)$ \\
			$\eta_{i,a}(\tau,s)$ & For $a \in \{0,1\}$,  $s \in \mathcal{S}$, and $\tau \in \Upsilon$,, $\eta_{i,a}(\tau,s) = \tau - 1\{Y_i(a) \leq q_a(\tau)\} - m_a(\tau,s)$\\
			$f_a(\cdot)$ & For $a \in \{0,1\}$, $f_a(\cdot)$ denotes the density of $Y(a)$ \\
			$D_n(s)$ & For $s \in \mathcal{S}$, $D_n(s) = \sum_{i=1}^n (A_i - \pi(s))1\{S_i=s\}$ denotes the imbalance in stratum $s$  \\
			$\overline{\Delta}_a(\tau,s,x)$ & For $a \in \{0,1\}$,  $s \in \mathcal{S}$, $\tau \in \Upsilon$,
			and $x \in \Supp(X)$, $\overline{\Delta}_a(\tau,s,x)= \widehat{m}_a(\tau,s,x) - \overline{m}_a(\tau,s,x)$  \\
			\hline
		\end{tabular}
	\end{adjustbox}
	%	\caption{Table of notation}
\end{table}
We first derive the linear expansion of $\hat{q}_1^{adj}(\tau)$. By Knight's identity (\citep{K98}), we have
\begin{align*}
L_n(u,\tau) = & \sum_{i =1}^n\left[\frac{A_i}{\hat{\pi}(S_i)}\left[\rho_\tau(Y_i - q_1(\tau) - u/\sqrt{n})-\rho_\tau(Y_i - q_1(\tau))\right] + \frac{(A_i - \hat{\pi}(S_i))}{\hat{\pi}(S_i)\sqrt{n}}\widehat{m}_1(\tau,S_i,X_i)u\right] \\
\equiv & -L_{1,n}(\tau)u + L_{2,n}(u,\tau),
\end{align*}
where 
\begin{align*}
L_{1,n}(\tau) = \frac{1}{\sqrt{n}}\sum_{i =1}^n \left[\frac{A_i}{\hat{\pi}(S_i)}(\tau - 1\{Y_i \leq q_1(\tau)\})-  \frac{(A_i - \hat{\pi}(S_i))}{\hat{\pi}(S_i)}\widehat{m}_1(\tau,S_i,X_i)\right]
\end{align*}
and 
\begin{align*}
L_{2,n}(\tau) = \sum_{i =1}^n \frac{A_i}{\hat{\pi}(S_i)}\int_0^{\frac{u}{\sqrt{n}}}\left(1\{Y_i \leq q_1(\tau)+v\}-1\{Y_i \leq q_1(\tau)\}  \right)dv.
\end{align*}
By change of variables, we have 
\begin{align*}
\sqrt{n}(\hat{q}_1^{adj}(\tau)-q_1(\tau)) = \argmin_uL_n(u,\tau). 
\end{align*}

Note that $L_{2,n}(\tau)$ is exactly the same as that considered in the proof of Theorem 3.2 in \cite{ZZ20} and by their result we have
\begin{align*}
\sup_{\tau \in \Upsilon}\left|L_{2,n}(\tau) - \frac{f_1(q_1(\tau))u^2}{2}\right| = o_p(1). 
\end{align*}
Next, consider $L_{1,n}(\tau)$. Denote $m_1(\tau,s) = \mathbb{E}(m_1(\tau,S_i,X_i)|S_i=s)$, $\eta_{i,1}(\tau,s) = \tau - 1\{Y_i \leq q_1(\tau)\} - m_1(\tau,s)$, and 
\begin{align*}
L_{1,n}(\tau) = & \frac{1}{\sqrt{n}}\sum_{i =1}^n \left[\frac{A_i}{\hat{\pi}(S_i)}(\tau - 1\{Y_i \leq q_1(\tau)\})\right]-  \frac{1}{\sqrt{n}}\sum_{i =1}^n \left[\frac{(A_i - \hat{\pi}(S_i))}{\hat{\pi}(S_i)}\widehat{m}_1(\tau,S_i,X_i)\right] \\
\equiv & L_{1,1,n}(\tau) - L_{1,2,n}(\tau).
\end{align*}
First, note that $\hat{\pi}(s) - \pi(s) = \frac{D_n(s)}{n(s)}$. Therefore, 
\begin{align}
L_{1,1,n}(\tau) & = \frac{1}{\sqrt{n}}\sum_{i=1}^n \sum_{s \in \mathcal{S}}\frac{A_i}{\pi(s)}1\{S_i=s\}(\tau - 1\{Y_i(1) \leq q_1(\tau)\}) \notag \\
& - \sum_{i=1}^n \sum_{s \in \mathcal{S}}\frac{A_i1\{S_i = s\}(\hat{\pi}(s) - \pi(s))}{\sqrt{n}\hat{\pi}(s)\pi(s)}(\tau - 1\{Y_i(1) \leq q_1(\tau)\}) \notag \\
& = \frac{1}{\sqrt{n}}\sum_{i=1}^n \sum_{s \in \mathcal{S}}\frac{A_i}{\pi(s)}1\{S_i=s\}(\tau - 1\{Y_i(1) \leq q_1(\tau)\}) \notag \\
& - \sum_{i=1}^n \sum_{s \in \mathcal{S}}\frac{A_i1\{S_i = s\}D_n(s)}{n(s)\sqrt{n}\hat{\pi}(s)\pi(s)}\eta_{i,1}(\tau,s) -  \sum_{s \in \mathcal{S}} \frac{D_n(s)m_1(\tau,s)}{n(s)\sqrt{n}\hat{\pi}(s)\pi(s)} D_n(s) - \sum_{s \in \mathcal{S}} \frac{D_n(s)m_1(\tau,s)}{\sqrt{n}\hat{\pi}(s)} \notag \\ 
& = \sum_{s \in \mathcal{S}} \frac{1}{\sqrt{n}}\sum_{i =1}^n \frac{A_i 1\{S_i = s\}}{\pi(s)}\eta_{i,1}(\tau,s) + \sum_{s \in \mathcal{S}} \frac{D_n(s)}{\sqrt{n}\pi(s)}m_1(\tau,s) + \sum_{i=1}^n \frac{m_1(\tau,S_i)}{\sqrt{n}} \notag \\
& - \sum_{i=1}^n \sum_{s \in \mathcal{S}}\frac{A_i1\{S_i = s\}D_n(s)}{n(s)\sqrt{n}\hat{\pi}(s)\pi(s)}\eta_{i,1}(\tau,s) -  \sum_{s \in \mathcal{S}} \frac{D_n(s)m_1(\tau,s)}{n(s)\sqrt{n}\hat{\pi}(s)\pi(s)} D_n(s) - \sum_{s \in \mathcal{S}} \frac{D_n(s)m_1(\tau,s)}{\sqrt{n}\hat{\pi}(s)} \notag \\
& = \sum_{s \in \mathcal{S}} \frac{1}{\sqrt{n}}\sum_{i =1}^n \frac{A_i 1\{S_i = s\}}{\pi(s)}\eta_{i,1}(\tau,s)  + \sum_{i=1}^n \frac{m_1(\tau,S_i)}{\sqrt{n}}+ R_{1,1}(\tau),
\label{eq:L11}
\end{align}
where 
\begin{align*}
R_{1,1}(\tau) = & - \sum_{i=1}^n \sum_{s \in \mathcal{S}}\frac{A_i1\{S_i = s\}D_n(s)}{n(s)\sqrt{n}\hat{\pi}(s)\pi(s)}\eta_{i,1}(\tau,s) -  \sum_{s \in \mathcal{S}} \frac{D_n(s)m_1(\tau,s)}{n(s)\sqrt{n}\hat{\pi}(s)\pi(s)} D_n(s) \\
& + \sum_{s \in \mathcal{S}} \frac{D_n(s)m_1(\tau,s)}{\sqrt{n}}\left(\frac{1}{\pi(s)} - \frac{1}{\hat{\pi}(s)}\right) \\
& = - \sum_{i=1}^n \sum_{s \in \mathcal{S}}\frac{A_i1\{S_i = s\}D_n(s)}{n(s)\sqrt{n}\hat{\pi}(s)\pi(s)}\eta_{i,1}(\tau,s).
\end{align*}
In addition, note that 
\begin{align*}
\{\tau - 1\{Y_i(1) \leq q_1(\tau)\} - m_1(\tau,S_i):\tau \in \Upsilon\}
\end{align*}
is of the VC-type with fixed coefficients $(\alpha,v)$ and bounded envelope,  and $\mathbb{E}(\tau - 1\{Y_i(1) \leq q_1(\tau)\} - m_1(\tau,S_i)|S_i=s) = 0$. Therefore, Lemma \ref{lem:max_eq} implies 
\begin{align*}
\sup_{\tau \in \Upsilon,s\in \mathcal{S}}\left|\frac{1}{\sqrt{n}}\sum_{i =1}^n A_i 1\{S_i=s\}\eta_{i,1}(\tau,s)\right|  = O_p(1).
\end{align*}
By Assumption \ref{ass:assignment1} we have $\max_{s \in \mathcal{S}}|D_n(s)/n(s)| = o_p(1)$, $\max_{s \in \mathcal{S}}|\hat{\pi}(s) - \pi(s)| = o_p(1)$, and $\min_{s \in \mathcal{S}} \pi(s)>c>0$, which imply $\sup_{\tau \in \Upsilon}|R_{1,1}(\tau)| = o_p(1)$. 

Next, denote $\overline{m}_1(\tau,s) = \mathbb{E}(\overline{m}_1(\tau,s,X_i)|S_i=s)$. Then
\begin{align}
L_{1,2,n} & = \frac{1}{\sqrt{n}} \sum_{s \in \mathcal{S}} \sum_{i =1}^n\frac{A_i}{\hat{\pi}(s)} \overline{m}_1(\tau,s,X_i)1\{S_i=s\} - \frac{1}{\sqrt{n}}\sum_{i =1}^n  \overline{m}_1(\tau,S_i,X_i) \notag \\
& + \frac{1}{\sqrt{n}} \sum_{s \in \mathcal{S}} \frac{1}{\hat{\pi}(s)}\sum_{i =1}^n(A_i - \hat{\pi}(s))\left(\widehat{m}_1(\tau,s,X_i) - \overline{m}_1(\tau,s,X_i)\right)1\{S_i=s\} \notag \\
& =  \frac{1}{\sqrt{n}} \sum_{s \in \mathcal{S}} \sum_{i =1}^n\frac{A_i}{\hat{\pi}(s)} (\overline{m}_1(\tau,s,X_i) - \overline{m}_1(\tau,s))1\{S_i=s\} \notag \\
& - \frac{1}{\sqrt{n}}\sum_{i =1}^n  \left(\overline{m}_1(\tau,S_i,X_i)-\overline{m}_1(\tau,S_i)\right) \notag \\
& + \frac{1}{\sqrt{n}} \sum_{s \in \mathcal{S}} \frac{1}{\hat{\pi}(s)}\sum_{i =1}^n(A_i - \hat{\pi}(s))\left(\widehat{m}_1(\tau,s,X_i) - \overline{m}_1(\tau,s,X_i)\right)1\{S_i=s\} \notag \\
& =  \frac{1}{\sqrt{n}} \sum_{s \in \mathcal{S}} \sum_{i =1}^n\frac{A_i}{\pi(s)} (\overline{m}_1(\tau,s,X_i) - \overline{m}_1(\tau,s))1\{S_i=s\} \notag \\
& - \frac{1}{\sqrt{n}}\sum_{i =1}^n  \left(\overline{m}_1(\tau,S_i,X_i)-\overline{m}_1(\tau,S_i)\right) \notag \\
& + \frac{1}{\sqrt{n}} \sum_{s \in \mathcal{S}} \left(\frac{\pi(s) - \hat{\pi}(s)}{\hat{\pi}(s)\pi(s)}\right)\left(\sum_{i =1}^{n}A_i (\overline{m}_1(\tau,s,X_i) - \overline{m}_1(\tau,s))1\{S_i=s\}\right) \notag \\
& + \frac{1}{\sqrt{n}} \sum_{s \in \mathcal{S}} \frac{1}{\hat{\pi}(s)}\sum_{i =1}^n(A_i - \hat{\pi}(s))\left(\widehat{m}_1(\tau,s,X_i) - \overline{m}_1(\tau,s,X_i)\right)1\{S_i=s\} \notag \\
& \equiv \frac{1}{\sqrt{n}} \sum_{s \in \mathcal{S}} \sum_{i =1}^n\frac{A_i}{\pi(s)} (\overline{m}_1(\tau,s,X_i) - \overline{m}_1(\tau,s))1\{S_i=s\} \notag \\
& - \frac{1}{\sqrt{n}}\sum_{i =1}^n  \left(\overline{m}_1(\tau,S_i,X_i)-\overline{m}_1(\tau,S_i)\right)  + R_2(\tau),
\label{eq:L12}
\end{align}
where the second equality holds because 
\begin{align*}
\sum_{s \in \mathcal{S}}\sum_{i\in[n]}\frac{A_i}{\hat{ \pi}(s)}\overline{m}_1(\tau,s)1\{S_i=s\} = \sum_{s \in \mathcal{S}}n(s)\overline{m}_1(\tau,s) = \sum_{i =1}^n\overline{m}_1(\tau,S_i). 
\end{align*}

For the first term of $R_2(\tau)$, we have
\begin{align*}
&\sup_{\tau \in \Upsilon}\left|\frac{1}{\sqrt{n}} \sum_{s \in \mathcal{S}} \left(\frac{\pi(s) - \hat{\pi}(s)}{\hat{\pi}(s)\pi(s)}\right)\left(\sum_{i =1}^nA_i (\overline{m}_1(\tau,s,X_i) - \overline{m}_1(\tau,s))1\{S_i=s\}\right)\right|\\
& \leq \sum_{s \in \mathcal{S}} \left|\frac{D_n(s)}{n_1(s) \pi(s)}\right|\sup_{\tau \in \Upsilon,s\in \mathcal{S}}\left|\frac{1}{\sqrt{n}}\sum_{i =1}^nA_i1\{S_i=s\}(\overline{m}_1(\tau,s,X_i) - \overline{m}_1(\tau,s))\right|.
\end{align*}
Assumption \ref{ass:mhat} implies
\begin{align*}
\mathcal{F} = \{ \overline{m}_1(\tau,s,X_i) - \overline{m}_1(\tau,s) : \tau \in \Upsilon\}
\end{align*}
is of the VC-type with fixed coefficients $(\alpha,v)$ and an envelope $F_i$ such that $\mathbb{E}(|F_i|^d|S_i=s)<\infty$ for $d>2$. Therefore, 
\begin{align*}
\sup_{\tau \in \Upsilon,s\in \mathcal{S}}\left|\frac{1}{\sqrt{n}}\sum_{i =1}^nA_i1\{S_i=s\}(\overline{m}_1(\tau,s,X_i) - \overline{m}_1(\tau,s))\right| = O_p(n^{-1/2}). 
\end{align*}
It is also assumed that $D_n(s)/n(s) = o_p(1)$ and $n(s)/n_1(s) \convP 1/\pi(s)<\infty$. Therefore, we have
\begin{align*}
\sup_{\tau \in \Upsilon}\left|\frac{1}{\sqrt{n}} \sum_{s \in \mathcal{S}} \left(\frac{\pi(s) - \hat{\pi}(s)}{\hat{\pi}(s)\pi(s)}\right)\left(\sum_{i =1}^nA_i (\overline{m}_1(\tau,s,X_i) - \overline{m}_1(\tau,s))1\{S_i=s\}\right)\right| = o_p(1). 
\end{align*}

Recall $\overline{\Delta}_1(\tau,s,X_i) = \widehat{m}_1(\tau,s,X_i) -\overline{m}_1(\tau,s,X_i)$. Then, for the second term of $R_2(\tau)$, we have 
\begin{align*}
& \left|\frac{1}{\sqrt{n}} \sum_{s \in \mathcal{S}} \frac{1}{\hat{\pi}(s)}\sum_{i =1}^n(A_i - \hat{\pi}(s))\overline{\Delta}_1(\tau,s,X_i)1\{S_i=s\}\right| \\
= & \frac{1}{\sqrt{n}}\sum_{s \in \mathcal{S}}n_0(s)\sup_{\tau \in \Upsilon}\biggl|\frac{\sum_{i\in I_1(s)}\overline{\Delta}_1(\tau,s,X_i)}{n_1(s)} - \frac{\sum_{i \in I_0(s)}\overline{\Delta}_1(\tau,s,X_i)}{n_0(s)}\biggr| = o_p(1),
\end{align*}
where the last equality holds by Assumption \ref{ass:mhat}(i). Therefore, we have
\begin{align*}
\sup_{\tau \in \Upsilon}|R_{1,2}(\tau)| = o_p(1).
\end{align*} 

Combining  \eqref{eq:L11} and \eqref{eq:L12}, we have
\begin{align*}
L_{1,n}(\tau) & = \sum_{s \in \mathcal{S}} \frac{1}{\sqrt{n}}\sum_{i =1}^n A_i 1\{S_i = s\}\left[\frac{\eta_{i,1}(\tau,s)}{\pi(s)} + \left(1-\frac{1}{\pi(s)}\right)\left(\overline{m}_1(\tau,s,X_i)-\overline{m}_1(\tau,s)\right) \right] \\
&+ \sum_{s \in \mathcal{S}} \frac{1}{\sqrt{n}}\sum_{i =1}^n (1-A_i) 1\{S_i = s\}\left(\overline{m}_1(\tau,s,X_i)-\overline{m}_1(\tau,s)\right) \\
& + \sum_{i=1}^n \frac{m_1(\tau,S_i)}{\sqrt{n}} + R_{1,1}(\tau) - R_{1,2}(\tau). 
\end{align*}
Note by Assumption \ref{ass:mhat} that the classes of functions
\begin{align*}
\left\{\frac{\eta_{i,1}(\tau,s)}{\pi(s)} + \left(1-\frac{1}{\pi(s)}\right)\left(\overline{m}_1(\tau,s,X_i)-\overline{m}_1(\tau,s)\right):\tau \in \Upsilon \right\}
\end{align*}
and 
\begin{align*}
\left\{\overline{m}_1(\tau,s,X_i)-\overline{m}_1(\tau,s):\tau \in \Upsilon  \right\}
\end{align*}
are of the VC-type with fixed coefficients and envelopes belonging to $L_{\mathbb{P},d}$. In addition, 
\begin{align*}
\mathbb{E}\left[\frac{\eta_{i,1}(\tau,s)}{\pi(s)} + \left(1-\frac{1}{\pi(s)}\right)\left(\overline{m}_1(\tau,s,X_i)-\overline{m}_1(\tau,s)\right)\bigg|S_i=s \right] = 0
\end{align*}
and 
\begin{align*}
\mathbb{E}(\overline{m}_1(\tau,s,X_i)-\overline{m}_1(\tau,s)|S_i=s) = 0. 
\end{align*}

Therefore, Lemma \ref{lem:max_eq} implies, 
\begin{align*}
\sup_{\tau \in \Upsilon}\left|\sum_{s \in \mathcal{S}} \frac{1}{\sqrt{n}}\sum_{i =1}^n A_i 1\{S_i = s\}\left[\frac{\eta_{i,1}(\tau,s)}{\pi(s)} + \left(1-\frac{1}{\pi(s)}\right)\left(\overline{m}_1(\tau,s,X_i)-\overline{m}_1(\tau,s)\right) \right]\right| = O_p(1)
\end{align*} 
and 
\begin{align*}
\sup_{\tau \in \Upsilon}\left|\sum_{s \in \mathcal{S}} \frac{1}{\sqrt{n}}\sum_{i =1}^n (1-A_i) 1\{S_i = s\}\left(\overline{m}_1(\tau,s,X_i)-\overline{m}_1(\tau,s)\right) \right| = O_p(1). 
\end{align*}
This implies $\sup_{\tau \in \Upsilon}|L_{1,n}(\tau)| = O_p(1)$. Then by \citet[Theorem 2]{K09}, we have 
\begin{align*}
& \sqrt{n}(\hat{q}^{adj}_1(\tau)-q_1(\tau))\\
&= \frac{1}{f_1(q_1(\tau))}\biggl\{\sum_{s \in \mathcal{S}} \frac{1}{\sqrt{n}}\sum_{i =1}^n A_i 1\{S_i = s\}\left[\frac{\eta_{i,1}(\tau,s)}{\pi(s)} + \left(1-\frac{1}{\pi(s)}\right)\left(\overline{m}_1(\tau,s,X_i)-\overline{m}_1(\tau,s)\right) \right] \\
&+ \sum_{s \in \mathcal{S}} \frac{1}{\sqrt{n}}\sum_{i =1}^n (1-A_i) 1\{S_i = s\}\left(\overline{m}_1(\tau,s,X_i)-\overline{m}_1(\tau,s)\right) + \sum_{i=1}^n \frac{m_1(\tau,S_i)}{\sqrt{n}}\biggr\} + R_{q,1}(\tau),
\end{align*}
where $\sup_{\tau \in \Upsilon}|R_{q,1}(\tau)| = o_p(1)$. 
Similarly, we have
\begin{align*}
& \sqrt{n}(\hat{q}^{adj}_0(\tau)-q_0(\tau))\\
= &\frac{1}{f_0(q_0(\tau))}\biggl\{\sum_{s \in \mathcal{S}} \frac{1}{\sqrt{n}}\sum_{i =1}^n (1-A_i) 1\{S_i = s\}\left[\frac{\eta_{i,0}(\tau,s)}{1-\pi(s)} + \left(1-\frac{1}{1-\pi(s)}\right)\left(\overline{m}_0(\tau,s,X_i)-\overline{m}_0(\tau,s)\right) \right] \\
&+ \sum_{s \in \mathcal{S}} \frac{1}{\sqrt{n}}\sum_{i =1}^n A_i 1\{S_i = s\}\left(\overline{m}_0(\tau,s,X_i)-\overline{m}_0(\tau,s)\right) + \sum_{i=1}^n \frac{m_0(\tau,S_i)}{\sqrt{n}}\biggr\} + R_{q,0}(\tau),
\end{align*}
where $\eta_{i,0}(\tau,s) = \tau-1\{Y_i(0)\leq q_0(\tau)\} - m_0(\tau,s)$ and $\sup_{\tau\in \Upsilon}|R_{q,0}(\tau)| = o_p(1)$. Taking the difference of the above two displays gives 
\begin{align*}
& \sqrt{n}(\hat{q}^{adj}(\tau) - q(\tau)) \\
&=  \sum_{s \in \mathcal{S}}\frac{1}{\sqrt{n}}\sum_{i =1}^nA_i1\{S_i=s\}\left[ \frac{\eta_{i,1}(\tau,s) -(1-\pi(s))\left(\overline{m}_1(\tau,s,X_i)-\overline{m}_1(\tau,s)\right) }{\pi(s) f_1(q_1(\tau))} - \frac{\left( \overline{m}_0(\tau,s,X_i)-\overline{m}_0(\tau,s)\right)}{f_0(q_0(\tau))}  \right] \\
&- \sum_{s \in \mathcal{S}}\frac{1}{\sqrt{n}}\sum_{i =1}^n(1-A_i)1\{S_i=s\}\left[ \frac{\eta_{i,0}(\tau,s) - \pi(s) \left(\overline{m}_0(\tau,s,X_i)-\overline{m}_0(\tau,s)\right)}{(1-\pi(s))f_0(q_0(\tau))} - \frac{\left(\overline{m}_1(\tau,s,X_i)-\overline{m}_1(\tau,s)\right)}{f_1(q_1(\tau))} \right] \\
& + \frac{1}{\sqrt{n}} \sum_{i =1}^n\left( \frac{m_1(\tau,S_i)}{f_1(q_1(\tau))} -  \frac{m_0(\tau,S_i)}{f_0(q_0(\tau))}\right) + R_q(\tau)\\
&\equiv \sum_{s \in \mathcal{S}}\frac{1}{\sqrt{n}}\sum_{i =1}^nA_i1\{S_i=s\} \phi_1(\tau,s,Y_i(1),X_i) - \sum_{s \in \mathcal{S}}\frac{1}{\sqrt{n}}\sum_{i =1}^n(1-A_i)1\{S_i=s\}\phi_0(\tau,s,Y_i(0),X_i) \\
& + \frac{1}{\sqrt{n}} \sum_{i =1}^n\phi_s(\tau,S_i) + R_q(\tau),
\end{align*}
where $\sup_{\tau \in \Upsilon}|R_q(\tau)| = o_p(1)$.  Lemma \ref{lem:w_est} shows that, uniformly over $\tau \in \Upsilon$,  
\begin{align*}
\sqrt{n}(\hat{q}^{adj}(\tau) - q(\tau)) \convD \mathcal{B}(\tau),
\end{align*}
where $\mathcal{B}(\tau)$ is a Gaussian process with covariance kernel
\begin{align*}
\Sigma(\tau,\tau')& = \mathbb{E} \pi(S_i)\phi_1(\tau,S_i,Y_i(1),X_i)\phi_1(\tau',S_i,Y_i(1),X_i)  \\
& + \mathbb{E} (1-\pi(S_i))\phi_0(\tau,S_i,Y_i(0),X_i)\phi_0(\tau',S_i,Y_i(0),X_i) \\
& + \mathbb{E}\phi_s(\tau,S_i)\phi_s(\tau',S_i). 
\end{align*}

For the second result in Theorem \ref{thm:est}, we denote
\begin{align}
\delta_{a}(\tau,S_i,X_i) = m_a(\tau,S_i,X_i) - m_a(\tau,S_i) \quad \text{and} \quad \overline{\delta}_{a}(\tau,S_i,X_i) = (\overline{m}_a(\tau,S_i,X_i) - \overline{m}_a(\tau,S_i)),~a=0,1.
\label{eq:delta}
\end{align}

Then  
\begin{align*}
& \mathbb{E}\pi(S_i)\phi_1(\tau,S_i,Y_i(1),X_i)\phi_1(\tau',S_i,Y_i(1),X_i) \\
&= \mathbb{E}\frac{1}{\pi(S_i)}\left[\frac{(\tau - 1\{Y_i(1) \leq q_1(\tau)\} - m_1(\tau,S_i,X_i))}{f_1(q_1(\tau))} \right]\left[\frac{(\tau' - 1\{Y_i(1) \leq q_1(\tau')\} - m_1(\tau',S_i,X_i))}{f_1(q_1(\tau'))} \right] \\
& + \mathbb{E}\pi(S_i)\left[\frac{\delta_{1}(\tau,S_i,X_i) - \overline{\delta}_{1}(\tau,S_i,X_i) }{\pi(S_i)f_1(q_1(\tau))} + \left(\frac{\overline{\delta}_{1}(\tau,S_i,X_i) }{f_1(q_1(\tau))} - \frac{\overline{\delta}_{0}(\tau,S_i,X_i) }{f_0(q_0(\tau))} \right) \right] \\
& \times \left[\frac{\delta_{1}(\tau',S_i,X_i) - \overline{\delta}_{1}(\tau',S_i,X_i) }{\pi(S_i)f_1(q_1(\tau'))} + \left(\frac{\overline{\delta}_{1}(\tau',S_i,X_i) }{f_1(q_1(\tau'))} - \frac{\overline{\delta}_{0}(\tau',S_i,X_i) }{f_0(q_0(\tau'))} \right) \right],
\end{align*}
\begin{align*}
& \mathbb{E}(1-\pi(S_i))\phi_0(\tau,S_i,Y_i(1),X_i)\phi_0(\tau',S_i,Y_i(1),X_i) \\
&= \mathbb{E}\frac{1}{1-\pi(S_i)}\left[\frac{(\tau - 1\{Y_i(0) \leq q_0(\tau)\} - m_0(\tau,S_i,X_i))}{f_1(q_1(\tau))} \right]\left[\frac{(\tau' - 1\{Y_i(0) \leq q_0(\tau')\} - m_0(\tau',S_i,X_i))}{f_1(q_1(\tau'))} \right] \\
& + \mathbb{E}(1-\pi(S_i))\left[\frac{\delta_{0}(\tau,S_i,X_i) - \overline{\delta}_{0}(\tau,S_i,X_i) }{(1-\pi(S_i))f_0(q_0(\tau))} - \left(\frac{\overline{\delta}_{1}(\tau,S_i,X_i) }{f_1(q_1(\tau))} - \frac{\overline{\delta}_{0}(\tau,S_i,X_i) }{f_0(q_0(\tau))} \right) \right] \\
& \times \left[\frac{\delta_{0}(\tau',S_i,X_i) - \overline{\delta}_{0}(\tau',S_i,X_i) }{(1-\pi(S_i))f_0(q_0(\tau'))} - \left(\frac{\overline{\delta}_{1}(\tau',S_i,X_i) }{f_1(q_1(\tau'))} - \frac{\overline{\delta}_{0}(\tau',S_i,X_i) }{f_0(q_0(\tau'))} \right) \right],
\end{align*}
and 
\begin{align*}
\mathbb{E}\phi_s(\tau,S_i)\phi_s(\tau',S_i) &= \mathbb{E} \left( \frac{m_1(\tau,S_i)}{f_1(q_1(\tau))} -  \frac{m_0(\tau,S_i)}{f_0(q_0(\tau))}\right)\left( \frac{m_1(\tau',S_i)}{f_1(q_1(\tau'))} -  \frac{m_0(\tau',S_i)}{f_0(q_0(\tau'))}\right) \\
&= \mathbb{E} \left( \frac{m_1(\tau,S_i,X_i)}{f_1(q_1(\tau))} -  \frac{m_0(\tau,S_i,X_i)}{f_0(q_0(\tau))}\right)\left( \frac{m_1(\tau',S_i,X_i)}{f_1(q_1(\tau'))} -  \frac{m_0(\tau',S_i,X_i)}{f_0(q_0(\tau'))}\right) \\
& - \mathbb{E} \left( \frac{\delta_{1}(\tau,S_i,X_i)}{f_1(q_1(\tau))} -  \frac{\delta_{0}(\tau,S_i,X_i)}{f_0(q_0(\tau))}\right)\left( \frac{\delta_{1}(\tau',S_i,X_i)}{f_1(q_1(\tau'))} -  \frac{\delta_{0}(\tau',S_i,X_i)}{f_0(q_0(\tau'))}\right).
\end{align*}
Let 
\begin{align*}
& \Sigma^*(\tau,\tau') \\
&=  \mathbb{E}\frac{1}{\pi(S_i)}\left[\frac{(\tau - 1\{Y_i(1) \leq q_1(\tau)\} - m_1(\tau,S_i,X_i))}{f_1(q_1(\tau))} \right]\left[\frac{(\tau' - 1\{Y_i(1) \leq q_1(\tau')\} - m_1(\tau',S_i,X_i))}{f_1(q_1(\tau'))} \right] \\
& + \mathbb{E}\frac{1}{1-\pi(S_i)}\left[\frac{(\tau - 1\{Y_i(0) \leq q_0(\tau)\} - m_0(\tau,S_i,X_i))}{f_1(q_1(\tau))} \right]\left[\frac{(\tau' - 1\{Y_i(0) \leq q_0(\tau')\} - m_0(\tau',S_i,X_i))}{f_1(q_1(\tau'))} \right]\\
& +  \mathbb{E} \left( \frac{m_1(\tau,S_i,X_i)}{f_1(q_1(\tau))} -  \frac{m_0(\tau,S_i,X_i)}{f_0(q_0(\tau))}\right)\left( \frac{m_1(\tau',S_i,X_i)}{f_1(q_1(\tau'))} -  \frac{m_0(\tau',S_i,X_i)}{f_0(q_0(\tau'))}\right),
\end{align*}
which does not rely on the working models. Then,
\begin{align*}
& \Sigma(\tau,\tau') - \Sigma^*(\tau,\tau')  \\
&=  \mathbb{E}\pi(S_i)\left[\frac{\delta_{1}(\tau,S_i,X_i) - \overline{\delta}_{1}(\tau,S_i,X_i) }{\pi(S_i)f_1(q_1(\tau))} + \left(\frac{\overline{\delta}_{1}(\tau,S_i,X_i) }{f_1(q_1(\tau))} - \frac{\overline{\delta}_{0}(\tau,S_i,X_i) }{f_0(q_0(\tau))} \right) \right] \\
& \times \left[\frac{\delta_{1}(\tau',S_i,X_i) - \overline{\delta}_{1}(\tau',S_i,X_i) }{\pi(S_i)f_1(q_1(\tau'))} + \left(\frac{\overline{\delta}_{1}(\tau',S_i,X_i) }{f_1(q_1(\tau'))} - \frac{\overline{\delta}_{0}(\tau',S_i,X_i) }{f_0(q_0(\tau'))} \right) \right]\\
& + \mathbb{E}(1-\pi(S_i))\left[\frac{\delta_{0}(\tau,S_i,X_i) - \overline{\delta}_{0}(\tau,S_i,X_i) }{(1-\pi(S_i))f_0(q_0(\tau))} - \left(\frac{\overline{\delta}_{1}(\tau,S_i,X_i) }{f_1(q_1(\tau))} - \frac{\overline{\delta}_{0}(\tau,S_i,X_i) }{f_0(q_0(\tau))} \right) \right] \\
& \times \left[\frac{\delta_{0}(\tau',S_i,X_i) - \overline{\delta}_{0}(\tau',S_i,X_i) }{(1-\pi(S_i))f_0(q_0(\tau'))} - \left(\frac{\overline{\delta}_{1}(\tau',S_i,X_i) }{f_1(q_1(\tau'))} - \frac{\overline{\delta}_{0}(\tau',S_i,X_i) }{f_0(q_0(\tau'))} \right) \right] \\
& - \mathbb{E} \left( \frac{\delta_{1}(\tau,S_i,X_i)}{f_1(q_1(\tau))} -  \frac{\delta_{0}(\tau,S_i,X_i)}{f_0(q_0(\tau))}\right)\left( \frac{\delta_{1}(\tau',S_i,X_i)}{f_1(q_1(\tau'))} -  \frac{\delta_{0}(\tau',S_i,X_i)}{f_0(q_0(\tau'))}\right)\\
&= \mathbb{E}\left[\sqrt{\frac{1-\pi(S_i)}{\pi(S_i)}}\frac{\delta_{1}(\tau,S_i,X_i) - \overline{\delta}_{1}(\tau,S_i,X_i)}{f_1(q_1(\tau))} + \sqrt{\frac{\pi(S_i)}{1-\pi(S_i)}}\frac{\delta_{0}(\tau,S_i,X_i) -\overline{\delta}_{0}(\tau,S_i,X_i)}{f_0(q_0(\tau))}  \right] \\
& \times \left[\sqrt{\frac{1-\pi(S_i)}{\pi(S_i)}}\frac{\delta_{1}(\tau',S_i,X_i) - \overline{\delta}_{1}(\tau',S_i,X_i)}{f_1(q_1(\tau'))} + \sqrt{\frac{\pi(S_i)}{1-\pi(S_i)}}\frac{\delta_{0}(\tau',S_i,X_i) -\overline{\delta}_{0}(\tau',S_i,X_i)}{f_0(q_0(\tau'))}  \right] \\
&\equiv  \mathbb{E}u_i(\tau)u_i(\tau'),
\end{align*}
where $$u_i(\tau) = \sqrt{\frac{1-\pi(S_i)}{\pi(S_i)}}\frac{\left(\delta_{1}(\tau,S_i,X_i) - \overline{\delta}_{1}(\tau,S_i,X_i)\right)}{f_1(q_1(\tau))} + \sqrt{\frac{\pi(S_i)}{1-\pi(S_i)}}\frac{\left(\delta_{0}(\tau,S_i,X_i) -\overline{\delta}_{0}(\tau,S_i,X_i)\right)}{f_0(q_0(\tau))}.$$ 
Further, denote $\vec{u}_i = (a_i(\tau_1),\cdots,a_i(\tau_K))^\top$, the asymptotic variance covariance matrix of $(\hat{q}^{adj}(\tau_1),\cdots,\hat{q}^{adj}(\tau_K))$ as $[\Sigma_{kl}]_{k,l \in [K]}$, and the optimal variance covariance matrix as $[\Sigma^*_{kl}]_{k,l \in [K]}$. We have
\begin{align*}
[\Sigma_{kl}]_{k,l \in [K]} - [\Sigma^*_{kl}]_{k,l \in [K]} = [\mathbb{E}u_i(\tau_k)u_i(\tau_l)]_{k,l\in [K]} = \mathbb{E}\vec{u}_i\vec{u}_i^\top,
\end{align*}
which is positive semidefinite. In addition, $\mathbb{E}\vec{u}_i\vec{u}_i^\top = 0$ if $\overline{m}_a(\tau,s,x) = m_a(\tau,s,x)$ for $a = 0,1$, $\tau \in \{\tau_1,\cdots,\tau_K\}$, and $(s,x)$ in the joint support of $(S_i,X_i)$. This concludes the proof.

\section{Proof of Theorem \ref{thm:boot}}
\begin{table}[H]
	%\label{Table:constants}
	\begin{adjustbox}{max width=\textwidth}
		\begin{tabular}{c|l}
			\hline
			Name & Description \\ \hline
			$n^w_a(s)$ & For $a \in \{0,1\}$ and $s \in \mathcal{S}$, $n^w_a(s) = \sum_{i=1}^n \xi_i A_i 1\{S_i=s\}$, where $\{\xi_i\}_{i \in [n]}$ is the bootstrap weights \\ 
			$n^w(s)$ & For $s \in \mathcal{S}$, $n^w(s) = \sum_{i=1}^n \xi_i1\{S_i=s\}$ \\ 
			$\hat{\pi}^w(s)$ & For $s \in \mathcal{S}$, $\hat{\pi}(s) = n_1^w(s)/n^w(s)$ \\
			$\hat{q}_a^{w}(\tau)$ & For $a \in \{0,1\}$ and $\tau \in \Upsilon$, $\hat{q}_a^{w}(\tau)$ is the bootstrap estimator of $q_a(\tau)$ with a generic regression adjustment \\ 
			\hline
		\end{tabular}
	\end{adjustbox}
	%	\caption{Table of notation}
\end{table}

We focus on deriving the linear expansion of $\hat{q}^w_1(\tau).$ Let 
\begin{align*}
L_n^w(u,\tau) = & \sum_{i =1}^n\xi_i\left[\frac{A_i}{\hat{\pi}^w(S_i)}\left[\rho_\tau(Y_i - q_1(\tau) - u/\sqrt{n})-\rho_\tau(Y_i - q_1(\tau))\right] + \frac{(A_i - \hat{\pi}^w(S_i))}{\hat{\pi}^w(S_i)\sqrt{n}}\widehat{m}_1(\tau,S_i,X_i)u\right] \\
\equiv & -L_{1,n}^w(\tau)u + L_{2,n}^w(u,\tau),
\end{align*}
where 
\begin{align*}
L_{1,n}^w(\tau) = \frac{1}{\sqrt{n}}\sum_{i =1}^n \xi_i\left[\frac{A_i}{\hat{\pi}^w(S_i)}(\tau - 1\{Y_i \leq q_1(\tau)\})-  \frac{(A_i - \hat{\pi}^w(S_i))}{\hat{\pi}^w(S_i)}\widehat{m}_1(\tau,S_i,X_i)\right],
\end{align*}
and 
\begin{align*}
L_{2,n}^w(\tau) = \sum_{i =1}^n \frac{\xi_i A_i}{\hat{\pi}^w(S_i)}\int_0^{\frac{u}{\sqrt{n}}}\left(1\{Y_i \leq q_1(\tau)+v\}-1\{Y_i \leq q_1(\tau)\}  \right)dv.
\end{align*}

By change of variables, we have 
\begin{align*}
\sqrt{n}(\hat{q}_1^{w}(\tau)-q_1(\tau)) = \argmin_uL_n^w(u,\tau). 
\end{align*}

Note that $L_{2,n}^w(\tau)$ is exactly the same as that considered in the proof of Theorem 3.2 in \cite{ZZ20} and by their result we have
\begin{align*}
\sup_{\tau \in \Upsilon}\left|L_{2,n}^w(\tau) - \frac{f_1(q_1(\tau))u^2}{2}\right| = o_p(1). 
\end{align*}
Next consider $L_{1,n}^w(\tau)$. Recall $m_1(\tau,s) = \mathbb{E}(m_1(\tau,S_i,X_i)|S_i=s)$ and $\eta_{i,1}(\tau,s) = \tau - 1\{Y_i \leq q_1(\tau)\} - m_1(\tau,s)$. Denote
\begin{align*}
L_{1,n}^w(\tau) &=  \frac{1}{\sqrt{n}}\sum_{i =1}^n \xi_i\left[\frac{A_i}{\hat{\pi}^w(S_i)}(\tau - 1\{Y_i \leq q_1(\tau)\})\right]-  \frac{1}{\sqrt{n}}\sum_{i =1}^n \xi_i \left[\frac{(A_i - \hat{\pi}^w(S_i))}{\hat{\pi}^w(S_i)}\widehat{m}_1(\tau,S_i,X_i)\right] \\
&\equiv L_{1,1,n}^w(\tau) - L_{1,2,n}^w(\tau).
\end{align*}

First, note that 
\begin{align}
&L_{1,1,n}^w(\tau) = \frac{1}{\sqrt{n}}\sum_{i=1}^n \sum_{s \in \mathcal{S}}\frac{\xi_iA_i}{\pi(s)}1\{S_i=s\}(\tau - 1\{Y_i(1) \leq q_1(\tau)\}) \notag \\
& - \sum_{i=1}^n \sum_{s \in \mathcal{S}}\frac{\xi_iA_i1\{S_i = s\}(\hat{\pi}^w(s) - \pi(s))}{\sqrt{n}\hat{\pi}^w(s)\pi(s)}(\tau - 1\{Y_i(1) \leq q_1(\tau)\}) \notag \\
&= \frac{1}{\sqrt{n}}\sum_{i=1}^n \sum_{s \in \mathcal{S}}\frac{\xi_iA_i}{\pi(s)}1\{S_i=s\}(\tau - 1\{Y_i(1) \leq q_1(\tau)\}) \notag \\
& - \sum_{i=1}^n \sum_{s \in \mathcal{S}}\frac{\xi_iA_i1\{S_i = s\}D_n^w(s)}{n^w(s)\sqrt{n}\hat{\pi}^w(s)\pi(s)}\eta_{i,1}(\tau,s) -  \sum_{s \in \mathcal{S}} \frac{D_n^w(s)m_1(\tau,s)}{n^w(s)\sqrt{n}\hat{\pi}^w(s)\pi(s)} D_n^w(s) - \sum_{s \in \mathcal{S}} \frac{D_n^w(s)m_1(\tau,s)}{\sqrt{n}\hat{\pi}^w(s)} \notag \\ 
&= \sum_{s \in \mathcal{S}} \frac{1}{\sqrt{n}}\sum_{i =1}^n \frac{\xi_iA_i 1\{S_i = s\}}{\pi(s)}\eta_{i,1}(\tau,s) + \sum_{s \in \mathcal{S}} \frac{D^w_n(s)}{\sqrt{n}\pi(s)}m_1(\tau,s) + \sum_{i=1}^n \frac{\xi_im_1(\tau,S_i)}{\sqrt{n}} \notag \\
& - \sum_{i=1}^n \sum_{s \in \mathcal{S}}\frac{\xi_iA_i1\{S_i = s\}D_n^w(s)}{n^w(s)\sqrt{n}\hat{\pi}^w(s)\pi(s)}\eta_{i,1}(\tau,s) -  \sum_{s \in \mathcal{S}} \frac{D_n^w(s)m_1(\tau,s)}{n^w(s)\sqrt{n}\hat{\pi}^w(s)\pi(s)} D^w_n(s) - \sum_{s \in \mathcal{S}} \frac{D^w_n(s)m_1(\tau,s)}{\sqrt{n}\hat{\pi}^w(s)} \notag \\
&=  \sum_{s \in \mathcal{S}} \frac{1}{\sqrt{n}}\sum_{i =1}^n \frac{\xi_iA_i 1\{S_i = s\}}{\pi(s)}\eta_{i,1}(\tau,s)  + \sum_{i=1}^n \frac{\xi_im_1(\tau,S_i)}{\sqrt{n}}+ R_{1,1}^w(\tau),
\label{eq:L11w}
\end{align}
where $D_n^w(s) = \sum_{i =1}^n\xi_i(A_i - \pi(S_i))1\{S_i=s\} = (\pi^w(s) - \pi(s)) n^w(s)$, 
\begin{align*}
&R_{1,1}^w(\tau) = - \sum_{i=1}^n \sum_{s \in \mathcal{S}}\frac{\xi_iA_i1\{S_i = s\}D^w_n(s)}{n^w(s)\sqrt{n}\hat{\pi}^w(s)\pi(s)}\eta_{i,1}(\tau,s) -  \sum_{s \in \mathcal{S}} \frac{D^w_n(s)m_1(\tau,s)}{n^w(s)\sqrt{n}\hat{\pi}^w(s)\pi(s)} D^w_n(s) \\
& + \sum_{s \in \mathcal{S}} \frac{D^w_n(s)m_1(\tau,s)}{\sqrt{n}}\left(\frac{1}{\pi(s)} - \frac{1}{\hat{\pi}^w(s)}\right)
= - \sum_{i=1}^n \sum_{s \in \mathcal{S}}\frac{\xi_iA_i1\{S_i = s\}D^w_n(s)}{n^w(s)\sqrt{n}\hat{\pi}^w(s)\pi(s)}\eta_{i,1}(\tau,s).
\end{align*}
Note that 
\begin{align*}
\{\xi_i(\tau - 1\{Y_i(1) \leq q_1(\tau)\} - m_1(\tau,S_i)):\tau \in \Upsilon\}
\end{align*}
is of the VC-type with fixed coefficients $(\alpha,v)$ and the envelope $F_i = \xi_i$, and $$\mathbb{E}\left[\xi_i(\tau - 1\{Y_i(1) \leq q_1(\tau)\} - m_1(\tau,S_i))|S_i=s\right] = 0.$$ 
We can also let $\sigma_n^2 = \mathbb{E}(F_i^2|S_i=s)\leq C<\infty$ for some constant $C$. Then, Lemma \ref{lem:max_eq} implies 
\begin{align*}
\sup_{\tau \in \Upsilon,s\in \mathcal{S}}\left|\frac{1}{\sqrt{n}}\sum_{i =1}^n A_i 1\{S_i=s\}\xi_i\eta_{i,1}(\tau,s)\right|  = O_p(1).
\end{align*}
In addition, Lemma \ref{lem:Dw} implies $\max_{s \in \mathcal{S}}|D_n^w(s)/n^w(s)| = o_p(1)$, which further implies $\max_{s \in \mathcal{S}}|\hat{\pi}^w(s) - \pi(s)| = o_p(1)$. Combining these results, we have \begin{align*}
\sup_{\tau \in \Upsilon}|R^w_{1,1}(\tau)| = o_p(1). 
\end{align*}

Next, recall $\overline{m}_1(\tau,s) = \mathbb{E}(\overline{m}_1(q_1(\tau),X_i,S_i)|S_i=s)$. Then
\begin{align}
L_{1,2,n}^w &=  \frac{1}{\sqrt{n}} \sum_{s \in \mathcal{S}} \sum_{i =1}^n\frac{\xi_iA_i}{\hat{\pi}^w(s)} \overline{m}_1(\tau,s,X_i)1\{S_i=s\} - \frac{1}{\sqrt{n}}\sum_{i =1}^n \xi_i \overline{m}_1(\tau,S_i,X_i) \notag \\
& + \frac{1}{\sqrt{n}} \sum_{s \in \mathcal{S}} \frac{1}{\hat{\pi}^w(s)}\sum_{i =1}^n\xi_i(A_i - \hat{\pi}^w(s))\left(\widehat{m}_1(\tau,s,X_i) - \overline{m}_1(\tau,s,X_i)\right)1\{S_i=s\} \notag \\
&=  \frac{1}{\sqrt{n}} \sum_{s \in \mathcal{S}} \sum_{i =1}^n\frac{\xi_iA_i}{\hat{\pi}^w(s)} (\overline{m}_1(\tau,s,X_i) - \overline{m}_1(\tau,s))1\{S_i=s\} \notag \\
& - \frac{1}{\sqrt{n}}\sum_{i =1}^n  \xi_i\left(\overline{m}_1(\tau,S_i,X_i)-\overline{m}_1(\tau,S_i)\right) \notag \\
& + \frac{1}{\sqrt{n}} \sum_{s \in \mathcal{S}} \frac{1}{\hat{\pi}^w(s)}\sum_{i =1}^n\xi_i(A_i - \hat{\pi}^w(s))\left(\widehat{m}_1(\tau,s,X_i) - \overline{m}_1(\tau,s,X_i)\right)1\{S_i=s\}\notag \\
&= \frac{1}{\sqrt{n}} \sum_{s \in \mathcal{S}} \sum_{i =1}^n\frac{\xi_iA_i}{\pi(s)} (\overline{m}_1(\tau,s,X_i) - \overline{m}_1(\tau,s))1\{S_i=s\} \notag \\
& - \frac{1}{\sqrt{n}}\sum_{i =1}^n \xi_i \left(\overline{m}_1(\tau,S_i,X_i)-\overline{m}_1(\tau,S_i)\right) \notag \\
& + \frac{1}{\sqrt{n}} \sum_{s \in \mathcal{S}} \left(\frac{\pi(s) - \hat{\pi}^w(s)}{\hat{\pi}^w(s)\pi(s)}\right)\left(\sum_{i =1}^n\xi_iA_i (\overline{m}_1(\tau,s,X_i) - \overline{m}_1(\tau,s))1\{S_i=s\}\right) \notag \\
& + \frac{1}{\sqrt{n}} \sum_{s \in \mathcal{S}} \frac{1}{\hat{\pi}^w(s)}\sum_{i =1}^n\xi_i(A_i - \hat{\pi}^w(s))\left(\widehat{m}_1(\tau,s,X_i) - \overline{m}_1(\tau,s,X_i)\right)1\{S_i=s\} \notag \\
&\equiv \frac{1}{\sqrt{n}} \sum_{s \in \mathcal{S}} \sum_{i =1}^n\frac{\xi_iA_i}{\pi(s)} (\overline{m}_1(\tau,s,X_i) - \overline{m}_1(\tau,s))1\{S_i=s\} \notag \\
& - \frac{1}{\sqrt{n}}\sum_{i =1}^n  \xi_i\left(\overline{m}_1(\tau,S_i,X_i)-\overline{m}_1(\tau,S_i)\right)  + R_{1,2}^w(\tau),
\label{eq:L12w}
\end{align}
where the second equality holds because 
\begin{align*}
\sum_{s \in \mathcal{S}} \sum_{i =1}^n \xi_iA_i 1\{S_i=s\} \frac{\overline{m}_1(\tau,s)}{\hat{\pi}^w(s)} = & \sum_{s \in \mathcal{S}} n_1^w(s)  \frac{\overline{m}_1(\tau,s)}{\hat{\pi}^w(s)} = \sum_{s \in \mathcal{S}} n^w(s)  \overline{m}_1(\tau,s) \\
= & \sum_{i =1}^n  \sum_{s \in \mathcal{S}} \xi_i 1\{S_i=s\}\overline{m}_1(\tau,S_i) = \sum_{i =1}^n \xi_i \overline{m}_1(\tau,S_i). 
\end{align*}

For the first term in $R_2^w(\tau)$, we have 
\begin{align*}
&\sup_{\tau \in \Upsilon}\left|\frac{1}{\sqrt{n}} \sum_{s \in \mathcal{S}} \left(\frac{\pi(s) - \hat{\pi}^w(s)}{\hat{\pi}^w(s)\pi(s)}\right)\left(\sum_{i =1}^nA_i \xi_i(\overline{m}_1(\tau,s,X_i) - \overline{m}_1(\tau,s))1\{S_i=s\}\right)\right|\\
&\leq \sum_{s \in \mathcal{S}} \left|\frac{D_n^w(s)}{n^w(s) \hat{\pi}^w(s)\pi(s)}\right|\sup_{\tau \in \Upsilon,s\in \mathcal{S}}\left|\frac{1}{\sqrt{n}}\sum_{i =1}^n\xi_iA_i1\{S_i=s\}(\overline{m}_1(\tau,s,X_i) - \overline{m}_1(\tau,s))\right| = o_p(1),
\end{align*}
where the last equality holds due to Lemmas \ref{lem:max_eq} and \ref{lem:Dw}, and the fact that $\mathcal{F} = \{\xi(\overline{m}_1(\tau,s,X_i) - \overline{m}_1(\tau,s)): \tau \in \Upsilon\}$ is of the VC-type with fixed coefficients $(\alpha,v)$ and envelope $\xi_i F_i$ such that $\mathbb{E}((\xi_iF_i)^q|S_i=s)<\infty$ for $q>2.$

For the second term in $R_{1,2}^w(\tau)$, recall $\overline{\Delta}_1(\tau,s,X_i) = \widehat{m}_1(\tau,s,X_i) -\overline{m}_1(\tau,s,X_i)$.  Then
\begin{align*}
& \left|\frac{1}{\sqrt{n}} \sum_{s \in \mathcal{S}} \frac{1}{\hat{\pi}^w(s)}\sum_{i =1}^n\xi_i(A_i - \hat{\pi}^w(s))\overline{\Delta}_1(\tau,s,X_i)1\{S_i=s\}\right| \\
&= \frac{1}{\sqrt{n}}\sum_{s \in \mathcal{S}} n_0^w(s)\sup_{\tau \in \Upsilon}\biggl|\frac{\sum_{i\in I_1(s)}\xi_i\overline{\Delta}_1(\tau,s,X_i)}{n_1^w(s)} - \frac{\sum_{i \in I_0(s)}\xi_i\overline{\Delta}_1(\tau,s,X_i)}{n_0^w(s)}\biggr| = o_p(1),
\end{align*}
where the last equality holds by Assumption \ref{ass:mhat}. Therefore, we have
\begin{align*}
\sup_{\tau \in \Upsilon}|R_{1,2}^w(\tau)| = o_p(1).
\end{align*} 
Combining  \eqref{eq:L11w} and \eqref{eq:L12w}, we have
\begin{align*}
L_{1,n}^w(\tau) = & \sum_{s \in \mathcal{S}} \frac{1}{\sqrt{n}}\sum_{i =1}^n \xi_iA_i 1\{S_i = s\}\left[\frac{\eta_{i,1}(\tau,s)}{\pi(s)} + \left(1-\frac{1}{\pi(s)}\right)\left(\overline{m}_1(\tau,s,X_i)-\overline{m}_1(\tau,s)\right) \right] \\
&+ \sum_{s \in \mathcal{S}} \frac{1}{\sqrt{n}}\sum_{i =1}^n \xi_i(1-A_i) 1\{S_i = s\}\left(\overline{m}_1(\tau,s,X_i)-\overline{m}_1(\tau,s)\right) \\
& + \sum_{i=1}^n \xi_i\frac{m_1(\tau,S_i)}{\sqrt{n}} + R_{1,1}^w(\tau) - R_{1,2}^w(\tau),
\end{align*}
where $\sup_{\tau \in \Upsilon}(|R_{1,1}^w(\tau)| + |R_{1,2}^w(\tau)|) = o_p(1)$. In addition, Assumption \ref{ass:mhat} implies that the classes of functions
\begin{align*}
\left\{\xi_i\left[\frac{\eta_{i,1}(\tau,s)}{\pi(s)} + \left(1-\frac{1}{\pi(s)}\right)\left(\overline{m}_1(\tau,s,X_i)-\overline{m}_1(\tau,s)\right)\right]:\tau \in \Upsilon \right\}
\end{align*}
and 
\begin{align*}
\left\{\xi_i[\overline{m}_1(\tau,s,X_i)-\overline{m}_1(\tau,s)]:\tau \in \Upsilon  \right\}
\end{align*}
are of the VC-type with fixed coefficients and envelopes belonging to $L_{\mathbb{P},d}$. In addition, 
\begin{align*}
\mathbb{E}\left[\xi_i\left(\frac{\eta_{i,1}(\tau,s)}{\pi(s)} + \left(1-\frac{1}{\pi(s)}\right)\left(\overline{m}_1(\tau,s,X_i)-\overline{m}_1(\tau,s)\right)\right)\biggl|S_i=s \right] = 0,
\end{align*}
and 
\begin{align*}
\mathbb{E}\left[\xi_i(\overline{m}_1(\tau,s,X_i)-\overline{m}_1(\tau,s))|S_i=s\right] = 0. 
\end{align*}

Therefore, Lemma \ref{lem:max_eq} implies, 
\begin{align*}
\sup_{\tau \in \Upsilon}\left|\sum_{s \in \mathcal{S}} \frac{1}{\sqrt{n}}\sum_{i =1}^n A_i 1\{S_i = s\}\left[\frac{\eta_{i,1}(\tau,s)}{\pi(s)} + \left(1-\frac{1}{\pi(s)}\right)\left(\overline{m}_1(\tau,s,X_i)-\overline{m}_1(\tau,s)\right) \right]\right| = O_p(1),
\end{align*} 
and 
\begin{align*}
\sup_{\tau \in \Upsilon}\left|\sum_{s \in \mathcal{S}} \frac{1}{\sqrt{n}}\sum_{i =1}^n (1-A_i) 1\{S_i = s\}\left(\overline{m}_1(\tau,s,X_i)-\overline{m}_1(\tau,s)\right) \right| = O_p(1). 
\end{align*}
This implies $\sup_{\tau \in \Upsilon}|L^w_{1,n}(\tau)| = O_p(1)$. Then by \citet[Theorem 2]{K09} we have 
\begin{align*}
& \sqrt{n}(\hat{q}^{w}_1(\tau)-q_1(\tau))\\
&= \frac{1}{f_1(q_1(\tau))}\biggl\{\sum_{s \in \mathcal{S}} \frac{1}{\sqrt{n}}\sum_{i =1}^n \xi_iA_i 1\{S_i = s\}\left[\frac{\eta_{i,1}(\tau,s)}{\pi(s)} + \left(1-\frac{1}{\pi(s)}\right)\left(\overline{m}_1(\tau,s,X_i)-\overline{m}_1(\tau,s)\right) \right] \\
&+ \sum_{s \in \mathcal{S}} \frac{1}{\sqrt{n}}\sum_{i =1}^n \xi_i(1-A_i) 1\{S_i = s\}\left(\overline{m}_1(\tau,s,X_i)-\overline{m}_1(\tau,s)\right) + \sum_{i=1}^n \xi_i \frac{m_1(\tau,S_i)}{\sqrt{n}}\biggr\} + R_{q,1}^w(\tau),
\end{align*}
where $\sup_{\tau \in \Upsilon}|R_{q,1}^w(\tau)| = o_p(1)$. 
Similarly, we have
\begin{align*}
& \sqrt{n}(\hat{q}^{w}_0(\tau)-q_0(\tau))\\
&= \frac{1}{f_0(q_0(\tau))}\biggl\{\sum_{s \in \mathcal{S}} \frac{1}{\sqrt{n}}\sum_{i =1}^n \xi_i(1-A_i) 1\{S_i = s\}\left[\frac{\eta_{i,0}(\tau,s)}{1-\pi(s)} + \left(1-\frac{1}{1-\pi(s)}\right)\left(\overline{m}_0(\tau,s,X_i)-\overline{m}_0(\tau,s)\right) \right] \\
&+ \sum_{s \in \mathcal{S}} \frac{1}{\sqrt{n}}\sum_{i =1}^n \xi_iA_i 1\{S_i = s\}\left(\overline{m}_0(\tau,s,X_i)-\overline{m}_0(\tau,s)\right) + \sum_{i=1}^n \xi_i \frac{m_0(\tau,S_i)}{\sqrt{n}}\biggr\} + R_{q,0}(\tau),
\end{align*}
where $\sup_{\tau\in \Upsilon}|R_{q,0}^w(\tau)| = o_p(1)$. Taking the difference of the above two displays we obtain 
\begin{align*}
& \sqrt{n}(\hat{q}^{w}(\tau) - q(\tau)) \\
&=  \sum_{s \in \mathcal{S}}\frac{1}{\sqrt{n}}\sum_{i =1}^n\xi_iA_i1\{S_i=s\} \\
& \times \left[ \frac{\eta_{i,1}(\tau,s) -(1-\pi(s))\left(\overline{m}_1(\tau,s,X_i)-\overline{m}_1(\tau,s)\right) }{\pi(s) f_1(q_1(\tau))} - \frac{\left( \overline{m}_0(\tau,s,X_i)-\overline{m}_0(\tau,s)\right)}{f_0(q_0(\tau))}  \right] \\
&- \sum_{s \in \mathcal{S}}\frac{1}{\sqrt{n}}\sum_{i =1}^n\xi_i(1-A_i)1\{S_i=s\} \\
& \times \left[ \frac{\eta_{i,0}(\tau,s) - \pi(s) \left(\overline{m}_0(\tau,s,X_i)-\overline{m}_0(\tau,s)\right)}{(1-\pi(s))f_0(q_0(\tau))} - \frac{\left(\overline{m}_1(\tau,s,X_i)-\overline{m}_1(\tau,s)\right)}{f_1(q_1(\tau))} \right] \\
& + \frac{1}{\sqrt{n}} \sum_{i =1}^n\xi_i\left( \frac{m_1(\tau,S_i)}{f_1(q_1(\tau))} -  \frac{m_0(\tau,S_i)}{f_0(q_0(\tau))}\right) + R^w_q(\tau)\\
&= \sum_{s \in \mathcal{S}}\frac{1}{\sqrt{n}}\sum_{i =1}^n\xi_iA_i1\{S_i=s\} \phi_1(\tau,s,Y_i(1),X_i) - \sum_{s \in \mathcal{S}}\frac{1}{\sqrt{n}}\sum_{i =1}^n\xi_i(1-A_i)1\{S_i=s\}\phi_0(\tau,s,Y_i(0),X_i) \\
& + \frac{1}{\sqrt{n}} \sum_{i =1}^n\xi_i\phi_s(\tau,S_i) + R_q^w(\tau),
\end{align*}
where $\sup_{\tau \in \Upsilon}|R_q^w(\tau)| = o_p(1)$ and  $( \phi_1(\cdot), \phi_0(\cdot),  \phi_s(\cdot))$ are defined in Section \ref{sec:thm_est_pf}. Recalling the linear expansion of $\sqrt{n}(\hat{q}^{adj}(\tau) - q(\tau))$ established in Section \ref{sec:thm_est_pf}, we have
\begin{align*}
&\sqrt{n}(\hat{q}^{w}(\tau) - \hat{q}^{adj}(\tau)) = \sum_{s \in \mathcal{S}}\frac{1}{\sqrt{n}}\sum_{i =1}^n(\xi_i-1)A_i1\{S_i=s\} \phi_1(\tau,s,Y_i(1),X_i) \\
& - \sum_{s \in \mathcal{S}}\frac{1}{\sqrt{n}}\sum_{i =1}^n(\xi_i-1)(1-A_i)1\{S_i=s\}\phi_0(\tau,s,Y_i(0),X_i) + \frac{1}{\sqrt{n}} \sum_{i =1}^n(\xi_i-1)\phi_s(\tau,S_i) + R(\tau) \\
&= \varpi_{n,1}^w(\tau) - \varpi_{n,2}^w(\tau) + R(\tau), 
\end{align*}
where $\sup_{\tau \in \Upsilon}|R(\tau)| = o_p(1)$,
\begin{align*}
\varpi_{n,1}^w(\tau) = & \sum_{s \in \mathcal{S}}\frac{1}{\sqrt{n}}\sum_{i =1}^n(\xi_i-1)A_i1\{S_i=s\} \phi_1(\tau,s,Y_i(1),X_i)  \\
& -  \sum_{s \in \mathcal{S}}\frac{1}{\sqrt{n}}\sum_{i =1}^n(\xi_i-1)(1-A_i)1\{S_i=s\}\phi_0(\tau,s,Y_i(0),X_i),
\end{align*}
and 
\begin{align*}
\varpi_{n,2}^w(\tau) = \frac{1}{\sqrt{n}} \sum_{i =1}^n(\xi_i-1)\phi_s(\tau,S_i).
\end{align*}

Lemma \ref{lem:w_boot} shows that, uniformly over $\tau \in \Upsilon$ and conditionally on data,   
\begin{align*}
\varpi_{n,1}^w(\tau) +\varpi_{n,2}^w(\tau) \convD \mathcal{B}(\tau),
\end{align*}
where $\mathcal{B}(\tau)$ is the Gaussian process with covariance kernel
\begin{align*}
&\Sigma(\tau,\tau') = \mathbb{E} \pi(S_i)\phi_1(\tau,S_i,Y_i(1),X_i)\phi_1(\tau',S_i,Y_i(1),X_i)  \\
& + \mathbb{E} (1-\pi(S_i))\phi_0(\tau,S_i,Y_i(0),X_i)\phi_0(\tau',S_i,Y_i(0),X_i)  + \mathbb{E}\phi_s(\tau,S_i)\phi_s(\tau',S_i), 
\end{align*}
as defined in Theorem \ref{thm:est}. This concludes the proof.

\section{Proof of Theorem \ref{thm:par}}
\begin{table}[H]
	%\label{Table:constants}
	\begin{adjustbox}{max width=\textwidth}
		\begin{tabular}{c|l}
			\hline
			Name & Description \\ \hline
			$N(s)$ & For $s \in \mathcal{S}$, $N(s) = \sum_{i =1}^n1\{S_i <s\}$ \\
			$\Lambda_{\tau,s}(x,\theta_{a,s}(\tau))$ & For $a \in \{0,1\}$,  $s \in \mathcal{S}$, $\tau \in \Upsilon$, and $x \in \Supp(X)$, $\Lambda_{\tau,s}(x,\theta_{a,s}(\tau))$ is a parametric model for $\mathbb{P}(Y(a) \leq q_a(\tau)|S_i=s,X_i=x)$ \\
			&  with a pseudo true value $\theta_{a,s}(\tau)$ \\
			$\hat{\theta}_{a,s}(\tau)$  & For $a \in \{0,1\}$,  $s \in \mathcal{S}$, $\tau \in \Upsilon$, $\hat{\theta}_{a,s}(\tau)$ is a consistent estimator of $\theta_{a,s}(\tau)$
			\\ \hline
		\end{tabular}
	\end{adjustbox}
	%	\caption{Table of notation}
\end{table}

The proof is divided into two steps. In the first step, we show Assumption \ref{ass:mhatw}. Assumption \ref{ass:mhat}(i) can be shown in the same manner and is omitted. In the second step, we establish Assumptions \ref{ass:mhat}(ii) and \ref{ass:mhat}(iii).

\vspace{1.5mm}
\noindent \textbf{Step 1.} Recall $\widehat{m}_a(\tau,s,X_i) = \tau - \Lambda_{\tau,s}(x,\hat{\theta}_{a,s}(\tau))$, $\overline{m}_a(\tau,s,X_i) = \tau - \Lambda_{\tau,s}(x,\theta_{a,s}(\tau))$, 
\begin{align*}
\overline{\Delta}_a(\tau,s,X_i) = \widehat{m}_a(\tau,s,X_i) - \overline{m}_a(\tau,s,X_i) =\Lambda_{\tau,s}(X_i,\theta_{a,s}(\tau))- \Lambda_{\tau,s}(X_i,\hat{\theta}_{a,s}(\tau)), 
\end{align*}
and $\{X_i^s,\xi^s_i\}_{i \in [n]}$ is generated independently from the joint distribution of $(X_i,\xi_i)$ given $S_i=s$, and so is independent of  $\{A_i,S_i\}_{i \in [n]}$. 
Let $\Psi_{\tau,s}(\theta_1,\theta_2) = \mathbb{E}[\Lambda_{\tau,s}(X_i,\theta_1) - \Lambda_{\tau,s}(X_i,\theta_{2})|S_i=s] = \mathbb{E}[\Lambda_{\tau,s}(X_i^s,\theta_1) - \Lambda_{\tau,s}(X_i^s,\theta_{2})]$. We have
\begin{align}
& \sup_{\tau \in \Upsilon,s\in \mathcal{S}}\biggl|\frac{\sum_{i\in I_1(s)}\xi_i\overline{\Delta}_a(\tau,s,X_i)}{n_1^w(s)} - \frac{\sum_{i \in I_0(s)}\xi_i\overline{\Delta}_a(\tau,s,X_i)}{n_0^w(s)}\biggr| \notag  \\
&\leq  \left(\max_{s \in \mathcal{S}}\frac{n_1(s)}{n_1^w(s)} \right)\sup_{\tau \in \Upsilon,s\in \mathcal{S}}\left|\frac{\sum_{i\in I_1(s)}\xi_i[ \overline{\Delta}_a(\tau,s,X_i) - \Psi_{\tau,s}(\theta_{a,s}(\tau),\hat{\theta}_{a,s}(\tau))] }{n_1(s)} \right| \notag \\
& +  \left(\max_{s \in \mathcal{S}}\frac{n_0(s)}{n_0^w(s)}\right)\sup_{\tau \in \Upsilon,s\in \mathcal{S}}\left|\frac{\sum_{i\in I_0(s)}\xi_i[ \overline{\Delta}_a(\tau,s,X_i) - \Psi_{\tau,s}(\theta_{a,s}(\tau),\hat{\theta}_{a,s}(\tau))] }{n_0(s)} \right| \notag \\
&= o_p(n^{-1/2}).
\label{eq:Delta}
\end{align}
To see the last equality, we note that, for any $\eps>0$, with probability approaching one (w.p.a.1), we have
\begin{align*}
\sup_{\tau \in \Upsilon,s\in \mathcal{S}}||\hat{\theta}_{a,s}(\tau) - \theta_{a,s}(\tau)|| \leq \eps.
\end{align*}
Therefore, on the event 
\begin{align*}
\mathcal{A}_n(\eps) \equiv \left\{\sup_{\tau \in \Upsilon,s \in \mathcal{S}}||\hat{\theta}_{a,s}(\tau) - \theta_{a,s}(\tau)|| \leq \eps, \sup_{\tau \in \Upsilon}\max(||\hat{\theta}_{a,s}(\tau)||,||\theta_{a,s}(\tau)||) \leq C,\min_{s \in \mathcal{S}}n_1(s) \geq \eps n \right\}    
\end{align*}
we have
\begin{align*}
& \sup_{\tau \in \Upsilon}\left|\frac{\sum_{i\in I_1(s)}\xi_i[ \overline{\Delta}_a(\tau,s,X_i) - \Psi_{\tau,s}(\theta_{a,s}(\tau),\hat{\theta}_{a,s}(\tau))] }{n_1(s)} \right| \biggl|\{A_i,S_i\}_{i \in [n]} \\
&\stackrel{d}{=} \sup_{\tau \in \Upsilon}\left|\frac{\sum_{i = N(s)+1}^{N(s)+n_1(s)}\xi_i^s[ \overline{\Delta}_a(\tau,s,X_i^s) - \Psi_{\tau,s}(\theta_{a,s}(\tau),\hat{\theta}_{a,s}(\tau))] }{n_1(s)} \right| \biggl|\{A_i,S_i\}_{i \in [n]} \\
&\leq ||\mathbb{P}_{n_1(s)} - \mathbb{P}||_{\mathcal{F}}\biggl|\{A_i,S_i\}_{i \in [n]},
\end{align*}
where $N(s) = \sum_{i =1}^n1\{S_i <s\}$ and 
\begin{align*}
\mathcal{F} = \left\{ \xi_i^s[ \Lambda_{\tau,s}(X_i^s,\theta_1) - \Lambda_{\tau,s}(X_i^s,\theta_{2}) - \Psi_{\tau,s}(\theta_1,\theta_2)]: \tau \in \Upsilon, ||\theta_1-\theta_2||\leq \eps, \max(||\theta_1||,||\theta_2||) \leq C \right\}.
\end{align*}
By Assumption \ref{ass:par}, $\mathcal{F}$ is a VC-class with a fixed VC index and envelope $L_i$. In addition, 
\begin{align*}
\sup_{f \in \mathcal{F}}\mathbb{P}f^2 \leq \mathbb{E}L_i^2(\theta_1-\theta_2)^2 \leq C\eps^2.
\end{align*}
Therefore, for any $\delta>0$ we have
\begin{align*}
& \mathbb{P}\left( \sup_{\tau \in \Upsilon,s\in \mathcal{S}}\left|\frac{\sum_{i\in I_1(s)}\xi_i[ \overline{\Delta}_a(\tau,s,X_i) - \Psi_{\tau,s}(\theta_{a,s}(\tau),\hat{\theta}_{a,s}(\tau))] }{n_1(s)} \right|\geq \delta n^{-1/2}\right) \\
&\leq \mathbb{P}\left(\sup_{\tau \in \Upsilon,s\in \mathcal{S}}\left|\frac{\sum_{i\in I_1(s)}\xi_i[ \overline{\Delta}_a(\tau,s,X_i) - \Psi_{\tau,s}(\theta_{a,s}(\tau),\hat{\theta}_{a,s}(\tau))] }{n_1(s)} \right| \geq \delta n^{-1/2}, \mathcal{A}_n(\eps) \right) + \mathbb{P}(\mathcal{A}_n^c(\eps)) \\
&\leq  
\mathbb{E}\left[\mathbb{P}\left( \sup_{\tau \in \Upsilon,s\in \mathcal{S}}\left|\frac{\sum_{i\in I_1(s)}\xi_i[ \overline{\Delta}_a(\tau,s,X_i) - \Psi_{\tau,s}(\theta_{a,s}(\tau),\hat{\theta}_{a,s}(\tau))] }{n_1(s)} \right| \geq \delta n^{-1/2}, \mathcal{A}_n(\eps)\biggl|\{A_i,S_i\}_{i \in [n]}\right)\right] \\
& + \mathbb{P}(\mathcal{A}_n^c(\eps)) \\
&\leq  \sum_{s \in \mathcal{S}}\mathbb{E}\left[\mathbb{P}\left(||\mathbb{P}_{n_1(s)} - \mathbb{P}||_{\mathcal{F}} \geq \delta n^{-1/2} \biggl|\{A_i,S_i\}_{i \in [n]}\right)1\{n_1(s) \geq n\eps\}\right] + \mathbb{P}(\mathcal{A}_n^c(\eps)) \\
&\leq  \sum_{s \in \mathcal{S}}\mathbb{E}\left\{ \frac{n^{1/2}\mathbb{E}\left[||\mathbb{P}_{n_1(s)} - \mathbb{P}||_{\mathcal{F}}|\{A_i,S_i\}_{i \in [n]}\right]1\{n_1(s) \geq n\eps\}}{\delta}\right\} +  \mathbb{P}(\mathcal{A}_n^c(\eps)).
\end{align*}
By \citet[Corollary 5.1]{CCK14}, 
\begin{align*}
& n^{1/2}\mathbb{E}\left[||\mathbb{P}_{n_1(s)} - \mathbb{P}||_{\mathcal{F}}|\{A_i,S_i\}_{i \in [n]}\right]1\{n_1(s) \geq n\eps\} \\
&\leq  C(\sqrt{\frac{n}{n_1(s)}\eps^2} + n^{1/2}n_1^{1/d-1}(s))1\{n_1(s) \geq n\eps\} \\
&\leq C (\eps^{1/2} + n^{1/d-1/2}\eps^{1/d-1}).
\end{align*}
Therefore, 
\begin{align*}
\mathbb{E}\left\{ \frac{n^{1/2}\mathbb{E}\left[||\mathbb{P}_{n_1(s)} - \mathbb{P}||_{\mathcal{F}}|\{A_i,S_i\}_{i \in [n]}\right]1\{n_1(s) \geq n\eps\}}{\delta}\right\}  \leq C\mathbb{E}\left(\eps^{1/2} + n^{1/d-1/2}\eps^{1/d-1}\right)/\delta.
\end{align*}
By letting $n \rightarrow \infty$ followed by $\eps \rightarrow 0$, we have 
\begin{align*}
\lim_{n\rightarrow \infty} \mathbb{P}\left( \sup_{\tau \in \Upsilon,s\in \mathcal{S}}\left|\frac{\sum_{i\in I_1(s)}\xi_i[ \overline{\Delta}_a(\tau,s,X_i) - \Psi_{\tau,s}(\theta_{a,s}(\tau),\hat{\theta}_{a,s}(\tau))] }{n_1(s)} \right|\geq \delta n^{-1/2}\right) = 0,
\end{align*}
In addition, 
\begin{align*}
\max_{s \in \mathcal{S}}|n_1^w(s)/n_1(s)-1| = \max_{s\in \mathcal{S}}|(D_n^w(s) - D_n(s))/(\pi(s) n(s) + D_n(s))|\convP 1,
\end{align*}
as Lemma \ref{lem:Dw} shows that $\max_{s\in \mathcal{S}}|(D_n^w(s) - D_n(s))/n(s)| = o_p(1)$. 

Therefore, 
\begin{align*}
\sup_{\tau \in \Upsilon,s\in \mathcal{S}}\left|\frac{\sum_{i\in I_1(s)}\xi_i[ \overline{\Delta}_a(\tau,s,X_i) - \Psi_{\tau,s}(\theta_{a,s}(\tau),\hat{\theta}_{a,s}(\tau))] }{n^w_1(s)} \right| = o_p(n^{-1/2}).
\end{align*}
For the same reason, we have 
\begin{align*}
\sup_{\tau \in \Upsilon,s\in \mathcal{S}}\left|\frac{\sum_{i\in I_0(s)}\xi_i[ \overline{\Delta}_a(\tau,s,X_i) - \Psi_{\tau,s}(\theta_{a,s}(\tau),\hat{\theta}_{a,s}(\tau))] }{n^w_0(s)} \right| = o_p(n^{-1/2}),
\end{align*}
and \eqref{eq:Delta} holds. 
%
%
%Lemma \ref{lem:thetahat_lp} implies 
%\begin{align*}
%\sup_{\tau \in \Upsilon}|\theta_{a,s}^\textit{LP}(\tau)-\hat{\theta}_{a,s}^\textit{LP}(\tau)| = O_p(n^{-1/2}). 
%\end{align*}
%In addition, by the same argument in the proof of Lemma \ref{lem:max_eq}, we have 
%\begin{align*}
%\left|\frac{\sum_{i\in I_a(s)}\xi_i X_i^\top }{n_a(s)} - \mathbb{E}(X_i^\top|S_i=s)\right| \stackrel{d}{=}\left|\frac{\sum_{i\in I_a(s)}\xi^s_i X^s_i}{n_a(s)} - \mathbb{E}(\xi_i^sX_i^s)\right| = O_p(n^{-1/2}). 
%\end{align*}

\vspace{1.5mm}
\noindent \textbf{Step 2.}
By Assumption \ref{ass:par}, 
\begin{align*}
& |\overline{m}_a(\tau_2,S_i,X_i) - \overline{m}_a(\tau_1,S_i,X_i)| \\
&\leq  |\tau_2-\tau_1| + 
|\Lambda_{\tau_1,s}(X_i,\theta_{a,s}(\tau_1)) - \Lambda_{\tau_2,s}(X_i,\theta_{a,s}(\tau_2))| \\
&\leq (1+L_1)  |\tau_2-\tau_1| + L_i|\theta_{a,s}(\tau_1)-\theta_{a,s}(\tau_2)| \leq (CL_i+1)|\tau_2-\tau_1|. 
\end{align*} 
This implies Assumption \ref{ass:mhat}(iii). Furthermore, by Assumption $\ref{ass:par}$ we can let the envelope for the class of functions $\mathcal{F} = \{ \overline{m}_a(\tau_2,S_i,X_i):\tau \in \Upsilon\}$ be $F_i = \max(C,1)L_i+1$ where the constant $C$ is the one in the above display.  Then, we have
\begin{align*}
\sup_Q N(\mathcal{F},e_{Q},\eps ||F||_{Q,2}) \leq N(\Upsilon,d,\eps) \leq 1/\eps,
\end{align*}
where $d(\tau_1,\tau_2) = |\tau_1-\tau_2|$. This verifies Assumption \ref{ass:mhat}(ii). 
%\textbf{Step 3.} Given Assumptions \ref{ass:assignment1}--\ref{ass:mhatw} hold, Theorems \ref{thm:est} and \ref{thm:boot} lead to the desired result. Here we compute the covariance kernel. Following the proof of Theorem \ref{thm:est}, we have
%\begin{align*}
%& \phi_1^\textit{LP}(\tau,s,Y_i(1),X_i) = \frac{\eta_{i,1}(\tau,s) + (1-\pi(s))\tilde{X}_{i,s} \theta_{1,s}^\textit{LP}(\tau)}{\pi(s)f_1(q_1(\tau)) }+ \frac{\tilde{X}_{i,s}\theta_{0,s}^\textit{LP}(\tau)}{f_0(q_0(\tau))}, \\
%& \phi_0^\textit{LP}(\tau,s,Y_i(0),X_i) = \frac{\eta_{i,0}(\tau,s) + \pi(s)\tilde{X}_{i,s} \theta_{0,s}^\textit{LP}(\tau)}{(1-\pi(s))f_0(q_0(\tau))} + \frac{\tilde{X}_{i,s}\theta_{1,s}^\textit{LP}(\tau)}{f_1(q_1(\tau))}, \\
%& \phi_s(\tau,S_i) = \frac{m_1(\tau,S_i)}{f_1(q_1(\tau))}-\frac{m_0(\tau,S_i)}{f_0(q_0(\tau))},
%\end{align*}
%and 
%\begin{align*}
%\Sigma^\textit{LP}(\tau,\tau') = & \mathbb{E}\pi(S_i)\phi_1^\textit{LP}(\tau,s,Y_i(1),X_i)\phi_1^\textit{LP}(\tau',s,Y_i(1),X_i) \\
%& + \mathbb{E}(1-\pi(S_i))\phi_0^\textit{LP}(\tau,s,Y_i(0),X_i)\phi_0^\textit{LP}(\tau',s,Y_i(0),X_i) \\
%& + \mathbb{E}\phi_s(\tau,S_i)\phi_s(\tau',S_i). 
%\end{align*}
%This concludes the proof. 

\section{Proof of Theorem \ref{thm:theta_as^star}}
\begin{table}[H]
	%\label{Table:constants}
	\begin{adjustbox}{max width=\textwidth}
		\begin{tabular}{c|l}
			\hline
			Name & Description \\ \hline
			$W_{i,s}(\tau)$ & For $i \in [n]$, $s \in \mathcal{S}$, and $\tau \in \Upsilon$, $W_{i,s}(\tau)$ is the linear regressor in the linear adjustment so that  $\Lambda_{\tau,s}(X_i,\theta_{a,s}(\tau)) = W_{i,s}^\top(\tau) \theta_{a,s}(\tau)$ \\
			$\tilde{W}_{i,s}(\tau)$ & For $i \in [n]$, $s \in \mathcal{S}$, and $\tau \in \Upsilon$, $\tilde{W}_{i,s}(\tau) = W_{i,s}(\tau) - \mathbb{E}(W_{i,s}(\tau)|S_i=s)$ \\
			$\theta_{a,s}(\tau)$ & For $a \in \{0,1\}$, $s \in \mathcal{S}$, and $\tau \in \Upsilon$, $\theta_{a,s}(\tau)$ is the pseudo true value in the linear adjustment
			\\ \hline
		\end{tabular}
	\end{adjustbox}
	%	\caption{Table of notation}
\end{table}
Let $\Sigma^\textit{LP}(\tau_k,\tau_l)$ be the asymptotic covariance matrix of $\hat{q}^{adj}(\tau_k)$ and $\hat{q}^{adj}(\tau_l)$ with a linear adjustment and pseudo true values $(\theta_{1,s}(\tau_k),\theta_{0,s}(\tau_k))_{k \in [K]}$. Then, we have
\begin{align*}
\Sigma^\textit{LP}(\tau_k,\tau_l) = \Sigma^\textit{LP}_1(\tau_k,\tau_l) + \sum_{s \in \mathcal{S}}p(s) \mathbb{E}\left[(\tilde{W}_{i,s}^\top(\tau_k)\beta_s(\tau_k) - \overline{y}_{i,s}(\tau_k))(\tilde{W}_{i,s}^\top(\tau_l)\beta_s(\tau_l) - \overline{y}_{i,s}(\tau_l))|S_i=s\right],
\end{align*} 
where 
\begin{align*}
\Sigma^\textit{LP}_1(\tau_k,\tau_l) &=  \biggl\{\mathbb{E}\left[\frac{(\tau - 1\{Y_i(1) \leq q_1(\tau)\} - m_1(\tau,S_i,X_i))^2}{\pi(S_i)f_1^2(q_1(0))} \right] \\
& + \mathbb{E}\left[\frac{(\tau - 1\{Y_i(0) \leq q_0(\tau)\} - m_0(\tau,S_i,X_i))^2}{(1-\pi(S_i))f_0^2(q_0(0))} \right] \\
&+ \mathbb{E}\left(\frac{m_1(\tau,S_i,X_i)}{f_1(q_1(\tau))} -  \frac{m_0(\tau,S_i,X_i)}{f_0(q_0(\tau))}\right)^2 \biggr\},\\
\beta_s(\tau)& =  \sqrt{\frac{1-\pi(s)}{\pi(s)}}\frac{\theta_{1,s}(\tau)}{f_1(q_1(\tau))} +  \sqrt{\frac{\pi(s)}{1-\pi(s)}}\frac{\theta_{0,s}(\tau)}{f_0(q_0(\tau))}, \quad \text{and}\\
\overline{y}_{i,s}(\tau) &= \sqrt{\frac{1-\pi(s)}{\pi(s)}}\frac{\left[\mathbb{P}(Y_i(1) \leq q_1(\tau)|X_i,S_i=s) -\mathbb{P}(Y_i(1) \leq q_1(\tau)|S_i=s)\right] }{f_1(q_1(\tau))} \\
& +  \sqrt{\frac{\pi(s)}{1-\pi(s)}}\frac{\left[\mathbb{P}(Y_i(0) \leq q_0(\tau)|X_i,S_i=s) -\mathbb{P}(Y_i(0) \leq q_0(\tau)|S_i=s)\right]}{f_0(q_0(\tau))}.
\end{align*}

To minimize $[\Sigma^\textit{LP}(\tau_k,\tau_l) ]_{k,l\in [K]}$ (in the matrix sense) is the same as minimizing $$\left[\mathbb{E}\left[(\tilde{W}_{i,s}^\top(\tau_k)\beta_s(\tau_k) - \overline{y}_{i,s}(\tau_k))(\tilde{W}_{i,s}^\top(\tau_l)\beta_s(\tau_l) - \overline{y}_{i,s}(\tau_l))|S_i=s\right]\right]_{k,l\in [K]}$$ 
for each $s \in \mathcal{S}$, which is achieved if
\begin{align}
\beta_s(\tau_k) =  [\mathbb{E}\tilde{W}_{i,s}(\tau_k)\tilde{W}_{i,s}^\top(\tau_k)|S_i=s]^{-1}\mathbb{E}[\tilde{W}_{i,s}(\tau_k)\overline{y}_{i,s}(\tau_k)|S_i=s].
\label{eq:beta_sk}
\end{align} 
Because $\mathbb{E}[\tilde{W}_{i,s}(\tau)\mathbb{P}(Y_i(a) \leq q_a(\tau)|S_i=s)|S_i=s] = 0$ for $a=0,1$, \eqref{eq:beta_sk} implies 
\begin{align*}
& \sqrt{\frac{1-\pi(s)}{\pi(s)}}\frac{\theta_{1,s}(\tau_k)}{f_1(q_1(\tau_k))} +  \sqrt{\frac{\pi(s)}{1-\pi(s)}}\frac{\theta_{0,s}(\tau_k)}{f_0(q_0(\tau_k))} \\
&= \sqrt{\frac{1-\pi(s)}{\pi(s)}}\frac{\theta^\textit{LP}_{1,s}(\tau_k)}{f_1(q_1(\tau_k))} +  \sqrt{\frac{\pi(s)}{1-\pi(s)}}\frac{\theta^\textit{LP}_{0,s}(\tau_k)}{f_0(q_0(\tau_k))},
\end{align*}
or equivalently, 
\begin{align*}
& \frac{\theta_{1,s}(\tau_k)}{f_1(q_1(\tau_k))} +  \frac{\pi(s)}{1-\pi(s)}\frac{\theta_{0,s}(\tau_k)}{f_0(q_0(\tau_k))} = \frac{\theta^\textit{LP}_{1,s}(\tau_k)}{f_1(q_1(\tau_k))} +  \frac{\pi(s)}{1-\pi(s)}\frac{\theta^\textit{LP}_{0,s}(\tau_k)}{f_0(q_0(\tau_k))}.
\end{align*}
This concludes the proof. 

\section{Proof of Theorem \ref{prop:lp}}
\begin{table}[H]
	%\label{Table:constants}
	\begin{adjustbox}{max width=\textwidth}
		\begin{tabular}{c|l}
			\hline
			Name & Description \\ \hline
			$q_a(\tau)$ & For $a =0,1$ and $\tau \in \Upsilon$, $q_a(\tau)$ is the $\tau$th quantile of $Y(a)$ \\
			$\theta_{a,s}^\textit{LP}(\tau)$ & For $a \in \{0,1\}$, $s \in \mathcal{S}$, and $\tau \in \Upsilon$, $\theta_{a,s}^\textit{LP}(\tau) = \left[\mathbb{E}(\tilde{W}_{i,s}(\tau) \tilde{W}_{i,s}^\top(\tau)|S_i=s)\right]^{-1}\mathbb{E}\left[\tilde{W}_{i,s}(\tau) 1\{Y_i(a) \leq q_a(\tau)\}|S_i=s\right]$ is the optimal linear coefficient\\
			% $\tilde{X}_{i,s}$ & For $i \in [n]$, $s \in \mathcal{S}$, $\tilde{X}_{i,s} = X_i - \mathbb{E}(X_i|S_i=s)$ \\
			$\dot{W}_{i,a,s}(\tau)$ & For $a \in \{0,1\}$, $i \in [n]$, $s \in \mathcal{S}$, $\dot{W}_{i,a,s}(\tau) = W_{i,s}(\tau) - \frac{1}{n_a(s)}\sum_{i \in I_{a}(s)}W_{i,s}(\tau)$, $I_a(s) = \{i \in [n]: A_i =a,S_i=s\}$ \\
			$\hat{\theta}_{a,s}^\textit{LP}(\tau)$ & For $a \in \{0,1\}$, $s \in \mathcal{S}$, and $\tau \in \Upsilon$, $\hat{\theta}_{a,s}^\textit{LP}(\tau)$ is defined in \eqref{eq:thetahat_lp_est} \\
			$\hat{q}_a(\tau)$ & For $a \in \{0,1\}$ and $\tau \in \Upsilon$, $\hat{q}_a(\tau)$ is the estimator of $q_a(\tau)$ without any adjustments \\ 
			\hline
		\end{tabular}
	\end{adjustbox}
	%	\caption{Table of notation}
\end{table}
Assumption \ref{ass:par}(i) holds by Assumption \ref{ass:regX}. In addition, by Assumption \ref{ass:tau}, we have $\sup_{\tau \in \Upsilon}|\partial_\tau \theta_{a,s}^\textit{LP}(\tau)| < \infty$. This implies Assumption \ref{ass:par}(ii). Next, we aim to show 
\begin{align*}
\sup_{\tau \in \Upsilon, a = 0,1, s \in \mathcal{S}}|\hat{\theta}_{a,s}^\textit{LP}(\tau) - \theta_{a,s}^\textit{LP}(\tau)| = O_p(n^{-1/2}). 
\end{align*}
Focusing on $\hat{\theta}_{1,s}^\textit{LP}(\tau)$ we have
\begin{align}
\hat{\theta}_{1,s}^\textit{LP}(\tau) - \theta_{1,s}^\textit{LP}(\tau) & = \left[\frac{1}{n_{1}(s)}\sum_{i \in I_1(s)}\dot{W}_{i,1,s}(\tau)\dot{W}_{i,1,s}^\top(\tau) \right]^{-1} \notag \\
& \times \left[\frac{1}{n_{1}(s)}\sum_{i \in I_1(s)}\dot{W}_{i,1,s}(\tau)(1\{Y_i \leq \hat{q}_1(\tau)\} - \dot{W}_{i,1,s}^\top(\tau)\theta_{1,s}^\textit{LP}(\tau))\right].
\label{eq:thetahat_lp}
\end{align}
For the first term in \eqref{eq:thetahat_lp}, we have 
\begin{align*}
\frac{1}{n_{1}(s)}\sum_{i \in I_1(s)}\dot{W}_{i,1,s}(\tau)\dot{W}_{i,1,s}^\top(\tau)  \stackrel{d}{=} \frac{1}{n_{1}(s)}\sum_{i =N(s)+1}^{N(s)+n_1(s)}\dot{W}_{i,1,s}^s(\tau)\dot{W}_{i,1,s}^{s \top}(\tau), 
\end{align*}
where $\dot{W}_{i,1,s}^s(\tau) = W_{i,s}^s(\tau) - \frac{1}{n_{1}(s)}\sum_{i =N(s)+1}^{N(s)+n_1(s)}W_{i,s}^s(\tau)$ and $W_{i,s}^s(\tau)$ is i.i.d. across $i$ with common distribution equal to the conditional distribution of $W_{i,s}(\tau)$ given $S_i=s$ and independent of $N(s),n_1(s)$. Therefore, by Assumption \ref{ass:regX2}, we have
\begin{align*}
& \sup_{s \in \mathcal{S},\tau \in \Upsilon}\left\Vert \frac{1}{n_{1}(s)}\sum_{i =N(s)+1}^{N(s)+n_1(s)}\dot{W}_{i,1,s}^s(\tau)\dot{W}_{i,1,s}^{s \top}(\tau) -  \mathbb{E}(W_{i,s}^s(\tau) - \mathbb{E}W_{i,s}^s(\tau))(W_{i,s}^s(\tau) - \mathbb{E}W_{i,s}^s(\tau))^\top \right\Vert_F \\
& = \sup_{s \in \mathcal{S},\tau \in \Upsilon}\left\Vert \frac{1}{n_{1}(s)}\sum_{i =N(s)+1}^{N(s)+n_1(s)}\dot{W}_{i,1,s}^s(\tau)\dot{W}_{i,1,s}^{s \top}(\tau) -  \mathbb{E}(\tilde{W}_{i,s}(\tau)\tilde{W}_{i,s}^\top(\tau)|S_i=s) \right\Vert_F  = o_p(1),
\end{align*}
where $||\cdot||_F$ denotes the Frobenius norm and $\tilde{W}_{i,s}(\tau) = W_{i,s}(\tau)-\mathbb{E}(W_{i,s}(\tau)|S_i=s)$. For the second term in \eqref{eq:thetahat_lp}, we have
\begin{align*}
& \frac{1}{n_{1}(s)}\sum_{i \in I_1(s)}\dot{W}_{i,1,s}(\tau)\left(1\{Y_i \leq \hat{q}_1(\tau)\} - \dot{W}_{i,1,s}^\top(\tau)\theta_{1,s}^\textit{LP}(\tau)\right)  \\
&= \frac{1}{n_{1}(s)}\sum_{i \in I_1(s)}\tilde{W}_{i,s}(\tau)\left(1\{Y_i \leq \hat{q}_1(\tau)\} - \dot{W}_{i,1,s}^\top(\tau)\theta_{1,s}^\textit{LP}(\tau)\right)  + R_1(\tau) \\
&=  \frac{1}{n_{1}(s)}\sum_{i \in I_1(s)}\tilde{W}_{i,s}(\tau)\left(1\{Y_i \leq \hat{q}_1(\tau)\} - \tilde{W}_{i,s}^\top(\tau)\theta_{1,s}^\textit{LP}(\tau)\right)  + R_1(\tau)  + R_2(\tau) \\
&= \left[\frac{1}{n_{1}(s)}\sum_{i \in I_1(s)}\tilde{W}_{i,s}(\tau)(1\{Y_i \leq \hat{q}_1(\tau)\} -1\{Y_i \leq q_1(\tau)\}) \right] \\
& + \left[\frac{1}{n_{1}(s)}\sum_{i \in I_1(s)}\tilde{W}_{i,s}(\tau)(1\{Y_i \leq q_1(\tau)\} - \tilde{W}_{i,s}^\top(\tau)\theta_{1,s}^\textit{LP}(\tau)) \right] +  R_1(\tau)  + R_2(\tau) \\
&\equiv  I(\tau) + II(\tau) + R_1(\tau) + R_2(\tau),
\end{align*}
where 
\begin{align*}
R_1(\tau) = & -\left(\frac{1}{n_1(s)}\sum_{i \in I_1(s)}W_{i,s}(\tau) - \mathbb{E}(W_{i,s}(\tau)|S_i=s)\right)  \left(\frac{1}{n_1(s)}\sum_{i \in I_1(s)}(1\{Y_i \leq \hat{q}_1(\tau)\} - \dot{W}_{i,1,s}^\top(\tau)\theta_{1,s}^\textit{LP}(\tau)) \right),
\end{align*}
and 
\begin{align*}
R_2(\tau) = \left(\frac{1}{n_1(s)}\sum_{i \in I_1(s)}W_{i,s}(\tau) - \mathbb{E}(W_{i,s}(\tau)|S_i=s)\right) \left(\frac{1}{n_1(s)}\sum_{i \in I_1(s)}\tilde{W}_{i,s}^\top(\tau) \theta_{1,s}^\textit{LP}(\tau)\right).
\end{align*}
By Assumption \ref{ass:regX} we can show that $\sup_{\tau\in \Upsilon}|\theta_{1,s}^\textit{LP}(\tau)| \leq C<\infty$ for some constant $C>0$. Therefore, we have
\begin{align*}
\sup_{\tau \in \Upsilon}|R_1(\tau)| &= \sup_{\tau \in \Upsilon}\left|\frac{1}{n_1(s)}\sum_{i \in I_1(s)}W_{i,s}(\tau) - \mathbb{E}(W_{i,s}(\tau)|S_i=s)\right|\left|\frac{1}{n_1(s)}\sum_{i \in I_1(s)}(1\{Y_i \leq \hat{q}_1(\tau)\} - \dot{W}_{i,1,s}^\top(\tau)\theta_{1,s}^\textit{LP}(\tau)) \right| \\
&=  O_p(n^{-1/2}),
\end{align*}
and 
\begin{align*}
\sup_{\tau \in \Upsilon}|R_2(\tau)| &= \sup_{\tau \in \Upsilon}\left|\frac{1}{n_1(s)}\sum_{i \in I_1(s)}W_{i,s}(\tau) - \mathbb{E}(W_{i,s}(\tau)|S_i=s)\right| \left|\frac{1}{n_1(s)}\sum_{i \in I_1(s)}\tilde{W}_{i,s}^\top(\tau)\theta_{1,s}^\textit{LP}(\tau)\right|\\
&= O_p(n^{-1/2}), 
\end{align*}
where we use the fact that, by Assumption \ref{ass:regX2}, 
\begin{align*}
\sup_{\tau \in \upsilon}\left|\frac{1}{n_1(s)}\sum_{i \in I_1(s)}W_{i,s}(\tau) - \mathbb{E}(W_{i,s}(\tau)|S_i=s) \right| \stackrel{d}{=} \sup_{\tau \in \upsilon}\left|\frac{1}{n_1(s)}\sum_{i = N(s)+1}^{N(s)+n_1(s)}(W_{i,s}^s(\tau) - \mathbb{E}W_{i,s}^s(\tau)) \right| = O_p(n^{-1/2}). 
\end{align*}

Next, note that $\sup_{\tau \in \Upsilon}|\hat{q}(\tau) - q(\tau)| = O_p(n^{-1/2})$, which means for any $\eps>0$, there exists a constant $M>0$ such that $\sup_{\tau \in \Upsilon}|\hat{q}(\tau) - q(\tau)| \leq Mn^{-1/2}$ with probability greater than $1-\eps$. On the event set that $\sup_{\tau \in \Upsilon}|\hat{q}(\tau) - q(\tau)| \leq Mn^{-1/2}$, we have
\begin{align*}
& \sup_{\tau \in \Upsilon}|I(\tau)| \leq  \sup_{\tau \in \Upsilon}\biggl| \frac{1}{n_1(s)}\sum_{i \in I_1(s)}\tilde{W}_{i,s}(\tau)\biggl(1\{Y_i(1) \leq \hat{q}_1(\tau)\} -1\{Y_i(1) \leq q_1(\tau)\} \\
& - \mathbb{P}(Y_i(1) \leq \hat{q}_1(\tau)|X_i,S_i=s) + \mathbb{P}(Y_i(1) \leq q_1(\tau)|X_i,S_i=s)\biggr)\biggr|\\
& +  \sup_{\tau \in \Upsilon}\biggl| \frac{1}{n_1(s)}\sum_{i \in I_1(s)}\tilde{W}_{i,s}(\tau)(\mathbb{P}(Y_i(1) \leq \hat{q}_1(\tau)|X_i,S_i=s) - \mathbb{P}(Y_i(1) \leq q_1(\tau)|X_i,S_i=s))\biggr| \\
&\leq  \sup_{|q-q'|\leq M n^{-1/2}}\biggl| \frac{1}{n_1(s)}\sum_{i \in I_1(s)}\tilde{W}_{i,s}(\tau)\biggl(1\{Y_i(1) \leq q\} -1\{Y_i(1) \leq q'\} \\
& - \mathbb{P}(Y_i(1) \leq q|X_i,S_i=s) + \mathbb{P}(Y_i(1) \leq q'|X_i,S_i=s)\biggr)\biggr| + C\sup_{\tau \in \Upsilon}|\hat{q}(\tau) - q(\tau)| \\
&\leq   \sup_{|q-q'|\leq M n^{-1/2}}\biggl| \frac{1}{n_1(s)}\sum_{i \in I_1(s)}\tilde{W}_{i,s}(\tau)\biggl(1\{Y_i(1) \leq q\} -1\{Y_i(1) \leq q'\} \\
& - \mathbb{P}(Y_i(1) \leq q|X_i,S_i=s) + \mathbb{P}(Y_i(1) \leq q'|X_i,S_i=s)\biggr)\biggr|  + Cn^{-1/2} \\
&= O_p(n^{-1/2}),
\end{align*}
where the first inequality is due to the triangle inequality, the second inequality is due to the fact that $\sup_{\tau \in \Upsilon}|\hat{q}(\tau) - q(\tau)| \leq Mn^{-1/2}$, and the third inequality is due to the fact that $f_1(\cdot|X_i,S_i=s)$ is assumed to be bounded. To see the last equality in the above display, we define
\begin{align*}
\mathcal{F} = \begin{Bmatrix}
& (W_{i,s}(j) - \mathbb{E}W_{i,s}(\tau,j)|S_i=s)\biggl(1\{Y_i(1) \leq q\} -1\{Y_i(1) \leq q'\} \\
&- \mathbb{P}(Y_i(1) \leq q|X_i,S_i=s) + \mathbb{P}(Y_i(1) \leq q'|X_i,S_i=s)\biggr): \tau \in \Upsilon, |q-q'|\leq M n^{-1/2}
\end{Bmatrix} 
\end{align*}
with envelope $F_i = 2L_i + \mathbb{E}(L_i|S_i=s) \in L_{\mathbb{P},d}$ for some $d>2$, where $W_{i,s}(\tau,j)$ is the $j$th coordinate of $W_{i,s}(\tau)$. Clearly $\mathcal{F}$ is of the VC-type with fixed coefficients $(\alpha, v)$. In addition, 
\begin{align*}
\sup_{f \in \mathcal{F}}\mathbb{P}f^2 \leq Cn^{-1/2} \equiv \sigma_n^2. 
\end{align*}
Therefore, Lemma \ref{lem:max_eq} implies that $\sup_{\tau \in \Upsilon}|I(\tau)| = O_p(n^{-1/2})$. By the usual maximal inequality (e.g. \citep{VW96}, Theorem 2.14.1), we can show that 
\begin{align*}
\sup_{\tau \in \Upsilon}|II(\tau)| = O_p(n^{-1/2}). 
\end{align*}
Combining these results, we conclude that
\begin{align*}
\sup_{\tau \in \Upsilon, s\in \mathcal{S}} \left|\frac{1}{n_{1}(s)}\sum_{i \in I_1(s)}\dot{W}_{i,1,s}(\tau)(1\{Y_i \leq \hat{q}_1(\tau)\} - \dot{W}_{i,1,s}^\top(\tau)\theta_{1,s}^\textit{LP}(\tau)) \right| = O_p(n^{-1/2}),
\end{align*}
and hence 
\begin{align*}
\sup_{\tau \in \Upsilon, s\in \mathcal{S}}|\hat{\theta}_{1,s}^\textit{LP}(\tau) - \theta_{1,s}^\textit{LP}(\tau)| = O_p(n^{-1/2}). 
\end{align*}

\section{Proof of Theorem \ref{prop:lg2}}
\begin{table}[H]
	%\label{Table:constants}
	\begin{adjustbox}{max width=\textwidth}
		\begin{tabular}{c|l}
			\hline
			Name & Description \\ \hline
			$H_i$ & For $i \in [n]$, $H_i = H(X_i)$ for some function $H$  \\
			$\theta_{a,s}^\textit{ML}(\tau)$ & For $a \in \{0,1\}$ and $s \in \mathcal{S}$, $\theta_{a,s}^\textit{ML}(\tau)$ is the pseudo true value defined in \eqref{eq:ml}\\ 
			$\hat{\theta}_{a,s}^\textit{ML}(\tau)$ & For $a \in \{0,1\}$ and $s \in \mathcal{S}$, $\hat{\theta}_{a,s}^\textit{ML}(\tau)$ is the estimator of $\theta_{a,s}^\textit{ML}(\tau)$ in \eqref{eq:mlhat} \\ 
			\hline
		\end{tabular}
	\end{adjustbox}
	%	\caption{Table of notation}
\end{table}
Let $d_H$ be the dimension of $H_i$, $u \in \Re^{d_H}$, and  
\begin{align*}
Q_n(\tau,s,q,u) = \frac{1}{n_a(s)}\sum_{i \in I_a(s)}[1\{Y_i \leq q\}\log(\lambda(H_i^\top (\theta_a^\textit{ML}(\tau) + u))) + 1\{Y_i > q\}\log(1-\lambda(H_i^\top(\theta_a^\textit{ML}(\tau) + u)))],
\end{align*}
and
\begin{align*}
Q(\tau,s,q,u) = \mathbb{E}[1\{Y_i(a) \leq q\}\log(\lambda(H_i^\top(\theta_a^\textit{ML}(\tau) + u))) + 1\{Y_i(a) > q\}\log(1-\lambda(H_i^\top(\theta_a^\textit{ML}(\tau) + u)))|S_i=s].
\end{align*}
We note that 
\begin{align*}
\mathcal{F} = \begin{pmatrix}
1\{Y_i \leq q\}\log(\lambda(H_i^\top(\theta_a^\textit{ML}(\tau) + u)))) + 1\{Y_i > q\}\log(1-\lambda(H_i^\top(\theta_a^\textit{ML}(\tau) + u)))) - Q(\tau,s,q,u) \\
\tau \in \Upsilon, s\in \mathcal{S}, q\in \Re, ||u||_2 \leq \delta
\end{pmatrix}
\end{align*}
is a VC class with a fixed VC index. Then, Lemma \ref{lem:max_eq} implies 
\begin{align*}
\sup_{\tau \in \Upsilon, s \in \mathcal{S}, q \in \Re, ||u||_2 \leq \delta}|Q_n(\tau,s,q,u)-Q(\tau,s,q,u)| = o_p(1).
\end{align*}
In addition, 
\begin{align*}
\sup_{\tau \in \Upsilon, s \in \mathcal{S}, q \in \Re, ||u||_2\leq \delta}|\partial_q Q(\tau,s,q,u)| \leq C,
\end{align*}
and $\sup_{\tau \in \Upsilon}|\hat{q}_a(\tau) - q_a(\tau)| = O_p(n^{-1/2})$. Therefore, 
\begin{align}
&\Delta_n \equiv  \sup_{\tau \in \Upsilon, s \in \mathcal{S}, ||u||_2\leq \delta}|Q_n(\tau,s,\hat{q}_a(\tau),u)-Q(\tau,s,q_a(\tau),u)| \notag \\
&\leq \sup_{\tau \in \Upsilon, s \in \mathcal{S}, q \in \Re, ||u||_2\leq \delta}|Q_n(\tau,s,q,u)-Q(\tau,s,q,u)| \notag \\
& +  \sup_{\tau \in \Upsilon, s \in \mathcal{S}, ||u||_2\leq \delta}|Q(\tau,s,\hat{q}_a(\tau),u)-Q(\tau,s,q_a(\tau),u)| = o_p(1).
\label{eq:deltan2}
\end{align}

Further note that $Q_n(\tau,s,\hat{q}_a(\tau),u)$ is concave in $u$ for fixed $\tau$. Therefore, for $v \in S^{d_H-1}$ where $S^{d_H-1} = \{v\in \Re^{d_H}, ||v||_2 =1\}$, and $l > \delta$
\begin{align*}
Q_n(\tau,s,\hat{q}_a(\tau), \delta v) \geq \frac{\delta}{l}Q_n(\tau,s,\hat{q}_a(\tau),l v) + \left(1-\frac{\delta}{l}\right)Q_n(\tau,s,\hat{q}_a(\tau),0),
\end{align*}
which implies 
\begin{align*}
\frac{\delta}{l}\left(Q_n(\tau,s,\hat{q}_a(\tau),l v)-Q_n(\tau,s,\hat{q}_a(\tau),0)\right) &\leq  Q_n(\tau,s,\hat{q}_a(\tau),\delta v) - Q_n(\tau,s,\hat{q}_a(\tau),0) \notag \\
&\leq  Q(\tau,s,q_a(\tau),\delta v) - Q(\tau,s,q_a(\tau),0) + 2\Delta_n.
\end{align*}
Because $Q(\tau,s,q_a(\tau),\delta v) - Q(\tau,s,q_a(\tau),0)$ is continuous in $(\tau,v) \in \Upsilon \times S^{d_H-1}$, $\Upsilon \times S^{d_H-1}$ is compact, and $0$ is the unique maximizer of $ Q(\tau,s,q_a(\tau),u)$, we have 
\begin{align*}
\sup_{(\tau,v) \in \Upsilon \times S^{d_x-1}}Q(\tau,s,q_a(\tau),\delta v) - Q(\tau,s,q_a(\tau),0) \leq -\eta,
\end{align*}
for some $\eta>0$. In addition, if $\sup_{\tau \in \Upsilon}||\hat{\theta}_{a,s}^\textit{ML}(\tau)-\theta_{a,s}^\textit{ML}(\tau)||_2 > \delta$, then there exists $(\tau,l,v) \in \Upsilon \times (\delta,\infty) \times S^{d_x-1}$ such that 
\begin{align*}
\frac{\delta}{l}\left(Q_n(\tau,s,\hat{q}_a(\tau),l v)-Q_n(\tau,s,\hat{q}_a(\tau),0)\right) \geq 0.
\end{align*}
Therefore, 
\begin{align*}
\mathbb{P}\left(\sup_{\tau \in \Upsilon}||\hat{\theta}_{a,s}^\textit{ML}(\tau)-\theta_{a,s}^\textit{ML}(\tau)||_2 > \delta   \right) \leq \mathbb{P}(\eta \leq 2\Delta_n) \rightarrow 0,
\end{align*}
where the last step is due to \eqref{eq:deltan2}. This implies 
\begin{align*}
\sup_{\tau \in \Upsilon}||\hat{\theta}_{a,s}^\textit{ML}(\tau)-\theta_{a,s}^\textit{ML}(\tau)||_2 = o_p(1). 
\end{align*}

\section{Proof of Theorem \ref{prop:lpml}}
\begin{table}[H]
	%\label{Table:constants}
	\begin{adjustbox}{max width=\textwidth}
		\begin{tabular}{c|l}
			\hline
			Name & Description \\ \hline
			$\lambda(\cdot)$ & Logistic CDF\\
			$\theta_{a,s}^\textit{LPML}(\tau)$ & For $a \in \{0,1\}$, $s \in \mathcal{S}$, and $\tau \in \Upsilon$, $\theta_{a,s}^\textit{LPML}(\tau)$ is the pseudo true value defined in \eqref{eq:theta^lpml}\\ 
			$\hat{\theta}_{a,s}^\textit{LPML}(\tau)$ & For $a \in \{0,1\}$, $s \in \mathcal{S}$, and $\tau \in \Upsilon$, $\hat{\theta}_{a,s}^\textit{LPML}(\tau)$ is the estimator of $\theta_{a,s}^\textit{LPML}(\tau)$ in \eqref{eq:thetahat^lpml} \\
			$\omega_{i,a,s}(\tau)$ & For $a = 0,1$, $i \in [n]$, $s \in \mathcal{S}$, and $\tau \in \Upsilon$, $\omega_{i,a,s}(\tau) = \lambda(H_i^\top \theta_{a,s}^\textit{ML}(\tau))$ \\
			$W_{i,s}(\tau)$ & For $i \in [n]$, $s \in \mathcal{S}$, and $\tau \in \Upsilon$, $W_{i,s}(\tau) = (\omega_{i,1,s}(\tau),\omega_{i,0,s}(\tau))^\top$ \\
			$\omega_{i,a,a',s}(\tau)$ & For $a = 0,1$, $a' = 0,1$, $i \in [n]$, $s \in \mathcal{S}$, and $\tau \in \Upsilon$, $\omega_{i,a,a',s}(\tau) = \omega_{i,a',s}(\tau)- \frac{1}{n_a(s)}\sum_{i \in I_a(s)}\omega_{i,a',s}(\tau)$ \\
			$\dot{W}_{i,a,s}(\tau)$ & For $a =0,1$, $i \in [n]$, $s \in \mathcal{S}$, and $\tau \in \Upsilon$, $\dot{W}_{i,a,s}(\tau) = W_{i,s}(\tau) - \frac{1}{n_a(s)}\sum_{i \in I_a(s)}W_{i,s}(\tau) = (\omega_{i,a,1,s}(\tau),\omega_{i,a,0,s}(\tau))^\top$ \\
			$\hat{\omega}_{i,a,s}(\tau)$ & For $a = 0,1$, $i \in [n]$, $s \in \mathcal{S}$, and $\tau \in \Upsilon$, $\hat{\omega}_{i,a,s}(\tau) = \lambda(H_i^\top \hat{\theta}_{a,s}^\textit{ML}(\tau))$ \\
			$\hat{W}_{i,s}(\tau)$ & For $i \in [n]$, $s \in \mathcal{S}$, and $\tau \in \Upsilon$, $\hat{W}_{i,s}(\tau) = (\hat{\omega}_{i,1,s}(\tau),\hat{\omega}_{i,0,s}(\tau))^\top$ \\
			$\hat{\omega}_{i,a,a',s}(\tau)$ & For $a = 0,1$, $a' = 0,1$, $i \in [n]$, $s \in \mathcal{S}$, and $\tau \in \Upsilon$, $\hat{\omega}_{i,a,a',s}(\tau) = \hat{\omega}_{i,a',s}(\tau)- \frac{1}{n_a(s)}\sum_{i \in I_a(s)}\hat{\omega}_{i,a',s}(\tau)$ \\
			$\breve{W}_{i,a,s}(\tau)$ & For $a =0,1$, $i \in [n]$, $s \in \mathcal{S}$, and $\tau \in \Upsilon$, $\breve{W}_{i,a,s}(\tau) = \hat{W}_{i,a,s}(\tau) - \frac{1}{n_a(s)}\sum_{i \in I_a(s)}\hat{W}_{i,s}(\tau) = (\hat{\omega}_{i,a,1,s}(\tau),\hat{\omega}_{i,a,0,s}(\tau))^\top$ \\
			\hline
		\end{tabular}
	\end{adjustbox}
	%	\caption{Table of notation}
\end{table}
Recall $\widehat{m}_a(\tau,s,X_i) = \tau - \hat{W}_{i,s}^\top(\tau)\hat{\theta}_{a,s}^\textit{LPML}(\tau)$ and $\overline{m}_a(\tau,s,X_i) = \tau - W_{i,s}^\top(\tau)\theta_{a,s}^\textit{LPML}(\tau)$. Let $d_H$ be the dimension of $H_i$. Then, we have
\begin{align*}
\overline{\Delta}_a(\tau,s,X_i) & = \widehat{m}_a(\tau,s,X_i) - \overline{m}_a(\tau,s,X_i) \\
& \equiv \Lambda_{\tau,s}(X_i, \theta_{a,s}(\tau)) - \Lambda_{\tau,s}(X_i, \hat{\theta}_{a,s}(\tau)), 
\end{align*}
where the functional form $\Lambda_{\tau,s}(\cdot)$ is invariant to $(\tau,s)$, 
\begin{align*}
& \Lambda_{\tau,s}(X_i,\theta) = \Lambda(X_i,\theta) \equiv (\lambda(H(X_i)^\top \theta_1),\lambda(H(X_i)^\top \theta_2)) \theta_3, \quad \theta_1,\theta_2 \in \Re^{d_H}, \quad \theta_3 \in \Re^2,\\
&\theta_{a,s}(\tau) = ((\theta_{1,s}^\textit{ML}(\tau))^{\top},(\theta_{0,s}^\textit{ML}(\tau))^{\top},(\theta_{a,s}^\textit{LPML}(\tau))^{\top})^{\top}, \quad \text{and}\\
& \hat{\theta}_{a,s}(\tau) = ((\hat{\theta}_{1,s}^\textit{ML}(\tau))^{\top},(\hat{\theta}_{0,s}^\textit{ML}(\tau))^{\top},(\hat{\theta}_{a,s}^\textit{LPML}(\tau))^{\top})^{\top}. 
\end{align*}
Suppose 
\begin{align}
\sup_{\tau \in \Upsilon, a = 0,1, s \in \mathcal{S}}|\hat{ \theta}_{a,s}^\textit{LPML}(\tau) - \theta_{a,s}^\textit{LPML}(\tau)| = o_p(1),
\label{eq:theta_lpml}
\end{align}
and we also have
\begin{align*}
\sup_{\tau \in \Upsilon, a = 0,1, s \in \mathcal{S}}|\hat{ \theta}_{a,s}^\textit{ML}(\tau) - \theta_{a,s}^\textit{ML}(\tau)| = o_p(1)
\end{align*}
by Theorem \ref{prop:lg2}. Then Assumption \ref{ass:par}(iii) holds for $\hat{\theta}_{a,s}(\tau)$. Assumption \ref{ass:par}(i) holds automatically as $\Lambda(\cdot)$ does not depend on $\tau$, and Assumption \ref{ass:par}(ii) holds by Assumption \ref{ass:lpml}. Then, Theorem \ref{thm:par} implies Assumptions \ref{ass:mhat} and \ref{ass:mhatw} hold. In addition, Theorem \ref{thm:theta_as^star} implies $[\Sigma^\textit{LPML}(\tau_k,\tau_l)]_{k,l \in [K]}$ is the smallest among all linear adjustments with $W_{i,s}(\tau) = (\omega_{i,1,s}(\tau),\omega_{i,0,s}(\tau))^\top$ with $\omega_{i,a,s}(\tau) = \lambda(H_i^\top \theta_{a,s}^\textit{ML}(\tau))$ as the regressors. 

Therefore, the only thing left is to establish \eqref{eq:theta_lpml}. First, denote  
\begin{align*}
& \dot{\omega}_{i,a,a',s}(\tau) = \omega_{i,a',s}(\tau) - \frac{1}{n_a(s)}\sum_{i \in I_a(s)} \omega_{i,a',s}(\tau), \quad \dot{W}_{i,a,s}(\tau) = (\dot{\omega}_{i,a,1,s}(\tau),\dot{\omega}_{i,a,0,s}(\tau))^\top, \quad \text{and} \\
& \breve{\theta}_{a,s}^\textit{LPML}(\tau) = \left[\frac{1}{n_a(s)}\sum_{i \in I_a(s)} \dot{W}_{i,a,s}(\tau)\dot{W}_{i,a,s}(\tau)^\top \right]^{-1}\left[\frac{1}{n_a(s)}\sum_{i \in I_a(s)} \dot{W}_{i,a,s}(\tau)1\{Y_i \leq \hat{q}_a(\tau)\} \right].
\end{align*}
We note that Assumption \ref{ass:regX2} holds with $W_{i,s}(\tau) = (\lambda(H_i^\top \theta_{1,s}^\textit{ML}(\tau)),\lambda(H_i^\top \theta_{0,s}^\textit{ML}(\tau)))^\top$ by Assumption \ref{ass:lpml}(ii). Then, following the same argument as in the proof of Theorem \ref{prop:lp}, we can show that 
$$\sup_{\tau \in \Upsilon, a = 0,1, s \in \mathcal{S}}||\breve{\theta}_{a,s}^\textit{LPML}(\tau) - \theta_{a,s}^\textit{LPML}(\tau)||_2 = o_p(1).$$ Therefore, it suffices to show 
\begin{align*}
\sup_{\tau \in \Upsilon, a = 0,1, s \in \mathcal{S}}||\breve{\theta}_{a,s}^\textit{LPML}(\tau) - \hat{\theta}_{a,s}^\textit{LPML}(\tau)||_2 = o_p(1).
\end{align*}
Denote $\hat{W}_{i,s}(\tau) = (\hat{\omega}_{i,1,s}(\tau),\hat{\omega}_{i,0,s}(\tau))^\top$ and $\hat{\omega}_{i,a,s}(\tau) = \lambda(H_i^\top \hat{\theta}_{a,s}^\textit{ML}(\tau))$. We have 
\begin{align*}
& \sup_{\tau \in \Upsilon, a = 0,1, s \in \mathcal{S}}\left[\frac{1}{n_a(s)}\sum_{i \in I_a(s)} (\omega_{i,1,s}(\tau)-\hat{\omega}_{i,1,s}(\tau))^2 \right]^{1/2} \\
&\leq \max_{a=0,1, s \in \mathcal{S}}\lambda_{\max}^{1/2}\left(\frac{1}{n_a(s)}\sum_{i \in I_a(s)}H_iH_i^\top\right) \sup_{\tau \in \Upsilon, a = 0,1, s \in \mathcal{S}}||\hat{ \theta}_{1,s}^\textit{ML}(\tau) - \theta_{1,s}^\textit{ML}(\tau)||_2 = o_p(1). 
\end{align*}
In addition, denote $\breve{\omega}_{i,a,a',s}(\tau) = \hat{\omega}_{i,a',s}(\tau)-\frac{1}{n_a(s)}\sum_{i \in I_a(s)}\hat{\omega}_{i,a',s}(\tau)$ and $\breve{W}_{i,a,s}(\tau) = (\breve{\omega}_{i,a,1,s}(\tau),\breve{\omega}_{i,a,0,s}(\tau))^\top$. We first consider the case $a'=1$. We have 
\begin{align*}
& \sup_{\tau \in \Upsilon, a = 0,1, s \in \mathcal{S}}\left|\left[\frac{1}{n_a(s)}\sum_{i \in I_a(s)} \dot{\omega}_{i,a,1,s}^2(\tau) \right]^{1/2} - \left[\frac{1}{n_a(s)}\sum_{i \in I_a(s)} \breve{\omega}_{i,a,1,s}^2(\tau) \right]^{1/2}\right| \\
& \leq \sup_{\tau \in \Upsilon, a = 0,1, s \in \mathcal{S}}\left[\frac{1}{n_a(s)}\sum_{i \in I_a(s)} (\dot{\omega}_{i,a,1,s}(\tau)-\breve{\omega}_{i,a,1,s}^2(\tau))^2 \right]^{1/2}  \\
& \leq  \sup_{\tau \in \Upsilon, a = 0,1, s \in \mathcal{S}}\left[\frac{1}{n_a(s)}\sum_{i \in I_a(s)} (\omega_{i,1,s}(\tau)-\hat{\omega}_{i,1,s}(\tau))^2 \right]^{1/2} + \sup_{\tau \in \Upsilon, a = 0,1, s \in \mathcal{S}}\left| \frac{1}{n_a(s)}\sum_{i \in I_a(s)} (\omega_{i,1,s}(\tau)-\hat{\omega}_{i,1,s}(\tau))\right| \\
& \leq 2\sup_{\tau \in \Upsilon, a = 0,1, s \in \mathcal{S}}\left[\frac{1}{n_a(s)}\sum_{i \in I_a(s)} (\omega_{i,1,s}(\tau)-\hat{\omega}_{i,1,s}(\tau))^2 \right]^{1/2} = o_p(1). 
\end{align*}
In addition, Assumption \ref{ass:lpml} implies $\inf_{\tau \in \Upsilon, a = 0,1, s \in \mathcal{S}}\left[\frac{1}{n_a(s)}\sum_{i \in I_a(s)} \dot{\omega}_{i,a,1,s}^2(\tau) \right]^{1/2} \leq C < \infty$. Therefore, we have
\begin{align*}
\sup_{\tau \in \Upsilon, a = 0,1, s \in \mathcal{S}}\left|\left[\frac{1}{n_a(s)}\sum_{i \in I_a(s)} \dot{\omega}_{i,a,1,s}^2(\tau) \right] - \left[\frac{1}{n_a(s)}\sum_{i \in I_a(s)} \breve{\omega}_{i,a,1,s}^2(\tau) \right]\right| = o_p(1). 
\end{align*}
Similarly, we have 
\begin{align*}
\sup_{\tau \in \Upsilon, a = 0,1, s \in \mathcal{S}}\left|\left[\frac{1}{n_a(s)}\sum_{i \in I_a(s)} \dot{\omega}_{i,a,0,s}^2(\tau) \right] - \left[\frac{1}{n_a(s)}\sum_{i \in I_a(s)} \breve{\omega}_{i,a,0,s}^2(\tau) \right]\right| = o_p(1). 
\end{align*}
Last, 
\begin{align*}
& \sup_{\tau \in \Upsilon, a = 0,1, s \in \mathcal{S}}\left|\left[\frac{1}{n_a(s)}\sum_{i \in I_a(s)} \dot{\omega}_{i,a,0,s}(\tau)\dot{\omega}_{i,a,1,s}(\tau) \right] - \left[\frac{1}{n_a(s)}\sum_{i \in I_a(s)} \breve{\omega}_{i,a,0,s}(\tau)\breve{\omega}_{i,a,1,s}(\tau) \right]\right| \\
& \leq \sup_{\tau \in \Upsilon, a = 0,1, s \in \mathcal{S}}\left[\left(\frac{1}{n_a(s)}\sum_{i \in I_a(s)} \dot{\omega}_{i,a,0,s}^2(\tau) \right)^{1/2} + \left(\frac{1}{n_a(s)}\sum_{i \in I_a(s)} \breve{\omega}_{i,a,1,s}^2(\tau) \right)^{1/2}  \right] \\
& \times \biggl\{\left[\frac{1}{n_a(s)}\sum_{i \in I_a(s)} (\dot{\omega}_{i,a,1,s}(\tau)  - \breve{\omega}_{i,a,1,s}(\tau))^2 \right]^{1/2} +\left[\frac{1}{n_a(s)}\sum_{i \in I_a(s)} (\dot{\omega}_{i,a,0,s}(\tau)  - \breve{\omega}_{i,a,0,s}(\tau))^2 \right]^{1/2}\biggr\}= o_p(1). 
\end{align*}
This implies 
\begin{align*}
& \sup_{\tau \in \Upsilon, a = 0,1, s \in \mathcal{S}}\left|\left[\frac{1}{n_a(s)}\sum_{i \in I_a(s)} \dot{W}_{i,a,s}(\tau)\dot{W}_{i,a,s}(\tau)^\top \right] - \left[\frac{1}{n_a(s)}\sum_{i \in I_a(s)} \breve{W}_{i,a,s}(\tau)\breve{W}_{i,a,s}(\tau)^\top \right]\right| \\
& = \sup_{\tau \in \Upsilon, a = 0,1, s \in \mathcal{S}}\biggl|\left[\frac{1}{n_a(s)}\sum_{i \in I_a(s)} \begin{pmatrix}
\dot{\omega}_{i,a,1,s}^2(\tau) & \dot{\omega}_{i,a,1,s}(\tau)\dot{\omega}_{i,a,0,s}(\tau) \\
\dot{\omega}_{i,a,0,s}(\tau)\dot{\omega}_{i,a,1,s}(\tau) & \dot{\omega}_{i,a,0,s}^2(\tau) 
\end{pmatrix} \right] \\
& - \left[\frac{1}{n_a(s)}\sum_{i \in I_a(s)} \begin{pmatrix}
\breve{\omega}_{i,a,1,s}^2(\tau) & \breve{\omega}_{i,a,1,s}(\tau)\breve{\omega}_{i,a,0,s}(\tau) \\
\breve{\omega}_{i,a,0,s}(\tau)\breve{\omega}_{i,a,1,s}(\tau) & \breve{\omega}_{i,a,0,s}^2(\tau) 
\end{pmatrix} \right]\biggr| = o_p(1),
\end{align*}
and thus, 
\begin{align*}
\sup_{\tau \in \Upsilon, a = 0,1, s \in \mathcal{S}}\left|\left[\frac{1}{n_a(s)}\sum_{i \in I_a(s)} \dot{W}_{i,a,s}(\tau)\dot{W}_{i,a,s}(\tau)^\top \right]^{-1} - \left[\frac{1}{n_a(s)}\sum_{i \in I_a(s)} \breve{W}_{i,a,s}(\tau)\breve{W}_{i,a,s}(\tau)^\top \right]^{-1}\right| = o_p(1).
\end{align*}

In addition, 
\begin{align*}
&\sup_{\tau \in \Upsilon, a = 0,1, s \in \mathcal{S}}\left\Vert \frac{1}{n_a(s)}\sum_{i \in I_a(s)} (\dot{W}_{i,a,s}-\breve{W}_{i,a,s}(\tau))1\{Y_i \leq \hat{q}_a(\tau)\} \right\Vert_2 \\
& \leq  \sup_{\tau \in \Upsilon, a = 0,1, s \in \mathcal{S}}\frac{1}{n_a(s)}\sum_{i \in I_a(s)} \left\Vert \dot{W}_{i,a,s}(\tau)-\breve{W}_{i,a,s}(\tau)\right\Vert_2 \\
& \leq  2\sup_{\tau \in \Upsilon, a = 0,1, s \in \mathcal{S}}\frac{1}{n_a(s)}\sum_{i \in I_a(s)} \left\Vert W_{i,s}(\tau)-\hat{W}_{i,s}(\tau)\right\Vert_2 = o_p(1). 
\end{align*}
Therefore, we have
\begin{align*}
&\sup_{\tau \in \Upsilon, a = 0,1, s \in \mathcal{S}}||\breve{\theta}_{a,s}^\textit{LPML}(\tau) - \hat{\theta}_{a,s}^\textit{LPML}(\tau)||_2 \\
& = \sup_{\tau \in \Upsilon, a = 0,1, s \in \mathcal{S}}\biggl\Vert \left[\frac{1}{n_a(s)}\sum_{i \in I_a(s)} \dot{W}_{i,a,s}(\tau)\dot{W}_{i,a,s}(\tau)^\top \right]^{-1} \left[\frac{1}{n_a(s)}\sum_{i \in I_a(s)} \dot{W}_{i,a,s}(\tau)1\{Y_i \leq \hat{q}_a(\tau)\}  \right] \\
& - \left[\frac{1}{n_a(s)}\sum_{i \in I_a(s)} \breve{W}_{i,a,s}(\tau)\breve{W}_{i,a,s}(\tau)^\top \right]^{-1} \left[\frac{1}{n_a(s)}\sum_{i \in I_a(s)} \breve{W}_{i,a,s}(\tau)1\{Y_i \leq \hat{q}_a(\tau)\}  \right]\biggr\Vert_2 \\
& = o_p(1),
\end{align*}
which concludes the proof. 

\section{Proof of Theorem \ref{thm:np}}
\begin{table}[H]
	%\label{Table:constants}
	\begin{adjustbox}{max width=\textwidth}
		\begin{tabular}{c|l}
			\hline
			Name & Description \\ \hline
			$\theta_{a,s}^\textit{NP}(\tau)$ & For $a \in \{0,1\}$, $s \in \mathcal{S}$, and $\tau \in \Upsilon$, $\theta_{a,s}^\textit{NP}(\tau)$ is the pseudo true value defined in Assumption \ref{ass:np}(ii)\\ 
			$\hat{\theta}_{a,s}^\textit{NP}(\tau)$ & For $a \in \{0,1\}$, $s \in \mathcal{S}$, and $\tau \in \Upsilon$, $\hat{\theta}_{a,s}^\textit{NP}(\tau)$ is the estimator of $\theta_{a,s}^\textit{NP}(\tau)$ in \eqref{eq:mhat_np2} \\
			\hline
		\end{tabular}
	\end{adjustbox}
	%	\caption{Table of notation}
\end{table}
The proof strategy follows \cite{BCFH13} and details are given here for completeness. We divide the proof into three steps. In the first step, we show 
$$\sup_{\tau \in \Upsilon}||\hat{\theta}_{a,s}^\textit{NP}(\tau) - \theta_{a,s}^\textit{NP}(\tau)||_2 = O_p(\sqrt{h_n \log(n)/n}).$$ 
In the second step, we establish Assumption \ref{ass:mhatw}. By a similar argument, we can establish Assumption \ref{ass:mhat}(i). In the third step, we establish Assumptions \ref{ass:mhat}(ii) and \ref{ass:mhat}(iii). 

\vspace{1.5mm}
\noindent \textbf{Step 1.} 
Let $\hat{U}_\tau = \hat{\theta}^\textit{NP}_{a,s}(\tau) - \theta^\textit{NP}_{a,s}(\tau)$,  
\begin{align*}
Q_n(\tau,s,q,\theta) &= \frac{-1}{n_a(s)}\sum_{i \in I_a(s)}[1\{Y_i \leq q\}\log(\lambda(H_{h_n}^\top(X_i)\theta_a)) + 1\{Y_i > q\}\log(1-\lambda(H_{h_n}^\top(X_i)\theta_a))] \\
&= \frac{1}{n_a(s)}\sum_{i \in I_a(s)}[\log\left(1+\exp(H_{h_n}^\top(X_i)\theta_a)\right) - 1\{Y_i \leq q\} H_{h_n}^\top(X_i)\theta_a],
\end{align*}
and for an arbitrary $U_\tau \in \Re^{h_n}$, 
\begin{align*}
\ell_i(t) = \log(1+\exp(H_{h_n}^\top(X_i)(\theta^\textit{NP}_{a,s}(\tau) + tU_\tau) )).
\end{align*}
Then, we have
\begin{align*}
\hat{U}_\tau  = \argmin_{U_\tau} Q_n(\tau,s,\hat{q}_a(\tau),\theta^\textit{NP}_{a,s}(\tau) + U_\tau) - Q_n(\tau,s,\hat{q}_a(\tau),\theta^\textit{NP}_{a,s}(\tau)),
\end{align*}
\begin{align*}
\partial_t Q_n(\tau,s,\hat{q}_a(\tau),\theta^\textit{NP}_{a,s}(\tau) + tU_\tau)|_{t=0} = \frac{-1}{n_a(s)}\sum_{i \in I_a(s)}\left(1\{Y_i \leq \hat{q}_a(\tau)\} - \lambda(H_{h_n}^\top(X_i)\theta^\textit{NP}_{a,s}(\tau))\right)H_{h_n}^\top(X_i)U_\tau,
\end{align*}
and
\begin{align*}
& Q_n(\tau,s,\hat{q}_a(\tau),\theta^\textit{NP}_{a,s}(\tau) + U_\tau) - Q_n(\tau,s,\hat{q}_a(\tau),\theta^\textit{NP}_{a,s}(\tau)) - \partial_t Q_n^\top (\tau,s,\hat{q}_a(\tau),\theta^\textit{NP}_{a,s}(\tau) + tU_\tau)|_{t=0}  \\
= & \frac{1}{n_a(s)}\sum_{i \in I_a(s)}\left[\ell_i(1) - \ell_i(0) - \ell_i'(0)\right].
\end{align*}
In addition
\begin{align*}
|\ell_i^{'''}(t)| \leq |\ell_i^{''}(t)| |H_{h_n}^\top(X_i) U_\tau|.
\end{align*}
Therefore,  there exists a constant $\underline{c}>0$ such that 
\begin{align*}
& \ell_i(1) - \ell_i(0) -  \ell'_i(0) \\
&\geq  \frac{\ell_i^{''}(0)}{(H_{h_n}^\top(X_i) U_\tau)^2} \left[\exp(-|H_{h_n}^\top(X_i) U_\tau|) + |H_{h_n}^\top(X_i) U_\tau| -1\right] \\
&= \lambda(H_{h_n}^\top(X_i)\theta^\textit{NP}_{a,s}(\tau))(1-\Lambda(H_{h_n}^\top(X_i)\theta^\textit{NP}_{a,s}(\tau)))\left[\exp(-|H_{h_n}^\top(X_i) U_\tau|) + |H_{h_n}^\top(X_i) U_\tau| -1\right] \\
&\geq \underline{c}\left[\exp(-|H_{h_n}^\top(X_i) U_\tau|) + |H_{h_n}^\top(X_i) U_\tau| -1\right] \\
&\geq  \underline{c}\left(\frac{(H_{h_n}^\top(X_i) U_\tau)^2}{2} - \frac{|H_{h_n}^\top(X_i) U_\tau|^3}{6}\right),
\end{align*}
where the first inequality is due to \citet[Lemma 1]{B10} and the third inequality holds because 
\begin{align*}
e^{-x} + x - 1 \geq \frac{x^2}{2} - \frac{x^3}{6}, \; x>0.
\end{align*}
To see the second inequality, note that $e^{-x}+ x - 1 \geq 0$ for $x\geq 0$ and  by Assumption \ref{ass:np}, 
\begin{align*}
& \inf_{a = 0,1, s \in \mathcal{S},\tau \in \Upsilon, x \in \Supp(X)}\lambda(H_{h_n}^\top(x) \theta^\textit{NP}_{a,s}(\tau)) \\
&= \inf_{a = 0,1, s \in \mathcal{S},\tau \in \Upsilon, x \in \Supp(X)}(\mathbb{P}(Y_i(a)\leq q_a(\tau)|S_i=s,X_i=x) - R_a(\tau,s,x)) \geq c/2,
\end{align*} 
and 
\begin{align*}
& \sup_{a = 0,1, s \in \mathcal{S},\tau \in \Upsilon, x \in \Supp(X)}\lambda(H_{h_n}^\top(x) \theta^\textit{NP}_{a,s}(\tau))\\ 
= & \sup_{a = 0,1, s \in \mathcal{S},\tau \in \Upsilon, x \in \Supp(X)}(\mathbb{P}(Y_i(a)\leq q_a(\tau)|S_i=s,X_i=x) + R_a(\tau,s,x)) \leq 1-c/2.
\end{align*} 
This implies  
\begin{align*}
\inf_{\tau \in \Upsilon}\lambda(H_{h_n}^\top(X_i)\theta^\textit{NP}_{a,s}(\tau))(1-\Lambda(H_{h_n}^\top(X_i)\theta^\textit{NP}_{a,s}(\tau))) \geq \underline{c}>0,
\end{align*}
and thus, 
\begin{align*}
G_n(U_\tau) &\equiv Q_n(\tau,s,\hat{q}_a(\tau),\theta^\textit{NP}_{a,s}(\tau) + U_\tau) - Q_n(\tau,s,\hat{q}_a(\tau),\theta^\textit{NP}_{a,s}(\tau)) - \partial_t Q_n^\top (\tau,s,\hat{q}_a(\tau),\theta^\textit{NP}_{a,s}(\tau) + tU_\tau)|_{t=0}\\
&\geq  \frac{ \underline{c}}{n_a(s)}\sum_{i\in I_a(s)}\left(\frac{(H_{h_n}^\top(X_i) U_\tau)^2}{2} - \frac{|H_{h_n}^\top(X_i) U_\tau|^3}{6}\right).
\end{align*}
Let 
\begin{align}
\overline{\ell} = \inf_{U \in \Re^{h_n} }\frac{\left[\frac{1}{n_a(s)}\sum_{i \in I_a(s)}(H_{h_n}^\top(X_i) U)^2\right]^{3/2} }{\frac{1}{n_a(s)}\sum_{i \in I_a(s)}|H_{h_n}^\top(X_i) U|^3}.
\label{eq:overlineell} 
\end{align}
If $\left[\frac{1}{n_a(s)}\sum_{i \in I_a(s)}(H_{h_n}^\top(X_i)\hat{U}_\tau)^2\right]^{1/2} \leq \overline{\ell}$, then 
\begin{align*}
& \frac{1}{n_a(s)}\sum_{i \in I_a(s)}(H_{h_n}^\top(X_i)\hat{U}_\tau)^2 \\
&= \left[\frac{1}{n_a(s)}\sum_{i \in I_a(s)}(H_{h_n}^\top(X_i)\hat{U}_\tau)^2\right]^{-1/2} \frac{\left[\frac{1}{n_a(s)}\sum_{i \in I_a(s)}(H_{h_n}^\top(X_i)\hat{U}_\tau)^2\right]^{3/2} }{\frac{1}{n_a(s)}\sum_{i \in I_a(s)}(H_{h_n}^\top(X_i)\hat{U}_\tau)^3} \frac{1}{n_a(s)}\sum_{i \in I_a(s)}(H_{h_n}^\top(X_i)\hat{U}_\tau)^3 \\
&\geq \frac{1}{n_a(s)}\sum_{i \in I_a(s)}(H_{h_n}^\top(X_i)\hat{U}_\tau)^3,
\end{align*}
and thus 
\begin{align*}
G_n(\hat{U}_\tau) \geq \frac{ \underline{c}}{n_a(s)}\sum_{i\in I_a(s)}\left(\frac{(H_{h_n}^\top(X_i) U_\tau)^2}{2} - \frac{|H_{h_n}^\top(X_i) U_\tau|^3}{6}\right) \geq \frac{ \underline{c}}{n_a(s)}\sum_{i\in I_a(s)}\frac{(H_{h_n}^\top(X_i)\hat{U}_\tau)^2}{3}.
\end{align*}
On the other hand, if $\left[\frac{1}{n_a(s)}\sum_{i \in I_a(s)}(H_{h_n}^\top(X_i)\hat{U}_\tau)^2\right]^{1/2} > \overline{\ell}$, we can denote $\overline{U}_\tau = \frac{\overline{\ell} \hat{U}_\tau}{\left[\frac{1}{n_a(s)}\sum_{i \in I_a(s)}(H_{h_n}^\top(X_i)\hat{U}_\tau)^2\right]^{1/2}}$ such that 
\begin{align*}
\left[\frac{1}{n_a(s)}\sum_{i \in I_a(s)}(H_{h_n}^\top(X_i)\overline{U}_\tau)^2\right]^{1/2} \leq \overline{\ell}.
\end{align*}
Further, because $G_n(U_\tau)$ is convex in $U_\tau$ we have
\begin{align*}
G_n(\hat{U}_\tau) &= G_n\left(\frac{\left[\frac{1}{n_a(s)}\sum_{i \in I_a(s)}(H_{h_n}^\top(X_i)\hat{U}_\tau)^2\right]^{1/2} }{\overline{\ell}} \overline{U}_\tau\right)  \\
&\geq  \frac{\left[\frac{1}{n_a(s)}\sum_{i \in I_a(s)}(H_{h_n}^\top(X_i)\hat{U}_\tau)^2\right]^{1/2} }{\overline{\ell}}G_n(\overline{U}_\tau) \\
&\geq \frac{\left[\frac{1}{n_a(s)}\sum_{i \in I_a(s)}(H_{h_n}^\top(X_i)\hat{U}_\tau)^2\right]^{1/2} }{\overline{\ell}} \frac{ \underline{c}}{n_a(s)}\sum_{i\in I_a(s)}\frac{(H_{h_n}^\top(X_i)\overline{U}_\tau)^2}{3} \\
&= \frac{\underline{c} \overline{\ell}}{3}\left[\frac{1}{n_a(s)}\sum_{i \in I_a(s)}(H_{h_n}^\top(X_i)\hat{U}_\tau)^2\right]^{1/2}.
\end{align*}
Therefore, for some constant $\overline{c}$ that only depends on $\underline{c}$ and $\kappa_1$, we have
\begin{align}
G_n(\hat{U}_\tau) &\geq \min\left( \frac{ \underline{c}}{n_a(s)}\sum_{i\in I_a(s)}\frac{(H_{h_n}^\top(X_i)\hat{U}_\tau)^2}{3},  \frac{\underline{c} \overline{\ell}}{3}\left[\frac{1}{n_a(s)}\sum_{i \in I_a(s)}(H_{h_n}^\top(X_i)\hat{U}_\tau)^2\right]^{1/2}\right) \notag \\
&\geq  \frac{\overline{c}}{3}\min(||\hat{U}_\tau||_2^2, \overline{\ell}||\hat{U}_\tau||_2). 
\label{eq:lower}
\end{align}
In addition, by construction, 
\begin{align}
G_n(\hat{U}_\tau) &\leq  |\partial_t Q_n^\top (\tau,s,\hat{q}_a(\tau),\theta^\textit{NP}_{a,s}(\tau) + tU_\tau)|_{t=0}| \notag \\
&= \left|\frac{1}{n_a(s)}\sum_{i \in I_{a}(s)}(1\{Y_i \leq \hat{q}_a(\tau)\} - \lambda(H_{h_n}^\top(X_i) \theta_{a,s}^\textit{NP}(\tau)))H_{h_n}^\top(X_i)\hat{U}_\tau\right| \notag \\
&\leq \left\Vert \frac{1}{n_a(s)}\sum_{i \in I_{a}(s)}(1\{Y_i \leq \hat{q}_a(\tau)\} - m_a(\tau,s,X_i))H_{h_n}(X_i)\right\Vert_\infty \left\Vert \hat{U}_\tau\right\Vert_1 \notag \\
& + \left[\frac{1}{n_a(s)}\sum_{i \in I_{a}(s)}R_{a}^2(\tau,s,X_i)\right]^{1/2}\left[\frac{1}{n_a(s)}\sum_{i \in I_{a}(s)}(  H_{h_n}^\top(X_i)\hat{U}_\tau)\right]^{1/2} \notag \\
&\leq \left\Vert \frac{ h_n^{1/2}}{n_a(s)}\sum_{i \in I_{a}(s)}(1\{Y_i \leq \hat{q}_a(\tau)\} - m_a(\tau,s,X_i))H_{h_n}(X_i)\right\Vert_\infty \left\Vert \hat{U}_\tau\right\Vert_2 \notag \\
& + \left[\frac{\kappa_2}{n_a(s)}\sum_{i \in I_{a}(s)}R_{a}^2(\tau,s,X_i)\right]^{1/2}||\hat{U}_\tau||_2.
\label{eq:upper}
\end{align}
Combining \eqref{eq:lower} and \eqref{eq:upper}, we have 
\begin{align*}
\frac{\overline{c}}{3}\min(||\hat{U}_\tau||_2, \overline{\ell}) \leq \left\Vert \frac{ h_n^{1/2}}{n_a(s)}\sum_{i \in I_{a}(s)}(1\{Y_i \leq \hat{q}_a(\tau)\} - m_a(\tau,s,X_i))H_{h_n}(X_i)\right\Vert_\infty + \left[\frac{\kappa_2}{n_a(s)}\sum_{i \in I_{a}(s)}R_{a}^2(\tau,s,X_i)\right]^{1/2}.
\end{align*}
Taking $\sup_{\tau \in \Upsilon}$ on both sides, we have
\begin{align*}
& \frac{\overline{c}}{3}\min(\sup_{\tau \in \Upsilon}||\hat{U}_\tau||_2, \overline{\ell}) \\
& \leq  \sup_{\tau \in \Upsilon}\left\Vert \frac{ h_n^{1/2}}{n_a(s)}\sum_{i \in I_{a}(s)}(1\{Y_i \leq \hat{q}_a(\tau)\} - m_a(\tau,s,X_i))H_{h_n}(X_i)\right\Vert_\infty + \sup_{\tau \in \Upsilon}\left[\frac{\kappa_2}{n_a(s)}\sum_{i \in I_{a}(s)}R_{a}^2(\tau,s,X_i)\right]^{1/2} \\
& = O_p(\sqrt{\frac{h_n \log(n)}{n}}),
\end{align*}
where the last line holds due to Assumption \ref{ass:np} and Lemma \ref{lem:np1}. Finally, Lemma \ref{lem:np2} shows that $\overline{\ell}/\sqrt{\frac{h_n \log(n)}{n}} \rightarrow \infty$, which implies 
\begin{align*}
\sup_{\tau \in \Upsilon}||\hat{U}_\tau||_2 = O_p\left(\sqrt{\frac{h_n \log(n)}{n}}\right).
\end{align*}

\vspace{1.5mm}
\noindent \textbf{Step 2.}
Recall 
\begin{align*}
\overline{\Delta}_a(\tau,s,X_i) &=  \widehat{m}_a(\tau,s,X_i) - \overline{m}_a(\tau,s,X_i) \\
&=  \mathbb{P}(Y_i(a)\leq q_a(\tau)|X_i,S_i=s)- \lambda(H_{h_n}^\top(X_i)\hat{\theta}^\textit{NP}_{a,s}(\tau)) \\
&= \lambda(H_{h_n}^\top(X_i)\theta^\textit{NP}_{a,s}(\tau))- \lambda(H_{h_n}^\top(X_i)\hat{\theta}^\textit{NP}_{a,s}(\tau)) + R_a(\tau,s,X_i), 
\end{align*}
and $\{X_i^s,\xi^s_i\}_{i \in [n]}$ is generated independently from the joint distribution of $(X_i,\xi_i)$ given $S_i=s$, and so is independent of  $\{A_i,S_i\}_{i \in [n]}$. 
Let 
$$H(\theta_1,\theta_2) = \mathbb{E}[\lambda(H_{h_n}^\top(X_i)\theta_1) - \lambda(H_{h_n}^\top(X_i)\theta_{2})|S_i=s] = \mathbb{E}[\lambda(H_{h_n}^\top (X_i^s)\theta_1) - \lambda(H_{h_n}^\top (X_i^s)\theta_{2})].$$ 
We have
\begin{align}
& \sup_{\tau \in \Upsilon,s\in \mathcal{S}}\biggl|\frac{\sum_{i\in I_1(s)}\xi_i\overline{\Delta}_1(\tau,s,X_i)}{n_1^w(s)} - \frac{\sum_{i \in I_0(s)}\xi_i\overline{\Delta}_1(\tau,s,X_i)}{n_0^w(s)}\biggr| \notag  \\
&\leq  \sup_{\tau \in \Upsilon,s\in \mathcal{S}}\left|\frac{\sum_{i\in I_1(s)}\xi_i[ \overline{\Delta}_1(\tau,s,X_i) - H(\theta_{1,s}^\textit{NP}(\tau),\hat{\theta}^\textit{NP}_{1,s}(\tau))- \mathbb{E}(R_1(\tau,s,X_i)|S_i=s)] }{n_1^w(s)} \right| \notag \\
& +  \sup_{\tau \in \Upsilon,s\in \mathcal{S}}\left|\frac{\sum_{i\in I_0(s)}\xi_i[ \overline{\Delta}_1(\tau,s,X_i) - H(\theta^\textit{NP}_{1,s}(\tau),\hat{\theta}^\textit{NP}_{1,s}(\tau)) - \mathbb{E}(R_1(\tau,s,X_i)|S_i=s)] }{n_0^w(s)} \right| 
\label{eq:Delta_np}
\end{align}
We aim to bound the first term on the RHS of \eqref{eq:Delta_np}. Note for any $\eps>0$, there exists a constant $M>0$ such that 
\begin{align*}
\mathbb{P}\left(\sup_{\tau \in \Upsilon}||\hat{\theta}_{a,s}^\textit{NP}(\tau) - \theta_{a,s}^\textit{NP}(\tau)||_2 \leq M \sqrt{h_n \log(n)/n}\right) \geq 1-\eps. 
\end{align*}
On the set $\mathcal{A}(\eps) = \{\sup_{\tau \in \Upsilon}||\hat{\theta}_{a,s}^\textit{NP}(\tau) - \theta_{a,s}^\textit{NP}(\tau)||_2 \leq M \sqrt{h_n \log(n)/n}\}$, we have 
\begin{align*}
& \sup_{\tau \in \Upsilon,s\in \mathcal{S}}\left|\frac{\sum_{i\in I_1(s)}\xi_i[ \overline{\Delta}_1(\tau,s,X_i) - H(\theta_{1,s}^\textit{NP}(\tau),\hat{\theta}^\textit{NP}_{1,s}(\tau))-\mathbb{E}(R_1(\tau,s,X_i)|S_i=s)] }{n_1^w(s)} \right| \\
&\leq  \sup_{\tau \in \Upsilon,s\in \mathcal{S}}\left|\frac{\sum_{i\in I_1(s)}\xi_i[ \lambda(H_{h_n}^\top(X_i)\theta_{1,s}^\textit{NP}(\tau)) - \lambda(H_{h_n}^\top(X_i)\hat{\theta}_{1,s}^\textit{NP}(\tau)) - H(\theta_{1,s}^\textit{NP}(\tau),\hat{\theta}^\textit{NP}_{1,s}(\tau))] }{n_1(s)} \right| \\
& +  \sup_{\tau \in \Upsilon,s\in \mathcal{S}}\left|\frac{\sum_{i\in I_1(s)}\xi_i\left[R_1(\tau,s,X_i)-\mathbb{E}(R_1(\tau,s,X_i)|S_i=s)\right] }{n_1^w(s)} \right|\\
&\leq  \frac{n_1(s)}{n_1^w(s)}\biggl[\sup_{s\in \mathcal{S}, \theta_1,\theta_2 \in \Re^{h_n}, ||\theta_1 - \theta_2||_2 \leq M\sqrt{h_n \log(n)/n}}\left|\frac{\sum_{i\in I_1(s)}\xi_i[ \lambda(H_{h_n}^\top(X_i)\theta_1) - \lambda(H_{h_n}^\top(X_i)\theta_2) - H(\theta_1,\theta_2)] }{n_1(s)} \right| \\
& +  \sup_{\tau \in \Upsilon,s\in \mathcal{S}}\left|\frac{\sum_{i\in I_1(s)}\xi_i\left[R_1(\tau,s,X_i)-\mathbb{E}(R_1(\tau,s,X_i)|S_i=s)\right] }{n_1(s)} \right|\biggr]\\
&\equiv \frac{n_1(s)}{n_1^w(s)}(D_1 + D_2).
\end{align*}
For $D_1$, we have
\begin{align*}
& D_1|\{A_i,S_i\}_{i \in [n]} \\
&\stackrel{d}{=}  \sup\left|\frac{\sum_{i = N(s)}^{N(s)+n_1(s)}\xi_i^s[ \lambda(H_{h_n}^\top(X^s_i)\theta_1) - \lambda(H_{h_n}^\top(X^s_i)\theta_2) - H(\theta_1,\theta_2)] }{n_1(s)} \right| |\{A_i,S_i\}_{i \in [n]}  \\
&\stackrel{d}{=}  ||\mathbb{P}_{n_1(s)} - \mathbb{P}||_{\mathcal{F}} |\{A_i,S_i\}_{i \in [n]},
\end{align*}
where the supremum in the first equality is taken over $\{s\in \mathcal{S}, \theta_1,\theta_2 \in \Re^{h_n},||\theta_1 - \theta_2||_2 \leq M\sqrt{h_n \log(n)/n}\}$ and 
$$\mathcal{F} = \begin{Bmatrix}
&\xi_i^s[ \lambda(H_{h_n}^\top(X^s_i)\theta_1) - \lambda(H_{h_n}^\top(X^s_i)\theta_2) - H(\theta_1,\theta_2)]: \\
& s\in \mathcal{S}, \theta_1,\theta_2 \in \Re^{h_n},||\theta_1 - \theta_2||_2 \leq M\sqrt{h_n \log(n)/n}
\end{Bmatrix}$$
with the envelope $F = 2\xi_i^s$. We further note that $||\max_{i \in [n]}2\xi_i^s||_{\mathbb{P},2} \leq C\log(n)$, 
\begin{align*}
\sup_{f \in \mathcal{F}}\mathbb{E}f^2 \leq \sup \mathbb{E}(H_{h_n}^\top(X^s_i)(\theta_1-\theta_2))^2 \leq \kappa_2 M^2h_n \log(n)/n, 
\end{align*}
and 
\begin{align*}
\sup_Q N(\mathcal{F}, e_Q, \eps||F||_{Q,2}) \leq \left(\frac{a}{\eps}\right)^{ch_n},
\end{align*}
where $a,c$ are two fixed constants. Therefore, by \citet[Corollary 5.1]{CCK14}, 
\begin{align*}
\mathbb{E} \left[||\mathbb{P}_{n_1(s)} - \mathbb{P}||_{\mathcal{F}} |\{A_i,S_i\}_{i \in [n]}\right] = O_p\left( h_n  \log(n)/n + \frac{h_n \log^2(n)}{n}\right) = o_p(n^{-1/2}), 
\end{align*}
which implies $D_1 = o_p(n^{-1/2})$. 

Similarly, we have
\begin{align*}
D_2|\{A_i,S_i\}_{i \in [n]} \stackrel{d}{=} & \sup_{\tau \in \Upsilon,s\in \mathcal{S}}\left|\frac{\sum_{i = N(s)}^{N(s)+n_1(s)}\xi^s_i\left[R_1(\tau,s,X^s_i)-\mathbb{E}(R_1(\tau,s,X^s_i))\right] }{n_1(s)} \right| |\{A_i,S_i\}_{i \in [n]} \\
= & ||\mathbb{P}_{n_1(s)} - \mathbb{P}||_{\mathcal{F}} |\{A_i,S_i\}_{i \in [n]} |\{A_i,S_i\}_{i \in [n]},
\end{align*}
where $\mathcal{F} = \{\xi_i^s[ \tau - m_1(\tau,s,X_i^s) - \lambda(H_{h_n}^\top(X^s_i)\theta_{1,s}^\textit{NP}(\tau))]: \tau \in \Upsilon  \}$ with an envelope $F=\xi_i^s$. In addition, we note $\mathcal{F}$ is nested in 
\begin{align*}
\widetilde{\mathcal{F}} = \{\xi_i^s[ \tau - m_1(\tau,s,X_i^s) - \lambda(H_{h_n}^\top(X^s_i)\theta_1)]: \tau \in \Upsilon, \theta_1 \in \Re^{h_n}  \},
\end{align*}
so that 
\begin{align*}
\sup_Q N(\mathcal{F}, e_Q, \eps||F||_{Q,2}) \leq \sup_Q N(\widetilde{\mathcal{F}}, e_Q, \eps||F||_{Q,2}) \leq \left(\frac{a}{\eps}\right)^{ch_n}.
\end{align*}
Last, 
\begin{align*}
\sup_{f \in \mathcal{F}}\mathbb{E}f^2 = \sup_{\tau \in \Upsilon, a = 0,1, s \in \mathcal{S}}\mathbb{E}R_1^2(\tau,s,X_i^s) = O(h_n  \log(n)/n). 
\end{align*}
by \citet[Corollary 5.1]{CCK14}, 
\begin{align*}
\mathbb{E} \left[||\mathbb{P}_{n_1(s)} - \mathbb{P}||_{\mathcal{F}} |\{A_i,S_i\}_{i \in [n]}\right] = O_p\left( h_n  \log(n)/n + \frac{h_n \log^2(n)}{n}\right) = o_p(n^{-1/2}), 
\end{align*}
which implies $D_2 = o_p(n^{-1/2})$. This leads to \eqref{eq:Delta_np}.  

\textbf{Step 3.}
Note $|m_a(\tau_1,s,X_i)| \leq 1$ and 
\begin{align*}
& |m_a(\tau_1,s,X_i) - m_a(\tau_2,s,X_i)| \\
\leq & |\tau_1 - \tau_2| + |\mathbb{P}(Y_i(a) \leq q_a(\tau_1)|X_i,S_i=s) - \mathbb{P}(Y_i(a) \leq q_a(\tau_2)|X_i,S_i=s)| \\
\leq & \left(1+ \frac{\sup_y f_a(y|X_i,S_i=s)}{\inf_{\tau \in \Upsilon}f_a(q_a(\tau))}\right)|\tau_1 -\tau_2|\\ \leq & C|\tau_1-\tau_2|.  
\end{align*}
This implies Assumptions \ref{ass:mhat}(ii) and \ref{ass:mhat}(iii).

\section{Proof of Theorem \ref{thm:hd}}
\label{sec:pf_hd}
\begin{table}[H]
	%\label{Table:constants}
	\begin{adjustbox}{max width=\textwidth}
		\begin{tabular}{c|l}
			\hline
			Name & Description \\ \hline
			$H_{p_n}(X_i)$ & High-dimensional regressor constructed based on $X_i$ with dimension $p_n$\\
			$\theta_{a,s}^\textit{HD}(q)$ & For $a \in \{0,1\}$, $s \in \mathcal{S}$, and $\tau \in \Upsilon$, $\theta_{a,s}^\textit{HD}(q)$ is the pseudo true value defined in Assumption \ref{ass:hd}(i)\\ 
			$\hat{\theta}_{a,s}^\textit{HD}(q)$ & For $a \in \{0,1\}$, $s \in \mathcal{S}$, and $\tau \in \Upsilon$, $\hat{\theta}_{a,s}^\textit{HD}(q)$ is the estimator of $\theta_{a,s}^\textit{HD}(q)$ in \eqref{eq:thetahat^hd} \\
			$\varrho_{n,a}(s)$ & Lasso penalty defined after \eqref{eq:thetahat^hd} \\
			$\hat{\Omega}$ & Lasso penalty loading matrix defined after \eqref{eq:thetahat^hd} \\
			$\mathcal{M}_a(q,s,x)$ & For $a \in \{0,1\}$, $q \in Re$, $s \in \mathcal{S}$, and $x \in \Supp(X)$, $\mathcal{M}_a(q,s,x) = \mathbb{P}(Y_i(a) \leq q |S_i=s,X_i=x)$ \\
			\hline
		\end{tabular}
	\end{adjustbox}
	%	\caption{Table of notation}
\end{table}
We focus on the case with $a=1$. Note
\begin{align*}
\{X_i,Y_i(1)\}_{i \in I_1(s)}|\{A_i,S_i\}_{i \in [n]} \stackrel{d}{=} \{X_i^s,Y_i^s(1)\}_{ i = N(s)+1}^{N(s)+n_1(s)}|\{A_i,S_i\}_{i \in [n]}.
\end{align*}
where $\{X_i^s,Y_i(1)^s\}_{ i\in [n]}$ is an i.i.d. sequence that is independent of $\{A_i,S_i\}_{i \in [n]}$. Therefore, 
\begin{align*}
\hat{\theta}_{1,s}^\textit{HD}(q)|\{A_i,S_i\}_{i \in [n]} &\stackrel{d}{=}  \argmin_{\theta_a} \frac{-1}{n_{1}(s)}\sum_{i =N(s)+1}^{N(s)+n_1(s)} \biggl[1\{Y_i^s(1) \leq q\}\log(\lambda(H_{p_n}(X_i)^\top \theta_a)) \notag \\
& \hspace{-3mm}  + 1\{Y_i^s(1) > q\}\log(1-\lambda(H_{p_n}^\top(X^s_i)\theta_a))\biggr] + \frac{\varrho_{n,1}(s)}{n_1(s)}||\hat{ \Omega} \theta_a||_1\biggl|\{A_i,S_i\}_{i \in [n]},
\end{align*} 
and Assumption \ref{ass:hd}(vi) implies 
\begin{align*}
0< \kappa_1 \leq & \inf_{a = 0,1, s\in \mathcal{S}, |v|_0 \leq h_n \ell_n} \frac{v^T \left(\frac{1}{n_1(s)}\sum_{i = N(s)+1}^{N(s)+n_1(s)}H_{p_n}(X^s_i)H_{p_n}^\top(X^s_i) \right)v }{||v||_2^2} \\
\leq & \sup_{a = 0,1, s\in \mathcal{S}, |v|_0 \leq h_n \ell_n} \frac{v^T \left(\frac{1}{n_1(s)}\sum_{i = N(s)+1}^{N(s)+n_1(s)}H_{p_n}(X^s_i)H_{p_n}^\top(X^s_i) \right)v }{||v||_2^2} \leq \kappa_2 < \infty,
\end{align*}
and 
\begin{align*}
0< \kappa_1 \leq & \inf_{a = 0,1, s \in \mathcal{S}, |v|_0 \leq h_n \ell_n} \frac{v^T \mathbb{E}(H_{p_n}(X^s_i)H_{p_n}^\top(X^s_i))v }{||v||_2^2} \\
\leq & \sup_{a = 0,1, s\in \mathcal{S}, |v|_0 \leq h_n \ell_n} \frac{v^T \mathbb{E}(H_{p_n}(X^s_i)H_{p_n}^\top(X^s_i))v }{||v||_2^2} \leq \kappa_2 < \infty.
\end{align*}	
In addition, we have $n_1(s)/n \stackrel{a.s.}{\rightarrow} \pi(s)p(s)>0$. Therefore, based on the results established by \cite{BCFH13}, we have, conditionally on $\{A_i,S_i\}_{i \in [n]}$, and thus, unconditionally,  
\begin{align*}
\sup_{a=0,1,q \in \mathcal{Q}^\eps_a,s\in \mathcal{S}}||\hat{ \theta}_{a,s}^\textit{HD}(q) - \theta_{a,s}^\textit{HD}(q)||_2 = O_p\left(\sqrt{\frac{h_n \log(p_n)}{n}}\right),
\end{align*}
\begin{align*}
\sup_{a=0,1,q \in \mathcal{Q}^\eps_a,s\in \mathcal{S}}||\hat{ \theta}_{a,s}^{post}(q) - \theta_{a,s}^\textit{HD}(q)||_2 = O_p\left(\sqrt{\frac{h_n \log(p_n)}{n}}\right),
\end{align*}
\begin{align*}
\sup_{a=0,1,q \in \mathcal{Q}^\eps_a,s\in \mathcal{S}}||\hat{ \theta}_{a,s}^\textit{HD}(q)|| = O_p(h_n),
\end{align*}
and 
\begin{align*}
\sup_{a=0,1,q \in \mathcal{Q}^\eps_a,s\in \mathcal{S}}||\hat{ \theta}_{a,s}^{post}(q)|| = O_p(h_n). 
\end{align*}

In the following, we prove the results when $\hat{ \theta}_{a,s}^\textit{HD}(q)$ is used. The results corresponding to $\hat{ \theta}_{a,s}^{post}(q)$ can be proved in the same manner and are therefore omitted. Recall 
\begin{align*}
&\overline{\Delta}_1(\tau,s,X_i) =  \widehat{m}_1(\tau,s,X_i) - \overline{m}_1(\tau,s,X_i) \\
&=  \mathbb{P}(Y_i(1)\leq q_1(\tau)|X_i,S_i=s)- \lambda(H_{p_n}(X_i)^\top \hat{\theta}^\textit{HD}_{1,s}(\hat{q}_1(\tau))) \\
&=  \left[\mathcal{M}_1(q_1(\tau),s,X_i) - \mathcal{M}_1(\hat{q}_1(\tau),s,X_i) + r_a(\hat{q}_1(\tau),s,X_i)\right] \\
& + \left[\lambda(H_{p_n}(X_i)^\top \theta^\textit{HD}_{1,s}(\hat{q}_1(\tau))) -  \lambda(H_{p_n}(X_i)^\top \hat{\theta}^\textit{HD}_{1,s}(\hat{q}_1(\tau)))\right] \\
& \equiv \mathcal{R}_{a,s}(q_1(\tau),\hat{q}_1(\tau),X_i) + \lambda(H_{p_n}(X_i)^\top \theta^\textit{HD}_{1,s}(\hat{q}_1(\tau)))- \lambda(H_{p_n}(X_i)^\top \hat{\theta}^\textit{HD}_{1,s}(\hat{q}_1(\tau))), 
\end{align*} 
where 
\begin{align*}
\mathcal{R}_{a,s}(q,q',X_i) = \mathcal{M}_1(q,s,X_i) - \mathcal{M}_1(q',s,X_i) + r_a(q',s,X_i).
\end{align*}
Let
$$H_\lambda(\theta_1,\theta_2,s) = \mathbb{E}[\lambda(H_{p_n}(X_i)^\top \theta_1) - \lambda(H_{p_n}(X_i)^\top \theta_{2})|S_i=s],$$ 
and 
\begin{align*}
H_R(q,q',s) = \mathbb{E}(\mathcal{R}_{a,s}(q,q',X_i)|S_i=s).
\end{align*}
Then, we have
\begin{align}
& \sup_{\tau \in \Upsilon,s\in \mathcal{S}}\biggl|\frac{\sum_{i\in I_1(s)}\xi_i\overline{\Delta}_1(\tau,s,X_i)}{n_1^w(s)} - \frac{\sum_{i \in I_0(s)}\xi_i\overline{\Delta}_1(\tau,s,X_i)}{n_0^w(s)}\biggr| \notag  \\
&\leq \sup_{\tau \in \Upsilon,s\in \mathcal{S}}\left|\frac{\sum_{i\in I_1(s)}\xi_i[ \overline{\Delta}_1(\tau,s,X_i) - H_\lambda(\hat{\theta}^\textit{HD}_{1,s}(\hat{q}_1(\tau)),\theta^\textit{HD}_{1,s}(\hat{q}_1(\tau)),s)- H_R(q_1(\tau),\hat{ q}_1(\tau),s)]}{n_1^w(s)} \right| \notag \\
& +  \sup_{\tau \in \Upsilon,s\in \mathcal{S}}\left|\frac{\sum_{i\in I_0(s)}\xi_i[ \overline{\Delta}_1(\tau,s,X_i) - H_\lambda(\hat{\theta}^\textit{HD}_{1,s}(\hat{q}_1(\tau)),\theta^\textit{HD}_{1,s}(\hat{q}_1(\tau)),s)- H_R(q_1(\tau),\hat{ q}_1(\tau),s)] }{n_0^w(s)} \right| 
\label{eq:Delta_hd}
\end{align}
We aim to bound the first term on the RHS of \eqref{eq:Delta_hd}. Note for any $\eps>0$, there exists a constant $M>0$ such that 
\begin{align*}
\mathbb{P}\begin{pmatrix}
\sup_{q \in \mathcal{Q}_1^\eps}||\hat{\theta}_{1,s}^\textit{HD}(q) - \theta_{1,s}^\textit{HD}(q)||_2 \leq M \sqrt{\frac{h_n \log(p_n)}{n}},~\sup_{q \in \mathcal{Q}_1^\eps}||\hat{\theta}_{1,s}^\textit{HD}(q)||_0 \leq Mh_n,\\
\sup_{\tau \in \Upsilon}|\hat{q}_1(\tau) - q_1(\tau)| \leq M n^{-1/2}
\end{pmatrix} \geq 1-\eps. 
\end{align*}
On the set 
$$\mathcal{A}(\eps) = \begin{Bmatrix}
\sup_{q \in \mathcal{Q}_1^\eps}||\hat{\theta}_{1,s}^\textit{HD}(q) - \theta_{1,s}^\textit{HD}(q)||_2 \leq M \sqrt{\frac{h_n \log(p_n)}{n}},~\sup_{q \in \mathcal{Q}_1^\eps}||\hat{\theta}_{1,s}^\textit{HD}(q)||_0 \leq Mh_n,\\
\sup_{\tau \in \Upsilon}|\hat{q}_1(\tau) - q_1(\tau)| \leq M n^{-1/2}
\end{Bmatrix},$$ we have 
\begin{align*}
& \sup_{\tau \in \Upsilon,s\in \mathcal{S}}\left|\frac{\sum_{i\in I_1(s)}\xi_i[ \overline{\Delta}_1(\tau,s,X_i) - H_\lambda(\hat{\theta}^\textit{HD}_{1,s}(\hat{q}_1(\tau)),\theta^\textit{HD}_{1,s}(\hat{q}_1(\tau)),s)- H_R(q_1(\tau),\hat{ q}_1(\tau),s)]}{n_1^w(s)} \right| \\
&\leq  \sup_{\tau \in \Upsilon,s\in \mathcal{S}}\left|\frac{\sum_{i\in I_1(s)}\xi_i[ \lambda(H_{h_n}^\top(X_i)\theta_{1,s}^\textit{NP}(\tau)) - \lambda(H_{h_n}^\top(X_i)\hat{\theta}_{1,s}^\textit{NP}(\tau)) - H_\lambda(\theta_{1,s}^\textit{NP}(\tau),\hat{\theta}^\textit{NP}_{1,s}(\tau),s)] }{n_1(s)} \right| \\
& +  \sup_{\tau \in \Upsilon,s\in \mathcal{S}}\left|\frac{\sum_{i\in I_1(s)}\xi_i\left[\mathcal{R}_{1,s}(q_1(\tau),\hat{q}_1(\tau),X_i)-\mathbb{E}(\mathcal{R}_{1,s}(q_1(\tau),\hat{q}_1(\tau),X_i)|S_i=s)\right] }{n_1^w(s)} \right|\\
&\leq \frac{n_1(s)}{n_1^w(s)}\biggl[\sup\left|\frac{\sum_{i\in I_1(s)}\xi_i[ \lambda(H_{h_n}^\top(X_i)\theta_1) - \lambda(H_{h_n}^\top(X_i)\theta_2) - H_\lambda(\theta_1,\theta_2,s)] }{n_1(s)} \right| \\
& +  \sup_{q,q' \in \mathcal{Q}_1^\eps,|q-q'|\leq Mn^{-1/2},s\in \mathcal{S}}\left|\frac{\sum_{i\in I_1(s)}\xi_i\left[\mathcal{R}_{1,s}(q,q',X_i)-\mathbb{E}(\mathcal{R}_{1,s}(q,q',X_i)|S_i=s)\right] }{n_1(s)} \right|\biggr]\\
&\equiv \frac{n_1(s)}{n_1^w(s)}(D_1 + D_2),
\end{align*}
where the first supremum in the second inequality is taken over $\{s\in \mathcal{S}, \theta_1,\theta_2 \in \Re^{h_n}, ||\theta_1 - \theta_2||_2 \leq M\sqrt{h_n \log(n)/n},||\theta_1||_0 + ||\theta_2||_0 \leq Mh_n\}$. Denote
$$\mathcal{F} = \begin{Bmatrix}
\xi_i^s[ \lambda(H_{h_n}^\top(X^s_i)\theta_1) - \lambda(H_{h_n}^\top(X^s_i)\theta_2) - H_\lambda(\theta_1,\theta_2,s)]: \\
s\in \mathcal{S}, \theta_1,\theta_2 \in \Re^{h_n}, ||\theta_1 - \theta_2||_2 \leq M\sqrt{h_n \log(n)/n},||\theta_1||_0 + ||\theta_2||_0 \leq Mh_n
\end{Bmatrix}$$
with the envelope $F = 2\xi_i^s$. We further note that $||\max_{i \in [n]}2\xi_i^s||_{\mathbb{P},2} \leq C\log(n)$, 
\begin{align*}
\sup_{f \in \mathcal{F}}\mathbb{E}f^2 \leq \sup \mathbb{E}(H_{h_n}^\top(X^s_i)(\theta_1-\theta_2))^2 \leq \kappa_2 M^2h_n \log(p_n)/n, 
\end{align*}
and 
\begin{align*}
\sup_Q N(\mathcal{F}, e_Q, \eps||F||_{Q,2}) \leq \left(\frac{c_1 p_n}{\eps}\right)^{c_2h_n},
\end{align*}
where $c_1,c_2$ are two fixed constants. Therefore, Lemma \ref{lem:max_eq} implies  
\begin{align*}
D_1 = O_p\left( \frac{h_n  \log(p_n)}{n} + \frac{h_n \log(n) \log(p_n)}{n}\right) = o_p(n^{-1/2}).
\end{align*}

Similarly, denote 
\begin{align*}
\mathcal{F} = \begin{Bmatrix}
\xi^s_i\left[\mathcal{M}_1(q,s,X_i) - \mathcal{M}_1(q',s,X_i) + r_a(q',s,X_i)\right]:q,q' \in \mathcal{Q}_1^\eps, |q-q'| \leq Mn^{-1/2},s\in \mathcal{S}
\end{Bmatrix},
\end{align*}
with an envelope $F=\xi_i^s$. In addition, note that $\mathcal{F}$ is nested in 
\begin{align*}
\widetilde{\mathcal{F}} = \{\xi_i^s[ \mathcal{M}_1(q,s,X_i^s) - \lambda(H_{h_n}^\top(X^s_i)\theta_1)]: q \in \mathcal{Q}_1^\eps, \theta_1 \in \Re^{p_n}, ||\theta_1||_0 \leq h_n \},
\end{align*}
with the same envelope. Hence, 
\begin{align*}
\sup_Q N(\mathcal{F}, e_Q, \eps||F||_{Q,2}) \leq \sup_Q N(\widetilde{\mathcal{F}}, e_Q, \eps||F||_{Q,2}) \leq \left(\frac{a p_n}{\eps}\right)^{ch_n}.
\end{align*}
Last, 
\begin{align*}
\sup_{f \in \mathcal{F}}\mathbb{E}f^2 \leq C\sup_{q,q' \in \mathcal{Q}_1^\eps, |q-q'| \leq Mn^{-1/2},s\in \mathcal{S}} (|q-q'|^2 + \mathbb{E}r_a^2(q',s,X_i^s))= O(h_n  \log(p_n)/n). 
\end{align*}
Therefore, Lemma \ref{lem:max_eq} implies  
\begin{align*}
D_2= O_p\left( \frac{h_n  \log(p_n)}{n} + \frac{h_n \log(n)\log(p_n)}{n}\right) = o_p(n^{-1/2}). 
\end{align*}
This leads to \eqref{eq:Delta_hd}. We can establish Assumption \ref{ass:mhat}(i) in the same manner. Assumptions \ref{ass:mhat}(ii) and \ref{ass:mhat}(iii) can be established by the same argument used in Step 3 of the proof of Theorem \ref{thm:np}. This concludes the proof of Theorem \ref{thm:hd}.

\section{Technical Lemmas}
\label{sec:lem}
The first lemma was established in \cite{ZZ20}. 
\begin{lem}
	\label{lem:S}
	Let $S_k$ be the $k$-th partial sum of Banach space valued independent identically distributed random variables, then 
	\begin{align*}
	\mathbb{P}(\max_{1 \leq k \leq n}||S_k|| \geq \eps) \leq 3\max_{1 \leq k \leq n}\mathbb{P}(||S_k||\geq \eps/3) \leq 9\mathbb{P}(||S_n||\geq \eps/30).
	\end{align*}
\end{lem}
\begin{proof}
	The first inequality is due to \citet[Lemma E.1]{ZZ20} and the second inequality is due to \citet[Theorem 1]{M93}.
\end{proof}

The next lemma is due to \cite{CCK14} with our modification of their maximal inequality to the case with covariate-adaptive randomization. 
\begin{lem}
	Suppose Assumption \ref{ass:assignment1} holds. Let $w_i = 1$ or $\xi_i$ defined in Assumption \ref{ass:weight}. Denote $\mathcal{F}$ as a class of functions of the form $f(x,y_1,y_0)$ where, $f(x,y_1,y_0)$ is a measurable function and $\mathbb{E}(f(X_i,Y_i(1),Y_i(0))|S_i=s) = 0$. Further suppose $\max_{s \in \mathcal{S}}\mathbb{E}(|F_i|^q|S_i=s)<\infty$ for some $q \geq 2$, where  
	$$F_i = \sup_{f \in \mathcal{F}}|w_if(X_i,Y_i(1),Y_i(0))|,$$ 
	$\mathcal{F}$ is of the VC-type with coefficients $(\alpha_n,v_n)>0$, and $\sup_{f \in \mathcal{F}}\mathbb{E}(f^2|S=s) \leq \sigma_n^2$. Then, 
	\begin{align*}
	& \sup_{f \in \mathcal{F},s\in \mathcal{S}}\frac{1}{\sqrt{n}}\left|\sum_{i =1}^nA_i1\{S_i=s\}w_if(X_i,Y_i(1),Y_i(0))\right| \\
	&=  O_p\left( \sqrt{v_n \sigma_n^2 \log\left(\frac{\alpha_n||F||_{\mathbb{P},2}}{\sigma}\right)} + \frac{v_n ||\max_{i \in [n]} F_i||_{\mathbb{P},2} \log\left(\frac{\alpha_n||F||_{\mathbb{P},2}}{\sigma}\right)}{\sqrt{n}}\right),
	\end{align*}
	and 
	\begin{align*}
	& \sup_{f \in \mathcal{F},s\in \mathcal{S}}\frac{1}{\sqrt{n}}\left|\sum_{i =1}^n(1-A_i)1\{S_i=s\}w_if(X_i,Y_i(1),Y_i(0))\right| \\
	&= O_p\left( \sqrt{v_n \sigma_n^2 \log\left(\frac{\alpha_n||F||_{\mathbb{P},2}}{\sigma}\right)} + \frac{v_n ||\max_{i \in [n]} F_i||_{\mathbb{P},2} \log\left(\frac{\alpha_n||F||_{\mathbb{P},2}}{\sigma}\right)}{\sqrt{n}}\right).
	\end{align*}
	\label{lem:max_eq}
\end{lem}
\begin{proof}
	We focus on establishing the first statement. The proof of the second statement is similar and is omitted. Following \cite{BCS17}, we define the sequence of i.i.d. random variables $\{(w_i^s,X_i^s,Y_i^s(1),Y_i^s(0)): 1\leq i \leq n\}$  with marginal distributions equal to the distribution of $(w_i,X_i,Y_i(1),Y_i(0))|S_i = s$. The distribution of $\sum_{i =1}^nA_i1\{S_i=s\}w_if(X_i,Y_i(1),Y_i(0))$ is the same as the counterpart with units ordered by strata and then ordered by $A_i = 1$ first and $A_i = 0$ second within each stratum, i.e.,
	\begin{align*}
	\sum_{i =1}^nA_i1\{S_i=s\}w_if(X_i,Y_i(1),Y_i(0)) &\stackrel{d}{=} \sum_{i =N(s)+1}^{N(s)+n_1(s)}w^s_if(X^s_i,Y^s_i(1),Y^s_i(0)) \\
	&\equiv  \Gamma_n^s(N(s)+n_1(s),f) - \Gamma_n^s(N(s)+1,f),
	\end{align*}
	where $N(s) = \sum_{i =1}^n1\{S_i <s\}$ and 
	\begin{align*}
	\Gamma_n^s(k,f) = \sum_{i \in [k]}w^s_if(X^s_i,Y^s_i(1),Y^s_i(0)). 
	\end{align*}
	Let $\mu_n = \sqrt{v_n \sigma_n^2 \log\left(\frac{\alpha_n||F||_{\mathbb{P},2}}{\sigma}\right)} + \frac{v_n ||\max_{i \in [n]} F_i||_{\mathbb{P},2} \log\left(\frac{\alpha_n||F||_{\mathbb{P},2}}{\sigma}\right)}{\sqrt{n}}.$ Then, for some constant $C>0$, we have
	\begin{align*}
	& \mathbb{P}\left(  \sup_{f \in \mathcal{F},s\in \mathcal{S}}\frac{1}{\sqrt{n}}\left|\sum_{i =1}^nA_i1\{S_i=s\}w_if(X_i,Y_i(1),Y_i(0))\right| \geq t\mu_n \right) \\
	&= \mathbb{P}\left(   \sup_{f \in \mathcal{F},s\in \mathcal{S}}\frac{1}{\sqrt{n}}\left|\Gamma_n^s(N(s)+n_1(s),f) - \Gamma_n^s(N(s)+1,f)\right| \geq t \mu_n\right) \\
	&\leq \sum_{s \in \mathcal{S}}\mathbb{P}\left(   \max_{1\leq k \leq n}\sup_{f \in \mathcal{F}}\left|\frac{1}{\sqrt{n}}\Gamma_n^s(k,f)\right| \geq t\mu_n/2 \right) \\
	&\leq \sum_{s \in \mathcal{S}}9\;\mathbb{P}\left( \sup_{f \in \mathcal{F}}\left|\frac{1}{\sqrt{n}}\Gamma_n^s(n,f)\right| \geq t\mu_n/60\right) \\
	&\leq  \sum_{s \in \mathcal{S}}\frac{540 \;\mathbb{E}\left( \sup_{f \in \mathcal{F}}\left|\frac{1}{\sqrt{n}}\Gamma_n^s(n,f)\right|\right)}{t\mu_n} \\
	&= \sum_{s \in \mathcal{S}}\frac{540 \mathbb{E}\left( \sqrt{n} ||\mathbb{P}^s_n - \mathbb{P}^s||_{\mathcal{F}} \right)}{t\mu_n}\\
	&\leq C/t,
	\end{align*}
	where $\mathbb{P}_n^s$ and $\mathbb{P}^s$ are the empirical process and expectation w.r.t. i.i.d. data $\{w_i^s,X_i^s,Y_i^s(1),Y_i^s(0)\}_{i \in [n]}$, respectively, the second inequality is due to Lemma \ref{lem:S}, the last equality is due to the fact that 
	\begin{align*}
	\mathbb{E}w_i^sf(X_i^s,Y_i^s(1),Y_i^s(0)) = \mathbb{E}\left(w_if(X_i,Y_i(1),Y_i(0))|S_i=s\right) = 0, 
	\end{align*}
	and the last inequality is due to the fact that, by \citet[Corollary 5.1]{CCK14},
	\begin{align*}
	\mathbb{E}\left( \sqrt{n} ||\mathbb{P}^s_n - \mathbb{P}^s||_{\mathcal{F}}\right)  \leq C \mu_n.
	\end{align*}
	
	Then, for any $\eps>0$, we can choose $t \geq C/\eps$ so that 
	\begin{align*}
	\mathbb{P}\left(  \sup_{f \in \mathcal{F},s\in \mathcal{S}}\frac{1}{\sqrt{n}}\left|\sum_{i =1}^nA_i1\{S_i=s\}w_if(X_i,Y_i(1),Y_i(0))\right| \geq t \mu_n\right)  \leq \eps, 
	\end{align*}
	which implies the desired result. 
\end{proof}

The next lemma is similar to \citet[Lemma E.2]{ZZ20} but with additional covariates and regression adjustments. It is retained in the Supplement to make the paper self-contained. 
\begin{lem}
	Suppose Assumptions in Theorem \ref{thm:est} hold. Denote
	\begin{align*}
	\varpi_{n,1}(\tau) = & \sum_{s \in \mathcal{S}}\frac{1}{\sqrt{n}}\sum_{i =1}^nA_i1\{S_i=s\} \phi_1(\tau,s,Y_i(1),X_i)  \\
	- & \sum_{s \in \mathcal{S}}\frac{1}{\sqrt{n}}\sum_{i =1}^n(1-A_i)1\{S_i=s\}\phi_0(\tau,s,Y_i(0),X_i),
	\end{align*}
	and 
	\begin{align*}
	\varpi_{n,2}(\tau) = \frac{1}{\sqrt{n}} \sum_{i =1}^n\phi_s(\tau,S_i).
	\end{align*}
	Then, uniformly over $\tau \in \Upsilon$,  
	\begin{align*}
	(\varpi_{n,1}(\tau),\varpi_{n,2}(\tau)) \convD (\mathcal{B}_1(\tau),\mathcal{B}_2(\tau)),
	\end{align*}
	where $(\mathcal{B}_1(\tau),\mathcal{B}_2(\tau))$ are two independent Gaussian processes with covariance kernels $\Sigma_1(\tau,\tau')$ and $\Sigma_2(\tau,\tau')$, respectively, such that  
	\begin{align*}
	\Sigma_1(\tau,\tau') &=  \mathbb{E} \pi(S_i)\phi_1(\tau,S_i,Y_i(1),X_i)\phi_1(\tau',S_i,Y_i(1),X_i)  \\
	& + \mathbb{E} (1-\pi(S_i))\phi_0(\tau,S_i,Y_i(0),X_i)\phi_0(\tau',S_i,Y_i(0),X_i)
	\end{align*}
	and 
	\begin{align*}
	\Sigma_2(\tau,\tau') = \mathbb{E}\phi_s(\tau,S_i)\phi_s(\tau',S_i). 
	\end{align*}
	\label{lem:w_est}
\end{lem}
\begin{proof}
	We follow the general argument in the proof of \citet[Lemma B.2]{BCS17}. We divide the proof into two steps. In the first step, we show that 
	\begin{align*}
	(\varpi_{n,1}(\tau), \varpi_{n,2}(\tau)) \stackrel{d}{=}( \varpi^\star_{n,1}(\tau), \varpi_{n,2}(\tau)) + o_p(1), 
	\end{align*}
	where the $o_p(1)$ term holds uniformly over $\tau \in \Upsilon$, $ \varpi^\star_{n,1}(\tau) \indep \varpi_{n,2}(\tau)$, and, uniformly over $\tau \in \Upsilon$, 
	\begin{align*}
	\varpi^\star_{n,1}(\tau) \convD \mathcal{B}_1(\tau).
	\end{align*}
	In the second step, we show that 
	\begin{align*}
	\varpi_{n,2}(\tau) \convD \mathcal{B}_2(\tau)
	\end{align*}
	uniformly over $\tau \in \Upsilon$.  
	
	\vspace{1.5mm}
	\noindent \textbf{Step 1.} 
	Recall that we define $\{(X_i^s,Y_i^s(1),Y_i^s(0)): 1\leq i \leq n\}$ as a sequence of i.i.d. random variables with marginal distributions equal to the distribution of $(X_i,Y_i(1),Y_i(0))|S_i = s$ and $N(s) = \sum_{i =1}^n1\{S_i <s\}$. The distribution of $\varpi_{n,1}(\tau)$ is the same as the counterpart with units ordered by strata and then ordered by $A_i = 1$ first and $A_i = 0$ second within each stratum, i.e., 
	\begin{align*}
	\varpi_{n,1}(\tau)|\{ (A_i,S_i)_{i \in [n]} \}  & \stackrel{d}{=} \widetilde{\varpi}_{n,1}(\tau)|\{ (A_i,S_i)_{i \in [n]} \} 
	\end{align*}
	where
	\begin{align*}
	\widetilde{\varpi}_{n,1}(\tau) &\equiv   \sum_{s \in \mathcal{S}}\frac{1}{\sqrt{n}}\sum_{i = N(s)+1}^{N(s)+n_1(s)} \phi_1(\tau,s,Y^s_i(1),X^s_i) \\
	& -   \sum_{s \in \mathcal{S}}\frac{1}{\sqrt{n}}\sum_{i = N(s)+n_1(s)+1}^{N(s)+n(s)} \phi_0(\tau,s,Y^s_i(0),X^s_i),
	\end{align*}
	with
	\begin{align*}
	\phi_1(\tau,s,Y^s_i(1),X^s_i) = & \frac{\tau-1\{Y_i^s(1) \leq q_1(\tau)\} - m_1(\tau,s) -(1-\pi(s))\left(\overline{m}_1(\tau,s,X^s_i)-\overline{m}_1(\tau,s)\right)}{\pi(s) f_1(q_1(\tau))} \\
	- & \frac{\left( \overline{m}_0(\tau,s,X^s_i)-\overline{m}_0(\tau,s)\right)}{f_0(q_0(\tau))},  
	\end{align*}
	and 
	\begin{align*}
	\phi_0(\tau,s,Y^s_i(0),X^s_i) = & \frac{\tau-1\{Y_i^s(0) \leq q_0(\tau)\} - m_0(\tau,s) - \pi(s) \left(\overline{m}_0(\tau,s,X^s_i)-\overline{m}_0(\tau,s)\right)}{(1-\pi(s))f_0(q_0(\tau))} \\
	- & \frac{\left(\overline{m}_1(\tau,s,X^s_i)-\overline{m}_1(\tau,s)\right)}{f_1(q_1(\tau))}.
	\end{align*}
	As $\varpi_{n,2}(\tau)$ is only a function of $\{S_i\}_{i \in [n]}$, we have
	\begin{align*}
	(\varpi_{n,1}(\tau),\varpi_{n,2}(\tau)) \stackrel{d}{=} (\widetilde{\varpi}_{n,1}(\tau),\varpi_{n,2}(\tau)). 
	\end{align*}
	Let $F(s) = \mathbb{P}(S_i<s)$, $p(s) = \mathbb{P}(S_i=s)$, and 
	\begin{align*}
	\varpi^\star_{n,1}(\tau) \equiv &  \sum_{s \in \mathcal{S}}\frac{1}{\sqrt{n}}\sum_{i = \lfloor nF(s)\rfloor+1}^{\lfloor n(F(s)+\pi(s)p(s))\rfloor} \phi_1(\tau,s,Y^s_i(1),X^s_i) \\
	- & \sum_{s \in \mathcal{S}}\frac{1}{\sqrt{n}}\sum_{i = \lfloor n(F(s)+\pi(s)p(s))\rfloor+1}^{\lfloor n(F(s)+p(s))\rfloor} \phi_0(\tau,s,Y^s_i(0),X^s_i).
	\end{align*}
	Note $\varpi^\star_{n,1}(\tau)$ is a function of $(Y_i^s(1),Y_i^s(0),X_i^s)_{ i \in [n],s\in \mathcal{S}}$, which is independent of $\{A_i,S_i\}_{i \in [n]}$ by construction. Therefore, 
	\begin{align*}
	\varpi^\star_{n,1}(\tau) \indep \varpi_{n,2}(\tau).
	\end{align*} 
	Note that 
	\begin{align*}
	\frac{N(s)}{n} \convP F(s),\quad \frac{n_1(s)}{n} \convP \pi(s)p(s), \quad \text{and} \quad \frac{n(s)}{n} \convP p(s). 
	\end{align*}
	Denote $\Gamma_{n,a}(s,t,\tau) = \sum_{i =1}^{\lfloor nt \rfloor}\frac{1}{n}\phi_a(\tau,s,Y_i^s(a),X_i^s)$ for $a=0,1$. In order to show that $\sup_{\tau \in \Upsilon}|\widetilde{\varpi}_{n,1}(\tau) - \varpi^\star_{n,1}(\tau)| = o_p(1)$ and $\varpi_{n,1}^\star(\tau) \convD \mathcal{B}_1(\tau)$, it suffices to show that (1) for $a=0,1$ and $s \in \mathcal{S}$, the stochastic process
	\begin{align*}
	\{\Gamma_{n,a}(s,t,\tau): t \in (0,1),\tau \in \Upsilon\}
	\end{align*}
	is stochastically equicontinuous and (2) $\varpi_{n,1}^\star(\tau) \convD \mathcal{B}_1(\tau)$ converges to $B_1(\tau)$ in finite dimension. 
	
	\vspace{1mm}
	\noindent \textbf{Claim (1).} We want to bound
	\begin{align*}
	\sup|\Gamma_{n,a}(s,t_2,\tau_2) - \Gamma_{n,a}(s,t_1,\tau_1)|,
	\end{align*}
	where the supremum is taken over $0 < t_1<t_2<t_1+\eps < 1$ and $\tau_1<\tau_2<\tau_1+\eps$ such that $\tau_1,\tau_1+\eps \in \Upsilon.$ Note that,
	\begin{align}
	\label{eq:GW}
	& \sup|\Gamma_{n,a}(s,t_2,\tau_2) - \Gamma_{n,a}(s,t_1,\tau_1)| \notag \\
	\leq & \sup_{0 < t_1<t_2<t_1+\eps < 1,\tau \in \Upsilon}|\Gamma_{n,a}(s,t_2,\tau) - \Gamma_{n,a}(s,t_1,\tau)| +  \sup_{t \in (0,1),\tau_1,\tau_2 \in \Upsilon, \tau_1 < \tau_2 < \tau_1+\eps}|\Gamma_{n,a}(s,t,\tau_2) - \Gamma_{n,a}(s,t,\tau_1)|. 
	\end{align}
	Then,  for an arbitrary $\delta>0$, by taking $\eps = \delta^4$,  we have
	\begin{align*}
	& \mathbb{P}\left(\sup_{0 < t_1<t_2<t_1+\eps < 1,\tau \in \Upsilon}|\Gamma_{n,a}(s,t_2,\tau) - \Gamma_{n,a}(s,t_1,\tau)| \geq \delta\right) \\
	&= \mathbb{P}\left(\sup_{0 < t_1<t_2<t_1+\eps < 1,\tau \in \Upsilon} |\sum_{i=\lfloor nt_1 \rfloor+1}^{\lfloor nt_2 \rfloor}\phi_a(\tau,s,Y_i^s(a),X_i^s)| \geq \sqrt{n}\delta\right) \\
	&=  \mathbb{P}\left(\sup_{0 < t \leq \eps,\tau \in \Upsilon} |\sum_{i=1}^{\lfloor nt \rfloor}\phi_a(\tau,s,Y_i^s(a),X_i^s)| \geq \sqrt{n}\delta\right) \\
	&\leq \mathbb{P}\left(\max_{1 \leq k \leq \lfloor n\eps \rfloor}\sup_{\tau \in \Upsilon} |S_k(\tau)| \geq \sqrt{n}\delta\right) \\
	&\leq \frac{270\;\mathbb{E}\sup_{\tau \in \Upsilon} |\sum_{i=1}^{\lfloor n\eps \rfloor}\phi_a(\tau,s,Y_i^s(a),X_i^s)|}{\sqrt{n}\delta} \\
	&\lesssim \frac{\sqrt{n\eps}}{\sqrt{n} \delta} \lesssim \delta,
	\end{align*}
	where in the first inequality, $S_k(\tau) = \sum_{i=1}^k\phi_a(\tau,s,Y_i^s(a),X_i^s)$ and the second inequality holds due to Lemma \ref{lem:S}. To see the third inequality, denote 
	\begin{align*}
	\mathcal{F} = \{\phi_a(\tau,s,Y_i^s(a),X_i^s): \tau \in \Upsilon \}
	\end{align*}
	with an envelope function $F_i$ such that by Assumption \ref{ass:mhat}, $||F_i||_{\mathbb{P},q}<\infty$. In addition, by Assumption \ref{ass:mhat} again and the fact that 
	\begin{align*}
	\{\tau - 1\{Y_i^s(a) \leq q_a(\tau)\}  - m_a(\tau,s): \tau \in \Upsilon\} 
	\end{align*}
	is of the VC-type with fixed coefficients $(\alpha,v)$, so is $\mathcal{F}$. Then, we have
	\begin{align*}
	J(1,\mathcal{F}) < \infty,
	\end{align*}
	where 
	\begin{align*}
	J(\delta,\mathcal{F}) = \sup_Q \int_0^\delta \sqrt{1 + \log N(\mathcal{F},L_2(Q),\eps||F||_{Q,2})}d\eps,
	\end{align*} 
	$N(\mathcal{F},L_2(Q),\eps||F||_{Q,2})$ is the covering number, and the supremum is taken over all discrete probability measures $Q$. Therefore, by \citet[Theorem 2.14.1]{VW96}
	\begin{align*}
	\frac{270\mathbb{E}\sup_{\tau \in \Upsilon} |\sum_{i=1}^{\lfloor n\eps \rfloor}\phi_a(\tau,s,Y_i^s(a),X_i^s)|}{\sqrt{n}\delta} \lesssim \frac{\sqrt{\lfloor n\eps \rfloor}\left[ \mathbb{E}\sqrt{\lfloor n\eps \rfloor}||\mathbb{P}_{\lfloor n\eps \rfloor} - \mathbb{P}||_{\mathcal{F}}\right]}{\sqrt{n}\delta} \lesssim \frac{\sqrt{\lfloor n\eps \rfloor}J(1,\mathcal{F})}{\sqrt{n}\delta}.
	\end{align*}
	For the second term on the RHS of \eqref{eq:GW}, by taking $\eps = \delta^4$, we have 
	\begin{align*}
	& \mathbb{P}\left(\sup_{t \in (0,1),\tau_1,\tau_2 \in \Upsilon, \tau_1 < \tau_2 < \tau_1+\eps}|\Gamma_{n,a}(s,t,\tau_2) - \Gamma_{n,a}(s,t,\tau_1)| \geq \delta\right) \\
	&=  \mathbb{P}\left(\max_{1 \leq k \leq n}\sup_{\tau_1,\tau_2 \in \Upsilon, \tau_1 < \tau_2 < \tau_1+\eps}|S_k(\tau_1,\tau_2)| \geq \sqrt{n}\delta\right) \\
	& \leq \frac{270 \;\mathbb{E}\sup_{\tau_1,\tau_2 \in \Upsilon, \tau_1 < \tau_2 < \tau_1+\eps} |\sum_{i=1}^{n}(\phi_a(\tau_2,s,Y_i^s(a),X_i^s) - \phi_a(\tau_1,s,Y_i^s(a),X_i^s)|}{\sqrt{n}\delta} \\
	& \lesssim\delta \sqrt{\log(\frac{C}{\delta^2})},
	\end{align*}
	where in the first equality, $S_k(\tau_1,\tau_2) = \sum_{i=1}^k (\phi_a(\tau_2,s,Y_i^s(a),X_i^s)- \phi_a(\tau_1,s,Y_i^s(a),X_i^s))$ and the first inequality is due to Lemma \ref{lem:S}. To see the last inequality, denote 
	\begin{align*}
	\mathcal{F} = \{\phi_a(\tau_2,s,Y_i^s(a),X_i^s)- \phi_a(\tau_1,s,Y_i^s(a),X_i^s): \tau_1,\tau_2 \in \Upsilon, \tau_1 < \tau_2 < \tau_1+\eps\}
	\end{align*}
	with a constant envelope function $F_i$ such that $||F_i||_{\mathbb{P},q} <\infty$. In addition, due to Assumptions \ref{ass:tau}.2 and \ref{ass:mhat}.3, one can show that 
	\begin{align*}
	\sup_{f \in \mathcal{F}}\mathbb{E}f^2 \leq c\eps \equiv \sigma^2 
	\end{align*}
	for some constant $c>0$. Last, due to Assumption \ref{ass:mhat}.2, $\mathcal{F}$ is of the VC-type with fixed coefficients $(\alpha,v)$. Therefore, by \citet[Corollary 5.1]{CCK14}, 
	\begin{align*}
	& \frac{270\mathbb{E}\sup_{\tau_1,\tau_2 \in \Upsilon, \tau_1 < \tau_2 < \tau_1+\eps} |\sum_{i=1}^{n}(\phi_a(\tau_2,s,Y_i^s(a),X_i^s)- \phi_a(\tau_1,s,Y_i^s(a),X_i^s))|}{\sqrt{n}\delta} \\
	\lesssim & \frac{\sqrt{n}\mathbb{E}||\mathbb{P}_n - \mathbb{P}||_\mathcal{F}}{\delta} \lesssim \sqrt{\frac{\sigma^2 \log(\frac{C}{\sigma})}{ \delta^2}} + \frac{C\log(\frac{C}{\sigma})}{\sqrt{n}\delta} \lesssim \delta \sqrt{\log(\frac{C}{\delta^2})},
	\end{align*}
	where the last inequality holds by letting $n$ be sufficiently large. Note that $\delta \sqrt{\log(\frac{C}{\delta^2})} \rightarrow 0$ as $\delta \rightarrow 0$. This concludes the proof of Claim (1). 
	
	\vspace{1.5mm}
	\noindent\textbf{Claim (2).}  For a single $\tau$, by the triangular array CLT, 
	\begin{align*}
	\varpi_{n,1}^\star(\tau) \convD N(0,\Sigma_1(\tau,\tau)),
	\end{align*}
	where 
	\begin{align*}
	\Sigma_1(\tau,\tau) &=  \lim_{n\rightarrow \infty} \sum_{s \in \mathcal{S}} \frac{(\lfloor n(F(s)+\pi(s)p(s))\rfloor -\lfloor nF(s)\rfloor)}{n}\mathbb{E}\phi_1^2(\tau,s,Y_i^s(1),X_i^s) \\
	& + \lim_{n\rightarrow \infty} \sum_{s \in \mathcal{S}} \frac{(\lfloor n(F(s)+p(s))\rfloor -\lfloor n(F(s)+p(s)\pi(s))\rfloor)}{n}\mathbb{E}\phi_0^2(\tau,s,Y_i^s(0),X_i^s) \\
	&=  \sum_{s \in \mathcal{S}} p(s) \mathbb{E}(\pi(s)\phi_1^2(\tau,S_i,Y_i(1),X_i) +(1-\pi(s))\phi_0^2(\tau,S_i,Y_i(0),X_i) |S_i=s) \\
	&=  \mathbb{E}\pi(S_i)\phi_1^2(\tau,S_i,Y_i(1),X_i) + \mathbb{E}(1-\pi(S_i))\phi_0^2(\tau,S_i,Y_i(0),X_i).
	\end{align*}
	Finite dimensional convergence is proved by the Cram\'{e}r-Wold device. In particular, we can show that the covariance kernel is 
	\begin{align*}
	\Sigma_1(\tau,\tau') = & \mathbb{E}\pi(S_i)\phi_1(\tau,S_i,Y_i(1),X_i)\phi_1(\tau',S_i,Y_i(1),X_i) \\
	& +\mathbb{E}(1-\pi(S_i))\phi_0(\tau,S_i,Y_i(0),X_i)\phi_0(\tau',S_i,Y_i(0),X_i).
	\end{align*}
	This concludes the proof of Claim (2), and thereby leads to the desired results in Step 1.   
	
	\vspace{1.5mm}
	\noindent\textbf{Step 2.} As $m_a(\tau,S_i) = \tau - \mathbb{P}(Y_i(a) \leq q_a(\tau)|S_i)$ is Lipschitz continuous in $\tau$ with a bounded Lipschitz constant, $\{m_a(\tau,S_i):\tau \in \Upsilon \}$ is of the VC-type with fixed coefficients $(\alpha,v)$ and a constant envelope function. Therefore, $\{\frac{m_1(\tau,S_i)}{f_1(q_1(\tau))} - \frac{m_0(\tau,S_i)}{f_0(q_0(\tau))}:\tau \in \Upsilon \}$ is a Donsker class and we have 
	\begin{align*}
	\varpi_{n,2}(\tau) \convD \mathcal{B}_2(\tau),
	\end{align*}
	where $ \mathcal{B}_2(\tau)$ is a Gaussian process with covariance kernel 
	\begin{align*}
	\Sigma_2(\tau,\tau') = \mathbb{E}\left(\frac{m_1(\tau,S_i)}{f_1(q_1(\tau))} - \frac{m_0(\tau,S_i)}{f_0(q_0(\tau))}\right)\left(\frac{m_1(\tau',S_i)}{f_1(q_1(\tau'))} - \frac{m_0(\tau',S_i)}{f_0(q_0(\tau'))}\right) \equiv \mathbb{E}\phi_s(\tau,S_i)\phi_s(\tau',S_i). 
	\end{align*}
	This concludes the proof.
\end{proof} 

\begin{lem}
	Suppose the Assumptions in Theorem \ref{thm:boot} hold and recall $D_n^w(s) = \sum_{i =1}^n\xi_i(A_i - \pi(S_i))1\{S_i=s\}$. Then, $\max_{s\in \mathcal{S}}|(D_n^w(s) - D_n(s))/n(s)| = o_p(1)$ and $\max_{s \in \mathcal{S}}|D_n^w(s)/n^w(s)| = o_p(1)$. 
	\label{lem:Dw}
\end{lem}
\begin{proof}
	We note that $n^w(s)/n(s) \convP 1$ and $D_n(s)/n(s) \convP 0$. Therefore, we only need to show 
	\begin{align*}
	\frac{D_n^w(s) - D_n(s)}{n(s)} = \sum_{i=1}^n\frac{(\xi_i-1)(A_i-\pi(s))1\{S_i=s\}}{n(s)} \convP 0. 
	\end{align*}
	As $n(s) \rightarrow \infty$ a.s., given data, 
	
	\begin{align*}
	\frac{1}{n(s)}\sum_{i=1}^n(A_i - \pi(s))^21\{S_i = s\} = & \frac{1}{n}\sum_{i=1}^n\left(A_i - \pi(s) - 2\pi(s)(A_i - \pi(s)) + \pi(s) - \pi^2(s)\right)1\{S_i=s\} \\
	= & \frac{D_n(s) - 2\pi(s) D_n(s)}{n(s)} + \pi(s)(1-\pi(s)) \convP \pi(s) (1-\pi(s)).  
	\end{align*}
	Then, by the Lindeberg CLT, conditionally on data, 
	\begin{align*}
	\frac{1}{\sqrt{n(s)}} \sum_{i =1}^n (\xi_i - 1)(A_i - \pi(s))1\{S_i = s\} \convD N(0,\pi(s)(1-\pi(s))) = O_p(1),  
	\end{align*}
	and thus 
	\begin{align*}
	\frac{D_n^w(s) - D_n(s)}{n(s)} = O_p(n^{-1/2}(s)) = o_p(1). 
	\end{align*}
\end{proof}

\begin{lem} \label{lem:w_boot}
	Suppose the Assumptions in Theorem \ref{thm:boot} hold. Then, uniformly over $\tau \in \Upsilon$,   
	\begin{align*}
	\varpi_{n,1}^w(\tau) + \varpi_{n,2}^w(\tau) \underset{\xi}{\overset{\mathbb{P}}{\convD}} \mathcal{B}(\tau),
	\end{align*}
	where $\mathcal{B}(\tau)$ is a Gaussian process with the covariance kernel
	\begin{align*}
	&\Sigma(\tau,\tau') =  \mathbb{E} \pi(S_i)\phi_1(\tau,S_i,Y_i(1),X_i)\phi_1(\tau',S_i,Y_i(1),X_i)  \\
	& + \mathbb{E} (1-\pi(S_i))\phi_0(\tau,S_i,Y_i(0),X_i)\phi_0(\tau',S_i,Y_i(0),X_i) 
	+ \mathbb{E}\phi_s(\tau,S_i)\phi_s(\tau',S_i). 
	\end{align*}
\end{lem}
\begin{proof}
	We divide the proof into two steps. In the first step, we show the conditional stochastic equicontinuity of $\varpi_{n,1}^w(\tau)$ and $\varpi_{n,2}^w(\tau)$. In the second step, we show the finite-dimensional convergence of 	$\varpi_{n,1}^w(\tau) + \varpi_{n,2}^w(\tau)$ conditional on data. 
	
	\vspace{1.5mm}
	\noindent\textbf{Step 1.} Following the same idea in the proof of Lemma \ref{lem:w_est}, we define $\{(\xi_i^s,X_i^s,Y_i^s(1),Y_i^s(0)): 1\leq i \leq n\}$ as a sequence of i.i.d. random variables with marginal distributions equal to the distribution of $(\xi_i,X_i,Y_i(1),Y_i(0))|S_i = s$ and $N(s) = \sum_{i =1}^n1\{S_i <s\}$. The distribution of $\varpi_{n,1}(\tau)$ is the same as the counterpart with units ordered by strata and then ordered by $A_i = 1$ first and $A_i = 0$ second within each stratum, i.e., 
	\begin{align*}
	\varpi_{n,1}^w(\tau)|\{ (A_i,S_i)_{i \in [n]} \} &\stackrel{d}{=}  \widetilde{\varpi}^w_{n,1}(\tau)|\{ (A_i,S_i)_{i \in [n]} \}, 
	\end{align*}
	and thus, 
	\begin{align}
	\varpi_{n,1}^w(\tau) &\stackrel{d}{=}  \widetilde{\varpi}^w_{n,1}(\tau),
	\label{eq:Wn1w1}
	\end{align}
	where
	\begin{align*}
	\widetilde{\varpi}_{n,1}^w(\tau) &\equiv   \sum_{s \in \mathcal{S}}\frac{1}{\sqrt{n}}\sum_{i = N(s)+1}^{N(s)+n_1(s)} (\xi_i^s-1)\phi_1(\tau,s,Y^s_i(1),X^s_i) \\
	& -   \sum_{s \in \mathcal{S}}\frac{1}{\sqrt{n}}\sum_{i = N(s)+n_1(s)+1}^{N(s)+n(s)} (\xi_i^s-1)\phi_0(\tau,s,Y^s_i(0),X^s_i).
	\end{align*}
	In addition, let 
	\begin{align*}
	\widetilde{\varpi}^{w\star}_{n,1}(\tau) &\equiv   \sum_{s \in \mathcal{S}}\frac{1}{\sqrt{n}}\sum_{i = \lfloor nF(s)\rfloor+1}^{\lfloor n(F(s)+\pi(s)p(s))\rfloor} (\xi_i^s-1) \phi_1(\tau,s,Y^s_i(1),X^s_i) \\
	& -   \sum_{s \in \mathcal{S}}\frac{1}{\sqrt{n}}\sum_{i = \lfloor n(F(s)+\pi(s)p(s))\rfloor+1}^{\lfloor n(F(s)+p(s))\rfloor} (\xi_i^s-1)\phi_0(\tau,s,Y^s_i(0),X^s_i).
	\end{align*}
	Following exactly the same argument as in the proof of Lemma \ref{lem:w_est}, we have 
	\begin{align}
	\sup_{\tau \in \Upsilon}|\widetilde{\varpi}_{n,1}^w(\tau) -\widetilde{\varpi}^{w\star}_{n,1}(\tau)| = o_p(1). 
	\label{eq:W1wstar}
	\end{align}
	and $\widetilde{\varpi}^{w\star}_{n,1}(\tau)$ is \textit{unconditionally} stochastically equicontinuous, i.e., for any $\eps>0$, as $n \rightarrow \infty$ followed by $\delta \rightarrow 0$, we have 
	\begin{align*}
	& \mathbb{E}\mathbb{P}_{\xi}\left(\sup_{\tau_1,\tau_2 \in \Upsilon, \tau_1 < \tau_2 < \tau_1+\delta}|\widetilde{\varpi}^{w\star}_{n,1}(\tau_1)-\widetilde{\varpi}^{w\star}_{n,1}(\tau_2)| \geq \eps  \right) \\
	& = 
	\mathbb{P}\left(\sup_{\tau_1,\tau_2 \in \Upsilon, \tau_1 < \tau_2 < \tau_1+\delta}|\widetilde{\varpi}^{w\star}_{n,1}(\tau_1)-\widetilde{\varpi}^{w\star}_{n,1}(\tau_2)| \geq \eps  \right) \rightarrow 0,
	\end{align*}
	where $\mathbb{P}_\xi$ means the probability operator is with respect to the bootstrap weights $\{\xi_i\}_{i \in [n]}$ and is conditional on data. This implies the \textit{unconditional} stochastic equicontinuity of $\varpi^{w}_{n,1}(\tau)$ due to \eqref{eq:Wn1w1} and \eqref{eq:W1wstar}, which further implies the \textit{conditional} stochastic equicontinuity of $\varpi^{w}_{n,1}(\tau)$, i.e., for any $\eps>0$, as $n \rightarrow \infty$ followed by $\delta \rightarrow 0$, 
	\begin{align*}
	\mathbb{P}_{\xi}\left(\sup_{\tau_1,\tau_2 \in \Upsilon, \tau_1 < \tau_2 < \tau_1+\delta}|\widetilde{\varpi}^{w\star}_{n,1}(\tau_1)-\widetilde{\varpi}^{w\star}_{n,1}(\tau_2)| \geq \eps  \right) \convP 0.
	\end{align*}
	By a similar but simpler argument, the \textit{conditional} stochastic equicontinuity of $\varpi_{n,2}^w(\tau)$ holds as well. This concludes the first step. 
	
	\vspace{1.5mm}
	\noindent\textbf{Step 2.} We first show the asymptotic normality of $\varpi_{n,1}^{w}(\tau)+\varpi_{n,2}^w(\tau)$ conditionally on data for a fixed $\tau$. Note 
	\begin{align*}
	& \varpi_{n,1}^{w}(\tau)+\varpi_{n,2}^w(\tau) \\
	= & \frac{1}{\sqrt{n}}\sum_{i =1}^n (\xi_i - 1) \left[ A_i 1\{S_i=s\}\phi_1(\tau,S_i,Y_i(1),X_i) - (1-A_i)1\{S_i=s\}\phi_1(\tau,S_i,Y_i(1),X_i) + \phi_s(\tau,S_i)\right] \\
	\equiv & \frac{1}{\sqrt{n}}\sum_{i =1}^n (\xi_i - 1) \mathcal{J}_i(\tau,s).
	\end{align*}
	Conditionally on data, $\{(\xi_i - 1) \mathcal{J}_i(\tau)\}_{i \in [n]}$ is a sequence of i.n.i.d. random variables. In order to apply the Lindeberg-Feller central limit theorem, we only need to show that (1)
	\begin{align*}
	\frac{1}{n}\sum_{i =1}^n \mathcal{J}^2_{n,i}(\tau) \convP \Sigma(\tau,\tau),
	\end{align*}
	where $\Sigma(\tau,\tau)$ is defined in Theorem \ref{thm:est}, and (2) the Lindeberg condition holds, i.e., 
	\begin{align*}
	\frac{1}{n}\sum_{i =1}^n \mathcal{J}_{n,i}^2(\tau) \mathbb{E} (\xi_i-1)^21\{|(\xi_i-1) \mathcal{J}_{n,i}(\tau)| \geq \sqrt{n}\eps  \} \convP 0. 
	\end{align*}
	For part (1), we have 
	\begin{align*}
	\frac{1}{n}\sum_{i =1}^n \mathcal{J}^2_{n,i}(\tau) = \sigma_1^2 + 2\sigma_{12} + \sigma_2^2,
	\end{align*}
	where 
	\begin{align*}
	\sigma_1^2 =\frac{1}{n}\sum_{i =1}^n\left[ A_i 1\{S_i=s\}\phi_1(\tau,S_i,Y_i(1),X_i) - (1-A_i)1\{S_i=s\}\phi_0(\tau,S_i,Y_i(1),X_i) \right]^2,
	\end{align*}
	\begin{align*}
	\sigma_{12} =\frac{1}{n}\sum_{i =1}^n  \left[ A_i 1\{S_i=s\}\phi_1(\tau,S_i,Y_i(1),X_i) - (1-A_i)1\{S_i=s\}\phi_1(\tau,S_i,Y_i(1),X_i) \right] \phi_s(\tau,S_i),
	\end{align*}
	and
	\begin{align*}
	\sigma_2^2 = \frac{1}{n}\sum_{i =1}^n\phi_s^2(\tau,S_i). 
	\end{align*}
	Note 
	\begin{align*}
	\sigma_1^2 &=  \frac{1}{n}\sum_{i =1}^nA_i 1\{S_i=s\}\phi^2_1(\tau,S_i,Y_i(1),X_i) + \frac{1}{n}\sum_{i =1}^n (1-A_i)1\{S_i=s\}\phi_0^2(\tau,S_i,Y_i(1),X_i) \\
	&\stackrel{d}{=} \frac{1}{n}\sum_{s \in \mathcal{S}}\sum_{i = N(s)+1}^{N(s)+n_1(s)}\phi^2_1(\tau,s,Y_i^s(1),X^s_i) + \frac{1}{n}\sum_{i =N(s)+n_1(s)+1} ^{N(s)+n(s)}\phi_0^2(\tau,s,Y^s_i(1),X^s_i) \\
	&\convP \sum_{s \in \mathcal{S}} [\pi(s) \mathbb{E}\phi^2_1(\tau,s,Y_i^s(1),X^s_i) + (1-\pi(s)) \mathbb{E}\phi^2_0(\tau,s,Y_i^s(1),X^s_i)] \\
	&=  \mathbb{E}\left[\pi(S_i)\phi^2_1(\tau,S_i,Y_i(1),X_i) + (1-\pi(s))\phi^2_0(\tau,S_i,Y_i(1),X_i)\right],
	\end{align*}
	where the convergence holds since $N(s)/n \rightarrow F(s)$, $n_1(s)/n \convP \pi(s)p(s)$, $n(s)/n \convP p(s)$, and uniform convergence of the partial sum process. Similarly, 
	\begin{align*}
	\sigma_{12} &\stackrel{d}{=}\frac{1}{n}\sum_{s \in \mathcal{S}}\sum_{i = N(s)+1}^{N(s)+n_1(s)}\phi_1(\tau,s,Y_i^s(1),X^s_i)\phi_s(\tau,s) + \frac{1}{n}\sum_{i =N(s)+n_1(s)+1} ^{N(s)+n(s)}\phi_0(\tau,s,Y^s_i(1),X^s_i)\phi_s(\tau,s) \\
	&\convP  \sum_{s \in \mathcal{S}} [\pi(s) \mathbb{E}\phi_1(\tau,s,Y_i^s(1),X^s_i) + (1-\pi(s)) \mathbb{E}\phi_0(\tau,s,Y_i^s(1),X^s_i)]\phi_s(\tau,s) 
	= 0,
	\end{align*}
	where we use the fact that 
	\begin{align*}
	\mathbb{E}\phi_1(\tau,s,Y_i^s(1),X^s_i) =  \mathbb{E}\phi_0(\tau,s,Y_i^s(1),X^s_i) = 0. 
	\end{align*}
	By the standard weak law of large numbers, we have
	\begin{align*}
	\sigma_2^2 \convP \mathbb{E}\phi_s^2(\tau,S_i).
	\end{align*}
	Therefore, 
	\begin{align*}
	\frac{1}{n}\sum_{i =1}^n \mathcal{J}^2_{n,i}(\tau) \convP \mathbb{E}\left[\pi(S_i)\phi^2_1(\tau,S_i,Y_i(1),X_i) + (1-\pi(s))\phi^2_0(\tau,S_i,Y_i(1),X_i)\right] + \mathbb{E}\phi_s^2(\tau,S_i) = \Sigma(\tau,\tau).
	\end{align*}
	
	To verify the Lindeberg condition, we note that  
	\begin{align*}
	& \frac{1}{n}\sum_{i =1}^n \mathcal{J}_{n,i}^2(\tau) \mathbb{E} (\xi_i-1)^21\{|(\xi_i-1) \mathcal{J}_{n,i}(\tau)| \geq \sqrt{n}\eps  \} \\
	&\leq \frac{1}{n(\sqrt{n}\eps)^{q-2}}\sum_{i =1}^n \mathcal{J}_{n,i}^q(\tau) \mathbb{E} (\xi_i-1)^q  \\
	&\leq  \frac{c}{n(\sqrt{n}\eps)^{q-2}}\sum_{i =1}^n [\phi_1^q(\tau,S_i,Y_i(1),X_i)+\phi_0^q(\tau,S_i,Y_i(1),X_i) + \phi_s^q(\tau,S_i)] = o_p(1),
	\end{align*}
	where the last equality is due to Assumption \ref{ass:mhat}(ii) and the fact that $\eta_{i,a}(\tau,s)$ is bounded. 
	
	Finite dimensional convergence of $\varpi_{n,1}^{w}(\tau)+\varpi_{n,2}^w(\tau)$ across $\tau$ can be established in the same manner using the Cram\'{e}r-Wold device and the details are omitted. By the same calculation as that given above the covariance kernel is shown to be  
	\begin{align*}
	& \lim_{n\rightarrow \infty}\frac{1}{n}\sum_{i =1}^n\mathcal{J}_{n,i}(\tau_1)\mathcal{J}_{n,i}(\tau_2) \\
	&=  \mathbb{E}\left[\pi(S_i)\phi_1(\tau_1,S_i,Y_i(1),X_i)\phi_1(\tau_2,S_i,Y_i(1),X_i)\right] \\
	& + \mathbb{E}\left[(1-\pi(s))\phi_0(\tau_1,S_i,Y_i(1),X_i)\phi(\tau_2,S_i,Y_i(1),X_i)\right] \\
	& + \mathbb{E}\phi_s(\tau_1,S_i)\phi_s(\tau_2,S_i) = \Sigma(\tau_1,\tau_2), 
	\end{align*} 
	which concludes the proof. 
\end{proof}

\begin{lem}
	Suppose the Assumptions in Theorem \ref{thm:np} hold. Then, 
	\begin{align*}
	\sup_{\tau \in \Upsilon, a=0,1, s \in \mathcal{S}}\left\Vert \frac{1}{n_a(s)}\sum_{i \in I_{a}(s)}(1\{Y_i \leq \hat{q}_a(\tau)\} - m_a(\tau,s,X_i))H_{h_n}(X_i)\right\Vert_\infty = O_p\left( \sqrt{\frac{\log(n)}{n}}\right).
	\end{align*}
	\label{lem:np1}
\end{lem}
\begin{proof}
	We focus on $a=1$. We have 
	\begin{align}
	& \sup_{\tau \in \Upsilon}\left\Vert \frac{1}{n_1(s)}\sum_{i \in I_{1}(s)}(1\{Y_i \leq \hat{q}_1(\tau)\} - m_1(\tau,s,X_i))H_{h_n}(X_i)\right\Vert_\infty \notag \\
	&\leq  \sup_{\tau \in \Upsilon}\left\Vert \frac{1}{n_1(s)}\sum_{i \in I_{1}(s)}(1\{Y_i \le \hat{q}_1(\tau)\} - \mathbb{P}(Y_i(1) \leq \hat{q}_1(\tau)|X_i,S_i=s))H_{h_n}(X_i)\right\Vert_\infty \notag \\
	& + \sup_{\tau \in \Upsilon}\left\Vert \frac{1}{n_1(s)}\sum_{i \in I_{1}(s)}( \mathbb{P}(Y_i(1) \leq \hat{q}_1(\tau)|X_i,S_i=s) - m_1(\tau,s,X_i))H_{h_n}(X_i)\right\Vert_\infty  \notag \\
	&\leq  \sup_{q \in \Re}\left\Vert \frac{1}{n_1(s)}\sum_{i \in I_{1}(s)}(1\{Y_i \leq q\} - \mathbb{P}(Y_i(1) \leq q|X_i,S_i=s))H_{h_n}(X_i)\right\Vert_\infty \notag \\
	& + \sup_{\tau \in \Upsilon}\left\Vert \frac{1}{n_1(s)}\sum_{i \in I_{1}(s)}( \mathbb{P}(Y_i(1) \leq \hat{q}_1(\tau)|X_i,S_i=s) - m_1(\tau,s,X_i))H_{h_n}(X_i)\right\Vert_\infty.
	\label{eq:np1}
	\end{align}
	Define 
	\begin{align*}
	\mathcal{F}_h = \{1\{Y_i^s(1) \leq q\}H_{h_n,h}(X_i): q \in \Re\}, \quad \mathcal{F} = \cup_{h \in [h_n]}\mathcal{F}_{h_n},
	\end{align*}
	and let $H_{h_n,h}(X_i)$ be the $h$-th coordinate of $H_{h_n}(X_i)$. For each $h \in [h_n]$, $\mathcal{F}_{h}$ is of the VC-type with fixed coefficients $(\alpha,v)$ and a common envelope $F_i = ||H_{h_n}(X_i)||_2 \leq \zeta(h_n)$, i.e., 
	\begin{align*}
	\sup_Q N(\mathcal{F}_h,e_Q,\eps||F||_{Q,2}) \leq \left(\frac{\alpha}{\eps}\right)^v, \quad \forall \eps \in (0,1],
	\end{align*}
	where the supremum is taken over all finitely discrete probability measures. This implies 
	\begin{align*}
	\sup_Q N(\mathcal{F},e_Q,\eps||F||_{Q,2}) \leq \sum_{h \in [h_n]}\sup_Q N(\mathcal{F}_h,e_Q,\eps||F||_{Q,2}) \leq \left (\frac{\alpha h_n}{\eps}\right)^v \, \quad \forall \eps \in (0,1],
	\end{align*}
	i.e., $\mathcal{F}$ is also of the VC-type with coefficients $(\alpha h_n,v)$. In addition, 
	\begin{align*}
	\sup_{f \in \mathcal{F}}\mathbb{E}f^2 \leq \max_{h \in [h_n]}\mathbb{E}H_{h_n,h}^2(X_i) \leq C<\infty.
	\end{align*}
	Then, Lemma \ref{lem:max_eq} implies 
	\begin{align*}
	& \sup_{q \in \Re}\left\Vert \frac{1}{n_1(s)}\sum_{i \in I_{1}(s)}(1\{Y_i \leq q\} - \mathbb{P}(Y_i(1) \leq q|X_i,S_i=s))H_{h_n}(X_i)\right\Vert_\infty  \\
	&=  O_p\left( \sqrt{ \frac{\log(h_n\zeta(h_n))}{n}} + \frac{\zeta(h_n) \log (\zeta(h_n))}{n}\right) =O_p\left( \sqrt{\frac{\log(n)}{n}}\right).
	\end{align*}
	
	For the second term of \eqref{eq:np1}, because $\sup_{q \in \Re, x \in \Supp(X), s\in \mathcal{S}}f_1(q|x,s) < \infty$, we have 
	\begin{align*}
	& \sup_{\tau \in \Upsilon}\left\Vert \frac{1}{n_1(s)}\sum_{i \in I_{1}(s)}( \mathbb{P}(Y_i(1) \leq \hat{q}_1(\tau)|X_i,S_i=s) - m_1(\tau,s,X_i))H_{h_n}(X_i)\right\Vert_\infty \\
	&\leq  \sup_{\tau \in \Upsilon}|\hat{q}_1(\tau) - q_1(\tau)| \left\Vert\frac{1}{n_1(s)}\sum_{i \in I_1(s)}|H_{h_n}(X_i)| \right\Vert_\infty \\
	&\leq \sup_{\tau \in \Upsilon}|\hat{q}_1(\tau) - q_1(\tau)|  \left\Vert\frac{1}{n_1(s)}\sum_{i \in I_s(1)}\left[|H_{h_n}(X^s_i)| - \mathbb{E}(|H_{h_n}(X_i)||S_i=s)\right] \right\Vert_\infty \\
	& +  \sup_{\tau \in \Upsilon}|\hat{q}_1(\tau) - q_1(\tau)|  \left\Vert\mathbb{E}(|H_{h_n}(X_i)||S_i=s) \right\Vert_\infty\\
	&=  \sup_{\tau \in \Upsilon}|\hat{q}_1(\tau) - q_1(\tau)|  \left\Vert\frac{1}{n_1(s)}\sum_{i \in I_s(1)}\left[|H_{h_n}(X^s_i)| - \mathbb{E}(|H_{h_n}(X_i)||S_i=s)\right] \right\Vert_\infty  + O_p(n^{-1/2}) \\
	&= O_p(n^{-1/2}),
	\end{align*}
	where the second to last inequality holds because of Assumption \ref{ass:qhat} and $\left\Vert\mathbb{E}(|H_{h_n}(X_i)||S_i=s) \right\Vert_\infty \leq C<\infty$, and the last inequality holds because by a similar argument to the one used in bounding the first term on the RHS of \eqref{eq:np1}, we can show that 
	\begin{align*}
	\left\Vert\frac{1}{n_1(s)}\sum_{i \in I_s(1)}\left[|H_{h_n}(X^s_i)| - \mathbb{E}(|H_{h_n}(X_i)||S_i=s)\right] \right\Vert_\infty =O_p\left( \sqrt{\frac{\log(n)}{n}}\right). 
	\end{align*}
	This concludes the proof. 
\end{proof}

\begin{lem}
	Suppose the Assumptions in Theorem \ref{thm:np} hold and recall $\overline{\ell}$ defined in \eqref{eq:overlineell}. We have $\overline{\ell}/(\sqrt{ h_n \log(n)}/n) \rightarrow \infty$, w.p.a.1. 
	\label{lem:np2} 
\end{lem}
\begin{proof}
	Note that w.p.a.1, 
	\begin{align*}
	\overline{\ell} &=  \inf_{U \in \Re^{h_n} }\frac{\left[\frac{1}{n_a(s)}\sum_{i \in I_a(s)}(H_{h_n}^\top(X_i) U)^2\right]^{3/2} }{\frac{1}{n_a(s)}\sum_{i \in I_a(s)}|H_{h_n}^\top(X_i) U|^3} \\
	&\geq \inf_{U \in \Re^{h_n} }\frac{\left[\frac{1}{n_a(s)}\sum_{i \in I_a(s)}(H_{h_n}^\top(X_i) U)^2\right]^{1/2} }{\sup_{x \in \mathcal{X}}||H_{h_n}(x)||_2 ||U||_2} \geq \frac{\kappa_1^{1/2}}{\zeta(h_n)},
	\end{align*}
	where the last inequality is due to Assumption \ref{ass:np}. Therefore, 
	\begin{align*}
	\overline{\ell}/ (\sqrt{ h_n \log(n)}/n) \geq \sqrt{\frac{\kappa_1 n}{ \zeta^2(h_n) h_n \log(n)}} \rightarrow \infty~w.p.a.1. 
	\end{align*}
\end{proof}

\bibliographystyle{chicago}
\bibliography{BCAR}
\end{document}